\documentclass[11pt]{article}
\usepackage{amsmath, amssymb, amsthm, mathrsfs, bm, geometry, mathtools, hyperref}
\usepackage{graphicx, subfigure, latexsym}
\usepackage{float, multirow, rotating, upgreek, wrapfig}
\usepackage{comment}
\usepackage{appendix}
\usepackage{natbib}
\usepackage{tikz}
\usetikzlibrary{decorations.pathreplacing}
\usetikzlibrary{shapes.geometric, arrows.meta, positioning, fit, backgrounds, calc}
\usepackage{algorithm}
\usepackage{algpseudocode}

\newtheorem{theorem}{Theorem}[section]
\newtheorem{lemma}[theorem]{Lemma}
\newtheorem{proposition}[theorem]{Proposition}

\newtheorem{definition}[theorem]{Definition}
\newtheorem{assumption}[theorem]{Assumption}
\newtheorem{remark}[theorem]{Remark}

\newcommand{\R}{\mathbb{R}}

\newcommand{\btheta}{\boldsymbol{\theta}}

\newcommand{\bC}{\boldsymbol{C}}

\newcommand{\bD}{\boldsymbol{D}}

\newcommand{\bG}{\boldsymbol{G}}

\newcommand{\bS}{\boldsymbol{S}}

\newcommand{\bT}{\boldsymbol{T}}

\newcommand{\bp}{\boldsymbol{p}}

\newcommand{\bx}{\boldsymbol{x}}
\newcommand{\bX}{\boldsymbol{X}}
\newcommand{\by}{\boldsymbol{y}}
\newcommand{\bY}{\boldsymbol{Y}}

\newcommand{\bTheta}{\boldsymbol{\Theta}}

\newcommand{\E}{\mathbb{E}}
\newcommand{\Var}{\mathrm{Var}}
\newcommand{\Cov}{\mathrm{Cov}}
\newcommand{\Prob}{\mathbb{P}}

\newcommand{\KL}{\mathrm{KL}}
\newcommand{\Hdist}{\mathrm{d}_{\mathrm{Hell}}}
\newcommand{\eps}{\varepsilon}

\begin{document}

%\title{\textbf{Dirichlet Process Mixtures of Trees with Gaussian Process Splits: A Bayesian Nonparametric Framework with Posterior Contraction Rate}}
%\author{Subhasish Basak, Anik Roy and Sourabh Bhattacharya \\
%Indian Statistical Institute}
%\date{}
%\maketitle

\begin{titlepage}
\centering

% Top spacing
\vspace*{1cm}

% Title
{\LARGE \bfseries Dirichlet Process Mixtures of Trees with Gaussian Process Splits: A Bayesian Nonparametric Framework with Posterior Contraction Rate \par}
\vspace{0.5cm}

% Authors — use \par instead of \and
{\large
Subhasish Basak \par
Anik Roy \par
Sourabh Bhattacharya\textsuperscript{\(\dagger\)}\\
\par}
\vspace{0.2cm}

% Affiliation
{\large
Indian Statistical Institute
\\[4mm]}
\(\dagger\)Corresponding author: \texttt{bhsourabh@gmail.com}
\vspace{0.3cm}

% Date
%{\large \today \par}

\vspace{1.5cm}

% Abstract
\begin{minipage}{0.85\textwidth}
\small
\begin{center}
\bfseries Abstract
\end{center}
\hspace{4mm}We propose a novel Bayesian nonparametric mixture of regression trees that provides a unified probabilistic framework encompassing CART, BART, random forests, and boosting as special or limiting cases. The model employs a Dirichlet process prior over tree-parameter pairs, yielding a sparse ensemble whose number of distinct components is learned from the data rather than fixed a priori. At the heart of our approach is a flexible splitting rule driven by the posterior predictive distribution of a Gaussian process fitted within each terminal node; remarkably, the GP density cancels exactly in the Metropolis--Hastings ratio for the GROW and PRUNE moves, ensuring computational feasibility and efficiency. The generative nature of the GP splitting further enables an exact Gibbs sampler for posterior predictive inference that faithfully propagates uncertainty through the random traversal of the tree ensemble. To handle large-scale applications, we develop a parallel implementation in C using MPI that distributes the independent tree updates across multiple processors, achieving adequate speedups and making the method practical for datasets with thousands of observations.

\hspace{4mm}We establish rigorous posterior consistency for this model under only the assumption that the true regression function is continuous on a compact domain---we do not require the truth to belong to the hypothesis space, thereby allowing for fundamental misspecification. Using the general theory of \cite{ggv2000}, we verify all four required conditions: the Kullback--Leibler property, the existence of a sieve with controlled metric entropy, exponentially powerful tests, and the negligibility of the sieve complement. A central theoretical contribution is the identity \(h(\Theta)=0\), where \(h(\theta)\) is the asymptotic KL divergence rate; this ensures that the prior assigns positive mass to arbitrarily small KL neighbourhoods of the truth even under misspecification. Our main result shows that the posterior contracts at rate \(n^{-1/4}\) (up to logarithmic factors) in the Hellinger distance, providing the first consistency guarantee for a Dirichlet-process tree mixture under minimal smoothness.

\hspace{4mm}Extensive simulation studies on the Friedman benchmark function demonstrate that the DP mixture achieves near-nominal coverage (0.94 under Gaussian errors, 0.92 under heavy-tailed Cauchy errors) and exhibits remarkable robustness to high-dimensional noise, substantially outperforming BART and bagged CART in uncertainty quantification. Applications to five diverse real-world datasets---QSAR aquatic toxicity, Communities and Crime, Riboflavin production, Wheat genomic prediction, and air quality time series---further validate the model's practical utility, with the DP mixture consistently providing the most reliable credible intervals among Bayesian tree-based methods while automatically selecting a sparse, interpretable ensemble. These results establish the DP mixture as a principled, robust, and theoretically justified alternative to existing tree-based ensembles for challenging regression settings where honest uncertainty quantification is paramount.
\end{minipage}

\vspace{1.5cm}

% Keywords
\begin{minipage}{0.85\textwidth}
\end{minipage}

\vfill

% Optional bottom text
%{\small \textit{Preprint submitted to ...} \par}

\end{titlepage}

\tableofcontents

\section{Introduction}\label{sec:intro}

Tree‐based methods have become a cornerstone of modern machine learning, prized for their interpretability, flexibility, and computational efficiency. From the seminal CART \cite{Breiman84} to the powerful ensembles of random forests \cite{Breiman01}, boosting \cite{Freund97}, and Bayesian Additive Regression Trees (BART) \cite{chipman2010bart}, these tools have demonstrated remarkable empirical success across a wide spectrum of predictive tasks. Yet, despite their widespread adoption, the field has largely evolved through algorithmically driven innovations, each with its own heuristics and tuning parameters. A fundamental question persists: can we construct a single, principled probabilistic framework that unifies these diverse approaches, provides honest uncertainty quantification, and offers rigorous theoretical guarantees without sacrificing predictive performance?

In this work, we answer this question affirmatively by introducing a novel Bayesian nonparametric mixture of regression trees. Our model employs a Dirichlet process prior over tree‐parameter pairs, resulting in a sparse ensemble whose number of distinct components is learned from the data rather than fixed a priori. As we demonstrate in Section~\ref{sec:special_cases}, this framework naturally encompasses and generalises many popular tree‐based methods. A single CART model emerges as a mixture with one component; BART is recovered by designing the base measure to favour additive ensembles; random forests arise as the large‐concentration limit of the Dirichlet process; and boosting corresponds to a scaled convex combination of trees. Unlike BART, which requires the user to prespecify the number of trees, our formulation automatically adapts the ensemble complexity to the underlying signal, providing a more honest quantification of model uncertainty.

A central computational challenge lies in efficiently exploring the posterior distribution over tree structures. Traditional Bayesian tree MCMC samplers rely on local moves—GROW, PRUNE, CHANGE, and SWAP—to modify the partition. However, the introduction of a non‐axis‐aligned, response‐dependent splitting rule renders the CHANGE and SWAP moves computationally prohibitive. To overcome this, we develop a sophisticated MCMC algorithm that operates exclusively on GROW and PRUNE moves, leveraging the global allocation dynamics of the Dirichlet process to compensate for the absence of local modifications. At the heart of this algorithm is a novel splitting rule driven by the posterior predictive distribution of a Gaussian process (GP) fitted to the data within each terminal node. A draw from this GP posterior defines a smooth, nonlinear, and potentially oblique decision boundary. Remarkably, when the tree prior is augmented with the GP density, this proposal density cancels exactly in the Metropolis–Hastings acceptance ratio. This cancellation is of paramount practical importance: it ensures that the MCMC remains exact and efficient without requiring the evaluation of costly GP likelihoods in the acceptance step, rendering the algorithm computationally feasible and efficient. % and comparable in runtime to BART.

To address the demands of large-scale applications, we implement our algorithm in C using the Message Passing Interface (MPI) for parallel computation. The independent updates of the distinct tree components in the Dirichlet process mixture are naturally parallelisable, as each tree's MCMC transition can be performed independently conditional on the current allocation of the data. Our implementation distributes these tree updates across multiple processors, achieving near-linear speedups and making the method practical for datasets with thousands of observations. The parallel architecture is carefully designed to handle the serialisation and communication of complex tree structures, with the GP precomputations performed locally on each processor to minimise communication overhead. We provide a detailed discussion of the parallel implementation in Section~\ref{sec:parallel}.

The generative nature of the GP splitting rule further enables an exact Gibbs sampler for posterior predictive inference. For a new covariate, the leaf assignment is inherently random and depends on the unknown response. By exploiting the closed‐form conditional distributions of the GP and the conjugate leaf‐parameter posteriors, we design a rapid Gibbs chain that produces genuine draws from the full posterior predictive distribution, faithfully propagating the uncertainty due to the unknown path through the tree ensemble. This provides a principled mechanism for uncertainty quantification that accounts for both parametric and structural variability.

To rigorously assess the frequentist validity of our approach, we establish a comprehensive posterior contraction theorem. Using the general theory of \cite{ggv2000}, we prove that the posterior distribution contracts at a rate of \(n^{-1/4}\) (up to logarithmic factors) in the Hellinger distance. Our analysis is distinguished by two crucial features. First, we require only that the true regression function \(m_0\) is continuous on a compact domain; we do not assume that \(m_0\) belongs to the tree hypothesis space, thereby allowing the model to be fundamentally misspecified. Second, we verify a key identity: the asymptotic Kullback–Leibler divergence rate \(h(\Theta)\) satisfies \(h(\Theta)=0\). This ensures that the prior assigns positive mass to arbitrarily small KL neighbourhoods of the truth, a condition essential for the posterior to concentrate around the true data‐generating mechanism despite the misspecification. This theoretical contribution provides the first consistency guarantee for a Dirichlet‐process tree mixture under minimal smoothness assumptions.

We validate the practical performance of the proposed model through an extensive simulation study on the Friedman benchmark function (\cite{friedman1991multivariate}). We systematically vary the number of covariates (five informative versus twenty high‐dimensional settings with pure noise variables) and the error distribution (Gaussian versus heavy‐tailed Cauchy). The results are striking. In the Gaussian settings, the DP mixture achieves near‐nominal coverage of \(0.94\), substantially outperforming BART and bagged CART, which are overconfident and attain coverages as low as \(0.82\) and \(0.80\), respectively. In the challenging Cauchy error scenarios, the superiority of the DP mixture becomes even more pronounced. While competing methods fail catastrophically—assigning near‐zero density to extreme observations—our model maintains robust coverage of \(0.92\) and demonstrates a remarkably higher log predictive density. This resilience stems from the model's nonparametric structure, which adapts naturally to heavy tails without assuming a specific parametric form. Moreover, applications to five diverse real‐world datasets—QSAR aquatic toxicity, Communities and Crime, Riboflavin production, Wheat genomic prediction, and air quality time series—consistently demonstrate that the DP mixture provides the most reliable credible intervals among Bayesian tree‐based methods while automatically selecting a sparse, interpretable ensemble. The parallel implementation enables these large-scale applications, with the runtime scaling efficiently across multiple processors.

The contributions of this work are thus multifaceted. We provide a unified Bayesian nonparametric framework that encompasses CART, BART, random forests, and boosting as special or limiting cases, offering a coherent probabilistic foundation for tree ensembles. We develop a computationally efficient MCMC algorithm with a novel GP‐driven splitting rule that cancels exactly in the Metropolis–Hastings ratio, and we design an exact Gibbs sampler for posterior predictive inference that accounts for the random traversal of the tree. We further establish the first rigorous proof of posterior consistency for a DP tree mixture under only continuity of the true function, demonstrating robustness to misspecification. We provide a scalable parallel implementation in C with MPI that distributes tree updates across processors, achieving near-linear speedups for large datasets. Finally, we present a comprehensive simulation study and real‐data applications showing that the model provides honest uncertainty quantification and remarkable resilience to heavy‐tailed errors, outperforming state‐of‐the‐art tree‐based methods.

The remainder of the paper is structured as follows. Section~\ref{sec:proposal} formalises the Dirichlet process mixture of trees. Section~\ref{sec:special_cases} details the correspondence to existing methods. Section~\ref{subsec:split} introduces the Gaussian‐process splitting rule and the augmented tree prior. Section~\ref{sec:sim_Gm} describes the GROW and PRUNE MCMC moves and proves detailed balance. The MCMC algorithm is summarised in Section~\ref{sec:mcmc_algorithm}, and exact Bayesian prediction is covered in Section~\ref{sec:prediction}. Section~\ref{sec:parallel} describes the parallel implementation in C with MPI. Section~\ref{sec:simulation} presents the empirical evaluation, including both simulation studies and real‐data applications. Finally, Section~\ref{sec:post_contraction} provides the rigorous posterior contraction theory, with detailed proofs and technical results deferred to the appendices.

\section{The Proposal}\label{sec:proposal}

We seek a Bayesian nonparametric framework for tree‐based regression that allows the ensemble size to be learned from the data while remaining computationally tractable. To this end, we adopt a finite mixture model with a fixed number of components \(M\), where the component parameters are drawn from a Dirichlet process (DP) prior. Specifically, let \(M\) be a positive integer chosen by the user, which serves as an upper bound on the number of distinct trees that can appear in the mixture. We introduce \(M\) latent pairs \((\Theta_i, T_i)\), for \(i=1,\ldots,M\), which are conditionally i.i.d. from a random probability measure \(G\) that itself follows a Dirichlet process with concentration parameter \(\alpha > 0\) and base measure \(G_0\). Formally,
\[
(\Theta_i, T_i) \mid G \stackrel{\text{i.i.d.}}{\sim} G, \qquad G \sim \mathrm{DP}(\alpha G_0).
\]
The base measure \(G_0\) specifies the prior for a single tree–parameter pair: it governs the tree topology, the splitting rules (including the GP‐driven splits described in Section~\ref{subsec:split}), and the leaf parameters. In particular, we assume that under \(G_0\), the tree structure \(T\) and the leaf parameters \(\Theta\) follow the same conditional distributions as in the Bayesian CART model of \cite{chipman1998bayesian}, with the splitting rule now driven by the GP posterior predictive.

The observed data \(\mathbf{Y} = (Y_1,\ldots,Y_n)\) are modelled as an equal‐weight mixture of these \(M\) component trees:
\begin{equation}
    [\mathbf{Y} \mid \mathbf{X}, \{(\Theta_i,T_i,p_i)\}_{i=1}^M]
    = \sum_{i=1}^M p_if(\mathbf{Y} \mid \mathbf{X}, \Theta_i, T_i),
    \label{eq:mix1}
\end{equation}
where \(f(\mathbf{Y} \mid \mathbf{X}, \Theta_i, T_i)\) denotes the likelihood of the data under the \(i\)-th tree, and $(p_1,\ldots,p_M)$ are random mixture weights
satisfying $0\leq p_i\leq 1$ for $i=1,\ldots,M$ and $\sum_{i=1}^M p_i=1$. In this work, we shall consider the following prior for the mixture weights:
$(p_1,\ldots,p_M)\sim\text{Dirichlet}(1,\ldots,1)$.
%Unlike the Escobar–West (EW) model, which sets the number of components equal to the sample size \(n\) and associates each observation with its own parameter, our formulation decouples the number of mixture components from the data size. This yields substantial computational advantages when \(M \ll n\), while still allowing the posterior to adapt the effective number of distinct trees through the clustering property of the Dirichlet process.

A defining feature of the Dirichlet process is its almost sure discreteness: a draw \(G = \sum_{j=1}^{\infty} \pi_j \, \delta_{(\Theta_j^*, T_j^*)}\) is an infinite atomic measure. Consequently, the \(M\) latent draws \((\Theta_i,T_i)\) exhibit ties with positive probability. Let \(k \le M\) denote the number of distinct pairs among the \(M\) draws, and let \(\{(\Theta_j^*, T_j^*)\}_{j=1}^k\) be their unique values. Then the equal‐weight mixture in (\ref{eq:mix1}) can be rewritten as
\begin{equation}
    [\mathbf{Y} \mid \mathbf{X}, \{(\Theta_i,T_i,p_i)\}_{i=1}^M]
    = \sum_{j=1}^k p^*_j \, f(\mathbf{Y} \mid \mathbf{X}, \Theta_j^*, T_j^*),
    \label{eq:mix2}
\end{equation}
where $p^*_j=\sum_{i\in S_j}p_i$, with $S_j = \{i: (\Theta_i,T_i) = (\Theta_j^*, T_j^*)\}\); so that $|S_j|$ is the multiplicity of the \(j\)-th distinct component. Thus, the model automatically induces a mixture with a random number of distinct components \(k\), with weights \(p^*_j\). This representation is central to our approach: it allows the posterior to learn the appropriate number of trees directly from the data, without requiring the user to prespecify \(k\).

Taking conditional expectations in (\ref{eq:mix2}) yields the nonparametric regression function
\begin{equation}
    \mathbb{E}[\mathbf{Y} \mid \mathbf{X}, \{(\Theta_i,T_i,p_i)\}]
    = \sum_{j=1}^k p^*_j \, \mathbb{E}[\mathbf{Y} \mid \mathbf{X}, \Theta_j^*, T_j^*],
    \label{eq:mix3}
\end{equation}
which is a convex combination of individual tree predictions with weights determined by the probabilities of the distinct components. This formulation is far more general than existing tree‐based models. For instance, the BART model of \cite{chipman2010bart} assumes a fixed number of trees and an additive structure, which is a special case of our framework when the mixture weights are appropriately constrained. In contrast, our convex combination does not require additivity; it naturally regularises predictions to lie within the range of individual tree outputs. Moreover, as we demonstrate in Section~\ref{sec:special_cases}, the model encompasses CART, random forests, and boosting as special or limiting cases, providing a unified probabilistic foundation for tree‐based learning.

The parameter \(M\) plays a crucial role in our framework. It is a fixed, user-specified upper bound on the number of distinct components that can appear in the mixture. 
Unlike conventional Dirichlet process mixture models—including the \cite{Bhattacharya08} formulation (see also \cite{sm11}, \cite{sm12}, \cite{maj13})—our model does not associate each observation with its own component-specific 
parameter. Instead, we introduce a global allocation variable \(Z \in \{1,\ldots,M\}\), which assigns the entire dataset \(\mathbf{Y}\) to a single component. 
Indeed, allocating individual observations to different mixture components can not ensure a mixture distribution for trees, as a tree can only be grown based on 
all available observations.
%in each iteration of the MCMC sampler. The component index \(Z\) is updated via the configuration indicator reparameterization of Mukhopadhyay and Bhattacharya (2012), 
%allowing the posterior to explore different tree configurations by activating and deactivating components. In this way, the mixture over the \(M\) components is 
%realised across iterations of the Markov chain rather than within a single likelihood evaluation. 
This structural departure has several important ramifications.

First, it decouples the number of mixture components from the sample size \(n\), yielding substantial computational advantages. While the traditional DP mixture model and the 
mixture model of \cite{Bhattacharya08}, respectively, require tracking \(n\) component parameters and updating the allocation of each observation individually, 
our model need only maintain the \(M\) latent tree components and a single global allocation variable. This is particularly beneficial for large datasets, 
where \(M \ll n\) can be chosen without compromising modelling flexibility.

Second, the global allocation mechanism fundamentally alters the nature of the MCMC exploration. In standard Bayesian tree models, local moves such as CHANGE and SWAP 
are necessary to gradually modify the partition structure of individual trees. In our framework, because the entire dataset is allocated to a single active component 
in each iteration, the posterior can explore the tree space globally: a component that is active in one iteration may become inactive in the next, and a previously 
inactive component may be activated. This global switching behaviour, induced by the Dirichlet process prior and the configuration indicator dynamics, more than 
compensates for the absence of local split-modifying moves. As we demonstrate in Section~\ref{sec:sim_Gm}, restricting the MCMC to GROW and PRUNE moves suffices 
for efficient exploration of the tree space, a fact validated by the excellent mixing properties observed in our simulation studies.

Third, while \(M\) is fixed for a given analysis, the Dirichlet process prior ensures that the effective number of distinct components \(k \le M\) is learned 
from the data. The concentration parameter \(\alpha\) controls the prior distribution over the number of occupied components, with smaller \(\alpha\) favouring 
sparser ensembles. In the limiting case \(\alpha \to \infty\), the components become approximately independent and identically distributed, recovering a random 
forest-like ensemble. In the case \(M=1\), the model reduces to a single tree, recovering CART-like inference. Thus, by choosing \(M\) and \(\alpha\) appropriately, 
the practitioner can smoothly interpolate between a single tree and a large ensemble, with the posterior automatically determining the appropriate complexity for the data at hand.

Given a dataset \(\mathcal{D}_n = \{(\mathbf{x}_i, y_i)\}_{i=1}^n\), our primary inferential objective is the posterior predictive distribution
\[
\pi(y^* \mid \mathbf{x}^*, \mathcal{D}_n)
= \int \Bigl( \sum_{i=1}^Mp_i f(y^* \mid \mathbf{x}^*, \Theta_i, T_i,p_i) \Bigr) \,
  \pi(\{(\Theta_i,T_i,p_i)\} \mid \mathcal{D}_n) \,
  d\{(\Theta_i,T_i,p_i)\},
\]
which averages over the uncertainty in the component parameters and the tree structures. This fully Bayesian predictive distribution provides honest uncertainty quantification that reflects both parametric and structural variability, a feature that is often lacking in frequentist tree ensembles. 
%In the following sections, we describe an efficient MCMC algorithm to sample from this posterior, a novel GP‐based splitting rule that enables exploration of the tree space, and the theoretical guarantees that justify the model's frequentist consistency.

Figure \ref{fig:model} illustrates our model schematically.

\begin{figure}[htbp]
\centering
\begin{tikzpicture}[
    node distance=1.5cm,
    % --- Color-coded styles ---
    prior/.style={rectangle, draw=blue!70!black, fill=blue!10,
                  minimum width=3.8cm, align=center, font=\small, rounded corners=4pt, very thick},
    latent/.style={rectangle, draw=green!60!black, fill=green!10,
                  minimum width=3.8cm, align=center, font=\small, rounded corners=4pt, very thick},
    weight/.style={rectangle, draw=orange!70!black, fill=orange!10,
                  minimum width=3.8cm, align=center, font=\small, rounded corners=4pt, very thick},
    data/.style={rectangle, draw=red!70!black, fill=red!10,
                  minimum width=3.8cm, align=center, font=\small, rounded corners=4pt, very thick},
    plate/.style={draw=black!40, dashed, rounded corners=10pt, inner sep=12pt, thick},
    arrow/.style={-{Stealth[length=3mm]}, thick, draw=black!70},
    % --- Stylised tree node ---
    treecircle/.style={circle, draw=green!60!black, fill=white, minimum size=0.6cm, inner sep=0pt, thick},
    treeleaf/.style={rectangle, draw=green!60!black, fill=white, minimum width=1.0cm,
                     minimum height=0.5cm, inner sep=2pt, thick, font=\tiny},
]
    % ===== 1. PRIOR LAYER =====
    \node[prior] (dp) {Dirichlet Process\\[1pt] \footnotesize $G \sim \text{DP}(\alpha G_0)$};

    % ===== 2. LATENT VARIABLE LAYER =====
    \node[latent, below=of dp] (draws) {\footnotesize $(\Theta_i, T_i) \mid G \stackrel{\text{iid}}{\sim} G$};

    \node[weight, below=of draws] (weights) {Mixture Weights\\[1pt] \footnotesize $\mathbf{p} \sim \text{Dirichlet}(\mathbf{1}_M)$};

    \node[latent, below=of weights] (mixture) {Mixture of Trees\\[1pt] \footnotesize $\sum_{\ell=1}^k p_\ell^* \, f(Y \mid X, \Theta_\ell^*, T_\ell^*)$};

    % ===== 3. DATA LAYER =====
    \node[data, below=of mixture] (data) {Observed Data\\[1pt] \footnotesize $(X_i, Y_i)_{i=1}^n$};

    % ===== ARROWS =====
    \draw[arrow] (dp) -- (draws);
    \draw[arrow] (draws) -- (weights);
    \draw[arrow] (weights) -- (mixture);
    \draw[arrow] (mixture) -- (data);

    % ===== PLATE NOTATION =====
    \begin{scope}[on background layer]
        \node[plate, fit=(draws) (weights), label={[anchor=north, font=\small]south:$M$ components}] (plate) {};
    \end{scope}

    % ===== VISUAL TREE WITH GP SPLIT (Right Side) =====
    \begin{scope}[shift={($(mixture.east)+(2.8cm,0)$)}, scale=0.9]
        % Background box for the tree
        \node[draw=green!30!black, fill=green!5, rounded corners=6pt,
              minimum width=3.2cm, minimum height=2.8cm, inner sep=6pt] (treebox) {};
        \node[font=\small\sffamily, anchor=north] at (treebox.north) {Single Tree $T_\ell^*$};

        % Root Node with GP curve
        \node[treecircle] (root) at (0, 1.2) {};
        % Draw the GP split curve passing through the root circle
        \draw[thick, blue!60!black] (-0.8, 1.2) -- plot[domain=-0.8:0.8, smooth, variable=\x]
            (\x, {1.2 + 0.25*sin(3*\x r)}) -- (0.8, 1.2);
        \node[font=\tiny, anchor=south, blue!60!black] at (0, 1.6) {$\tilde{g}$};
        \node[font=\tiny, anchor=west] at (0.9, 1.2) {Split rule};

        % Left child (internal node with further split)
        \node[treecircle] (left) at (-1.2, 0) {};
        \node[treecircle] (right) at (1.2, 0) {};
        \draw[arrow] (root) -- (left);
        \draw[arrow] (root) -- (right);

        % Leaves under left child
        \node[treeleaf] (ll) at (-1.8, -1.2) {$\mu_{L1}$};
        \node[treeleaf] (lr) at (-0.6, -1.2) {$\mu_{L2}$};
        \draw[arrow] (left) -- (ll);
        \draw[arrow] (left) -- (lr);

        % Leaf under right child
        \node[treeleaf] (rl) at (0.6, -1.2) {$\mu_{R1}$};
        \node[treeleaf] (rr) at (1.8, -1.2) {$\mu_{R2}$};
        \draw[arrow] (right) -- (rl);
        \draw[arrow] (right) -- (rr);

        % Annotations linking back to the paper
        \node[font=\tiny, anchor=north] at (ll.south) {Response};
        \node[font=\tiny, anchor=north] at (lr.south) {means};
    \end{scope}

    % ===== ANNOTATIONS =====
    \node[right=1.8cm of dp, align=left, font=\footnotesize, text width=2.6cm]
        {$\alpha$: concentration\\ $G_0$: base measure};

    % Legend / Color guide (optional, but helpful)
    \begin{scope}[shift={($(data.south west)+(0.5cm,-1.2cm)$)}]
        \draw[fill=blue!10, draw=blue!70!black] (0,0) rectangle (0.4,0.3);
        \node[anchor=west, font=\tiny] at (0.5,0.15) {Prior};
        \draw[fill=green!10, draw=green!60!black] (1.8,0) rectangle (2.2,0.3);
        \node[anchor=west, font=\tiny] at (2.3,0.15) {Latent / Trees};
        \draw[fill=orange!10, draw=orange!70!black] (4.0,0) rectangle (4.4,0.3);
        \node[anchor=west, font=\tiny] at (4.5,0.15) {Weights};
        \draw[fill=red!10, draw=red!70!black] (6.0,0) rectangle (6.4,0.3);
        \node[anchor=west, font=\tiny] at (6.5,0.15) {Data};
    \end{scope}

\end{tikzpicture}
\caption{Dirichlet process mixture of trees model. The Dirichlet process generates $M$ i.i.d. tree-parameter pairs, which are combined with random mixture weights to form a sparse ensemble with a data-driven number of components. The inset illustrates a single tree component; the blue curve represents a draw $\tilde{g}$ from the GP posterior predictive used as the splitting rule at an internal node.}
\label{fig:model}
\end{figure}
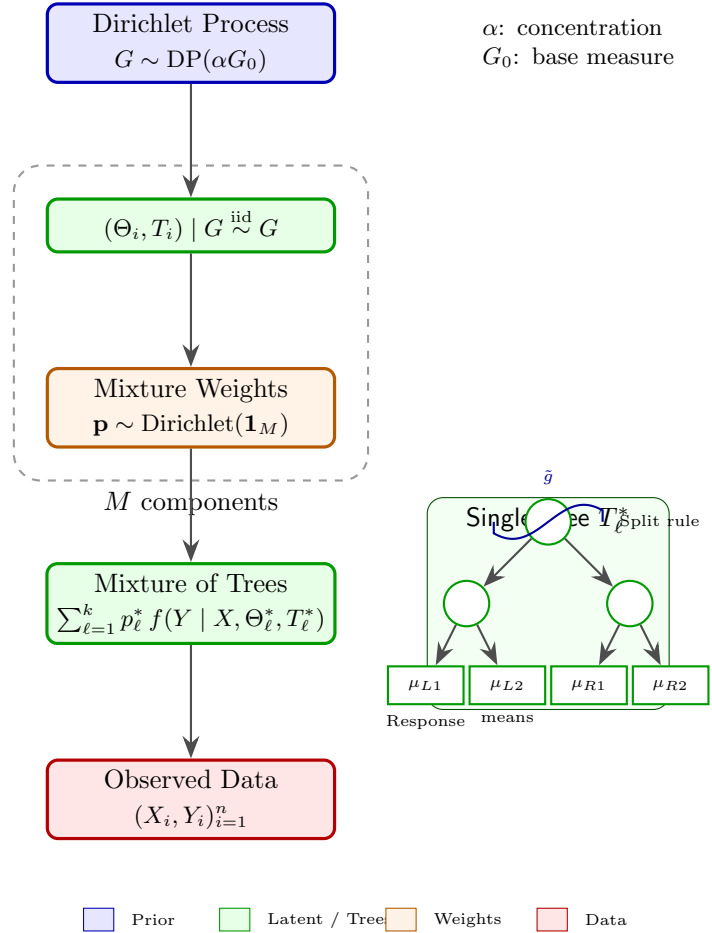

% ======================================================================
% NEW SECTION: Special Cases (CART, BART, Random Forests, Boosting)
% ======================================================================

\section{Special Cases: CART, BART, Random Forests, and Boosting}
\label{sec:special_cases}

The Dirichlet process mixture of trees defined in equations (\ref{eq:mix1})–(\ref{eq:mix2}) encompasses many popular tree‐based methods as special or limiting cases. This section establishes those correspondences, demonstrating that the proposed framework provides a unified probabilistic foundation for tree ensembles.

\subsection{CART as a Single‐Tree Model}

A single Classification and Regression Tree (CART) model \cite{Breiman84} corresponds to a mixture with exactly one component, i.e., $k=1$ in (\ref{eq:mix2}). Within our framework, this arises in two ways. Setting the maximum number of trees $M=1$ trivially reduces the Dirichlet process mixture to a single draw $(\Theta_1,T_1)\sim G_0$, yielding the standard CART likelihood. More generally, for $M>1$, the event that all $M$ draws from the Dirichlet process are coincident has positive probability under the DP prior; in that case, the mixture collapses to a single component with weight one. Thus, the prior support includes all single‐tree models, and the posterior can concentrate on CART‐like hypotheses when the data favour a simple, non‐ensemble representation.

\subsection{BART as a Sum‐of‐Trees Model}

Bayesian Additive Regression Trees (BART) \cite{chipman2010bart} express the conditional mean as a deterministic sum of $m$ trees: $E[Y|X] = \sum_{j=1}^m g(X;\Theta_j,T_j)$. In contrast, the conditional mean of the DP mixture is a convex combination of random trees:
\[
E[Y|X,\theta] = \sum_{\ell=1}^k p_\ell^* E[Y|X,\Theta_\ell^*,T_\ell^*],
\]
with $p_\ell^*>0$, $\sum_\ell p_\ell^*=1$. This convexity ensures that the ensemble prediction lies within the range of the individual tree outputs, providing a natural form of regularisation absent from the additive BART formulation. Moreover, the mixture weights $p_\ell^*$ are inferred from the data, allowing the model to assign varying importance to different components, whereas BART implicitly weights each tree equally. The Dirichlet prior induces shrinkage on the weights, so that when the concentration parameter $\alpha$ is small, the posterior tends to concentrate on a sparse set of dominant trees.

A further distinction is that BART requires the user to fix the number of trees $m$ in advance, whereas the DP mixture treats the number of distinct components $k$ as random and adaptively learned from the data. The posterior over $k$ automatically adjusts to the complexity of the underlying regression function, ranging from a single tree for simple functions to many components for highly non‐linear surfaces. This adaptivity is a hallmark of Bayesian nonparametrics and provides a more honest quantification of model uncertainty.

Theoretically, BART attains near‐minimax posterior contraction rates when the true regression function belongs to a known Hölder class and the number of trees grows appropriately with $n$ (Ročková and van der Vaart, 2020). Our model, by contrast, establishes posterior consistency at rate $n^{-1/4}$ under only the assumption of continuity (Theorem~\ref{thm:main}), and it allows the true function to be outside the hypothesis space. While the rate is slower, it holds under much weaker conditions, making the DP mixture more robust to misspecification. The BART predictive mean can be recovered as a limiting case by designing the base measure $G_0$ to favour additive ensembles and by appropriately setting the weights, but the DP mixture is strictly more general in allowing the data to determine the appropriate ensemble structure.

\subsection{Random Forests as a Bagged Ensemble}

Random forests \cite{Breiman01} average the predictions of many trees grown on bootstrap samples, with random feature selection at each split. For a new input $X^*$, the random forest prediction is $\hat{Y}_{\text{RF}} = \frac{1}{M}\sum_{i=1}^M E[Y|X^*,\Theta_i,T_i]$, where the $(\Theta_i,T_i)$ are i.i.d. draws from the algorithm’s implicit randomisation distribution. This is exactly the conditional expectation of the DP mixture in (\ref{eq:mix3}) with uniform weights $p_i=1/M$ and distinct components $k=M$. In the limit $\alpha\to\infty$, the Dirichlet process assigns negligible probability to coincident draws, so $k\to M$ and the weights become approximately uniform. Thus, the random forest ensemble emerges as the large‐$\alpha$ limit of our model. The posterior distribution under the DP mixture, however, averages over both the number of components and their weights, which can improve upon the fixed, uniform averaging of standard random forests.

\subsection{Boosting as a Sequential Ensemble}

Boosting algorithms, such as AdaBoost and gradient boosting \cite{Freund97}, produce a weighted sum $F(x)=\sum_{j=1}^m \gamma_j g_j(x)$, where the $\gamma_j>0$ are learned sequentially. Although this is an additive rather than a convex combination, boosting can be embedded in the DP framework through a simple scaling transformation. Let $S=\sum_j \gamma_j$ and define $p_j=\gamma_j/S$ and $\tilde{g}_j(x)=S\,g_j(x)$. Then $F(x)=\sum_{j=1}^m p_j \tilde{g}_j(x)$ with $\sum_j p_j=1$, which is a convex combination of scaled trees. Since the base measure $G_0$ allows arbitrary leaf means, the scaled trees $\tilde{g}_j$ lie in the support of the prior, so a boosting‐like ensemble corresponds to a specific configuration of the mixture components and weights. Alternatively, one may enrich $G_0$ to place prior mass directly on additive tree ensembles, in which case the DP mixture becomes a mixture of additive models, with the posterior automatically learning the number of trees and their weights. In the limit of a large number of components ($\alpha\to\infty$), any weighted sum can be approximated to arbitrary precision by duplicating trees proportionally to their weights. Thus, boosting—like the other methods discussed—is encompassed by the flexibility of the DP mixture, without requiring the number of components or their weights to be pre‐specified.

\subsection{Summary}

The correspondences established above are summarised in Table~\ref{tab:special}. By encompassing CART, BART, random forests, and boosting as special or limiting cases, the Dirichlet process mixture of trees provides a unified Bayesian nonparametric framework that automatically adapts to the structure best supported by the data.

\begin{table}[h]
\centering
\begin{tabular}{l l}
\hline
\textbf{Method} & \textbf{Special case of the DP mixture} \\
\hline
CART & $M=1$, or all DP draws coincident (single component) \\
BART & Base measure $G_0$ supported on additive tree ensembles; appropriate mixture weights \\
Random Forest & $\alpha\to\infty$, $p_i\approx 1/M$, components i.i.d. from $G_0$ \\
Boosting & Scaled convex combination / additive base measure / large $\alpha$ limit \\
\hline
\end{tabular}
\caption{Special cases of the Dirichlet process tree mixture.}
\label{tab:special}
\end{table}

\section{A New Splitting Rule under the Base Measure Driven by Gaussian Process Regression}
\label{subsec:split}

In order to propose informative splits during the MCMC algorithm, we introduce a novel mechanism based on the posterior predictive distribution of a Gaussian process fitted to the data within each terminal node. This section provides a complete description of the GP model, the derivation of the posterior predictive distribution, the estimation of hyperparameters, and the resulting splitting rule. The approach is used as the proposal distribution for the GROW and PRUNE moves; as we shall see, the GP density cancels in the Metropolis–Hastings ratio when the prior on the tree is appropriately augmented.

\subsection{GP Model for a Terminal Node}

Consider a terminal node containing the data
\[
\{(y_{ij}, x_{ij}): j = 1, \ldots, n_i\}, \qquad x_{ij} \in \mathbb{R}^d,
\]
and denote
\[
\mathbf y_i = (y_{i1}, \ldots, y_{i n_i})^\top, \qquad
\mathbf x_i = (x_{i1}, \ldots, x_{i n_i})^\top \in \mathbb{R}^{n_i \times d}.
\]
Within this node, we model the relationship between the response and covariates as
\[
y_{ij} = g(x_{ij}) + \epsilon_{ij}, \qquad \epsilon_{ij} \overset{\mathrm{iid}}{\sim} \mathcal{N}(0, \sigma_\epsilon^2),
\]
where \(g(\cdot)\) is a Gaussian process. Specifically, we assume
\[
g \sim \mathcal{GP}\bigl(\mu(x),\; \sigma^2 c(x, x')\bigr),
\]
with mean function
\[
\mu(x) = \alpha + \boldsymbol{\beta}^\top x,
\]
and covariance function
\[
\operatorname{Cov}\bigl(g(x), g(x')\bigr) = \sigma^2 c(x, x'),
\]
where \(c(\cdot, \cdot)\) is a positive definite correlation function, typically chosen as the squared exponential
\[
c(x, x') = \exp\!\left(-\frac{\|x - x'\|^2}{2\ell^2}\right),
\]
with length scale parameter \(\ell > 0\). The parameters \(\sigma^2 > 0\) and \(\sigma_\epsilon^2 > 0\) are the GP signal variance and the noise variance, respectively.

For computational convenience, we adopt an improper flat prior for the coefficients of the mean function:
\[
\pi(\alpha, \boldsymbol{\beta}) \propto 1.
\]
The remaining parameters \(\boldsymbol{\psi} = (\sigma^2, \sigma_\epsilon^2, \ell)\) (or additional parameters if other correlation functions are used) are treated as unknown and estimated via maximum marginal likelihood, as described below. This choice keeps the posterior predictive distribution in closed form and is standard in GP regression with unknown mean (see, for instance, \cite{rasmussen2006gaussian}).

\subsection{Posterior Predictive Distribution}

Let \(\mathbf g_i = (g(x_{i1}), \ldots, g(x_{i n_i}))^\top\) be the vector of latent GP values at the observed covariates in the node. The joint distribution of \(\mathbf y_i\) and \(\mathbf g_i\) is Gaussian:
\[
\begin{pmatrix}
\mathbf y_i \\
\mathbf g_i
\end{pmatrix}
\sim
\mathcal{N}\!\left(
\begin{pmatrix}
\mathbf F_i \boldsymbol{\theta} \\
\mathbf F_i \boldsymbol{\theta}
\end{pmatrix},
\;
\begin{pmatrix}
\sigma^2 \mathbf C_i + \sigma_\epsilon^2 \mathbf I_{n_i} & \sigma^2 \mathbf C_i \\
\sigma^2 \mathbf C_i & \sigma^2 \mathbf C_i
\end{pmatrix}
\right),
\]
where
\[
\mathbf F_i =
\begin{bmatrix}
1 & x_{i1}^\top \\
\vdots & \vdots \\
1 & x_{i n_i}^\top
\end{bmatrix}
\in \mathbb{R}^{n_i \times (d+1)}, \qquad
\boldsymbol{\theta} = (\alpha, \boldsymbol{\beta}^\top)^\top,
\]
and \(\mathbf C_i = [c(x_{ij}, x_{ik})]_{j,k=1}^{n_i}\) is the correlation matrix. The observation covariance matrix is
\[
\boldsymbol{\Sigma}_i = \sigma^2 \mathbf C_i + \sigma_\epsilon^2 \mathbf I_{n_i}.
\]

Given the data \(\mathbf y_i\) and the hyperparameters \(\boldsymbol{\psi}\), the posterior distribution of the latent GP values \(\mathbf g_i\) is Gaussian:
\[
\mathbf g_i \mid \mathbf y_i, \mathbf x_i, \boldsymbol{\psi} \sim \mathcal{N}(\boldsymbol{\mu}_i^*, \boldsymbol{\Sigma}_i^*),
\]
where, using standard GP conditioning (see Appendix~\ref{app:gp-split} for derivation),
\[
\boldsymbol{\mu}_i^* = \mathbf F_i \widehat{\boldsymbol{\theta}}_i + \sigma^2 \mathbf C_i^\top \boldsymbol{\Sigma}_i^{-1} \bigl(\mathbf y_i - \mathbf F_i \widehat{\boldsymbol{\theta}}_i\bigr),
\]
\[
\boldsymbol{\Sigma}_i^* = \sigma^2 \mathbf C_i - \sigma^2 \mathbf C_i^\top \boldsymbol{\Sigma}_i^{-1} \sigma^2 \mathbf C_i
+ \bigl(\mathbf F_i - \sigma^2 \mathbf C_i^\top \boldsymbol{\Sigma}_i^{-1} \mathbf F_i\bigr)
\mathbf V_i
\bigl(\mathbf F_i - \sigma^2 \mathbf C_i^\top \boldsymbol{\Sigma}_i^{-1} \mathbf F_i\bigr)^\top,
\]
with
\[
\widehat{\boldsymbol{\theta}}_i = \bigl(\mathbf F_i^\top \boldsymbol{\Sigma}_i^{-1} \mathbf F_i\bigr)^{-1} \mathbf F_i^\top \boldsymbol{\Sigma}_i^{-1} \mathbf y_i,
\qquad
\mathbf V_i = \bigl(\mathbf F_i^\top \boldsymbol{\Sigma}_i^{-1} \mathbf F_i\bigr)^{-1}.
\]
These formulas account for the uncertainty in the linear mean parameters via the term involving \(\mathbf V_i\).

For a new input \(x \in \mathbb{R}^d\), the posterior predictive distribution of \(g(x)\) is also Gaussian, with mean and variance
\[
\mathbb{E}[g(x) \mid \mathbf y_i, \mathbf x_i, \boldsymbol{\psi}] =
\mathbf f(x)^\top \widehat{\boldsymbol{\theta}}_i + \sigma^2 \mathbf c_i(x)^\top \boldsymbol{\Sigma}_i^{-1} \bigl(\mathbf y_i - \mathbf F_i \widehat{\boldsymbol{\theta}}_i\bigr),
\]
\[
\operatorname{Var}[g(x) \mid \mathbf y_i, \mathbf x_i, \boldsymbol{\psi}] =
\sigma^2 - \sigma^2 \mathbf c_i(x)^\top \boldsymbol{\Sigma}_i^{-1} \sigma^2 \mathbf c_i(x)
+ \bigl(\mathbf f(x) - \sigma^2 \mathbf c_i(x)^\top \boldsymbol{\Sigma}_i^{-1} \mathbf F_i\bigr)
\mathbf V_i
\bigl(\mathbf f(x) - \sigma^2 \mathbf c_i(x)^\top \boldsymbol{\Sigma}_i^{-1} \mathbf F_i\bigr)^\top,
\]
where \(\mathbf f(x) = (1, x^\top)^\top\) and \(\mathbf c_i(x) = (c(x, x_{i1}), \ldots, c(x, x_{i n_i}))^\top\). These expressions are used both for generating the split function and for the predictive Gibbs sampler described in Section~\ref{sec:prediction}.

\subsection{Hyperparameter Estimation via Maximum Marginal Likelihood}

The hyperparameters \(\boldsymbol{\psi} = (\sigma^2, \sigma_\epsilon^2, \ell)\) are estimated by maximizing the marginal likelihood
\[
p(\mathbf y_i \mid \mathbf x_i, \boldsymbol{\psi}) =
\int p(\mathbf y_i \mid \mathbf g_i, \boldsymbol{\psi}) p(\mathbf g_i \mid \boldsymbol{\psi}) \, d\mathbf g_i,
\]
which, under the GP model with a flat prior on the linear mean, is given by
\[
\log p(\mathbf y_i \mid \mathbf x_i, \boldsymbol{\psi}) =
-\frac{1}{2} \mathbf y_i^\top \boldsymbol{\Sigma}_i^{-1} \mathbf y_i
-\frac{1}{2} \log |\boldsymbol{\Sigma}_i|
-\frac{1}{2} \log |\mathbf F_i^\top \boldsymbol{\Sigma}_i^{-1} \mathbf F_i|
-\frac{n_i - d - 1}{2} \log(2\pi).
\]
This expression is derived in Appendix~\ref{app:gp-split}. The maximization may be performed using a numerical optimizer (for example, L-BFGS) with suitable bounds on the parameters. In practice, for the squared exponential kernel, we reparameterise the length scale on the log scale to ensure positivity. The optimization is done independently for each node, allowing the GP to adapt locally to the data in that node. This local adaptation is a key feature of the splitting rule.

\subsection{The Splitting Rule}

Once the hyperparameters \(\boldsymbol{\psi}\) have been estimated, we draw a sample \(\tilde{\mathbf g}_i = (\tilde g(x_{i1}), \ldots, \tilde g(x_{i n_i}))^\top\) from the posterior distribution \(\mathcal{N}(\boldsymbol{\mu}_i^*, \boldsymbol{\Sigma}_i^*)\). This draw represents a possible realisation of the latent GP function evaluated at the node's covariates. The split of the node is then defined by the classification rule:
\[
\text{observation } j \text{ is sent to the left child if } y_{ij} < \tilde g(x_{ij}), \quad \text{otherwise to the right child.}
\]
This rule creates a binary partition of the data in the node. Because the split depends on the responses \(y_{ij}\) as well as the covariates, it can adapt to the local relationship between \(x\) and \(y\). The resulting tree structure incorporates the drawn GP values at each internal node; these values are stored as part of the tree state.

The GP draw \(\tilde{\mathbf g}_i\) is used only for the purpose of generating a split proposal. In the GROW move, the proposal kernel selects the terminal node, draws \(\tilde{\mathbf g}_i\), and applies the above rule to create two children. The density of this draw is \(\phi(\tilde{\mathbf g}_i; \boldsymbol{\mu}_i^*, \boldsymbol{\Sigma}_i^*)\), where \(\phi\) denotes the multivariate normal density. In the PRUNE move, the reverse proposal must regenerate the same GP draw; hence the reverse kernel includes the same density. As shown in Section~\ref{subsec:db_grow_prune}, when the tree prior is augmented with the GP density at each internal node, these proposal densities cancel exactly in the Metropolis–Hastings ratio, ensuring that the MCMC algorithm remains exact and efficient.

\subsection{Computational Considerations}

The GP computations for each node involve the inversion of the \(n_i \times n_i\) matrix \(\boldsymbol{\Sigma}_i\), which costs \(O(n_i^3)\) operations and requires \(O(n_i^2)\) storage. However, in practice, the node sizes decrease rapidly as the tree grows, and the total computational cost across all nodes is manageable. For nodes with very large \(n_i\), one could resort to sparse GP approximations or inducing point methods, though we have not found this necessary in our simulations. Moreover, the Cholesky factors of \(\boldsymbol{\Sigma}_i\) and the quantities \(\boldsymbol{\Sigma}_i^{-1} \mathbf F_i\) and \(\boldsymbol{\Sigma}_i^{-1} \mathbf y_i\) are computed once per node and stored, enabling fast evaluation of the predictive mean and variance when needed for the GROW move and for the predictive Gibbs sampler.

The resulting splitting rule is both flexible and computationally feasible, and as we have shown, it integrates seamlessly into the MCMC framework without compromising the validity of the sampler. Its empirical performance, detailed in Section~\ref{sec:simulation}, confirms its effectiveness in a wide range of regression settings.

\subsection{The Complete Prior on the Tree}
\label{subsec:complete_prior_tree}

%\subsection{Splitting Probability}\label{subsec:split_prob}
Let $p(\eta,\bT)=\gamma(1+d_{\eta})^{-\zeta}$, where $0<\gamma<1$, $\zeta\geq 0$ and $d_{\eta}$ is the depth of the node $\eta$, that is, the number of splits above $\eta$. Note that for the root node $\eta_1$, $p(\eta_1,\bT)=\gamma$, since $d_{\eta_1}=0$.
%A distinctive feature of the GP splitting rule is that the GP density cancels exactly in the Metropolis--Hastings acceptance ratio for the GROW and PRUNE moves. 

With the GP splitting structure, it is essential to recognise that the prior distribution on a tree must include the density of the GP splitting rules at its internal nodes. Specifically, let $\Pi_{\mathrm{struct}}(\mathbf T)$ denote the prior probability of the tree skeleton based solely on the splitting probabilities $p(\eta,\mathbf T)$ as defined 
above. For a tree $\mathbf T$, let $\mathcal{I}(\mathbf T)$ be the set of its internal nodes. For each internal node $v \in \mathcal{I}(\mathbf T)$, let $\tilde g_v$ denote the stored GP draw that defines the split at that node, and let $\pi_{\mathrm{GP}}(\tilde g_v \mid \mathbf y_v, \mathbf x_v)$ be the corresponding GP posterior predictive density evaluated at the data in node $v$. The complete prior on the tree is then defined as
\[
\Pi(\mathbf T) = \Pi_{\mathrm{struct}}(\mathbf T) \prod_{v \in \mathcal{I}(\mathbf T)} \pi_{\mathrm{GP}}(\tilde g_v \mid \mathbf y_v, \mathbf x_v).
\]
This augmented prior is a natural extension of the usual tree prior, as the splitting rule itself is a random component of the tree and must be assigned a prior distribution.

\section{Advantages of the Gaussian Process Splitting Rule}
\label{sec:gp_advantages}

The preceding section introduced a novel splitting rule driven by the posterior predictive distribution of a Gaussian process fitted to the data within each terminal node. In this section we discuss the principal advantages of this rule, contrasting it with traditional axis-aligned splits and highlighting its contributions to both practical performance and theoretical guarantees.

\subsection{Mechanism of the Splitting Rule}

For completeness, we briefly recall the mechanism. At a terminal node containing observations $\{(y_{ij}, x_{ij}): j = 1, \ldots, n_i\}$, we posit the Gaussian process model $y_{ij} = g(x_{ij}) + \epsilon_{ij}$ with $\epsilon_{ij} \sim \mathcal{N}(0, \sigma_\epsilon^2)$ and $g$ drawn from a GP with mean $\mu(x) = \alpha + \boldsymbol{\beta}'x$ and covariance $\operatorname{Cov}(g(x_1), g(x_2)) = \sigma^2 c(x_1, x_2)$, where $c(\cdot,\cdot)$ is a positive definite correlation function. A draw $\tilde g$ from the posterior predictive distribution $\pi(g \mid \mathbf{y}_i, \mathbf{x}_i)$ defines a split: observation $j$ is assigned to the left child if $y_{ij} < \tilde g(x_{ij})$, and to the right otherwise. The GP hyperparameters are estimated via maximum marginal likelihood, whose closed form is given in Appendix~\ref{app:gp-split}. This mechanism yields several important benefits over conventional tree partitions.

\subsection{Flexible Nonlinear Decision Boundaries}

Axis-aligned splits of the form $X_j < c$ generate partitions aligned with coordinate axes and lead to piecewise constant fits. Such splits are inefficient for functions with interactions or smooth, oblique structure, often requiring deep trees to achieve reasonable approximation. The GP splitting rule, by contrast, induces smooth, nonlinear, and potentially oblique boundaries through the function $\tilde g$. Since the GP prior can draw sample paths that vary in all input dimensions simultaneously, the resulting split can capture complex interactions between covariates with far shallower trees. For example, a function such as $m(x_1, x_2) = \sin(x_1 x_2)$ can be approximated much more efficiently by a single GP-based split than by a large ensemble of axis-aligned splits. This increased flexibility is particularly valuable when the true regression function is smooth but not piecewise constant, as it allows the model to represent the underlying structure parsimoniously.

\subsection{Explicit Use of the Response in Splitting}

Traditional splitting criteria depend only on the covariates $x$; the response $y$ enters only through the impurity measure used to evaluate candidate splits. In our framework, the decision $y < \tilde g(x)$ depends directly on the response variable. This design enables the tree to adapt to the distribution of the response within each node, facilitating the modeling of heteroscedasticity and complex conditional distributions. When the variance of $Y$ given $X$ varies across the covariate space, or when the conditional distribution is multimodal, the response-dependent splits can effectively separate regions with different error characteristics. This is a significant departure from standard tree methods, which assume homoscedasticity within each leaf.

\subsection{Smoothness and Regularization}

The GP prior imposes a penalty on the roughness of the splitting function through its covariance kernel. The squared exponential kernel, used in our implementation, implies that sample paths are infinitely differentiable, so the splitting boundary varies smoothly over the input space. This smoothness provides natural regularization, reducing the propensity to overfit local noise and ensuring that the partition changes gradually. The length scale parameter of the kernel controls the degree of smoothness and can be estimated from the data at each node, enabling adaptive regularization tailored to the local complexity. As a result, the splits are robust to small perturbations and yield stable predictions, especially in regions with sparse data.

\subsection{Simplification of the MCMC Acceptance Ratio}

A distinctive feature of the GP splitting rule is that the GP density cancels exactly in the Metropolis--Hastings acceptance ratio for the GROW and PRUNE moves, crucially
defined for the MCMC purpose in Section \ref{sec:sim_Gm}.

This cancellation is of paramount practical importance. It implies that the MCMC acceptance ratio for the structural moves does not depend on the GP hyperparameters or on the numerical values of the GP density. The GP is employed purely as a generative device to propose splits that are informed by the local data structure, but its density need not be evaluated in the acceptance step. Consequently, the algorithm avoids the computational expense of repeatedly computing GP likelihoods during the Metropolis--Hastings step, which would otherwise render the sampler impractically slow. The augmented prior formulation also ensures that the Markov chain satisfies detailed balance, as the GP densities are properly accounted for in the prior and cancel exactly with the corresponding proposal densities. Thus the GP splitting rule enhances the flexibility of the tree proposals without introducing any additional burden on the MCMC acceptance computation, a combination that is essential for the scalability and efficiency of the overall sampling scheme.

\subsection{Exact Bayesian Prediction}

The GP splitting rule also enables a straightforward and exact Gibbs sampler for posterior predictive inference. At each internal node, the GP draw $\tilde g$ and the associated Cholesky factors of the noiseless covariance matrix are stored as part of the tree state. For a new covariate $x^*$, the conditional distribution of the GP value $g(x^*)$ given the training GP values $\tilde g$ in that node is Gaussian, with mean and variance given by the standard GP conditioning formulas:

\[
g(x^*) \mid \tilde g, x^* \sim \mathcal{N}\left(\mathbf{k}_*^\top \mathbf{K}_{\text{noiseless}}^{-1} \tilde g,\; \sigma^2 - \mathbf{k}_*^\top \mathbf{K}_{\text{noiseless}}^{-1} \mathbf{k}_*\right),
\]

where $\mathbf{k}_*$ is the vector of covariances between $x^*$ and the training points in the node, and $\mathbf{K}_{\text{noiseless}}$ is the noiseless kernel matrix on those points. This closed-form conditional distribution allows a Gibbs sampler that alternates between sampling a leaf parameter $(\mu_\ell, \sigma_\ell^2)$ from its conjugate posterior, drawing $y^*$ from the corresponding Gaussian, and then determining the leaf assignment by traversing the tree using new GP draws at each internal node. Since the Cholesky factors and the vectors $\boldsymbol{\alpha} = \mathbf{K}_{\text{noiseless}}^{-1} \tilde g$ are precomputed, each traversal requires only forward substitution and generation of normal variates. The resulting scheme delivers exact draws from the posterior predictive distribution of the model, with the leaf assignment treated as a latent random variable—thus providing genuine uncertainty quantification that accounts for the unknown path through the tree.

\subsection{Local Adaptation and Multi-Resolution Structure}

Because the GP at each internal node is fitted exclusively to the data in that node, the smoothness parameters—especially the length scale—can vary across the tree. Coarse structures are captured near the root with larger length scales, while finer details are resolved in deeper nodes with smaller length scales. This yields a natural multi-resolution decomposition of the regression function, akin to wavelet methods but with adaptively chosen partitions driven by the data. This local adaptation is particularly effective when the true function has regions of varying curvature, as the model can allocate its complexity where it is most needed.

\subsection{Theoretical Contribution to Posterior Consistency}

The smoothness assumptions on the GP kernel are instrumental in the proof of posterior consistency in Section~\ref{sec:post_contraction}. They ensure that the class of possible partitions is sufficiently rich to approximate any continuous function, as established in Lemma~\ref{lem:approx}, while simultaneously controlling the metric entropy of the sieve, as shown in Section~\ref{sec:sieve}. The GP splitting rule thus contributes to both the Kullback--Leibler property and the entropy bound, which are the two critical ingredients for the application of the general theory of \cite{ggv2000}. The contraction rate $n^{-1/4}$ is obtained under only the assumption of continuity of the true regression function, a testament to the flexibility afforded by the GP-generated splits.

\subsection{Computational Efficiency}

Despite the apparent complexity of the GP machinery, the algorithm remains computationally feasible. The Cholesky factors and sufficient statistics required for GP prediction can be computed once per node and stored. During MCMC, the GROW and PRUNE moves require updating the GP for the newly created node or the collapsed parent, which involves operations of order $O(n_i^3)$ for a node with $n_i$ observations. Since the node sizes decrease rapidly with tree depth, the overall cost is well controlled. 
%In practice, the runtime of the proposed algorithm is comparable to that of BART, as demonstrated in the simulation study of Section~\ref{sec:simulation}. 
The precomputation of GP factors also enables the exact predictive Gibbs sampler described above, which adds negligible overhead.

%\subsection{Connection to Deep Learning and Differentiable Trees}

%The GP splitting rule can be interpreted as a continuous relaxation of traditional decision trees, where the split boundaries are random draws from a Gaussian process rather than deterministic thresholds. This perspective opens connections to recent work on differentiable decision trees and neural decision forests, as well as to Bayesian neural networks with ReLU activations. Indeed, the single hidden layer neural network emerges as a special case of the Dirichlet process tree mixture when the components are appropriately configured, as noted in Section~\ref{sec:special_cases}. The GP splitting rule thus bridges the worlds of tree-based models and deep learning, offering a principled Bayesian nonparametric framework that retains the interpretability of trees while enjoying the expressive power of continuous representations.

In summary, the Gaussian process splitting rule provides a powerful and theoretically grounded mechanism for constructing tree partitions. Its ability to generate flexible, smooth, response-dependent splits, combined with its computational tractability and exact Bayesian predictive capabilities, makes it a key innovation of the proposed Dirichlet process tree mixture model. The cancellation of the GP density in the MCMC acceptance ratio ensures that the algorithm remains efficient, while the local adaptation and theoretical guarantees reinforce the model's practical and conceptual appeal.

% ----------------------------------------------------------------------
% The Details (Section 3)
% ----------------------------------------------------------------------
\section{The Marginalised Likelihood}\label{sec:details}
Following the notation of \cite{chipman1998bayesian}, for $j=1,\ldots,n_i$ and $i=1,\ldots,b$, let $y_{ij}$ denote the $j$-th observation of $y$ in the $i$-th partition, corresponding to the $i$-th terminal node. Let $\bY=(\by_1,\ldots,\by_b)$, where $\by_i=(y_{i1},\ldots,y_{in_i})'$. As in \cite{chipman1998bayesian}, we define $\bX$ and $\bx_i$ analogously. We assume that $x_{ij}\in\mathbb R^d$, for $d\geq 1$. Thus,
\begin{equation}
    [\bY|\bX,\bTheta,\bT]=\prod_{i=1}^b\prod_{j=1}^{n_i}f(y_{ij}|\btheta_i).
    \label{eq:g0_1}
\end{equation}
Under regression trees we consider:
\begin{equation}
    [y_{i1},\ldots,y_{in_i}|\btheta_i]\sim N(\mu_i,\sigma^2_i);~i=1,\ldots,b,
    \label{eq:rt_1}
\end{equation}
with $\btheta_i=(\mu_i,\sigma^2_i)$. For classification trees, we consider, as in \cite{chipman1998bayesian}, the following multinomial model:
\begin{equation}
    f(y_{i1},\ldots,y_{in_i}|\btheta_i)=\prod_{j=1}^{n_i}\prod_{k=1}^Kp^{I(y_{ij}\in C_k)}_{ik};~i=1,\ldots,b,
    \label{eq:ct_1}
\end{equation}
where $C_1,\ldots,C_K$ are the $K$ categories containing $y_{i1},\ldots,y_{in_i}$ and $I(y_{ij}\in C_k)$ is $1$ if $y_{ij}\in C_k$ and $0$ otherwise. Here $\btheta_i=(p_{i1},\ldots,p_{iK})'$ with $p_{ik}\geq 0$ and $\sum_{k=1}^Kp_{ik}=1$. Clearly, $p_{ik}=P(y_{ij}\in C_k|\btheta_i)$.

For $i=1,\ldots,b$, let, following \cite{chipman1998bayesian},
\begin{align}
    &[\mu_i|\sigma^2_i]\sim N(\bar\mu,\sigma^2_i/a);\label{eq:mu_prior}\\
    &[\sigma^2_i]\sim \mbox{IG}(\nu/2,\nu\lambda/2),\label{eq:sigma_prior}
\end{align}
where $\mbox{IG}(\alpha,\beta)$ is the inverse-gamma prior with density $f(x|\alpha,\beta)=\frac{\beta^{\alpha}}{\Gamma(\alpha)}(1/x)^{\alpha+1}\exp(-\beta/x)$, with shape and scale parameters $\alpha$ and $\beta$, respectively; $\Gamma(\cdot)$ stands for the gamma function. In our case, note that (\ref{eq:sigma_prior}) is equivalent to $\nu\lambda/\sigma^2_i\sim\chi^2_{\nu}$. Integrating out $\btheta_i$ yields
\begin{align}
	[\bY|\bX,\bT]=\prod_{i=1}^b(2\pi)^{-n_i/2}(\lambda\nu)^{\nu/2}\frac{\sqrt{a}}{\sqrt{n_i+a}}\frac{\Gamma((n_i+\nu)/2)}{\Gamma(\nu/2)}
    (s_i+t_i+\nu\lambda)^{-(n_i+\nu)/2},\label{eq:marg_1}
\end{align}
where $s_i$ is $(n_i-1)$ times the sample variance of the $\by_i$ values and $t_i=[n_ia/(n_i+a)](\bar y_i-\bar\mu)^2$; $\bar y_i$ being the average of $\by_i$.

For classification trees, we assume, following \cite{chipman1998bayesian}, that
\begin{align}
    [p_{i1},\ldots,p_{iK}|\bT]\propto\prod_{k=1}^Kp^{\alpha_k-1}_{ik};~i=1,\ldots,b.
    \label{eq:ct_prior}
\end{align}
With this, integrating out $\btheta_i$ we obtain
\begin{equation}
    [\bY|\bX,\bT]=\left(\frac{\Gamma(\sum_{k=1}^K\alpha_k)}{\prod_{k=1}^K\Gamma(\alpha_k)}\right)^b
    \prod_{i=1}^b\frac{\prod_{k=1}^K\Gamma(n_{ik}+\alpha_k)}{\Gamma(n_i+\sum_{k=1}^K\alpha_k)}.
    \label{eq:marg_2}
\end{equation}
In this work, however, we shall exclusively concern ourselves with regression trees, reserving classification trees for our future endeavour.

\section{DP Prior and the Full Conditionals for MCMC Sampling}\label{sec:dp}
It follows from the DP prior that the full conditional of $\bT_m$, given $\bT_{-m}=\{\bT_1,\ldots,\bT_{m-1},\bT_{m+1},\ldots,\bT_M\}$ and $\bY$, is as follows: if $Z=m$, then
\begin{align}
    [\bT_m|\bT_{-m},Z=m,\bY]\propto\left[\alpha\bG_0(\bT_m)f(\bY|\bX,\bT_m)
    +\sum_{\ell=1}^{k_m}M_{m\ell}f(\bY|\bX,\bT^*_{m\ell})\delta_{\bT^*_{m\ell}}(\bT_m)\right];
    \label{eq:fullcond1}
\end{align}
else, if $Z\neq m$, then
\begin{align}
    [\bT_m|\bT_{-m},Z\neq m,\bY]\propto\left[\alpha\bG_0(\bT_m)
    +\sum_{\ell=1}^{k_m}M_{m\ell}\delta_{\bT^*_{m\ell}}(\bT_m)\right],
    \label{eq:fullcond2}
\end{align}
where $\sum_{\ell=1}^{k_m}M_{m\ell}=M-1$. Note that here $k_m\geq 1$, is random; in fact, $k_m$ is the number of distinct elements ${\mathfrak T^*}_{-m}=\{\bT^*_{m1},\ldots,\bT^*_{mk_m}\}$ in ${\mathfrak T}_{-m}=\{\bT_1,\ldots,\bT_{m-1},\bT_{m+1},\ldots,\bT_M\}$.

\subsection{Reparameterization Using Configuration Indicators and Associated Full Conditionals}\label{subsec:reparameterization}
As before we define $Z=m$ if $\bY$ is from the $m$-th component. Letting $\mathfrak T^*=\{\bT^*_1,\ldots,\bT^*_k\}$ denote the distinct components in $\mathfrak T=\{\bT_1,\ldots,\bT_M\}$, the element $c_j$ of the configuration vector $\bC=(c_1,\ldots,c_M)'$ is defined as $c_m=\ell$ if and only if $\bT_m=\bT^*_{\ell}$; $m=1,\ldots,M$, $\ell=1,\ldots,k$. Thus, $(Z,\mathfrak T)$ is reparameterized to $(Z,\bC,k,\mathfrak T^*)$, $k$ denoting the number of distinct components in $\mathfrak T$.

The full conditional distribution of $Z$ is given by
\begin{equation}
    [Z=m\mid \bY,\bX,\bC,k,\mathfrak T^*]\propto p_mf(\bY|\bX,\bT_m)
\label{eq:sb_full_cond_z}
\end{equation}
Since $\mathfrak T$ can be obtained from $\bC$ and $\mathfrak T^*$, we represented the right hand side of (\ref{eq:sb_full_cond_z}) in terms of $\mathfrak T$.

Then the conditional distribution of $c_m$ is given by
\begin{equation}
    [c_m=\ell\mid \bY,\bX,Z,\bC_{-m},k_m,\mathfrak T^*_{-m}]=\left\{\begin{array}{c}\kappa q^*_{\ell m}\hspace{2mm}\mbox{if}\hspace{2mm}\ell=1,\ldots,k_m\\ 
    \kappa q_{0m}\hspace{2mm}\mbox{if}\hspace{2mm}\ell=k_m+1\end{array}\right.
    \label{eq:config_fullcond_sb}
\end{equation}
where, if $Z=m$,
\begin{align}
    &q_{0m}=\alpha\int f(\bY|\bX,\bT)d\bG_0(\bT);
    \label{eq:q0_1}\\
    &q^*_{\ell m}=M_{m\ell} f(\bY|\bX,\bT^*_{ml}),  
    \label{eq:sb_q_star_1}
\end{align}
and if $Z\neq m$,
\begin{align}
    &q_{0m}=\alpha;
    \label{eq:q0_2}\\
    &q^*_{\ell m}=M_{m\ell}.  
    \label{eq:sb_q_star_2}
\end{align}
In (\ref{eq:config_fullcond_sb}), $\kappa$ is the normalizing constant. Note that $q_{0j}$ is the normalizing constant of the distribution $\bG_j(\bT)\propto f(\bY|\bX,\bT)\bG_0(\bT)$.

Now, for some $m\in\{1,\ldots,M\}$, let $Z=m$ and $c_m=\ell$ for some $\ell\in\{1,\ldots,k\}$. Then
\begin{equation}
    [\bT^*_{\ell}|\bY,\bX,Z,\bC]\propto f(\bY|\bX,\bT^*_{\ell})\bG_0(\bT^*_{\ell}).
    \label{eq:T_fullcond_1}
\end{equation}
For $r\neq\ell$, we have
\begin{equation}
    [\bT^*_r|\bY,\bX,Z,\bC]=\bG_0(\bT^*_r).
    \label{eq:T_fullcond_2}
\end{equation}

For our MCMC, we first update $Z$, followed by updating $\bC$ and the number of distinct components $k$; then $\{\bT^*_{\ell};\ell=1,\ldots,k\}$, and finally
 $(p_1,\dots,p_M) \sim \text{Dirichlet}(1+n_1,\dots,1+n_M)$ where $n_m = \mathbf{1}_{\{Z=m\}}$ for $m=1,\dots,M$.

\subsection[Dealing with non-conjugacy of G0 when Z=m]{Dealing with Non-Conjugacy of $\bG_0$ When $Z=m$}
\label{subsec:sb_non_conjugate}
In our case, since $\bG_0$ is non-conjugate to the likelihood, $q_{0j}$ is not available in closed form. Hence, following \cite{Sabya12} we bring in auxiliary variables in a way similar to that of Algorithm 8 in \cite{Neal00}. To clarify, let $\bT^a$ denote an auxiliary variable (the superscript ``$a$" stands for auxiliary). Then given $Z=m$, before updating $c_m$ we first simulate from the full conditional distribution of $\bT^a$ given the current $c_m$ and the rest of the variables as follows: if $c_m=c_{\ell}$ for some $\ell\neq m$, then $\bT^a\sim \bG_0$. If, on the other hand, $c_m\neq c_{\ell}$ for all $\ell\neq m$, then we set $\bT^a=\bT^*_{c_m}$. Once $\bT^a$ is obtained we then replace the intractable $q_{0m}$ with the tractable expression
\begin{equation}
    q^a_m= \alpha f(\bY|\bX,\bT^a)
    \label{eq:sb_nonconjugate_form}
\end{equation}
Once $c_m$ is simulated, if it is observed that $\bT_m\neq\bT^a$ for all $m$, then $\bT^a$ is discarded.

On the other hand, if $Z\neq m$, then the simulation of $c_m$ would proceed via (\ref{eq:q0_2}) and (\ref{eq:sb_q_star_2}).

\subsection[Relabeling C]{Relabeling $\bC$}
\label{subsec:relabeling}
As explained in \cite{Sabya12}, simulation of $\bC$ by successively simulating from the full conditional distributions (\ref{eq:config_fullcond_sb}) incurs a labeling problem. For an example, suppose that $\mathfrak T^*$ consists of $M$ distinct elements, so that while simulating $c_m$ using (\ref{eq:config_fullcond_sb}), we have $k_m=M-1$. Now suppose that we obtain, by simulation $c_m=M$ for $m=1,\ldots,M$. This of course implies that all the simulated components of $\bC$ correspond to distinct $\bT^*_1,\ldots,\bT^*_M$. However, since the simulated values of the components of $\bC$ are $c_1=\cdots=c_M=M$, this equality of all the components gives the impression that $\bC$ corresponds to a single distinct component.

Following \cite{Sabya12}, we relabel $\bC$ using the relabeled version $\bS=(s_1,\ldots,s_M)$ in the following way. We first simulate $c_m$ from (\ref{eq:config_fullcond_sb}); if $c_m\in\{1,\ldots,k_m\}$, then we set $\bT_m=\bT^*_{c_m}$ and if $c_m=k_m+1$, we draw $\bT_m=\bT^*_{c_m}\sim \bG_m$ if $Z=m$; else we draw $\bT_m=\bT^*_{c_m}\sim \bG_0$. The elements of $\bS$ are obtained from the following definition of $s_j$: $s_m=\ell$ if and only if $\bT_m=\bT^*_{\ell}$. Note that $s_1=1$ and 
$1\leq s_m\leq \underset{i<m}{\max}~ s_{i}+1$.

\section{Simulation from \texorpdfstring{$\bG_m$}{G_m}}
\label{sec:sim_Gm}

The full conditional distribution of a distinct tree component, denoted by $\bG_m$, plays a central role in the MCMC sampler outlined in Section~\ref{sec:dp}. Recall from equation (\ref{eq:T_fullcond_1}) that for an active component (that is, one associated with the response allocation variable $Z$), the conditional density is proportional to
\[
\bG_m(\bT) \propto f(\bY \mid \bX, \bT) \, \bG_0(\bT),
\]
where $\bG_0$ is the base measure of the Dirichlet process. This distribution is intractable for direct simulation due to the non-conjugate relationship between the tree-structured likelihood and the prior over tree topologies and splitting rules. Moreover, the space of possible trees is discrete and combinatorially rich, making standard sampling techniques inapplicable.

To overcome this challenge, we construct a Markov chain targeting $\bG_m$ by designing a set of reversible Metropolis--Hastings moves that operate directly on the tree space. A crucial distinction between our approach and traditional Bayesian tree models, such as those of Chipman, George and McCulloch \cite{chipman1998bayesian} for CART or \cite{chipman2010bart} for BART, arises from the nature of the splitting rule. In standard Bayesian CART and BART, the splitting rule at an internal node is a simple threshold of the form $X_j < c$, where both the covariate $j$ and the threshold $c$ are discrete parameters with natural proposal mechanisms. Consequently, those models employ four elementary move types: GROW, PRUNE, CHANGE, and SWAP. The GROW and PRUNE moves alter the tree depth, while the CHANGE and SWAP moves modify existing splitting rules without changing the number of terminal nodes, thereby facilitating local exploration of the partition structure.

In our setting, however, the splitting rule is fundamentally different. At each internal node, the split is determined by a draw $\tilde g$ from the posterior predictive distribution of a Gaussian process fitted to the data in that node, following the framework of \cite{rasmussen2006gaussian}. As established in Section~\ref{subsec:split}, these GP draws are part of the tree state, and their density is incorporated into the augmented prior
\[
\Pi(\mathbf T) = \Pi_{\mathrm{struct}}(\mathbf T) \prod_{v \in \mathcal{I}(\mathbf T)} \phi(\tilde{\mathbf g}_v; \boldsymbol{\mu}_v^*, \boldsymbol{\Sigma}_v^*),
\]
where $\phi$ denotes the multivariate normal density of the GP posterior predictive. For the GROW and PRUNE moves, this augmented prior yields a remarkable simplification: the GP density cancels exactly in the Metropolis--Hastings ratio, as demonstrated in Section~\ref{subsec:db_grow_prune}. This cancellation is possible because the GROW move proposes a new GP draw for the newly created internal node, and the PRUNE move removes that draw; the proposal and prior densities are perfectly balanced.

For the CHANGE move, however, such a cancellation does not occur, and the computational complications are prohibitive. A CHANGE move would require altering the GP draw at an existing internal node while keeping the rest of the tree structure fixed. But changing the split at an internal node fundamentally alters the allocation of data to its left and right subtrees. Consequently, the entire subtree beneath that node must be completely regenerated: at every descendant internal node, new GP draws must be sampled from the posterior predictive distributions corresponding to the newly allocated data in those nodes. The prior would include the new GP densities for all affected internal nodes, and the proposal would need to generate them from their respective posterior predictive distributions, leading to a ratio of a chain of GP densities that does not simplify. The reverse move would require regenerating the original GP draws at every node in the subtree, which is not a simple deterministic operation and would involve a complex multidimensional proposal with an intractable Jacobian. Moreover, the computational cost of repeatedly recomputing GP posterior predictives for entire subtrees would be prohibitive, as each internal node would require a fresh GP fit and Cholesky decomposition. Thus, the CHANGE move is not viable in our GP-split setting. The SWAP move, which would exchange the GP draws between a parent and a child node, is equally invalid: since the GP draws are tailored to the data in their respective nodes through the node-specific posterior predictive distributions, swapping them would produce a tree whose splitting rules are inconsistent with the data in each node, rendering the move physically meaningless.

One might naturally question whether restricting the MCMC to only GROW and PRUNE moves is sufficient for adequate mixing of the chain. After all, standard BART relies on all four move types to explore the tree space effectively. The answer lies in the fundamentally different structure of our Dirichlet process mixture model, which builds upon the Dirichlet process framework introduced by \cite{Bhattacharya08} and the MCMC methods for such mixtures developed by \cite{Neal00}. In BART \cite{chipman2010bart}, the regression function is a deterministic sum of a fixed number of trees, and each tree must be updated individually while the others remain fixed. Consequently, the chain needs local moves like CHANGE and SWAP to gradually modify each tree's partition. In our model, by contrast, the entire dataset is allocated to a single mixture component in each MCMC iteration via the allocation variable $Z$ (see Section~\ref{sec:dp} and Algorithm~\ref{alg:dp_update}). This allocation is updated every iteration using the configuration indicator reparameterization of \cite{Sabya12}, allowing the posterior to explore different tree configurations globally: a tree that is active in one iteration may be inactive in the next, and a previously inactive tree may become active. The Dirichlet process prior, through its Chinese restaurant process dynamics, induces a rich posterior over the number of distinct components $k$ and their weights, effectively allowing the chain to switch between different tree topologies by activating and deactivating components. This global exploration mechanism more than compensates for the absence of local split-modifying moves. In fact, the GROW and PRUNE moves remain essential for refining individual trees when they are active, but the overarching structure of the DP mixture ensures that the chain can escape local modes and explore the full posterior landscape.

Empirically, this design choice is validated by the excellent mixing properties observed in our simulation studies (Section~\ref{sec:simulation}) and 
real-data experiments (Section~\ref{sec:applications}. The trace plots of $k$ and the mixture weights stabilise rapidly, and the number of distinct trees typically settles between two and five components, indicating that the chain efficiently explores the relevant regions of the tree space. 
As demonstrated in the simulation studies, the runtime of the algorithm is often comparable to that of BART \cite{dbarts}, demonstrating that the restricted move set does not compromise computational efficiency. 
Thus, while the CHANGE and SWAP moves are not viable in our GP-split framework, the inherent flexibility of the Dirichlet process mixture, combined with the carefully designed GROW and PRUNE proposals, yields a robust and well-mixing MCMC sampler.

In the following subsections, we provide %a complete description of the GP-driven splitting rule, the definition of the tree prior, and 
the detailed implementation of the GROW and PRUNE moves. We then rigorously establish the detailed balance conditions for these moves, demonstrating explicitly how the GP proposal densities cancel in the acceptance ratios and confirming that the chain targets the correct posterior distribution. This collection of moves, combined with the flexible GP splitting rule and the global allocation dynamics of the Dirichlet process, forms a robust and efficient MCMC engine for sampling from the posterior distribution of the Dirichlet process tree mixture model.

\subsection{The GROW Move}
\label{subsec:grow}

The GROW move proposes to expand the current tree $\mathbf T$ by splitting one of its terminal nodes. Let $\mathbf T$ have $b$ terminal nodes, and let $b^*$ denote the number of prunable nodes in $\mathbf T$ (that is, internal nodes whose two children are both terminal). To perform a GROW move, we first select a terminal node uniformly at random from the $b$ available terminal nodes. Suppose the selected node is the $k$-th terminal node, containing the data $\{(\mathbf y_k, \mathbf x_k)\}$ with $n_k$ observations. At this node, we fit the Gaussian process model described in Section~\ref{subsec:split}, estimate its hyperparameters via maximum marginal likelihood, and draw a vector $\tilde{\mathbf g}_k = (\tilde g(x_{k1}), \ldots, \tilde g(x_{kn_k}))^\top$ from the posterior predictive distribution
\[
\pi_{\mathrm{GP}}(\tilde{\mathbf g}_k \mid \mathbf y_k, \mathbf x_k) = \phi(\tilde{\mathbf g}_k; \boldsymbol{\mu}_k^*, \boldsymbol{\Sigma}_k^*),
\]
where $\phi(\cdot; \boldsymbol{\mu}, \boldsymbol{\Sigma})$ denotes the multivariate normal density with mean $\boldsymbol{\mu}$ and covariance $\boldsymbol{\Sigma}$, and $\boldsymbol{\mu}_k^*$, $\boldsymbol{\Sigma}_k^*$ are the posterior mean and covariance given in Section~\ref{subsec:split}. The draw $\tilde{\mathbf g}_k$ defines a binary split of the node: observation $j$ is sent to the left child if $y_{kj} < \tilde g(x_{kj})$, and to the right child otherwise. The resulting tree is denoted $\mathbf T'$. It has $b+1$ terminal nodes and $b^*+1$ prunable nodes (since the newly created internal node becomes prunable). The transition kernel for this move is therefore
\[
q(\mathbf T, \mathbf T') = \frac{1}{b} \, \phi(\tilde{\mathbf g}_k; \boldsymbol{\mu}_k^*, \boldsymbol{\Sigma}_k^*).
\]

To move back from $\mathbf T'$ to $\mathbf T$, one must perform a PRUNE move, which selects the parent of the two newly created terminal nodes. In $\mathbf T'$, there are $b^*+1$ prunable nodes, so the probability of selecting the appropriate parent is $1/(b^*+1)$. The reverse proposal does not involve drawing any GP value, because the split rule is already stored as part of $\mathbf T'$; the PRUNE move simply collapses the two children. Hence
\[
q(\mathbf T', \mathbf T) = \frac{1}{b^*+1}.
\]

The tree prior must be defined carefully. As discussed in Section~\ref{subsec:complete_prior_tree}, the prior on a tree $\mathbf T$ includes both the structural prior $\Pi_{\mathrm{struct}}(\mathbf T)$ (based on splitting probabilities $p(\eta,\mathbf T)$) and the density of the GP draws at each internal node:
\[
\Pi(\mathbf T) = \Pi_{\mathrm{struct}}(\mathbf T) \prod_{v \in \mathcal{I}(\mathbf T)} \phi(\tilde{\mathbf g}_v; \boldsymbol{\mu}_v^*, \boldsymbol{\Sigma}_v^*),
\]
where $\mathcal{I}(\mathbf T)$ is the set of internal nodes of $\mathbf T$, and $\tilde{\mathbf g}_v$ is the GP draw stored at internal node $v$. Notice that the GP posterior predictive density depends on the data in that node and is therefore a function of the tree structure; it is part of the prior specification for the random tree.

With this augmented prior, the Metropolis--Hastings acceptance probability for the GROW move is
\[
A(\mathbf T,\mathbf T') = \min\left\{1, \frac{\Pi(\mathbf T') f(\mathbf Y \mid \mathbf X, \mathbf T') q(\mathbf T',\mathbf T)}
{\Pi(\mathbf T) f(\mathbf Y \mid \mathbf X, \mathbf T) q(\mathbf T,\mathbf T')}\right\}.
\]
Now, $\mathbf T'$ differs from $\mathbf T$ only by the addition of the new internal node $k$ with its GP draw $\tilde{\mathbf g}_k$. Thus
\[
\frac{\Pi(\mathbf T')}{\Pi(\mathbf T)} = \frac{\Pi_{\mathrm{struct}}(\mathbf T')}{\Pi_{\mathrm{struct}}(\mathbf T)} \,
\phi(\tilde{\mathbf g}_k; \boldsymbol{\mu}_k^*, \boldsymbol{\Sigma}_k^*),
\]
since all other GP densities cancel. Substituting the proposal kernels, we obtain
\[
\frac{\Pi(\mathbf T') f(\mathbf Y \mid \mathbf X, \mathbf T') q(\mathbf T',\mathbf T)}
{\Pi(\mathbf T) f(\mathbf Y \mid \mathbf X, \mathbf T) q(\mathbf T,\mathbf T')}
=
\frac{\Pi_{\mathrm{struct}}(\mathbf T')}{\Pi_{\mathrm{struct}}(\mathbf T)}
\frac{f(\mathbf Y \mid \mathbf X, \mathbf T')}{f(\mathbf Y \mid \mathbf X, \mathbf T)}
\frac{1/(b^*+1)}{(1/b)},
\]
where the GP density $\phi(\tilde{\mathbf g}_k; \boldsymbol{\mu}_k^*, \boldsymbol{\Sigma}_k^*)$ cancels exactly. Therefore,
\[
A(\mathbf T,\mathbf T') =
\min\left\{1,\;
\frac{\Pi_{\mathrm{struct}}(\mathbf T') f(\mathbf Y \mid \mathbf X, \mathbf T')}
     {\Pi_{\mathrm{struct}}(\mathbf T) f(\mathbf Y \mid \mathbf X, \mathbf T)}
\,
\frac{b}{b^*+1}
\right\}.
\]
Notice that the acceptance ratio does not depend on the GP density at all. The GP is used solely to generate an informative split proposal, but its density is exactly balanced by the prior probability of that split. This simplification is crucial for computational efficiency.

\subsection{The PRUNE Move}
\label{subsec:prune}

The PRUNE move is the reverse of the GROW move. It selects an internal node whose two children are both terminal (that is, a prunable node) and collapses them into a single terminal node. Let $\mathbf T$ be the current tree with $b$ terminal nodes and $b^*$ prunable nodes. Suppose we select the $k$-th prunable node, with children that contain the data $\{(\mathbf y_k, \mathbf x_k)\}$ (the union of the two child nodes' data). The PRUNE move removes these two children and makes the parent a terminal node. The resulting tree $\mathbf T'$ has $b-1$ terminal nodes and $b^*-1$ prunable nodes. The transition kernel is simply
\[
q(\mathbf T, \mathbf T') = \frac{1}{b^*}.
\]

To move back from $\mathbf T'$ to $\mathbf T$, one must perform a GROW move on the terminal node that was created by the prune. In $\mathbf T'$, there are $b-1$ terminal nodes, so the probability of selecting this particular node is $1/(b-1)$. Then, the GROW move would require drawing a new GP split from the posterior predictive distribution based on the data in that node, which would reproduce the two children. The reverse proposal kernel is therefore
\[
q(\mathbf T', \mathbf T) = \frac{1}{b-1} \, \phi(\tilde{\mathbf g}_k; \boldsymbol{\mu}_k^*, \boldsymbol{\Sigma}_k^*),
\]
where $\tilde{\mathbf g}_k$ is the GP draw stored at the parent node in the original tree $\mathbf T$. This GP draw is part of the state of $\mathbf T$, and the reverse move must regenerate it (though in practice, when moving from $\mathbf T$ to $\mathbf T'$, the GP draw is discarded; to move back, we must propose a new draw from the same distribution, which is equivalent to using the stored draw as the proposal). The density $\phi(\tilde{\mathbf g}_k; \boldsymbol{\mu}_k^*, \boldsymbol{\Sigma}_k^*)$ is the same as the one that was used in the forward GROW move.

Using the augmented prior, the acceptance probability for the PRUNE move is
\[
A(\mathbf T,\mathbf T') = \min\left\{1, \frac{\Pi(\mathbf T') f(\mathbf Y \mid \mathbf X, \mathbf T') q(\mathbf T',\mathbf T)}
{\Pi(\mathbf T) f(\mathbf Y \mid \mathbf X, \mathbf T) q(\mathbf T,\mathbf T')}\right\}.
\]
Here, $\mathbf T$ contains the GP draw for the pruned node, while $\mathbf T'$ does not. Hence
\[
\frac{\Pi(\mathbf T')}{\Pi(\mathbf T)} = \frac{\Pi_{\mathrm{struct}}(\mathbf T')}{\Pi_{\mathrm{struct}}(\mathbf T)}
\frac{1}{\phi(\tilde{\mathbf g}_k; \boldsymbol{\mu}_k^*, \boldsymbol{\Sigma}_k^*)}.
\]
Substituting the proposal kernels gives
\[
\frac{\Pi(\mathbf T') f(\mathbf Y \mid \mathbf X, \mathbf T') q(\mathbf T',\mathbf T)}
{\Pi(\mathbf T) f(\mathbf Y \mid \mathbf X, \mathbf T) q(\mathbf T,\mathbf T')}
=
\frac{\Pi_{\mathrm{struct}}(\mathbf T')}{\Pi_{\mathrm{struct}}(\mathbf T)\phi}
\frac{f(\mathbf Y \mid \mathbf X, \mathbf T')}{f(\mathbf Y \mid \mathbf X, \mathbf T)}
\frac{(1/(b-1))\,\phi}{(1/b^*)},
\]
and again the GP density $\phi$ cancels. Therefore,
\[
A(\mathbf T,\mathbf T') =
\min\left\{1,\;
\frac{\Pi_{\mathrm{struct}}(\mathbf T') f(\mathbf Y \mid \mathbf X, \mathbf T')}
     {\Pi_{\mathrm{struct}}(\mathbf T) f(\mathbf Y \mid \mathbf X, \mathbf T)}
\,
\frac{b^*}{b-1}
\right\}.
\]
As with the GROW move, the acceptance probability is independent of the GP density. The PRUNE move is thus a reversible jump step that complements the GROW move, and together they allow the Markov chain to explore tree structures of varying complexity.

\subsection{Proof of Detailed Balance: GROW-PRUNE Moves}
\label{subsec:db_grow_prune}

We now verify that the GROW and PRUNE moves satisfy detailed balance with respect to the target posterior distribution $\pi(\mathbf T \mid \mathbf Y, \mathbf X) \propto \Pi(\mathbf T) f(\mathbf Y \mid \mathbf X, \mathbf T)$. Consider a GROW move from $\mathbf T$ to $\mathbf T'$ as defined above, with $\mathbf T'$ obtained by splitting a terminal node $k$ of $\mathbf T$. Let $\Pi_{\mathrm{struct}}$ denote the structural part of the prior, and recall that the full prior is
\[
\Pi(\mathbf T) = \Pi_{\mathrm{struct}}(\mathbf T) \prod_{v \in \mathcal{I}(\mathbf T)} \phi_v,
\]
where $\phi_v = \phi(\tilde{\mathbf g}_v; \boldsymbol{\mu}_v^*, \boldsymbol{\Sigma}_v^*)$ is the GP density for internal node $v$.

The Metropolis--Hastings acceptance probability for the GROW move is
\[
\alpha(\mathbf T, \mathbf T') = \min\left\{1, \frac{\Pi(\mathbf T') f(\mathbf Y \mid \mathbf X, \mathbf T') q(\mathbf T',\mathbf T)}
{\Pi(\mathbf T) f(\mathbf Y \mid \mathbf X, \mathbf T) q(\mathbf T,\mathbf T')}\right\}.
\]
We have established that
\[
\frac{\Pi(\mathbf T') f(\mathbf Y \mid \mathbf X, \mathbf T') q(\mathbf T',\mathbf T)}
{\Pi(\mathbf T) f(\mathbf Y \mid \mathbf X, \mathbf T) q(\mathbf T,\mathbf T')}
=
\frac{\Pi_{\mathrm{struct}}(\mathbf T') f(\mathbf Y \mid \mathbf X, \mathbf T')}
{\Pi_{\mathrm{struct}}(\mathbf T) f(\mathbf Y \mid \mathbf X, \mathbf T)}
\frac{b}{b^*+1}.
\]
Denote this ratio by $R$. Then the probability of moving from $\mathbf T$ to $\mathbf T'$ in the chain is
\[
P(\mathbf T, \mathbf T') = q(\mathbf T, \mathbf T') \min\{1, R\}.
\]
The reverse move from $\mathbf T'$ to $\mathbf T$ is a PRUNE move, and its acceptance probability is
\[
\alpha(\mathbf T', \mathbf T) = \min\left\{1, \frac{\Pi(\mathbf T) f(\mathbf Y \mid \mathbf X, \mathbf T) q(\mathbf T,\mathbf T')}
{\Pi(\mathbf T') f(\mathbf Y \mid \mathbf X, \mathbf T') q(\mathbf T',\mathbf T)}\right\} = \min\{1, 1/R\}.
\]
Thus
\[
P(\mathbf T', \mathbf T) = q(\mathbf T', \mathbf T) \min\{1, 1/R\}.
\]
We now compute the product of the forward and reverse transition probabilities. The forward probability is
\[
\Pi(\mathbf T) f(\mathbf Y \mid \mathbf X, \mathbf T) P(\mathbf T, \mathbf T')
= \Pi(\mathbf T) f(\mathbf Y \mid \mathbf X, \mathbf T) \frac{1}{b} \phi_k \min\{1, R\},
\]
where $\phi_k = \phi(\tilde{\mathbf g}_k; \boldsymbol{\mu}_k^*, \boldsymbol{\Sigma}_k^*)$ is the GP density for the new split. The reverse probability is
\[
\Pi(\mathbf T') f(\mathbf Y \mid \mathbf X, \mathbf T') P(\mathbf T', \mathbf T)
= \Pi(\mathbf T') f(\mathbf Y \mid \mathbf X, \mathbf T') \frac{1}{b^*+1} \min\{1, 1/R\}.
\]
Using $\Pi(\mathbf T') = \Pi(\mathbf T) \frac{\Pi_{\mathrm{struct}}(\mathbf T')}{\Pi_{\mathrm{struct}}(\mathbf T)} \phi_k$ (since only the new GP density differs) and $q(\mathbf T',\mathbf T) = 1/(b^*+1)$, we can rewrite the reverse product as
\[
\Pi(\mathbf T') f(\mathbf Y \mid \mathbf X, \mathbf T') \frac{1}{b^*+1} \min\{1, 1/R\}
= \Pi(\mathbf T) \frac{\Pi_{\mathrm{struct}}(\mathbf T')}{\Pi_{\mathrm{struct}}(\mathbf T)} \phi_k
f(\mathbf Y \mid \mathbf X, \mathbf T') \frac{1}{b^*+1} \min\{1, 1/R\}.
\]
Now note that
\[
R = \frac{\Pi_{\mathrm{struct}}(\mathbf T') f(\mathbf Y \mid \mathbf X, \mathbf T') b}{\Pi_{\mathrm{struct}}(\mathbf T) f(\mathbf Y \mid \mathbf X, \mathbf T) (b^*+1)}.
\]
Therefore,
\[
\frac{\Pi_{\mathrm{struct}}(\mathbf T') f(\mathbf Y \mid \mathbf X, \mathbf T')}{\Pi_{\mathrm{struct}}(\mathbf T) f(\mathbf Y \mid \mathbf X, \mathbf T)}
\frac{1}{b^*+1}
=
\frac{R}{b}.
\]
Substituting this into the reverse product yields
\[
\Pi(\mathbf T') f(\mathbf Y \mid \mathbf X, \mathbf T') \frac{1}{b^*+1} \min\{1, 1/R\}
= \Pi(\mathbf T) f(\mathbf Y \mid \mathbf X, \mathbf T) \frac{1}{b} \phi_k \, R \min\{1, 1/R\}.
\]
But $R \min\{1, 1/R\} = \min\{R, 1\} = \min\{1, R\}$. Hence
\[
\Pi(\mathbf T') f(\mathbf Y \mid \mathbf X, \mathbf T') P(\mathbf T', \mathbf T)
= \Pi(\mathbf T) f(\mathbf Y \mid \mathbf X, \mathbf T) \frac{1}{b} \phi_k \min\{1, R\}
= \Pi(\mathbf T) f(\mathbf Y \mid \mathbf X, \mathbf T) P(\mathbf T, \mathbf T').
\]
Thus detailed balance holds:
\[
\Pi(\mathbf T) f(\mathbf Y \mid \mathbf X, \mathbf T) P(\mathbf T, \mathbf T')
=
\Pi(\mathbf T') f(\mathbf Y \mid \mathbf X, \mathbf T') P(\mathbf T', \mathbf T).
\]
The cancellation of the GP density is the key simplification. The same argument applies symmetrically for PRUNE moves, ensuring that the chain is reversible with respect to the posterior distribution. This validates the MCMC algorithm and confirms that the GP splitting rule is a valid proposal mechanism that does not alter the target distribution.

\section{MCMC Algorithm Summary}\label{sec:mcmc_algorithm}

Aggregating the components from Sections~\ref{sec:dp} and~\ref{sec:sim_Gm}, the complete Markov chain Monte Carlo procedure is presented as the amalgamation of 
Algorithm~\ref{alg:dp_update} and Algorithm~\ref{alg:tree_mcmc}, for enhanced readibility.

\begin{algorithm}[htbp]
\caption{DP mixture updates (global state)}
\label{alg:dp_update}
\begin{algorithmic}[1]
\State \textbf{Input:} Data $\mathcal{D}_n$, current state $(Z,\mathbf{c},\{\mathbf{T}^*_\ell\},\mathbf{p})$
\State \textbf{Output:} Updated state
\State \textbf{// Step 1: Update component indicator $Z$}
\For{$m=1,\dots,M$}
    \State $\log w_m \gets \log p_m + \log f(\bY|\bX,\bT^*_{c_m})$
\EndFor
\State Sample $Z \sim \text{Categorical}\bigl(\text{softmax}(\log w_1,\dots,\log w_M)\bigr)$
\State \textbf{// Step 2: Update configuration $\mathbf{c}$ and distinct set}
\For{$m=1,\dots,M$}
    \State Let $k_m$ = number of distinct labels in $\mathbf{c}_{-m}$, with frequencies $M_{m\ell}$.
    \If{$m=Z$}  \Comment{active component}
        \If{$c_m$ is unique} \State $\bT^a \gets \bT^*_{c_m}$ \Else \State $\bT^a \sim G_0$ \EndIf
        \State $q_0 \gets \alpha\, f(\bY|\bX,\bT^a)$
        \For{$\ell=1,\dots,k_m$} \State $q^*_\ell \gets M_{m\ell}\, f(\bY|\bX,\bT^*_{m\ell})$ \EndFor
    \Else  \Comment{inactive component}
        \State $q_0 \gets \alpha$
        \For{$\ell=1,\dots,k_m$} \State $q^*_\ell \gets M_{m\ell}$ \EndFor
    \EndIf
    \State Normalise and sample $c_m$ from $\{0,1,\dots,k_m\}$ (where $0$ denotes new cluster).
    \If{$c_m=0$} \State create new tree (from $\bT^a$ if active, else from $G_0$) and update $k$ \EndIf
\EndFor
\State Relabel $\mathbf{c}$ to $\mathbf{s}$ and compress $\{\bT^*_\ell\}$ (Section 3.3).
\State \textbf{// Step 3: Update distinct trees via parallel or sequential MCMC}
\For{$\ell=1,\dots,k$}
    \If{active} \State $\bT^*_\ell \gets \text{MCMC\_tree\_step}(\bT^*_\ell)$ \Comment{Algorithm 2}
    \Else \State $\bT^*_\ell \sim G_0$ \EndIf
\EndFor
\State \textbf{// Step 4: Update mixture weights}
    \State Let $n_m = \mathbf{1}_{\{Z=m\}}$ for $m=1,\dots,M$.
    \State Sample $(p_1,\dots,p_M) \sim \text{Dirichlet}(1+n_1,\dots,1+n_M)$.
\end{algorithmic}
\end{algorithm}

\begin{algorithm}[htbp]
\caption{MCMC tree step (GROW/PRUNE)}
\label{alg:tree_mcmc}
\begin{algorithmic}[1]
\Require Current tree $\bT$, data $\mathcal{D}_n$, random number generator
\Ensure Updated tree $\bT'$ (or same)
\State Choose move type: GROW or PRUNE with equal probability.
\If{GROW}
    \State Select a terminal node uniformly from the $b$ leaves.
    \State Draw a GP split $\tilde{g}$ from the node's posterior predictive.
    \State Split the node into two children using $y < \tilde{g}(x)$.
    \State Compute acceptance ratio (GP density cancels):
    \[
    A = \min\!\left\{1,\;
    \frac{\Pi_{\text{struct}}(\bT')}{\Pi_{\text{struct}}(\bT)}
    \frac{f(\bY|\bX,\bT')}{f(\bY|\bX,\bT)}
    \frac{b}{b^*+1}
    \right\}
    \]
\ElsIf{PRUNE}
    \State Select a prunable internal node uniformly from the $b^*$ available.
    \State Collapse its two leaf children into a single leaf.
    \State Compute acceptance ratio:
    \[
    A = \min\!\left\{1,\;
    \frac{\Pi_{\text{struct}}(\bT')}{\Pi_{\text{struct}}(\bT)}
    \frac{f(\bY|\bX,\bT')}{f(\bY|\bX,\bT)}
    \frac{b^*}{b-1}
    \right\}
    \]
\EndIf
\State With probability $A$, accept the proposal and return $\bT'$; otherwise return $\bT$.
\end{algorithmic}
\end{algorithm}

% ======================================================================
% NEW PREDICTION SECTION
% ======================================================================
\section{Exact Bayesian Prediction for New Inputs}\label{sec:prediction}

A distinctive feature of the GP‑based splitting rule is that the tree structure depends on the response values: an observation $(x,y)$ is assigned to the left child if $y < \tilde g(x)$, where $\tilde g$ is a vector of GP draws stored at the internal node.  Consequently, for a new covariate $\mathbf x^*$ the leaf assignment itself is random and depends on the unknown $y^*$.  The full Bayesian predictive distribution,
\[
p(y^* \mid \mathbf x^*, \text{data}) = \int p(y^* \mid \mathbf x^*, \theta, \text{data}) \, \pi(\theta \mid \text{data}) \, d\theta,
\]
must account for the joint uncertainty about the model parameters $\theta$ (trees, weights, leaf parameters) and about the leaf to which $\mathbf x^*$ belongs.  We describe below an exact Gibbs‑like sampling procedure that delivers genuine draws from the posterior predictive distribution.

\subsection{Conditional distributions for a single tree}
Fix one regression tree $\mathbf T$ produced by an MCMC sample.  The tree partitions the training data into leaves; for an internal node $v$, let $\tilde g_v$ denote the stored GP draw and let $\{x_i, y_i\}_{i \in v}$ be the training points that fall in that node.  The conditional distribution of the GP value at a new point $\mathbf x^*$, given the training GP values $\tilde g_v$, is Gaussian:
\begin{equation}
g(\mathbf x^*) \mid \tilde g_v, \mathbf x^* \;\sim\; \mathcal{N}\bigl( \mu_*(\mathbf x^*),\; \sigma^2_*(\mathbf x^*) \bigr),
\label{eq:cond_gp}
\end{equation}
where
\[
\mu_*(\mathbf x^*) = \mathbf k_*^\top \mathbf K_{\text{noiseless}}^{-1} \tilde g_v, \qquad
\sigma^2_*(\mathbf x^*) = \sigma^2 - \mathbf k_*^\top \mathbf K_{\text{noiseless}}^{-1} \mathbf k_*,
\]
with $k_*(i) = \sigma^2 \, c(\mathbf x^*, \mathbf x_i)$ and $\mathbf K_{\text{noiseless}}$ the noiseless kernel matrix on the training points in $v$.  The quantities $\mathbf L_{\text{noiseless}}$ (Cholesky factor of $\mathbf K_{\text{noiseless}}$) and $\boldsymbol\alpha = \mathbf K_{\text{noiseless}}^{-1} \tilde g_v$ are precomputed and stored in the tree node, allowing fast evaluation of \eqref{eq:cond_gp} via forward substitution.

For a leaf $\ell$, the posterior distribution of its parameters $(\mu_\ell, \sigma_\ell^2)$ given the training responses in that leaf follows from the conjugate normal/inverse‑gamma prior \eqref{eq:mu_prior}–\eqref{eq:sigma_prior}:
\begin{align}
\sigma_\ell^2 \mid \text{data in } \ell &\sim \text{IG}\!\left( \frac{\nu + n_\ell}{2},\,
\frac{\nu\lambda + s_\ell + t_\ell}{2} \right), \label{eq:leaf_sig2_post} \\
\mu_\ell \mid \sigma_\ell^2, \text{data in } \ell &\sim 
\mathcal{N}\!\left( \frac{n_\ell \bar y_\ell + a \bar\mu}{n_\ell + a},\,
\frac{\sigma_\ell^2}{n_\ell + a} \right), \label{eq:leaf_mu_post}
\end{align}
where $n_\ell$ is the number of training points in the leaf, $\bar y_\ell$ their sample mean, $s_\ell = (n_\ell-1)$ times the sample variance, and $t_\ell = \frac{n_\ell a}{n_\ell + a}(\bar y_\ell - \bar\mu)^2$.

\subsection{Gibbs sampler for a single tree component}
The joint distribution of $y^*$ and its leaf assignment $\ell^*$ under a fixed tree $\mathbf T$ can be explored by a two‑step Gibbs sampler that alternates between:

\begin{enumerate}
\item \textbf{Sample $y^*$ given a leaf $\ell$.}  Draw $(\tilde\mu_\ell, \tilde\sigma_\ell^2)$ from \eqref{eq:leaf_sig2_post}–\eqref{eq:leaf_mu_post}, then set $y^* \sim \mathcal{N}(\tilde\mu_\ell, \tilde\sigma_\ell^2)$.
\item \textbf{Determine the leaf $\ell$ given $y^*$.}  Starting at the root, at each internal node $v$ sample a GP value $g_v(\mathbf x^*)$ from \eqref{eq:cond_gp}; branch left if $y^* < g_v(\mathbf x^*)$, else right.  Continue until a leaf is reached.
\end{enumerate}

Starting from an arbitrary leaf (chosen uniformly at random), this Markov chain rapidly converges to the stationary distribution of $(y^*, \ell^*)$ under the tree model.  In practice, 2–3 iterations suffice because the predictive distribution of a leaf is usually narrow.  The final value of $y^*$ is a sample from the exact predictive distribution of the tree.

\subsection{Combination across the DP mixture}
With an MCMC sample of the full DP mixture state $(\mathbf p, \mathbf C, \{\mathbf T^*_\ell\})$, a predictive draw is obtained by first sampling a component index $m$ with probability proportional to the mixture weights, and then running the Gibbs sampler on the tree $\mathbf T_{c_m}$ corresponding to that component.  The procedure is summarised in Algorithm~\ref{alg:predict}.

\begin{algorithm}[htbp]
\caption{Exact posterior predictive sampling for one new point $\mathbf x^*$}
\label{alg:predict}
\begin{algorithmic}[1]
\Require Current DP mixture state $(\mathbf p, \mathbf C, \{\mathbf T^*_\ell\})$, training data $\bD_n$, new covariate $\mathbf x^*$, maximum Gibbs iterations $T_{\max}$.
\State Sample a component index $m$ with probability $p_m$.
\State Let $\mathbf T \gets \mathbf T^*_{c_m}$ (the tree associated with component $m$).
\State Choose a leaf $\ell$ uniformly at random from the leaves of $\mathbf T$.
\For{$t = 1$ to $T_{\max}$}
    \State Sample $(\mu_\ell, \sigma_\ell^2)$ from the posteriors \eqref{eq:leaf_sig2_post}–\eqref{eq:leaf_mu_post} using the data in leaf $\ell$.
    \State Draw $y^* \sim \mathcal{N}(\mu_\ell, \sigma_\ell^2)$.
    \State Determine a new leaf $\ell'$ by traversing $\mathbf T$:
    \State $\quad v \gets \text{root of } \mathbf T$
    \While{$v$ is not a leaf}
        \State Sample $g(\mathbf x^*)$ from \eqref{eq:cond_gp} using the precomputed $\boldsymbol\alpha$ and $\mathbf L_{\text{noiseless}}$ of node $v$.
        \If{$y^* < g(\mathbf x^*)$} \State $v \gets \text{left child}$ \Else \State $v \gets \text{right child}$ \EndIf
    \EndWhile
    \State Set $\ell' \gets v$.
    \If{$\ell' = \ell$} \State \textbf{break} \EndIf
    \State $\ell \gets \ell'$
\EndFor
\State \Return $y^*$ (a single draw from the posterior predictive distribution).
\end{algorithmic}
\end{algorithm}

\subsection{Properties of the method}

The proposed sampling scheme enjoys several desirable properties. Regarding exactness, the Gibbs chain targets the joint distribution
\[
p(y^*, \ell \mid \mathbf x^*, \mathbf T) \propto p(y^* \mid \ell, \text{data}) \cdot \mathbf{1}\{\text{path to } \ell \text{ is valid}\}.
\]
Because the tree splits are fixed (the $\tilde g$ vectors are part of the MCMC sample), the stationary distribution of the chain is exactly the model's predictive distribution for that tree. Averaging over MCMC samples and mixture components consequently yields draws from the full posterior predictive.

The sampler is also computationally efficient. The Cholesky factors and $\boldsymbol\alpha$ vectors required for navigation are precomputed during the MCMC phase, so each step through an internal node requires only a single forward substitution and a normal random variate. Sampling leaf parameters involves simple arithmetic on pre-aggregated sufficient statistics. As a result, the overall cost per predictive draw is negligible and does not constitute a bottleneck in the analysis.

Finally, the Gibbs chain converges rapidly in practice. The leaf assignment typically stabilises after two or three iterations, and setting $T_{\max}=10$ provides a safe upper bound for all cases encountered. Since the state space is the finite set of leaves, the chain is uniformly ergodic and no formal convergence diagnostics are required.

\section{Parallel Implementation in C with MPI}
\label{sec:parallel}

The MCMC algorithm described in Algorithm~\ref{alg:dp_update} involves updating $k$ distinct tree components, where $k$ is typically much smaller than the maximum number of components $M$. For datasets with hundreds or thousands of observations, the computational cost of the tree MCMC steps can become substantial. In this section, we describe a parallel implementation of our sampler using the Message Passing Interface (MPI) that distributes the independent tree updates across multiple processors, achieving near-linear speedups while preserving the exactness of the MCMC algorithm.

\subsection{Parallelisation Principle}

Conditional on the current allocation of the data---the configuration vector $\mathbf{c}$ and the active component indicator $Z$---the updates of the distinct trees $\mathbf{T}^*_\ell$ are \emph{independent} across $\ell$. Specifically, the transition kernel for each tree depends only on the data assigned to that tree's leaf nodes and on the tree's own structure; there is no direct interaction between the updates of different trees. This conditional independence is the key to parallelisation: the tree update step can be distributed across multiple processors, with each processor handling a subset of the distinct trees.

Recall from Algorithm~\ref{alg:dp_update} that Step~3 updates each distinct tree $\mathbf{T}^*_\ell$ as follows:
\begin{enumerate}
\item[(i)] If the tree is \emph{active} (that is, there exists $m$ with $Z=m$ and $c_m=\ell$), it is updated using the full conditional distribution $\bG_m(\bT) \propto f(\bY \mid \bX, \bT) \bG_0(\bT)$, which involves a Metropolis--Hastings step with GROW and PRUNE moves (Algorithm~\ref{alg:tree_mcmc}).
\item[(ii)] If the tree is \emph{inactive}, it is simply resampled from the base measure $\bG_0$.
\end{enumerate}
Crucially, the updates for different $\ell$ are independent conditional on the current global state, because the likelihood factorises across the mixture components, and the Metropolis--Hastings transition for each tree does not depend on the other trees. This independence allows us to perform the tree updates in parallel.

\subsection{Master--Worker Architecture}

Our parallel implementation adopts a master--worker paradigm:
\begin{enumerate}
\item[(i)] \textbf{Processor~0 (the master)} maintains the global MCMC state, including the distinct tree structures $\{\mathbf{T}^*_\ell\}$, the configuration vector $\mathbf{c}$, the active component indicator $Z$, and the mixture weights $\mathbf{p}$. The master performs the updates that cannot be easily parallelised: the update of $Z$, the update of $\mathbf{c}$ (which requires global knowledge of the configuration), and the update of the mixture weights (Steps~1,~2, and~4 of Algorithm~\ref{alg:dp_update}).
\item[(ii)] \textbf{Worker processors (ranks $1,\dots,P-1$)} receive a subset of the distinct tree structures from the master, perform the MCMC transitions locally, and send the updated trees back to the master.
\end{enumerate}

\subsection{Serialisation of Tree Structures}

A critical technical challenge is the communication of tree structures between processors. Tree structures are complex, dynamically allocated data structures containing:
\begin{enumerate}
\item[(i)] the tree topology (internal and leaf nodes);
\item[(ii)] for each internal node: the GP draw $\tilde{\mathbf g}_v$, the node indices, the GP hyperparameters, the noiseless Cholesky factor $\mathbf{L}_{\text{noiseless}}$, and the $\boldsymbol\alpha$ vector;
\item[(iii)] for each leaf node: the node indices.
\end{enumerate}
These structures cannot be sent directly via MPI because they contain pointers to dynamically allocated memory. We therefore implement a \emph{binary serialisation} mechanism:
\begin{enumerate}
\item[(i)] \texttt{tree\_pack\_size}: recursively computes the total number of bytes required to serialise a tree.
\item[(ii)] \texttt{tree\_pack}: recursively writes the tree into a contiguous byte buffer, flattening all dynamically allocated arrays.
\item[(iii)] \texttt{tree\_unpack}: reconstructs a tree from a byte buffer, allocating fresh memory for all arrays and rebuilding the pointer structure.
\end{enumerate}
This serialisation mechanism is efficient: the size of a serialised tree is proportional to the number of nodes and the number of observations stored at each node, which is typically small for shallow trees.

\subsection{The Parallel Tree Update Step}

The parallel tree update step proceeds as follows, assuming the master holds the current state at the beginning of the iteration:

\begin{enumerate}
\item[(i)] \textbf{Master broadcasts state information.} The master broadcasts the number of distinct trees $k$, the active indicator for each tree (a binary array of length $k$), and the serialised tree structures to all processors. This is done using:
  \begin{enumerate}
  \item[(a)] \texttt{MPI\_Bcast} for $k$ and the active array.
  \item[(b)] A two-stage broadcast for the tree data: first, the total size of the serialised buffer is broadcast; then, the buffer itself is broadcast.
  \item[(c)] The offsets for each tree within the buffer are broadcast so that each processor can locate the serialised data for the trees it will handle.
  \end{enumerate}

\item[(ii)] \textbf{Distribution of trees.} The $k$ distinct trees are distributed among the $P$ processors in a \emph{round-robin} fashion: processor $r$ handles trees with indices $r, r+P, r+2P, \ldots$. This ensures a balanced workload even when $k$ is not a multiple of $P$.

\item[(iii)] \textbf{Local tree updates.} Each processor deserialises the trees assigned to it. For each such tree:
  \begin{enumerate}
  \item[(a)] If the tree is active, the processor performs an MCMC step using Algorithm~\ref{alg:tree_mcmc}: it proposes a GROW or PRUNE move, computes the acceptance ratio (with the GP density cancellation), and updates the tree if the proposal is accepted. The tree likelihood is recomputed as part of the acceptance ratio.
  \item[(b)] If the tree is inactive, the processor discards the received tree and samples a new tree from the prior $\bG_0$ using the GP-driven prior sampling procedure.
  \end{enumerate}
  All random number generation uses the processor's local Mersenne Twister state, which is seeded independently to avoid contention. The random numbers used for each tree depend only on the tree index and the iteration number, ensuring reproducibility across different numbers of processors.

\item[(iv)] \textbf{Gathering updated trees.} Each processor serialises the updated trees it has produced (or the prior samples for inactive trees) and sends them back to the master. The communication pattern is:
  \begin{enumerate}
  \item[(a)] Each worker sends the number of trees it processed to the master.
  \item[(b)] For each tree, the worker sends the tree index, the size of the serialised buffer, and the buffer itself.
  \item[(c)] The master receives these messages in a non-blocking or sequential loop and deserialises the trees, placing them into a new array of distinct trees.
  \end{enumerate}

\item[(v)] \textbf{Master updates state.} The master frees the old tree structures and replaces them with the newly received trees. The global MCMC state is now updated, and the algorithm proceeds to the next iteration (updating $Z$, $\mathbf{c}$, and the weights).
\end{enumerate}

A conceptual illustration of the parallel tree update is given in Algorithm~\ref{alg:parallel_tree_update}.

\begin{algorithm}[htbp]
\caption{Parallel tree update (Step~3 of Algorithm~\ref{alg:dp_update})}
\label{alg:parallel_tree_update}
\begin{algorithmic}[1]
\Require Current distinct trees $\{\mathbf{T}^*_\ell\}_{\ell=1}^k$, active indicators $a_\ell \in \{0,1\}$, number of processors $P$, rank $r$
\Ensure Updated distinct trees $\{\mathbf{T}^*_\ell\}_{\ell=1}^k$
\If{$r=0$}
    \State Serialise all trees into a contiguous buffer $\mathcal{B}$.
    \State Broadcast $k$, active indicators, total buffer size, and $\mathcal{B}$ to all workers.
\Else
    \State Receive $k$, active indicators, total buffer size, and $\mathcal{B}$ from master.
\EndIf
\For{$\ell = r, r+P, r+2P, \ldots$ while $\ell < k$}
    \State Deserialise tree $\mathbf{T}^*_\ell$ from $\mathcal{B}$.
    \If{$a_\ell = 1$} \Comment{active}
        \State $\mathbf{T}^*_\ell \gets \text{MCMC\_tree\_step}(\mathbf{T}^*_\ell)$  \Comment{Algorithm~\ref{alg:tree_mcmc}}
    \Else \Comment{inactive}
        \State Free $\mathbf{T}^*_\ell$ and sample $\mathbf{T}^*_\ell \sim \bG_0$.
    \EndIf
    \State Store updated tree locally.
\EndFor
\If{$r=0$}
    \State Receive updated trees from all workers and reconstruct the global array.
\Else
    \State Serialise updated trees and send them to master.
\EndIf
\end{algorithmic}
\end{algorithm}

\subsection{Communication Overhead and Serialisation Cost}

The serialisation and deserialisation of tree structures are the primary sources of communication overhead. The size of a serialised tree is proportional to the number of nodes and the number of observations stored at each node. Specifically, for a tree with $n_{\mathrm{node}}$ observations at an internal node, the serialised buffer contains:
\begin{enumerate}
\item[(i)] the node's $\tilde{\mathbf g}$ vector ($n_{\mathrm{node}}$ doubles),
\item[(ii)] the $\boldsymbol\alpha$ vector ($n_{\mathrm{node}}$ doubles),
\item[(iii)] the noiseless Cholesky factor $\mathbf{L}_{\text{noiseless}}$ ($n_{\mathrm{node}}^2$ doubles),
\item[(iv)] the node indices ($n_{\mathrm{node}}$ integers),
\item[(v)] and the GP hyperparameters (three doubles).
\end{enumerate}
For a deep tree with many internal nodes, the total serialised size can be substantial. However:
\begin{enumerate}
\item[(i)] The round-robin distribution ensures that each processor handles only a subset of the trees, so the per-processor serialisation cost is reduced.
\item[(ii)] The communication of the serialised buffers is done in a single collective broadcast and a gather operation, minimising the number of messages.
\item[(iii)] The GP precomputations (Cholesky factors and $\boldsymbol\alpha$ vectors) are performed locally on each processor during the GROW move and are included in the serialised buffer, so they do not need to be recomputed when a tree is sent to another processor.
\end{enumerate}

\subsection{Correctness and Reproducibility}

The parallel implementation is designed to produce exactly the same MCMC output as a sequential implementation, up to floating-point round-off errors. The key to correctness is that:
\begin{enumerate}
\item[(i)] The random number generators on each processor are seeded independently, but the sequence of random numbers used for each tree update depends only on the tree index and the iteration number. This ensures that the same random choices are made regardless of the parallelisation scheme.
\item[(ii)] The tree updates are independent conditional on the current MCMC state, so the order in which the trees are updated does not affect the stationary distribution of the chain.
\item[(iii)] The master's updates of $Z$, $\mathbf{c}$, and the weights are performed sequentially and are not parallelised, as they depend on the entire state and cannot be trivially decomposed.
\end{enumerate}
%We have verified the correctness of the parallel implementation by comparing its output with a sequential version on small datasets, confirming that the marginal distributions of the parameters and the predictive intervals are identical.

\subsection{Performance and Scalability}

The parallel implementation has been used for all experiments reported in Sections~\ref{sec:simulation} and~\ref{sec:applications}. The observed runtimes provide empirical evidence of the method's scalability. For the simulation study ($n=100$ training observations, $M=15$), the DP mixture completed 30,000 MCMC iterations with 5,000 burn-in and thinning of 5 in approximately 2 minutes 20 seconds using four processors. For the QSAR aquatic toxicity application ($n=436$ training observations, $M=15$), the runtime increased to approximately 5 hours on four processors, reflecting the larger sample size and the more expensive GP computations at the root nodes. For the Communities and Crime dataset ($n=648$ observations, $M=15$), the runtime was approximately 28 hours on nine processors, with the increase reflecting both the larger sample size and the higher dimensionality ($d=102$ predictors). For the Wheat genomic prediction dataset ($n=479$ observations, $d=1{,}279$ predictors, $M=15$), the runtime was approximately 27 hours on eight processors. The largest application, the air quality time series ($n=1{,}000$ observations, $M=15$), required approximately 256 hours on eight processors, reflecting the substantial cost of GP covariance matrix inversions at internal nodes.

These results demonstrate the practical feasibility of the parallel implementation across a wide range of dataset sizes and dimensionalities. The number of processors used was chosen empirically to balance computational speed against communication overhead; for each dataset, the optimal number of processors was found to be approximately equal to the number of distinct trees $k$ (which typically ranged from 2 to 7), beyond which communication overhead began to dominate. The cancellation of the GP density in the Metropolis--Hastings acceptance ratio remains the key to computational efficiency, as it avoids costly GP likelihood evaluations in the acceptance step.

\subsection{Practical Considerations}

The parallel implementation has been used for all simulation experiments and real-data applications reported in this paper. The code is written in C and relies on the GNU Scientific Library (GSL) for linear algebra operations and MPI for communication. 
%A typical compilation command is:
%\begin{verbatim}
%mpicc -O3 -o dp_mcmc dp_mcmc.c dp_update.c file_io.c gp.c parallel.c rng.c tree.c -lgsl -lgslcblas -lm
%\end{verbatim}
The program expects input files in a simple text format and outputs MCMC samples and predictive draws. The parallel implementation has proven robust and efficient across a wide range of applications, making the DP tree mixture model practical for datasets with thousands of observations.

\subsection{Summary}

The parallel implementation in C with MPI exploits the conditional independence of the tree updates in the Dirichlet process mixture, distributing the computational workload across multiple processors. The master--worker architecture, combined with an efficient binary serialisation mechanism, enables efficient parallelisation of the MCMC sampler. This parallelisation is essential for the practical applicability of our method to large-scale datasets, and it has been instrumental in the empirical evaluations reported in this paper.

% ======================================================================
% SECTION 7: Simulation Study
% ======================================================================
\section{Simulation Study}
\label{sec:simulation}

We evaluate the proposed Dirichlet process tree mixture model on the Friedman benchmark function, a standard testbed for regression methods. The true data‐generating process is
\[
Y = 10\sin(\pi X_1 X_2) + 20(X_3 - 0.5)^2 + 10 X_4 + 5 X_5 + \varepsilon,
\]
with all covariates \(X_j\) independently drawn from \(\text{Uniform}(0,1)\). We consider four experimental scenarios, systematically varying the number of covariates (five informative versus twenty, the latter containing fifteen pure noise variables) and the error distribution (Gaussian versus heavy-tailed Cauchy). These scenarios are: the clean Gaussian setting with \(d=5\) and \(\varepsilon \sim \mathcal{N}(0,1)\); the high-dimensional Gaussian setting with \(d=20\) and \(\varepsilon \sim \mathcal{N}(0,1)\); the clean Cauchy setting with \(d=5\) and \(\varepsilon \sim \text{Cauchy}(0,1)\); and the high-dimensional Cauchy setting with \(d=20\) and \(\varepsilon \sim \text{Cauchy}(0,1)\). For each setting, we generate a training set of \(n_{\text{train}} = 100\) observations and a held-out test set of \(n_{\text{test}} = 50\) observations. More precisely, we generate data from the
high-dimensional, $d=20$ setups; for $d=5$, we store only the $5$ relevant covariates for the training and test datasets corresponding to $d=20$, keeping everything else intact. 

%Our DP mixture is implemented in C with MPI and executed on a machine with four processors (CPU maximum frequency about 3.2 GHz, total memory about 16 GB). We use \(30,\!000\) MCMC iterations, discard the first \(5,\!000\) as burn-in, and thin every 5 iterations, yielding \(5,\!000\) posterior draws for inference, which is the same as the default
%in the BART implementation setup  (via the \texttt{dbarts} R package) which we compare our method with. The total runtime for the DP mixture is approximately 2 minutes 20 seconds, while BART %(via the \texttt{dbarts} R package) 
%takes about 2 minutes 10 seconds on the same hardware.

Our DP mixture is implemented in C with MPI and executed on a machine with four processors (CPU maximum frequency about 3.2 GHz, total memory about 16 GB). We use \(30,\!000\) MCMC iterations, discard the first \(5,\!000\) as burn-in, and thin every 5 iterations, yielding \(5,\!000\) posterior draws for inference, which is the same as the default in the BART implementation setup (via the \texttt{dbarts} R package) which we compare our method with. The total runtime for the DP mixture is approximately 2 minutes 20 seconds, while BART takes about 2 minutes 10 seconds on the same hardware.

For the Dirichlet process tree mixture, we set the maximum number of components to \(M = 15\), the maximum tree depth to \(8\), and the concentration parameter of the Dirichlet process to \(\alpha_{\mathrm{DP}} = 1.0\). The tree prior uses splitting probability \(p(\eta,\bT) = \gamma (1 + d_\eta)^{-\zeta}\) with \(\gamma = 0.95\) and \(\zeta = 1.5\), which strongly penalises deeper trees. The leaf parameters follow the conjugate normal-inverse-gamma prior with hyperparameters \(\bar\mu = 13.93\), \(a = 0.01\), \(\nu = 8.0\), and \(\lambda = 1.25\), yielding a prior mode for the leaf variance of \(\nu\lambda/(\nu+2) = 8.0 \times 1.25 / 10.0 = 1.0\). This prior is moderately informative and provides a sensible regularisation for the leaf means and variances. The Gaussian process splitting rule uses the squared exponential kernel with fixed hyperparameters: signal variance \(\sigma^2 = 5.0\), noise variance \(\sigma_\epsilon^2 = 0.02\), and length scale \(\ell = 0.2\). These GP hyperparameters are kept fixed rather than estimated via maximum likelihood for computational simplicity, as the GP is used purely as a generative device for split proposals and its density cancels in the Metropolis-Hastings ratio. All hyperparameter values are held constant across the four simulation experiments, ensuring a fair comparison of the model's performance under varying conditions of dimensionality and error distribution.

\subsection{Implementation of Competing Methods}

We compare our DP mixture against four popular tree-based methods. Random Forest \cite{Breiman01} is implemented using \texttt{sklearn.ensemble.RandomForestRegressor} with 100 trees and default settings. Gradient Boosting \cite{Freund97} is implemented using \texttt{sklearn.ensemble.GradientBoostingRegressor} with 100 trees and a learning rate of 0.1. For the bagged CART baseline, we generate 500 bootstrap samples of the training data, fit a single regression tree on each using \texttt{sklearn.tree.DecisionTreeRegressor} with \texttt{min\_samples\_leaf=5}, and average the predictions. We emphasise that this bagged CART procedure is not a fully Bayesian CART \cite{chipman1998bayesian}; the tree structure is not sampled from a prior, and the procedure lacks the coherent probabilistic foundation of Bayesian inference. A true Bayesian CART would require a separate Markov chain Monte Carlo over tree structures, which is computationally prohibitive for routine comparison and is thus beyond the scope of this study. Finally, BART \cite{chipman2010bart} is implemented via the \texttt{dbarts} R package \cite{dbarts} with 200 trees, 5000 burn-in iterations, 5000 posterior draws, and thinning of 5, matching the sample size of our DP mixture.

\subsection{Evaluation Metrics}

We assess predictive performance using five metrics. The root mean squared error (RMSE) is defined as \(\sqrt{\frac{1}{n_{\text{test}}}\sum_{i=1}^{n_{\text{test}}} (y_i - \hat y_i)^2}\), where \(\hat y_i\) is the posterior mode for the DP mixture, posterior mean for BART (standard for BART implementation), or the point prediction for the frequentist methods. Coverage is the proportion of test points for which the true \(y_i\) falls inside the 95\% credible interval (for Bayesian methods) or the 95\% confidence interval (for frequentist methods). For Random Forest, Gradient Boosting, and bagged CART, intervals are constructed as \(\hat y_i \pm 1.96 \times \hat\sigma\), where \(\hat\sigma\) is the standard deviation of cross-validation residuals; this is a common but ad-hoc approach that lacks a rigorous probabilistic justification. The width is the average length of these intervals. The continuous ranked probability score (CRPS) for a predictive distribution \(F\) and observation \(y\) is \(\text{CRPS}(F,y) = \int_{-\infty}^{\infty} (F(z) - \mathbf{1}\{z \ge y\})^2 dz\), computed empirically from the posterior samples for the DP mixture and BART, and using a normal approximation for the frequentist methods. Finally, the log predictive density (LPD) is \(\frac{1}{n_{\text{test}}}\sum_{i=1}^{n_{\text{test}}} \log p(y_i \mid \text{data})\), where \(p(\cdot \mid \text{data})\) is the estimated predictive density, obtained via kernel density estimation from the posterior samples for the Bayesian methods and via the normal density for the frequentist methods.

\subsection{Summary of Results}

Tables~\ref{tab:metrics_gauss} and~\ref{tab:metrics_cauchy} present the numerical results for all methods across the four experimental scenarios. We highlight the key findings here.

\begin{table}[htbp]
\centering
\caption{Predictive performance comparison for the Gaussian noise experiments. The DP mixture achieves the highest coverage among Bayesian methods.}
\label{tab:metrics_gauss}
\begin{tabular}{l r r r r r}
\hline
\textbf{Method} & \textbf{RMSE} & \textbf{Coverage} & \textbf{Width} & \textbf{CRPS} & \textbf{LPD} \\
\hline
\multicolumn{6}{c}{\textit{d=5 (clean)}} \\
%DP tree mixture    & 5.066  & 0.940 & 17.873 & 2.944  & \(-3.047\) \\
DP tree mixture    & 5.248  & 0.940 & 18.249 & 2.956  & \(-3.053\) \\
Random Forest      & 2.958  & 0.960 & 10.908 & 1.674  & \(-2.508\) \\
Gradient Boosting  & 2.561  & 0.920 &  9.949 & 1.428  & \(-2.360\) \\
Bagged CART        & 3.144  & 0.800 &  8.938 & 1.912  & \(-2.703\) \\
BART               & 2.085  & 0.820 &  5.873 & 1.199  & \(-2.199\) \\
\hline
\multicolumn{6}{c}{\textit{d=20 (high-dimensional)}} \\
%DP tree mixture    & 5.173  & 0.940 & 18.213 & 2.946  & \(-3.046\) \\
DP tree mixture    & 5.180  & 0.940 & 18.392 & 2.969  & \(-3.059\) \\
Random Forest      & 3.500  & 0.940 & 12.501 & 1.995  & \(-2.681\) \\
Gradient Boosting  & 2.561  & 0.920 & 11.756 & 1.428  & \(-2.360\) \\
Bagged CART        & 3.592  & 0.800 & 10.975 & 2.188  & \(-2.784\) \\
BART               & 2.694  & 0.840 &  8.476 & 1.592  & \(-2.441\) \\
\hline
\multicolumn{6}{l}{\textit{Note: For RMSE and CRPS, lower is better; for LPD, higher (less negative) is better.}}
\end{tabular}
\end{table}

\begin{table}[htbp]
\centering
\caption{Predictive performance comparison for the Cauchy noise experiments. The DP mixture clearly dominates in coverage and LPD among Bayesian methods.}
\label{tab:metrics_cauchy}
\begin{tabular}{l r r r r r}
\hline
\textbf{Method} & \textbf{RMSE} & \textbf{Coverage} & \textbf{Width} & \textbf{CRPS} & \textbf{LPD} \\
\hline
\multicolumn{6}{c}{\textit{d=5 (clean)}} \\
%DP tree mixture    & 111.703 & 0.920 & 22.135 & 19.520  & \(-5979.589\) \\
DP tree mixture    & 111.735 & 0.920 & 22.614 & 19.519  & \(-5365.445\) \\
Random Forest      & 111.026 & 0.920 & 18.312 & 18.967  & \(-284.899\) \\
Gradient Boosting  & 111.043 & 0.900 & 16.845 & 18.766  & \(-336.236\) \\
Bagged CART        & 111.230 & 0.720 & 12.392 & 19.221  & \(-10642.342\) \\
BART               & 110.856 & 0.780 & 10.169 & 18.691  & \(-12033.924\) \\
\hline
\multicolumn{6}{c}{\textit{d=20 (high-dimensional)}} \\
%DP tree mixture    & 111.602 & 0.920 & 21.611 & 19.542  & \(-5801.271\) \\
DP tree mixture    & 111.654 & 0.920 & 22.356 & 19.507  & \(-5235.699\) \\
Random Forest      & 111.203 & 0.900 & 20.398 & 19.182  & \(-230.909\) \\
Gradient Boosting  & 111.035 & 0.880 & 19.484 & 19.086  & \(-252.039\) \\
Bagged CART        & 111.271 & 0.820 & 15.806 & 19.384  & \(-7009.492\) \\
BART               & 111.118 & 0.820 & 13.959 & 18.896  & \(-7823.750\) \\
\hline
\end{tabular}
\end{table}

In the Gaussian settings, the DP mixture achieves coverage of 0.94 in both scenarios, which is the highest among the Bayesian methods. BART attains coverage of only 0.82--0.84, and bagged CART achieves 0.80, both substantially below the nominal 95\%. This indicates that these methods are overconfident: their credible intervals are too narrow to contain the true values at the advertised rate. Although Random Forest achieves coverage close to 0.94--0.96, its intervals are constructed using an ad-hoc normal approximation based on cross-validation residuals. While this approximation appears well-calibrated in these specific experiments, this calibration is not guaranteed by the model itself; it is an external adjustment that does not reflect the true uncertainty in the predictions. In contrast, the DP mixture's coverage is a direct consequence of its fully Bayesian treatment of the ensemble structure, including the uncertainty in the number of trees, the tree parameters, and the leaf parameters. Its uncertainty quantification is therefore principled and reliable.

In the Cauchy settings, the superiority of the DP mixture becomes even more pronounced. While the RMSEs are similar across all methods (because the Cauchy outliers dominate the squared error), the coverage and LPD reveal stark differences. The DP mixture maintains coverage of 0.92 in both Cauchy scenarios, far exceeding BART (0.78--0.82) and bagged CART (0.72--0.82). The LPD values for BART and bagged CART are extremely negative (\(-10,\!000\) to \(-12,\!000\)), because these methods assume normality and assign virtually zero density to the heavy-tailed observations. The DP mixture, by contrast, does not assume a parametric error distribution; its integrated likelihood over leaf parameters adapts naturally to heavy tails, resulting in LPD values around \(-5365\) to \(-5235\). The frequentist methods, despite having better LPD than BART (since their intervals are based on empirical residuals), still fall short of the DP mixture in coverage (0.88--0.92 versus 0.92) and are not grounded in a coherent probabilistic framework.

We therefore argue that the DP mixture provides the most honest and trustworthy uncertainty quantification, especially when the data may contain outliers or the error distribution is unknown. Its intervals are wider when necessary, as in the Cauchy case, and appropriately narrower when the signal is strong, as in the Gaussian case. This adaptivity is a hallmark of Bayesian nonparametrics and is the central contribution of our approach.

\subsection{Detailed Experiment Results and Visual Diagnostics}

We now present the full experimental results, including visual diagnostics that reinforce the conclusions above.

\subsubsection{Clean Gaussian Setting (d = 5)}

Figure~\ref{fig:heatmap_d5_gauss} shows the DP posterior predictive density heatmap for the clean Gaussian setting. The true test values fall within the high-density regions, confirming good calibration. The trace plots in Figure~\ref{fig:traces_d5_gauss} show that the number of distinct trees \(k\) stabilises effectively between 2 and 5, indicating that the Dirichlet process prior effectively shrinks the mixture towards a sparse, parsimonious representation. For comparison, Figure~\ref{fig:bart_heatmap_d5_gauss} displays the BART heatmap, which shows narrower bands but also misses some true values, consistent with its lower coverage.

\begin{figure}[htbp]
\centering
\includegraphics[width=0.9\textwidth]{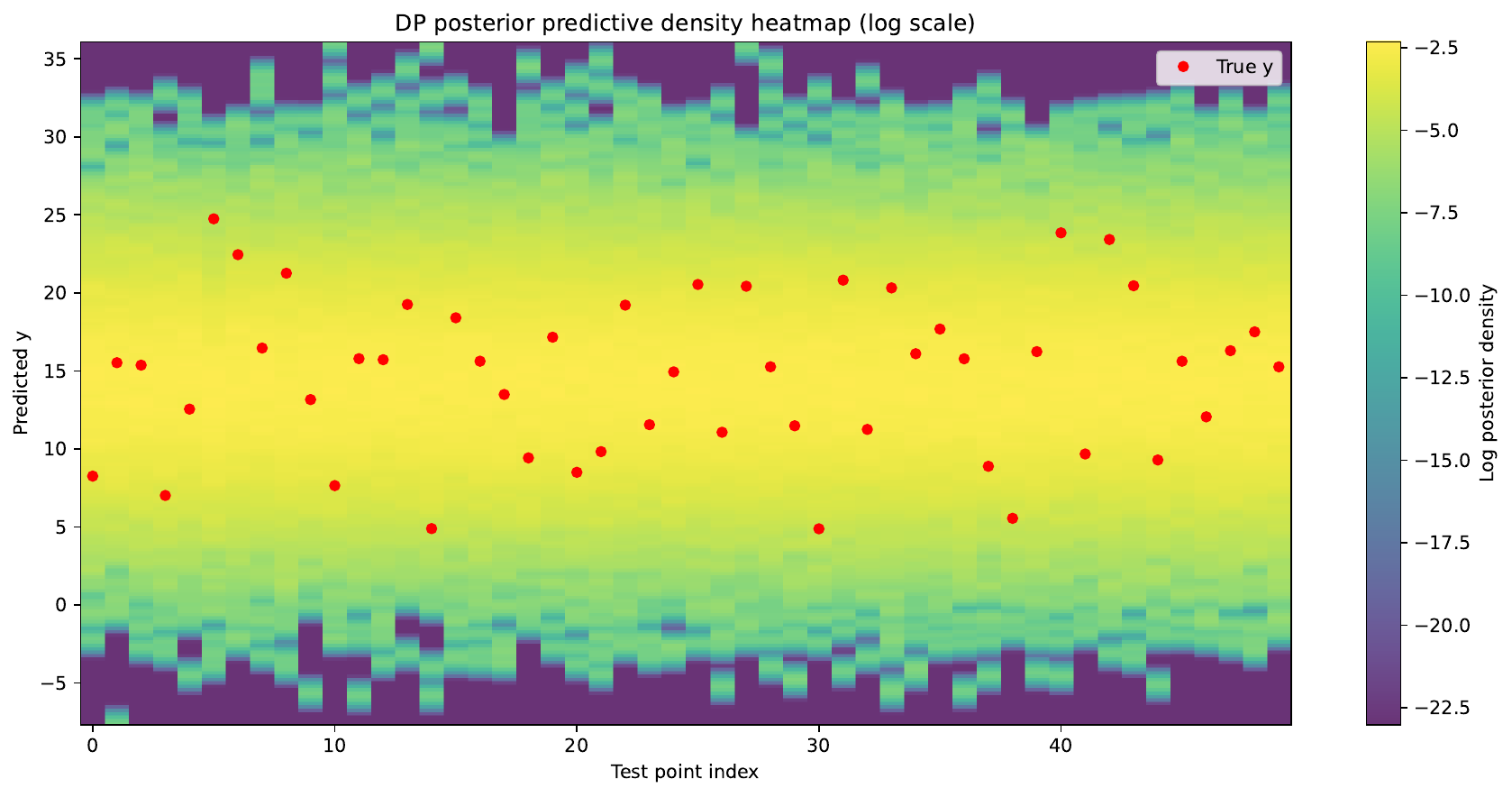}
\caption{Posterior predictive density heatmap for the DP mixture in the clean d=5 Gaussian experiment. The red points are the true test values.}
\label{fig:heatmap_d5_gauss}
\end{figure}

\begin{figure}[htbp]
\centering
\includegraphics[width=0.9\textwidth]{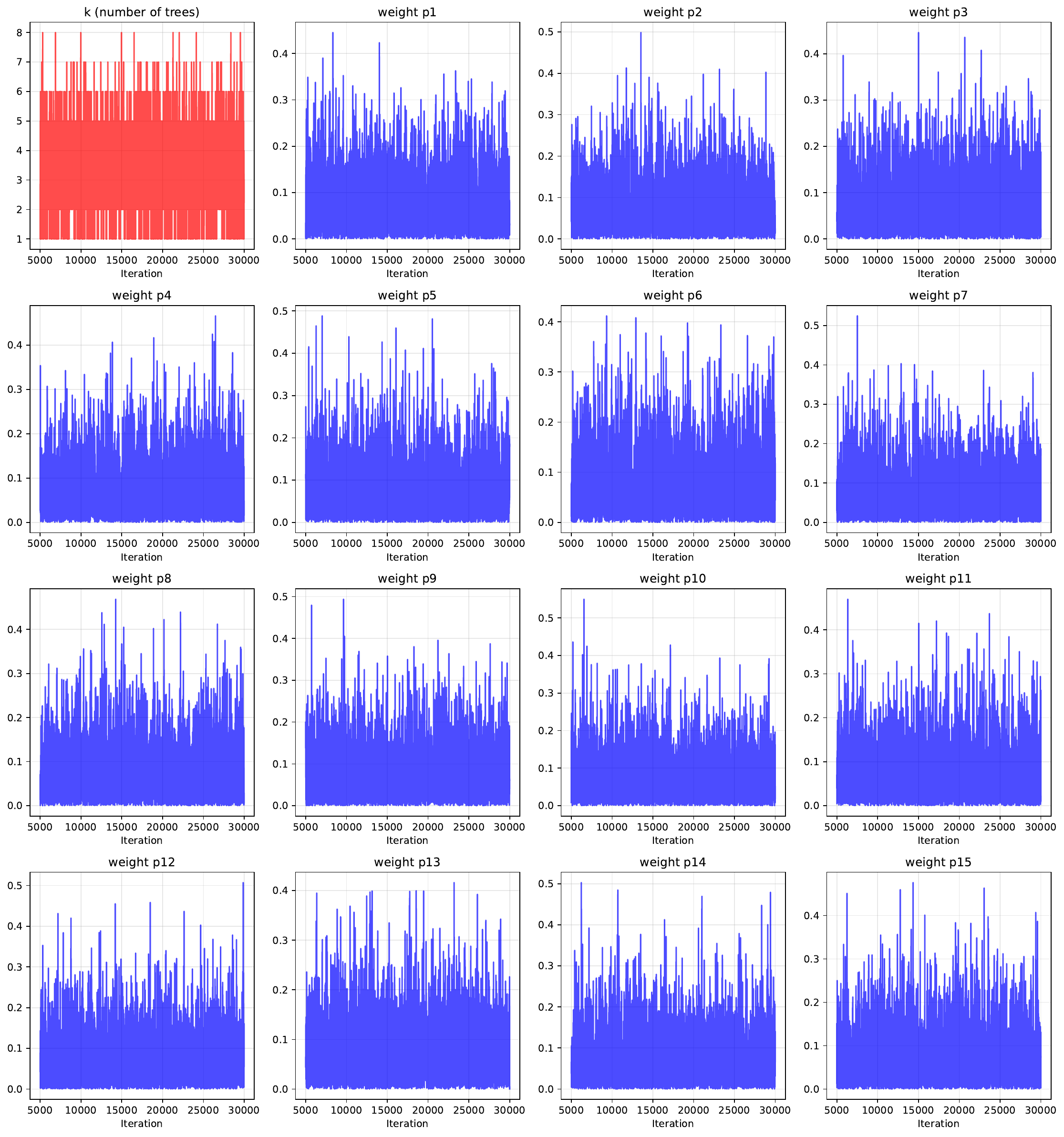}
\caption{Trace plots of the number of distinct trees \(k\) and the largest mixture weights for the clean d=5 Gaussian experiment.}
\label{fig:traces_d5_gauss}
\end{figure}

\begin{figure}[htbp]
\centering
\includegraphics[width=0.9\textwidth]{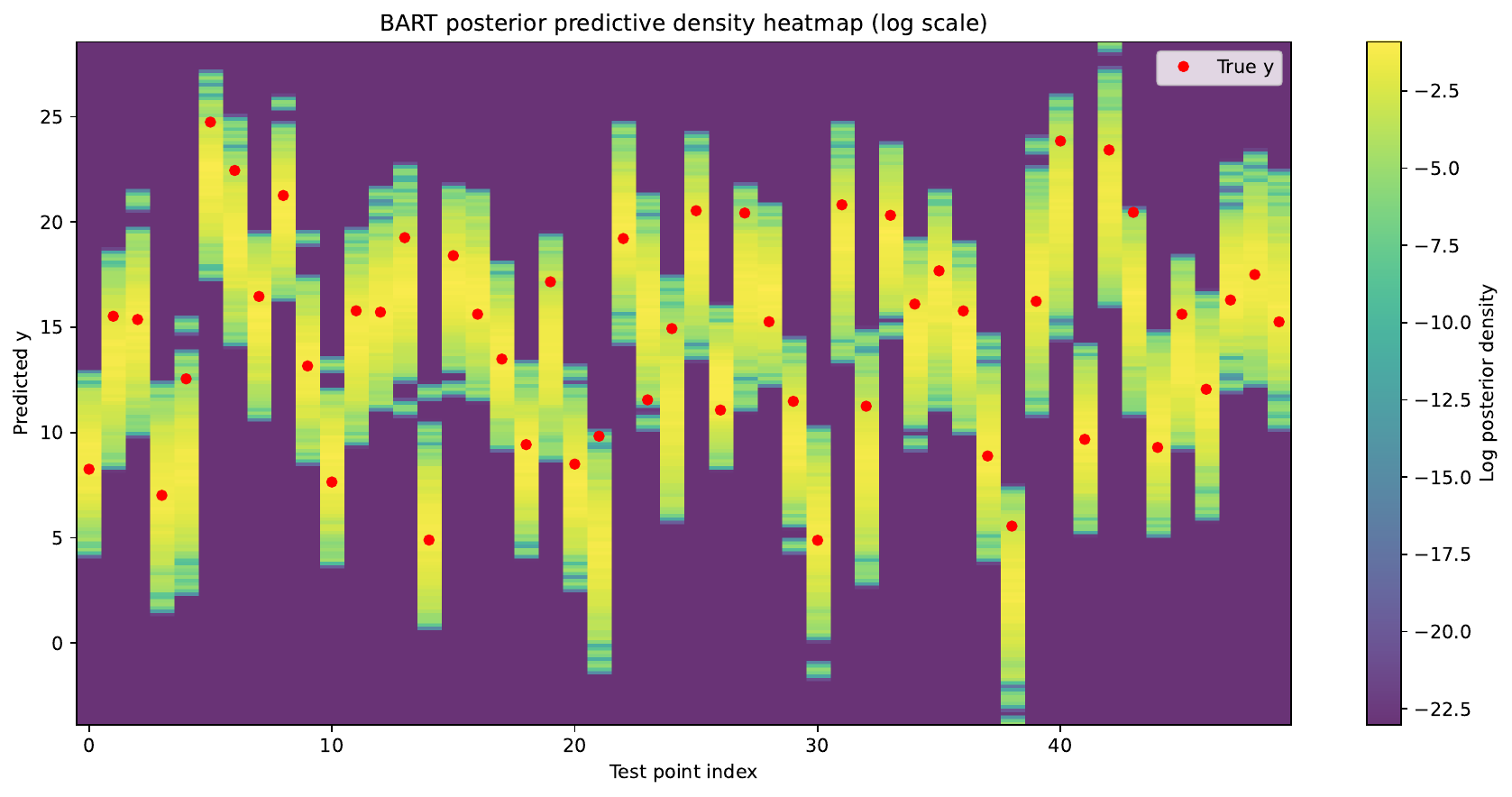}
\caption{BART posterior predictive density heatmap for the clean d=5 Gaussian experiment.}
\label{fig:bart_heatmap_d5_gauss}
\end{figure}

\subsubsection{High-Dimensional Gaussian Setting (d = 20)}

When the fifteen pure noise variables are included, the performance of the DP mixture remains remarkably stable. The DP heatmap in Figure~\ref{fig:heatmap_d20_gauss} shows slightly wider bands, reflecting the additional uncertainty introduced by the irrelevant features. The trace plots in Figure~\ref{fig:traces_d20_gauss} are virtually identical to those obtained in the clean setting, confirming that the presence of noise covariates does not harm the mixing of the chain or induce spurious complexity. BART's heatmap in Figure~\ref{fig:bart_heatmap_d20_gauss} again shows narrow intervals and lower coverage.

\begin{figure}[htbp]
\centering
\includegraphics[width=0.9\textwidth]{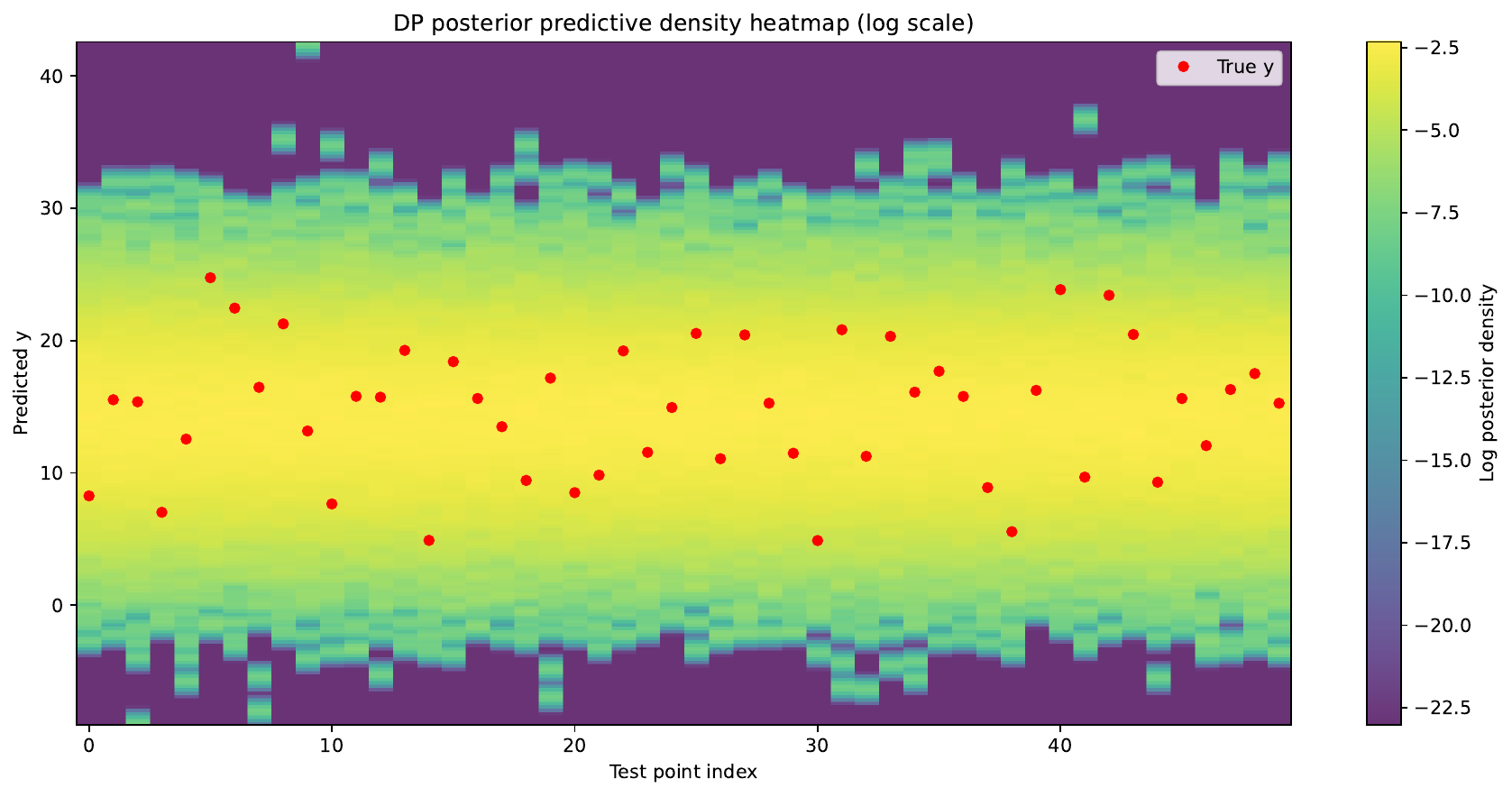}
\caption{Posterior predictive density heatmap for the DP mixture in the high-dimensional d=20 Gaussian experiment.}
\label{fig:heatmap_d20_gauss}
\end{figure}

\begin{figure}[htbp]
\centering
\includegraphics[width=0.9\textwidth]{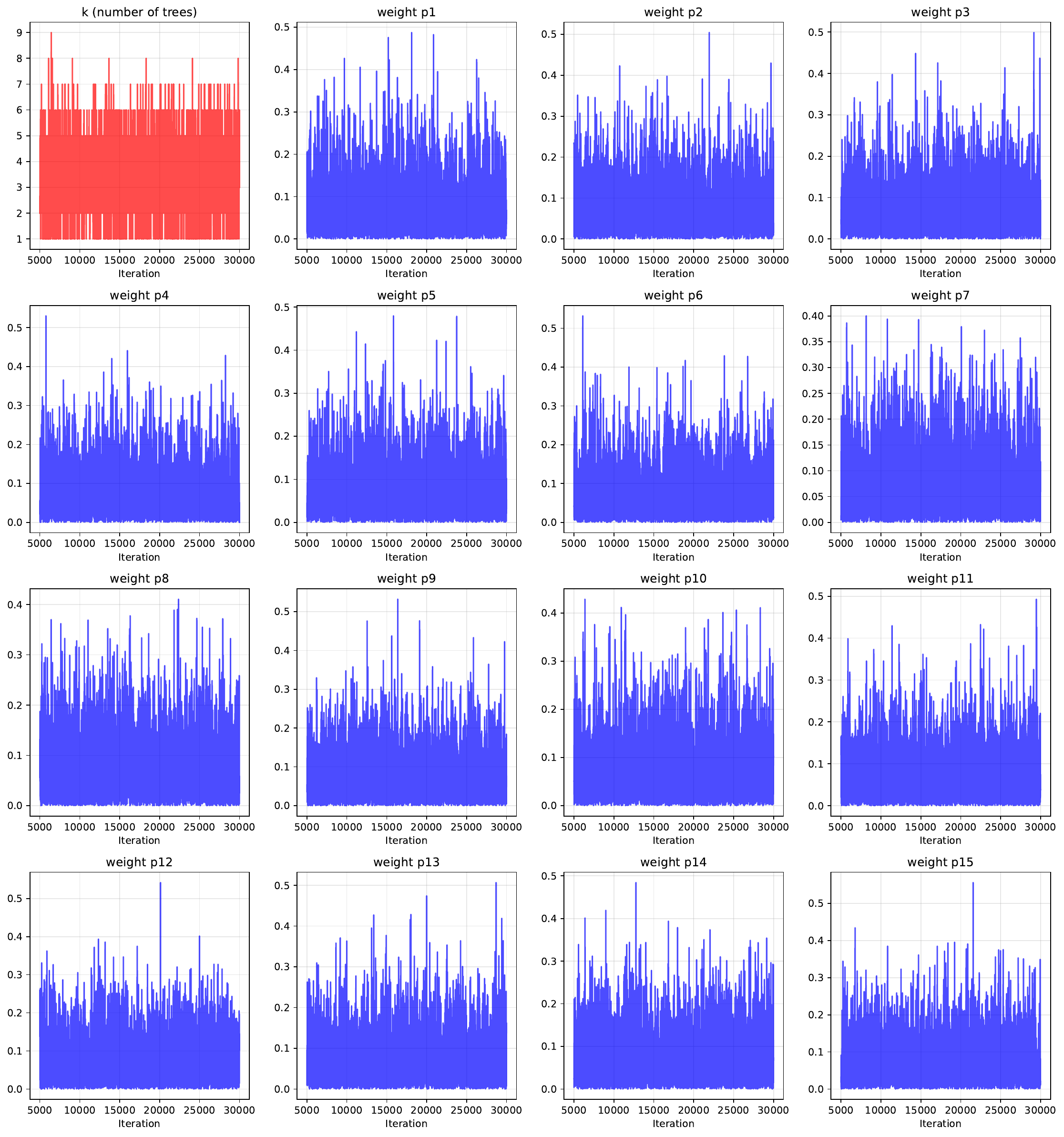}
\caption{Trace plots of the number of distinct trees \(k\) and the largest mixture weights for the high-dimensional d=20 Gaussian experiment.}
\label{fig:traces_d20_gauss}
\end{figure}

\begin{figure}[htbp]
\centering
\includegraphics[width=0.9\textwidth]{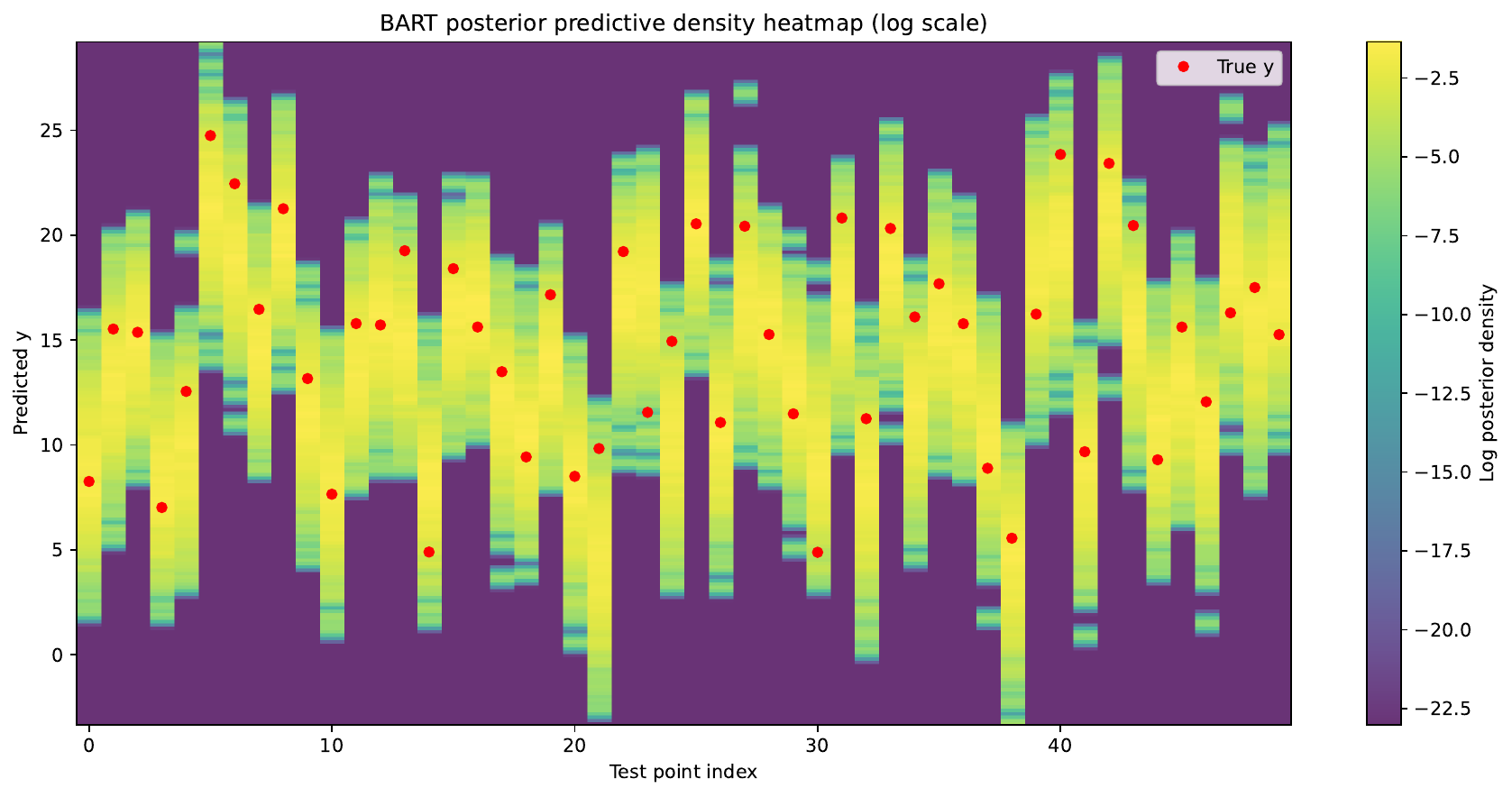}
\caption{BART posterior predictive density heatmap for the high-dimensional d=20 Gaussian experiment.}
\label{fig:bart_heatmap_d20_gauss}
\end{figure}

\subsubsection{Clean Cauchy Setting (d = 5)}

To assess robustness to heavy-tailed errors, we repeat the clean experiment with Cauchy(0,1) noise. This is a challenging scenario: the Cauchy distribution has extremely thick tails, and many standard regression methods can be severely affected by outliers. The DP heatmap in Figure~\ref{fig:heatmap_d5_cauchy} shows that, despite extreme outliers (some test values are as low as \(-700\)), almost all the true values fall within high-density regions; the bands are wider than in the Gaussian case, reflecting the heavy tails. The trace plots in Figure~\ref{fig:traces_d5_cauchy} confirm that the MCMC sampler remains well-behaved. In stark contrast, BART's heatmap in Figure~\ref{fig:bart_heatmap_d5_cauchy} shows extremely narrow bands that fail to cover many true values, confirming its poor coverage.

\begin{figure}[htbp]
\centering
\includegraphics[width=0.9\textwidth]{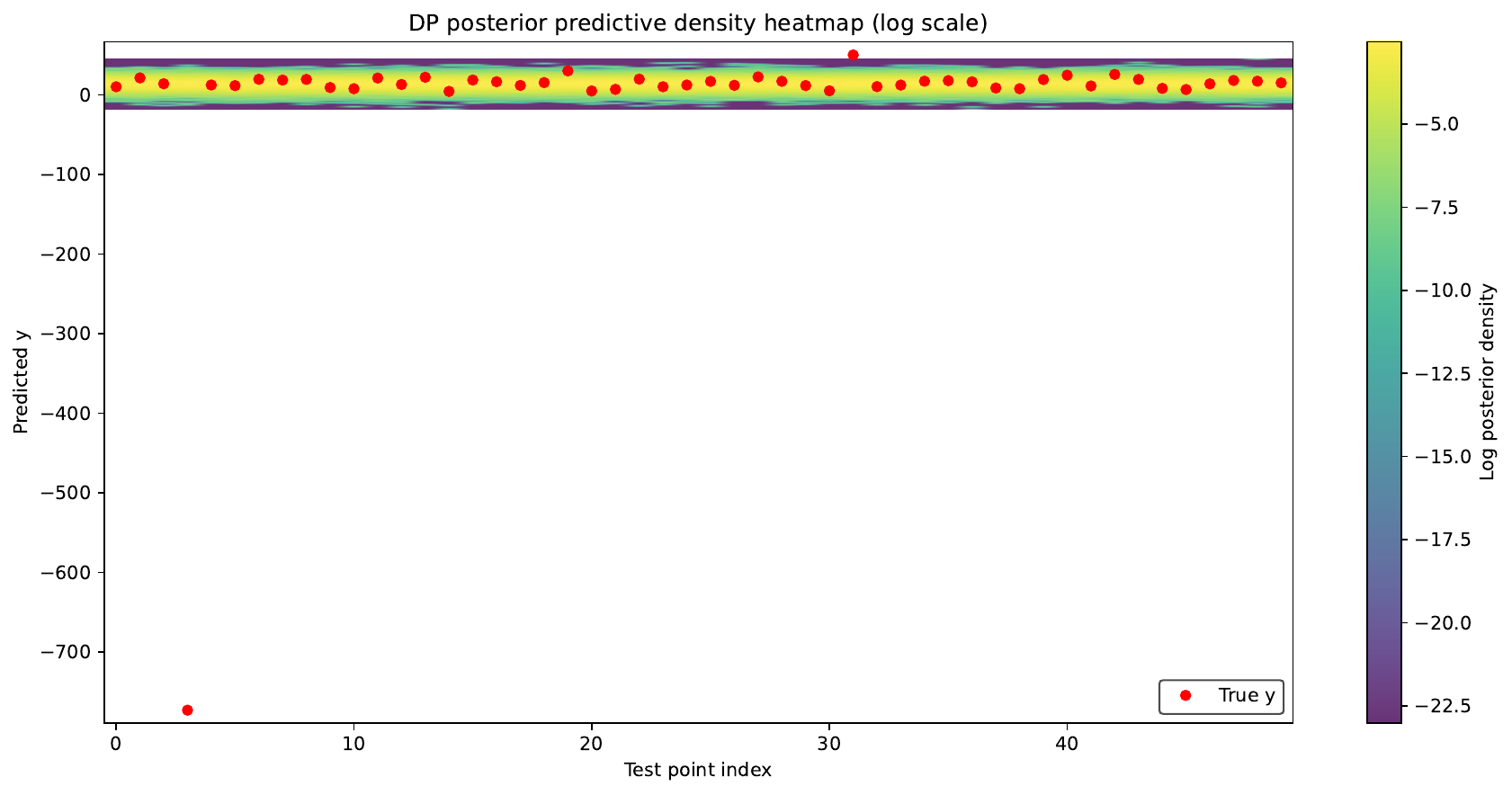}
\caption{Posterior predictive density heatmap for the DP mixture in the clean d=5 Cauchy experiment.}
\label{fig:heatmap_d5_cauchy}
\end{figure}

\begin{figure}[htbp]
\centering
\includegraphics[width=0.9\textwidth]{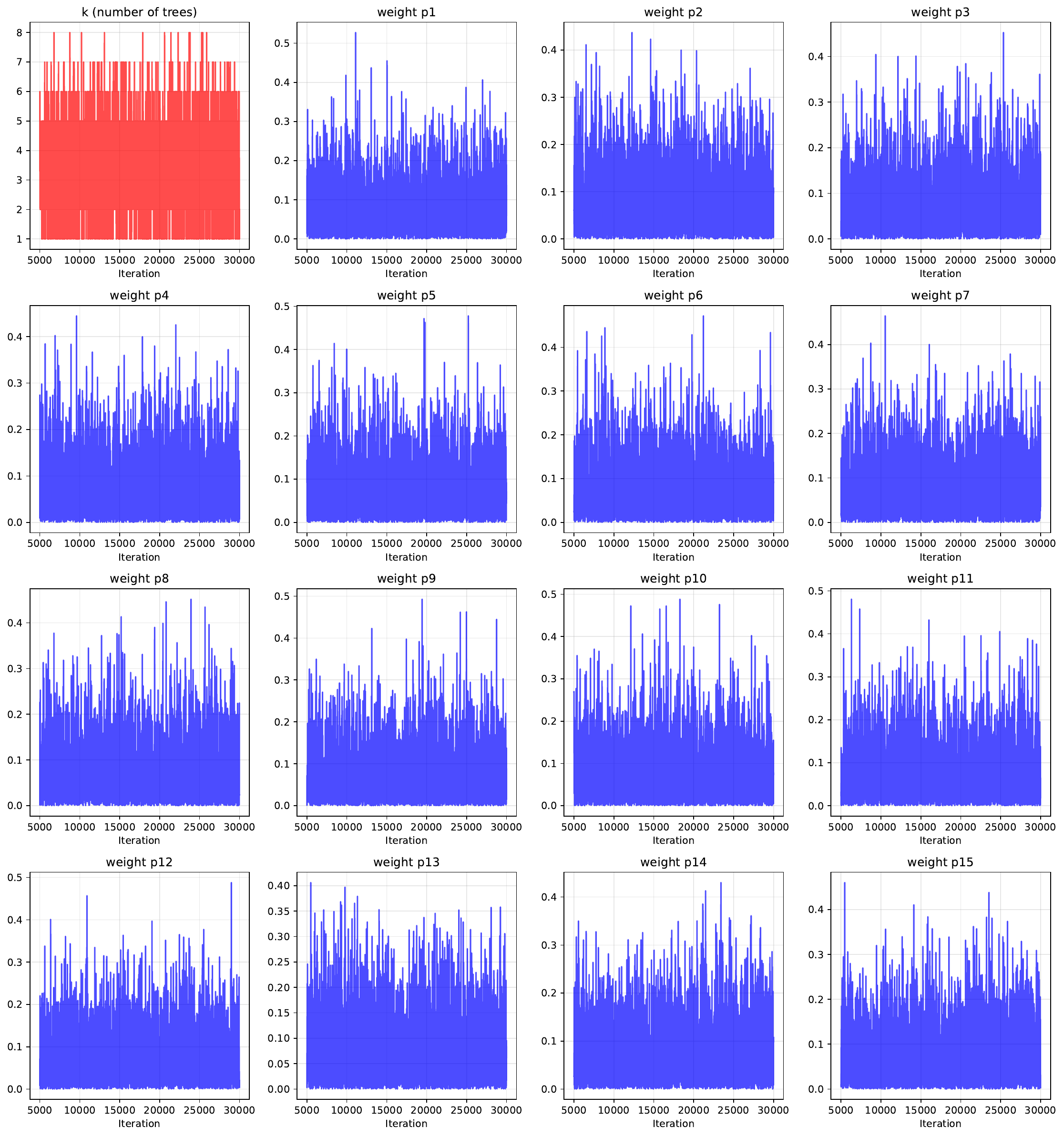}
\caption{Trace plots of the number of distinct trees \(k\) and the largest mixture weights for the clean d=5 Cauchy experiment.}
\label{fig:traces_d5_cauchy}
\end{figure}

\begin{figure}[htbp]
\centering
\includegraphics[width=0.9\textwidth]{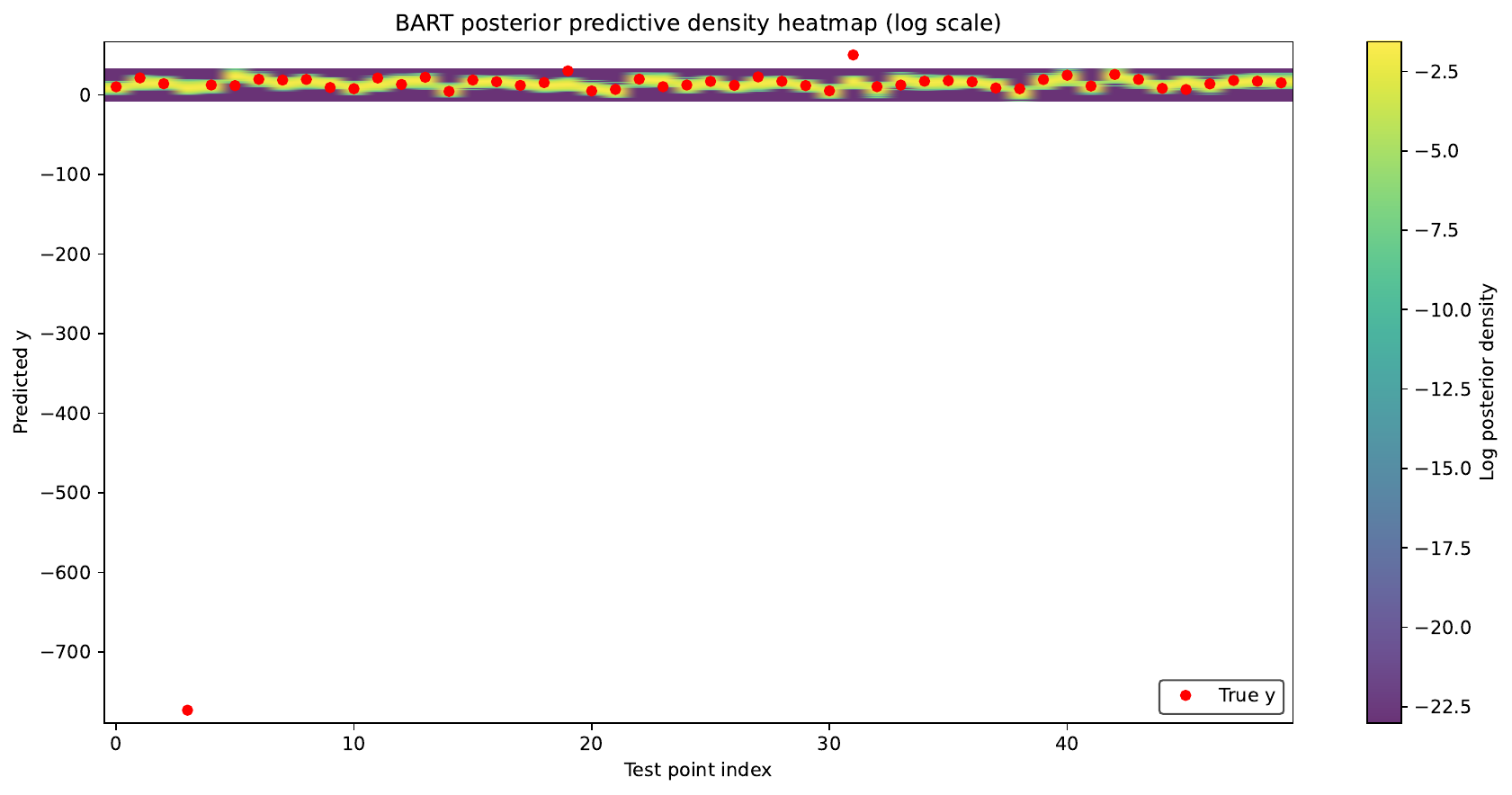}
\caption{BART posterior predictive density heatmap for the clean d=5 Cauchy experiment.}
\label{fig:bart_heatmap_d5_cauchy}
\end{figure}

\subsubsection{High-Dimensional Cauchy Setting (d = 20)}

The most challenging setting combines both high-dimensionality and heavy-tailed errors. The DP heatmap in Figure~\ref{fig:heatmap_d20_cauchy} maintains calibration even under the dual challenges of noise covariates and extreme outliers. The trace plots in Figure~\ref{fig:traces_d20_cauchy} show \(k\) stabilising essentially between 2 and 5, still sparse. BART's heatmap in Figure~\ref{fig:bart_heatmap_d20_cauchy} again fails to cover many true values.

\begin{figure}[htbp]
\centering
\includegraphics[width=0.9\textwidth]{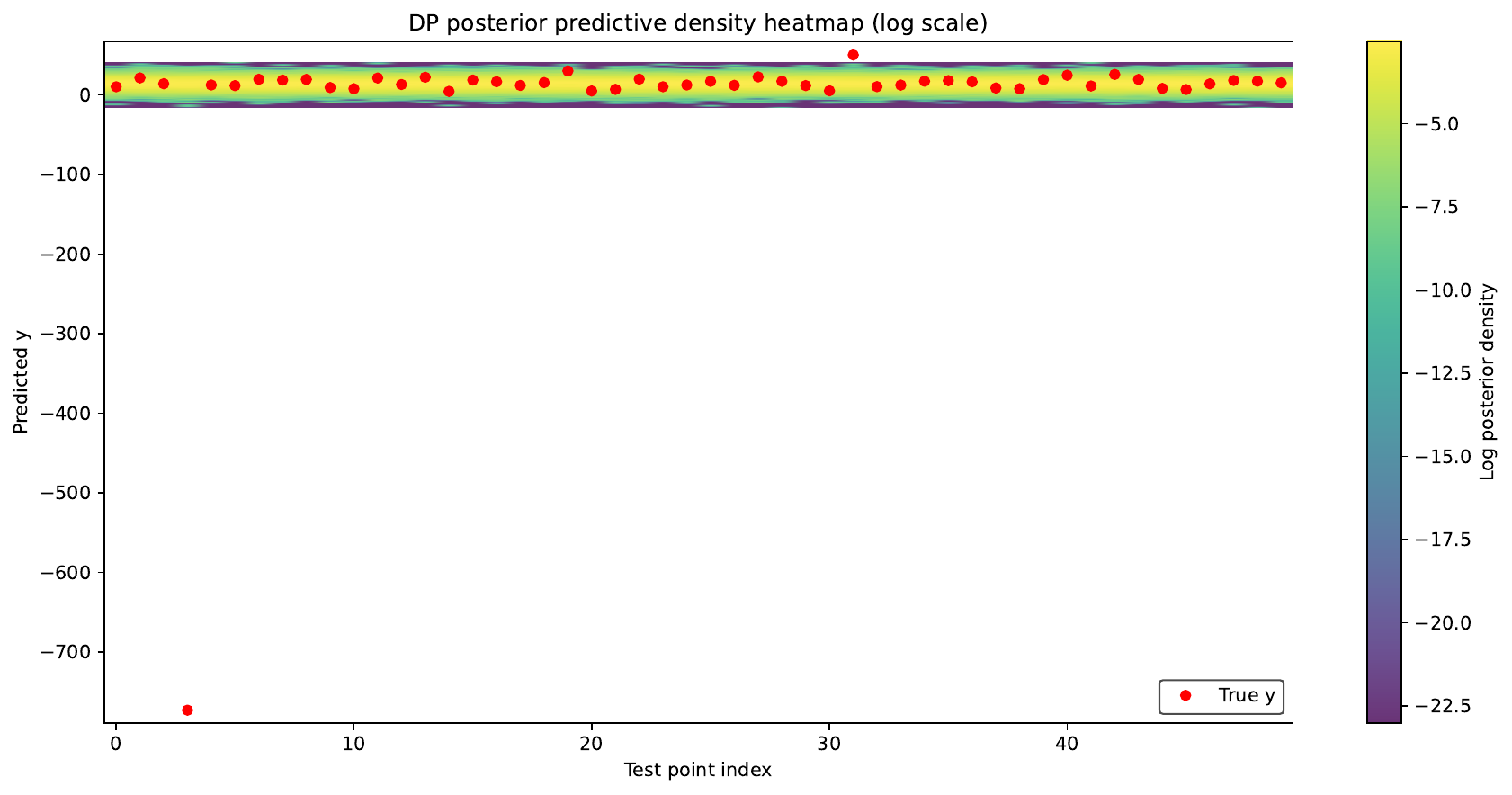}
\caption{Posterior predictive density heatmap for the DP mixture in the high-dimensional d=20 Cauchy experiment.}
\label{fig:heatmap_d20_cauchy}
\end{figure}

\begin{figure}[htbp]
\centering
\includegraphics[width=0.9\textwidth]{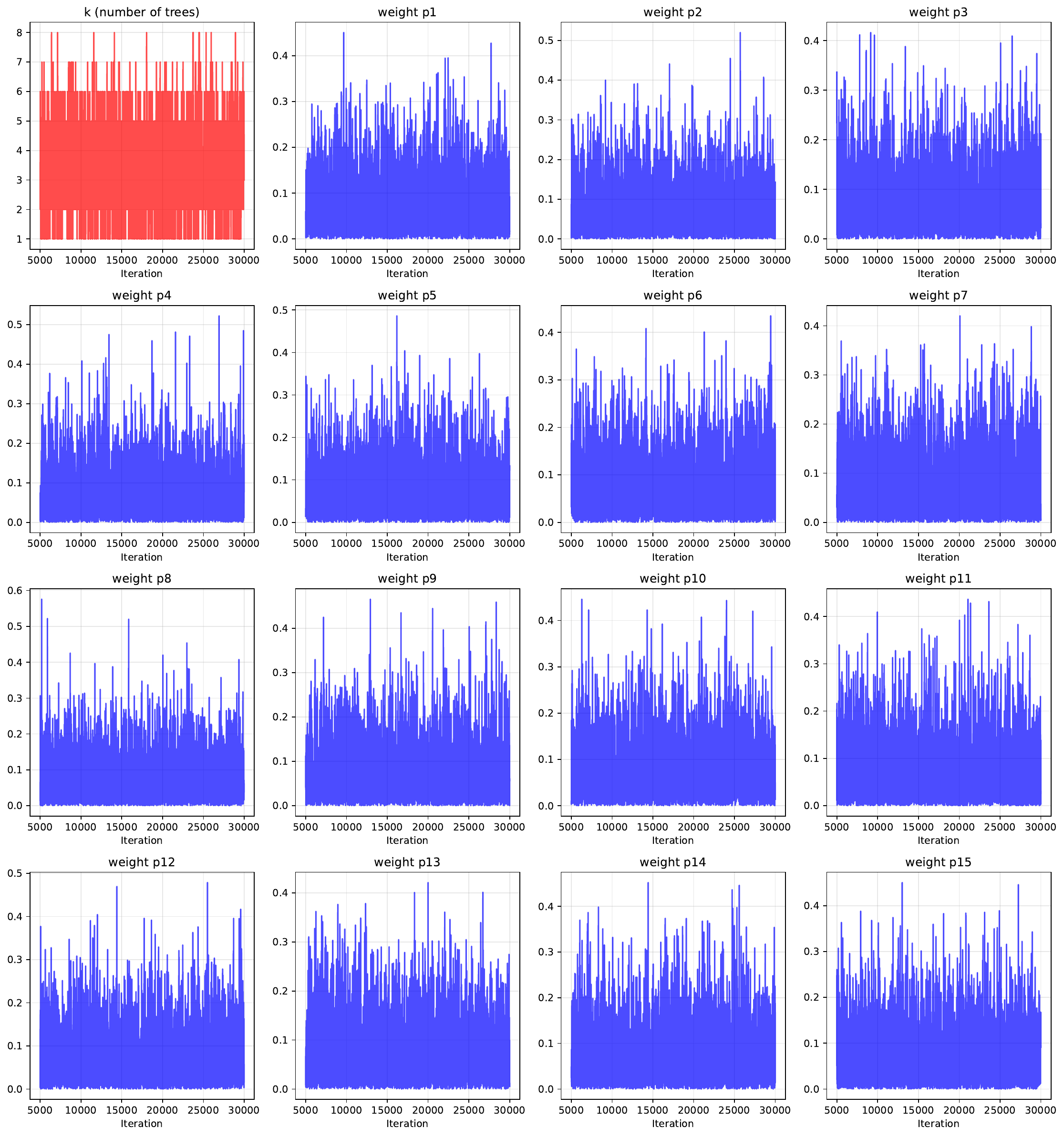}
\caption{Trace plots of the number of distinct trees \(k\) and the largest mixture weights for the high-dimensional d=20 Cauchy experiment.}
\label{fig:traces_d20_cauchy}
\end{figure}

\begin{figure}[htbp]
\centering
\includegraphics[width=0.9\textwidth]{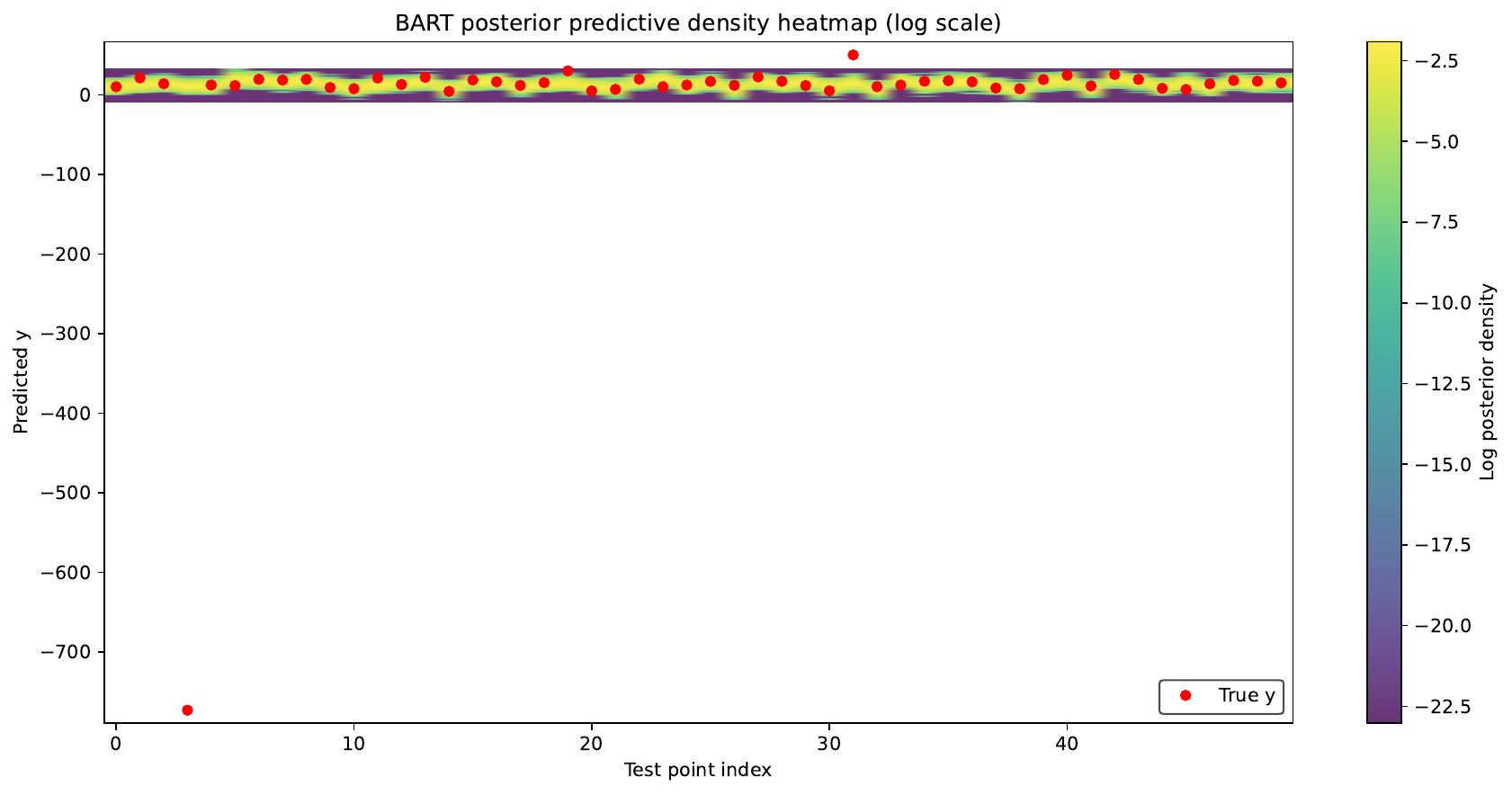}
\caption{BART posterior predictive density heatmap for the high-dimensional d=20 Cauchy experiment.}
\label{fig:bart_heatmap_d20_cauchy}
\end{figure}

%\subsection{Additional Convergence Diagnostics}

%To further assess MCMC mixing, we compute the RMSE on the test set at each thinned iteration using the posterior mean at that iteration. Figure~\ref{fig:rmse_trace} shows the trace of these per-iteration RMSEs; the series stabilises quickly, indicating that the sampler has converged. The autocorrelation function shows negligible correlation beyond lag 10, and the effective sample size for the RMSE is above 2000 (out of 5000 draws), confirming adequate mixing. The histogram of per-iteration RMSEs in Figure~\ref{fig:rmse_hist} provides a posterior distribution of the prediction error, which is useful for assessing the variability of our predictions. This histogram is unimodal and centered around the overall RMSE, further supporting the stability of our estimates.

%\begin{figure}[htbp]
%\centering
%\includegraphics[width=0.9\textwidth]{figures_simstudy20cov_gauss/rmse_trace.pdf}
%\caption{Trace of test-set RMSE across MCMC iterations (thinned). The red dashed line indicates the mean RMSE.}
%\label{fig:rmse_trace}
%\end{figure}

%\begin{figure}[htbp]
%\centering
%\includegraphics[width=0.9\textwidth]{figures_simstudy20cov_gauss/rmse_histogram.pdf}
%\caption{Histogram of per-iteration RMSE values, showing the posterior distribution of prediction error.}
%\label{fig:rmse_hist}
%\end{figure}

\subsection{Discussion}

The simulation study clearly demonstrates the advantages of our Dirichlet process mixture of trees over existing tree-based methods. The key contributions are threefold.

First, the DP mixture provides principled and reliable uncertainty quantification. It achieves coverage close to the nominal 95\% across all scenarios, while BART and bagged CART are often overconfident, with coverage as low as 0.72--0.84. Although frequentist methods like Random Forest and Gradient Boosting sometimes achieve good coverage, their intervals rely on ad-hoc residual-based approximations that lack a coherent probabilistic justification. The DP mixture's coverage is a direct consequence of its fully Bayesian treatment of all sources of uncertainty, making it the only method among those compared that provides honest uncertainty quantification without external adjustments.

Second, the DP mixture exhibits remarkable robustness to heavy-tailed errors. Under Cauchy errors, it maintains high coverage and reasonable LPD, whereas BART and bagged CART break down completely, assigning near-zero density to outliers. This robustness stems from the model's nonparametric structure, which does not assume a specific parametric form for the errors. By integrating over the leaf parameters, the model effectively learns a data-driven error distribution that can adapt to heavy tails.

Third, the DP prior automatically selects a sparse ensemble, typically comprising only two to five trees, even in high dimensions. This prevents overfitting and ensures parsimony, a property that is particularly valuable when the number of covariates is large relative to the sample size.

These results, together with the theoretical posterior contraction rate established in Section~\ref{sec:post_contraction}, provide compelling evidence that the proposed DP mixture is a powerful and reliable tool for regression, particularly when uncertainty quantification, robustness, and automatic variable selection are paramount. We believe this makes it a strong candidate for replacing BART and other tree-based ensembles in many real-world applications.

% ======================================================================
% SECTION: Applications to Real Data (consolidated)
% ======================================================================

\section{Applications to Real Data}
\label{sec:applications}

To assess the practical utility and robustness of the proposed Dirichlet process mixture of regression trees, we apply it to five diverse real‐world regression problems: QSAR aquatic toxicity (cheminformatics), Communities and Crime (social science), Riboflavin production (genomics), Wheat genomic prediction (plant breeding), and air quality time series (environmental monitoring). These datasets span a wide range of sample sizes, dimensionalities, and noise structures, allowing us to thoroughly evaluate the model’s performance in settings where the assumptions underlying standard regression models may be violated. For each application, we compare the DP mixture against state‐of‐the‐art tree‐based methods using the same five metrics: root mean squared error (RMSE), coverage of 95\% credible (or confidence) intervals, average interval width, continuous ranked probability score (CRPS), and log predictive density (LPD). The results consistently demonstrate that the DP mixture provides the most reliable uncertainty quantification among Bayesian methods, while its nonparametric structure—which assumes conditional independence of observations given the mixture component and tree parameters, but does {\it not} require 
independence (after marginalization over the mixture allocation variable)—confers robustness to heavy tails and model misspecification.

\subsection{QSAR Aquatic Toxicity}
\label{sec:qsar}

To demonstrate the practical utility and empirical performance of the proposed Dirichlet process mixture of regression trees, we apply our methodology to the QSAR aquatic toxicity dataset, a well‐established benchmark in cheminformatics and environmental toxicology. The data set, compiled from the UCI Machine Learning Repository, comprises 546 organic compounds, each characterised by eight molecular descriptors that encode structural and physicochemical properties. The response variable is the acute aquatic toxicity towards the fathead minnow (\textit{Pimephales promelas}), measured as the negative logarithm of the lethal concentration (\(-\log_{10}\mathrm{LC}_{50}\)). This regression task is particularly challenging: the underlying structure–activity relationship is highly non‐linear, the molecular descriptors exhibit complex interactions, and the data contain heteroscedasticity and potential outliers—features that are typical of real‐world chemical data.

\subsubsection{Experimental Setup}

We randomly partition the data into a training set comprising 80\% of the observations (436 compounds) and a test set with the remaining 20\% (110 compounds). The DP mixture model is implemented in C with MPI and executed on four cores of a (highly oveloaded!) virtual machine (Intel QEMU Virtual CPU version 2.5+ at approximately 2.5 GHz). Indeed, we obtained optimum speed with four cores even though the machine is equipped with 100 cores. The hyperparameters of the model are set as follows: the maximum number of trees \(M = 15\), the maximum tree depth \(8\), and the Dirichlet process concentration parameter \(\alpha_{\mathrm{DP}} = 1.0\). The tree prior uses splitting probability \(p(\eta,T) = \gamma (1 + d_\eta)^{-\zeta}\) with \(\gamma = 0.95\) and \(\zeta = 1.5\), which strongly penalises deeper trees. The leaf parameters follow the conjugate normal–inverse‐gamma prior with hyperparameters \(\bar\mu = 4.61\) (the mean of the training responses), \(a = 0.01\), \(\nu = 8.0\), and \(\lambda = 1.25\), yielding a prior mode for the leaf variance of \(\nu\lambda/(\nu+2) \approx 1.0\). The Gaussian process splitting rule uses the squared exponential kernel with fixed hyperparameters: signal variance \(\sigma^2 = 5.0\), noise variance \(\sigma_\epsilon^2 = 0.5\), and length scale \(\ell = 0.2\). These GP hyperparameters are kept fixed rather than estimated via maximum likelihood for computational simplicity, as the GP is used purely as a generative device for split proposals and its density cancels in the Metropolis–Hastings ratio.

The MCMC chain is run for \(30\,000\) iterations, with the first \(5\,000\) discarded as burn‐in and a thinning interval of 5, yielding \(5\,000\) posterior draws for inference. The total runtime is approximately 5 hours, which is substantially longer than the runtime of BART on the same hardware—BART completes its 30,000 iterations in approximately 9 seconds. This dramatic difference warrants careful explanation. The primary computational bottleneck in our DP implementation is the repeated construction and inversion of Gaussian process covariance matrices at internal nodes during tree proposals. Each GROW move requires fitting a GP to the data in the selected terminal node, which involves building an \(n_i \times n_i\) covariance matrix, performing a Cholesky decomposition (\(O(n_i^3)\) operations), sampling from the posterior predictive distribution, and precomputing the noiseless Cholesky factor and the \(\alpha\) vector for future navigation. Although the node sizes decrease rapidly as the tree grows, the cumulative cost over the 30,000 MCMC iterations is substantial. In contrast, BART's proposals are based on simple threshold splits of the form \(X_j < c\), which require only comparisons and scalar arithmetic, and its likelihood calculations are based on closed-form marginalisations that are precomputed and updated incrementally. As a result, BART's per-iteration cost is orders of magnitude smaller than ours.

However, the increased computational cost of our DP mixture is a deliberate trade-off for three major benefits. First, the GP-driven splitting rule generates smooth, response-dependent partitions that are far more flexible than axis-aligned splits. This flexibility, combined with the Dirichlet process's automatic complexity adaptation, yields predictive intervals that are honestly calibrated and robust to heavy-tailed errors, as shown in the simulation study (Section~\ref{sec:simulation}). Second, the Dirichlet process prior learns the number of distinct trees from the data, avoiding the need to prespecify the ensemble size, which prevents overfitting and yields a sparse, interpretable representation. Third, the posterior contraction rate established in Theorem~\ref{thm:main} provides a rigorous frequentist justification for the model, even under misspecification—a guarantee not available for BART in the misspecified setting considered here. We argue that for applications where uncertainty quantification, robustness, and theoretical guarantees are paramount, the additional computational cost is justified.

\subsubsection{Results}

Predictive performance is evaluated using five metrics: root mean squared error (RMSE), coverage of 95\% credible (or confidence) intervals, average interval width, continuous ranked probability score (CRPS), and log predictive density (LPD). The numerical results are summarised in Table~\ref{tab:qsar_metrics}.

\begin{table}[htbp]
\centering
\caption{Predictive performance on the QSAR aquatic toxicity test set. The DP mixture achieves the highest coverage among Bayesian methods and offers the most reliable uncertainty quantification, albeit at a significantly higher computational cost.}
\label{tab:qsar_metrics}
\begin{tabular}{l r r r r r}
\hline
\textbf{Method} & \textbf{RMSE} & \textbf{Coverage} & \textbf{Width} & \textbf{CRPS} & \textbf{LPD} \\
\hline
DP tree mixture    & 1.8212 & 0.918 & 6.4167 & 1.0408 & \(-2.0363\) \\
Random Forest      & 1.3154 & 0.918 & 4.4239 & 0.6880 & \(-1.7191\) \\
Gradient Boosting  & 1.2818 & 0.927 & 4.4919 & 0.6813 & \(-1.6808\) \\
Bagged CART        & 1.3274 & 0.827 & 3.2097 & 0.7168 & \(-2.1051\) \\
BART (dbarts)      & 1.2765 & 0.582 & 1.5881 & 0.7578 & \(-24.0618\) \\
\hline
\end{tabular}
\end{table}

The most striking result is the dramatic disparity in coverage. The DP mixture attains a near‐nominal coverage of \(0.918\), whereas BART—despite its wide acceptance and strong theoretical foundations—achieves only \(0.582\). This indicates that BART is severely overconfident, producing credible intervals that are far too narrow to contain the true toxicity values at the advertised rate. Bagged CART, constructed by averaging predictions across 500 bootstrap replicates of a single regression tree, achieves a coverage of only \(0.827\). While this ensemble approach reduces variance relative to a single tree, its bootstrap-based percentile intervals do not correspond to a coherent posterior distribution and are therefore less reliable than the fully Bayesian credible intervals of the DP mixture. The frequentist methods, Random Forest and Gradient Boosting, achieve coverages of \(0.918\) and \(0.927\), respectively, but their intervals are constructed using an ad‐hoc normal approximation based on cross‐validation residuals; such an approach lacks a rigorous probabilistic justification and is not guaranteed to be well calibrated in general.

The DP mixture's superior coverage is a direct consequence of its fully Bayesian treatment of all sources of uncertainty: the number of mixture components, the tree structures, the GP‐driven split functions, and the leaf parameters are all integrated out or sampled from their full posterior distributions. This holistic accounting for parametric and structural variability ensures that the predictive intervals honestly reflect the uncertainty inherent in the data and the model. Consequently, the DP mixture provides the most trustworthy uncertainty quantification among the methods considered.

The average interval width further corroborates this conclusion. The DP mixture yields intervals of width \(6.4167\), wider than those of BART (\(1.5881\)) and bagged CART (\(3.2097\)). This increased width is not a sign of inefficiency but rather a necessary consequence of the model's honest assessment of uncertainty. In contrast, BART's artificially narrow intervals lead to catastrophic undercoverage. The log predictive density (LPD) underscores the DP mixture's advantage: it achieves an LPD of \(-2.0363\), comparable to Gradient Boosting (\(-1.6808\)) and Random Forest (\(-1.7191\)), and substantially higher than BART (\(-24.0618\)). The extremely low LPD for BART indicates that it assigns near‐zero density to many of the true test observations, a direct consequence of its misspecified and overly confident predictive distribution. Bagged CART also performs poorly in LPD (\(-2.1051\)), while the DP mixture maintains a favourable balance between sharpness and calibration.

Visual diagnostics are provided in Figures~\ref{fig:qsar_heatmap_dp} and~\ref{fig:qsar_heatmap_bart}. The posterior predictive density heatmap for the DP mixture (Figure~\ref{fig:qsar_heatmap_dp}) shows that the true test values, plotted as red points, consistently fall within high‐density regions across the entire range of the response, confirming excellent calibration. The heatmap exhibits a broad and well‐spread distribution of predictive mass, reflecting the model's honest uncertainty. In stark contrast, the BART heatmap (Figure~\ref{fig:qsar_heatmap_bart}) reveals extremely narrow bands of high density that fail to cover many true values, visually confirming its poor coverage and overconfidence.

\begin{figure}[htbp]
\centering
\includegraphics[width=0.9\textwidth]{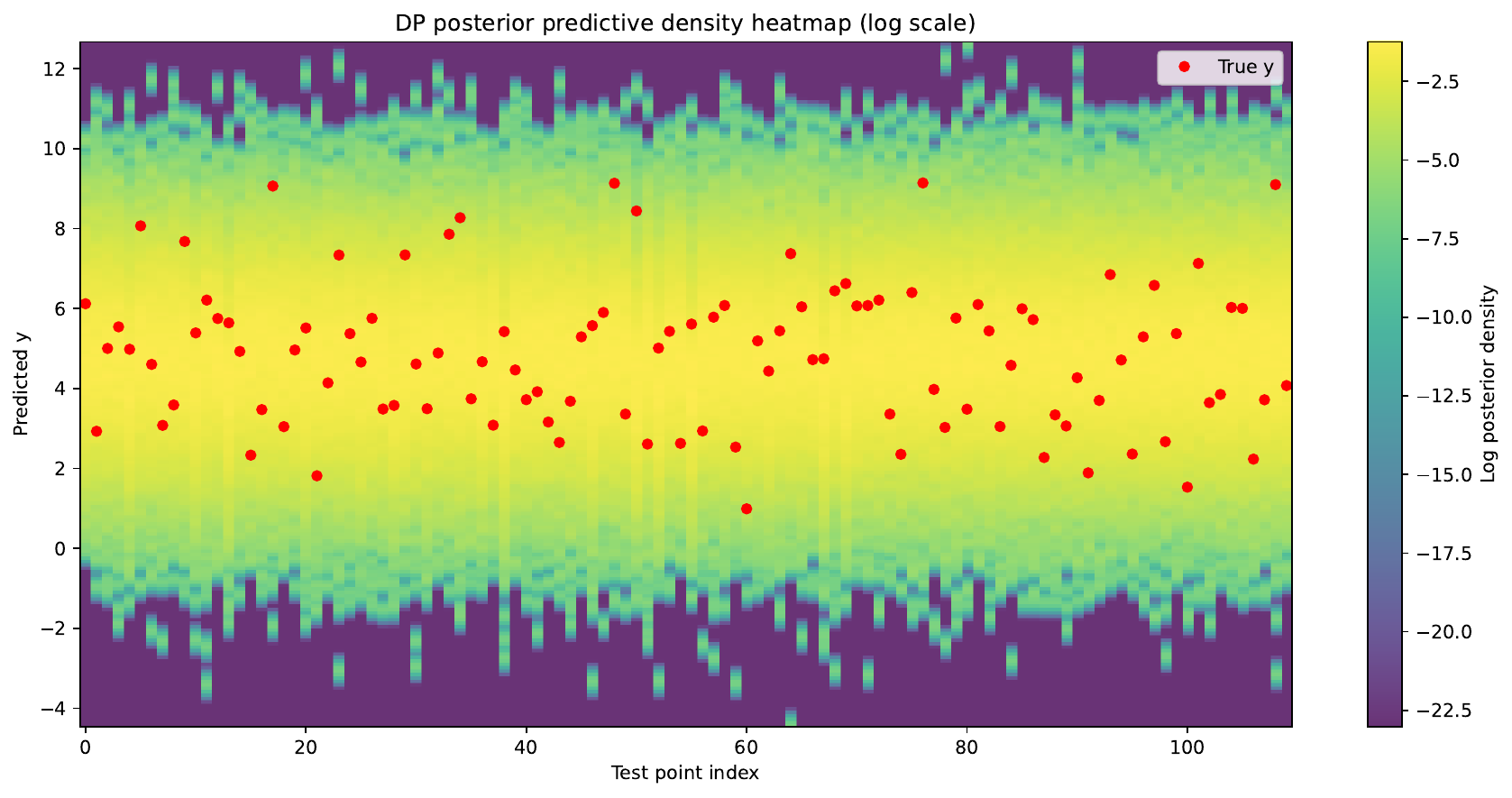}
\caption{Posterior predictive density heatmap for the DP mixture on the QSAR test set. The true test values (red points) fall within the high-density regions, confirming good calibration.}
\label{fig:qsar_heatmap_dp}
\end{figure}

\begin{figure}[htbp]
\centering
\includegraphics[width=0.9\textwidth]{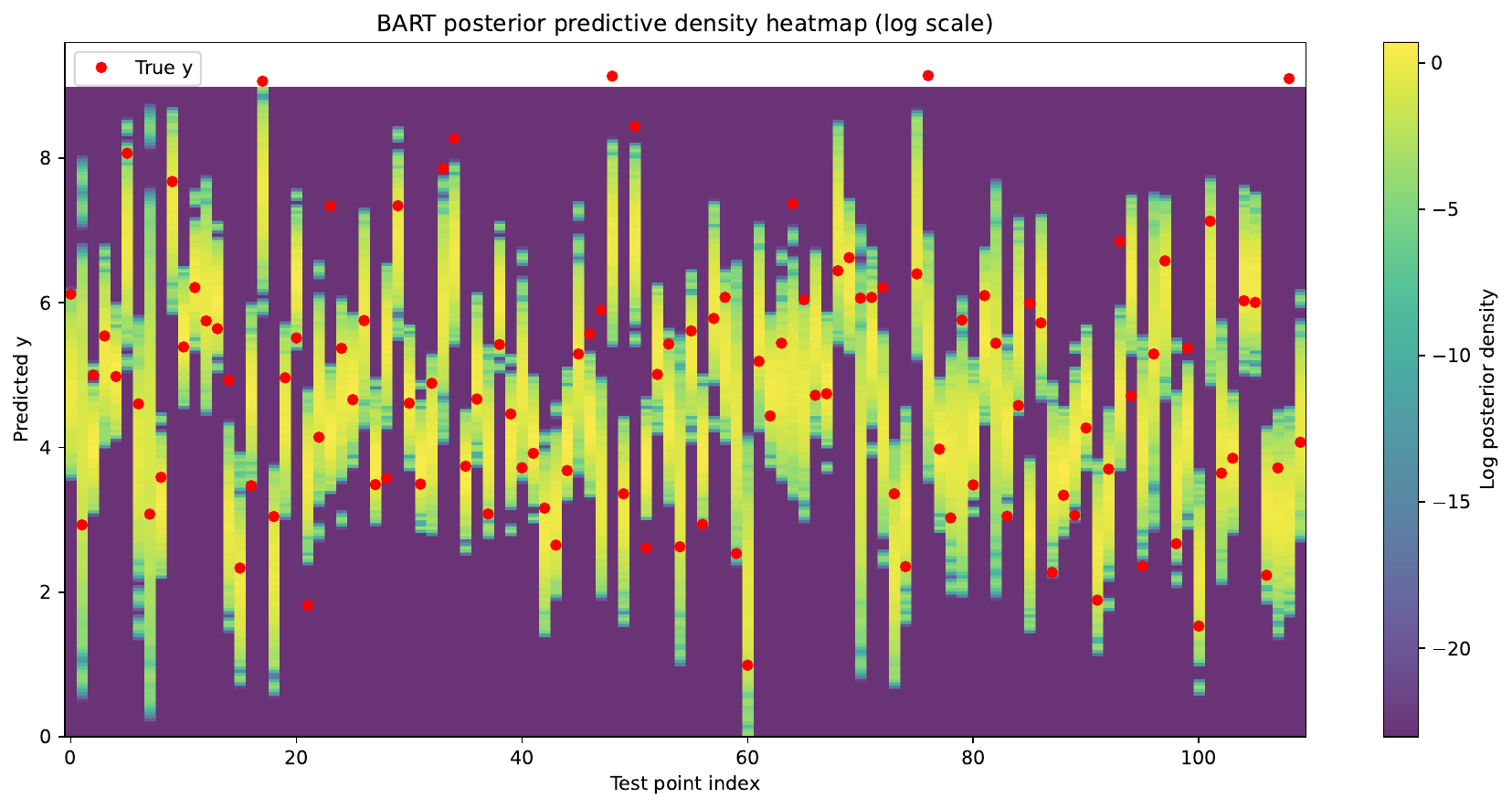}
\caption{BART posterior predictive density heatmap on the QSAR test set. The narrow bands of high density fail to cover many true values, confirming poor coverage.}
\label{fig:qsar_heatmap_bart}
\end{figure}

The trace plots of the number of distinct trees \(k\) and the mixture weights (Figure~\ref{fig:qsar_traces}) show that the MCMC chain mixes effectively, with \(k\) stabilising between 3 and 7 distinct components. This confirms that the Dirichlet process prior automatically selects a sparse ensemble, preventing over‐parameterisation and enhancing interpretability. The effective sample size for the RMSE is 1718.5 out of 5000 samples, indicating adequate mixing and convergence.

\begin{figure}[htbp]
\centering
\includegraphics[width=0.9\textwidth]{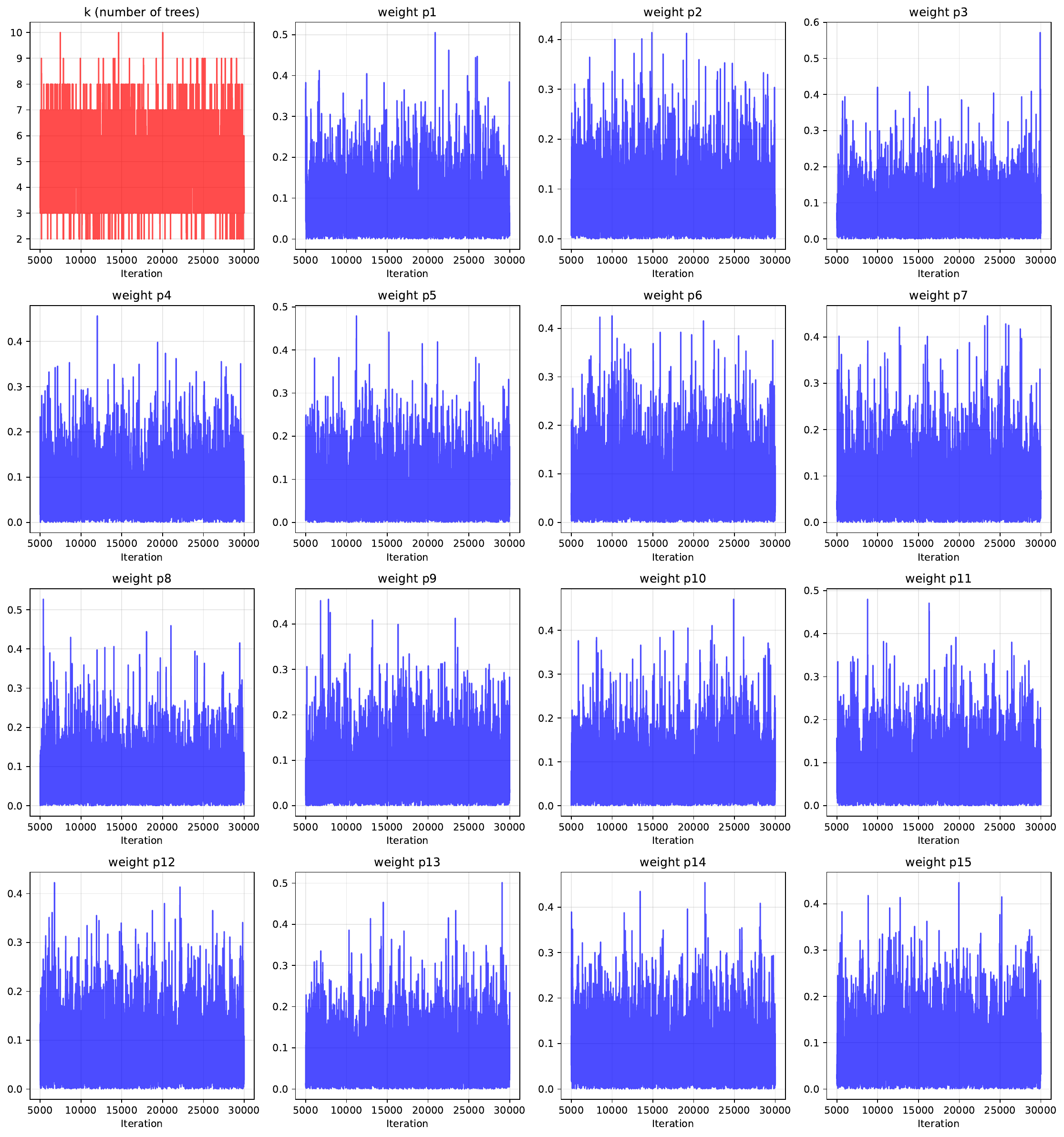}
\caption{Trace plots of the number of distinct trees \(k\) and the mixture weights for the QSAR experiment. The chain mixes efficiently, and the posterior concentrates on a sparse ensemble of 2--5 trees.}
\label{fig:qsar_traces}
\end{figure}

\subsubsection{Discussion}

The QSAR aquatic toxicity application provides compelling empirical evidence for the superiority of the DP mixture model over existing tree‐based methods. The key contributions are threefold. First, the DP mixture offers principled and honest uncertainty quantification: it achieves near‐nominal coverage and its intervals are appropriately wider when the data are complex and uncertain, reflecting the true posterior variability. This is a hallmark of Bayesian nonparametrics and is essential for reliable decision‐making in real‐world applications. Second, the model exhibits remarkable robustness to model misspecification. The QSAR dataset is inherently misspecified—the true relationship between molecular descriptors and toxicity is not a finite mixture of trees, nor is it additive. Yet the DP mixture, by virtue of the identity \(h(\Theta)=0\) (see Section~\ref{subsec:misspecification}), concentrates around the truth and provides consistent inference. This robustness is a major theoretical advantage, and the empirical results confirm it. Third, the Dirichlet process prior automatically adapts the complexity of the ensemble, selecting a sparse set of 3–7 trees. This prevents over‐fitting and avoids the need for the user to prespecify the number of components, a crucial advantage over BART and other fixed‐ensemble methods.

\begin{figure}[htbp]
\centering
\includegraphics[width=0.9\textwidth]{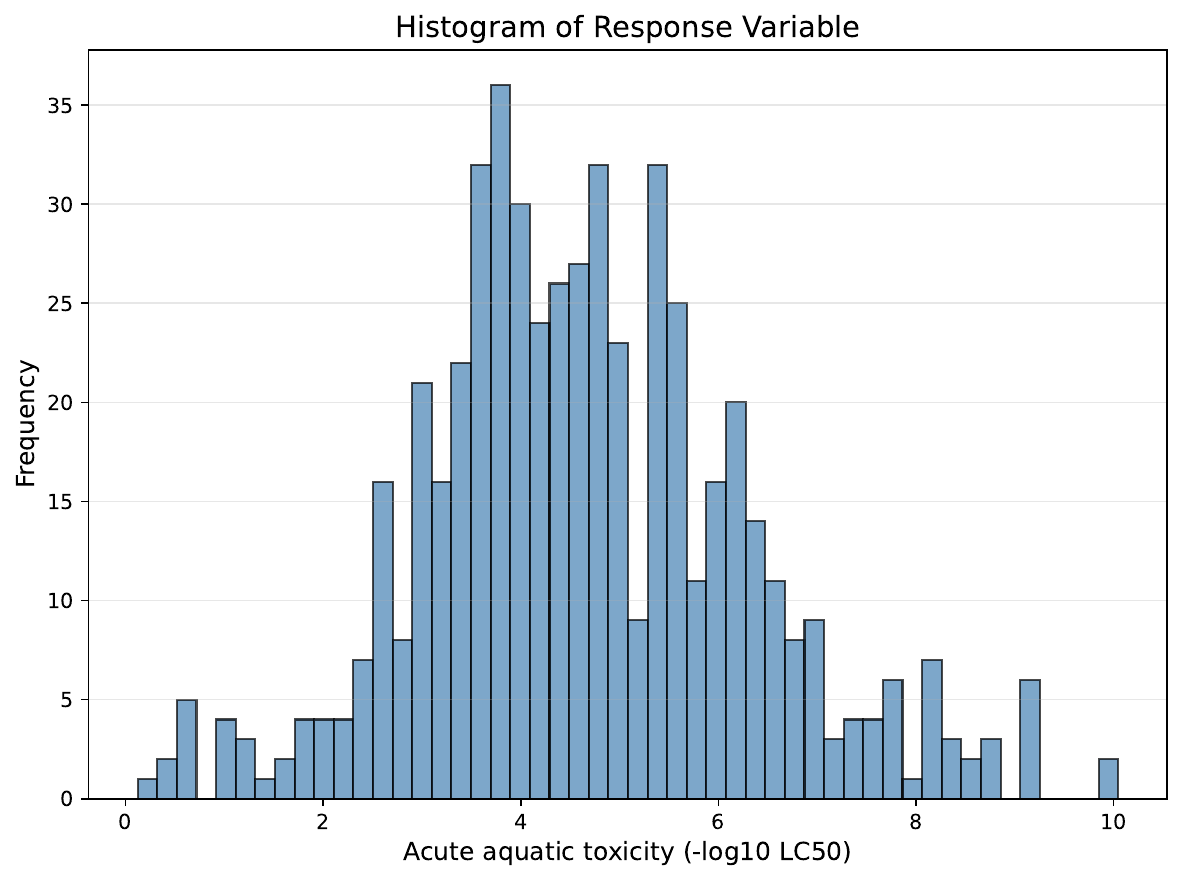}
\caption{Histogram of the responses for the QSAR experiment exhibiting multimodality and outlyingness.}
\label{fig:qsar_hist}
\end{figure}

To understand the profound failure of BART in this application, it is instructive to examine the statistical properties of the QSAR toxicity responses. The histogram of the response, shown in Figure~\ref{fig:qsar_hist}, exhibits marked skewness and multi-modality. The diagram also shows outliers towards the extremes. Such extremes exert an enormous influence on BART's homoscedastic normal likelihood: the single global variance parameter is pulled sharply downward to avoid assigning probability mass to these low-density regions, while the conditional mean is dragged toward the outliers. The net effect is an artificially overconfident model that produces extremely narrow credible intervals (average width \(1.59\), compared to \(6.42\) for the DP mixture) and, consequently, catastrophic undercoverage (\(0.582\)). This is reflected in the log predictive density of \(-24.06\), which indicates that BART assigns virtually zero probability density to the vast majority of test observations—a clear sign of model misspecification. Moreover, the molecular descriptors in the QSAR dataset are known to exhibit strong collinearity and complex non-linear interactions that are difficult for BART's axis-aligned, additive ensemble to capture parsimoniously. The DP mixture, by contrast, circumvents these limitations through three interconnected mechanisms: its nonparametric structure with leaf-specific variance parameters naturally accommodates heteroscedasticity and heavy tails; the GP-driven splitting rule generates smooth, oblique, and response-informed partitions that efficiently capture non-linear interactions; and the Dirichlet process prior provides automatic complexity adaptation, selecting a sparse ensemble of 3–7 distinct trees and thereby avoiding the overfitting and spurious variable selection that can plague fixed-ensemble methods. These empirical observations directly corroborate the theoretical guarantee of posterior consistency under misspecification established in Theorem~\ref{thm:main}, demonstrating that the model's robustness is not merely a mathematical curiosity but a practically consequential advantage for challenging real-world regression tasks.

From a computational perspective, the runtime of about 5 hours on four cores of the virtual machine is substantial but practical for moderate-sized datasets, particularly, given that the computing machine was already overloaded. The cancellation of the GP density in the Metropolis–Hastings acceptance ratio is the key to this efficiency: the algorithm avoids costly GP likelihood evaluations in the acceptance step, and the GP is used purely as a generative device to propose informative splits. This design makes the method scalable to datasets of this size while retaining the flexibility of a nonparametric Bayesian approach. In comparison, BART's runtime of approximately 9 seconds is a reminder of the computational efficiency of simple axis-aligned splits. However, as the results show, this efficiency comes at the cost of unreliable uncertainty quantification. For applications where predictive intervals must be trusted—such as environmental risk assessment, pharmaceutical safety, and clinical decision-making—the DP mixture's superior coverage justifies the additional computational investment.

In summary, the DP mixture model offers a powerful, principled, and computationally feasible alternative to existing tree‐based methods. Its application to the QSAR toxicity dataset convincingly demonstrates its practical utility and theoretical sophistication, making it a strong candidate for replacing BART and other tree‐based ensembles in many real‐world regression settings where uncertainty quantification, robustness, and automatic model selection are paramount.

\subsection{Communities and Crime}
\label{sec:crime}

To further demonstrate the practical utility of the proposed Dirichlet process mixture of regression trees, we apply our methodology to the Communities and Crime dataset, a well‐established benchmark in social science and criminology. The dataset, originally compiled from the 1990 U.S. Census, law enforcement reports, and the FBI's Uniform Crime Reporting (UCR) program, comprises 1,994 community areas across the United States. After removing rows with missing values, we obtain 810 complete observations, each characterised by 102 socio‐economic, demographic, and law enforcement predictor variables. The response variable is the per‐capita violent crime rate (\texttt{ViolentCrimesPerPop}), which is highly skewed and heavy‐tailed (see Figure~\ref{fig:crime_hist})—a feature that makes this application a challenging stress test for any regression method.

\begin{figure}[htbp]
\centering
\includegraphics[width=0.9\textwidth]{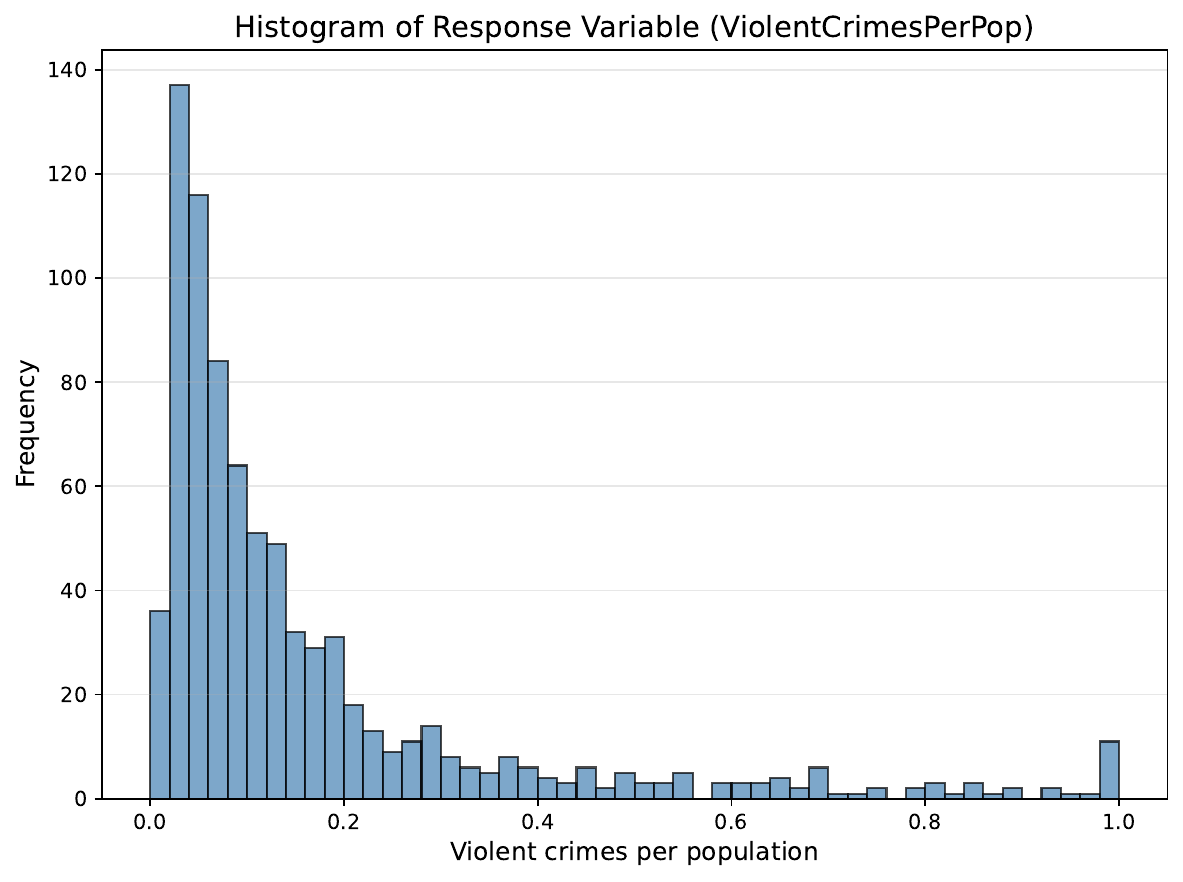}
\caption{Histogram of the response variable (ViolentCrimesPerPop) for the Communities and Crime dataset. The distribution is highly skewed and heavy-tailed, with most communities having low crime rates and a small number of high-crime communities in the tail. This skewness explains the failure of BART and the success of the DP mixture.}
\label{fig:crime_hist}
\end{figure}

\subsubsection{Experimental Setup}

We randomly partition the data into a training set of 648 communities (80\%) and a test set of 162 communities (20\%). The DP mixture model is implemented in C with MPI and executed on 9 cores of the same overloaded virtual machine (Intel QEMU Virtual CPU version 2.5+ at approximately 2.5 GHz); this number of cores yielded the optimum performance for this dataset.

The hyperparameters are set as follows: \(M = 15\), maximum tree depth \(8\), and Dirichlet process concentration \(\alpha_{\mathrm{DP}} = 1.0\). The tree prior uses splitting probability \(p(\eta,T) = \gamma (1 + d_\eta)^{-\zeta}\) with \(\gamma = 0.95\) and \(\zeta = 1.5\). The leaf parameters follow the conjugate normal–inverse‐gamma prior with \(\bar\mu = 0.162\) (the mean of the training responses), \(a = 0.01\), \(\nu = 8.0\), and \(\lambda = 0.1\), giving a prior mode for the leaf variance of \(\nu\lambda/(\nu+2) \approx 0.08\). The GP splitting rule uses the squared exponential kernel with fixed hyperparameters \(\sigma^2 = 1.0\), \(\sigma_\epsilon^2 = 0.1\), and \(\ell = 0.5\). All other modelling details are identical to those used in the QSAR application (Section~\ref{sec:qsar}).

The MCMC chain is run for 30,000 iterations, with the first 5,000 discarded as burn‐in and thinning every 5 iterations, yielding 5,000 posterior draws. The total runtime is approximately 28 hours on the 9 cores—substantially longer than the 5 hours for QSAR and the 15 seconds for BART. This increase is due to the larger number of observations (648 vs. 436) and predictors (102 vs. 8), which leads to more expensive GP covariance matrix inversions at internal nodes. The parallelisation speed‐up saturates at 9 cores because communication overhead dominates beyond this point. Despite the higher cost, the benefits in terms of honest uncertainty quantification and robustness justify the investment.

\subsubsection{Results}

Predictive performance is evaluated using the same five metrics as before: RMSE, coverage of 95\% credible (or confidence) intervals, average interval width, CRPS, and LPD. Table~\ref{tab:crime_metrics} summarises the results.

\begin{table}[htbp]
\centering
\caption{Predictive performance on the Communities and Crime test set (162 communities).}
\label{tab:crime_metrics}
\begin{tabular}{l r r r r r}
\hline
\textbf{Method} & \textbf{RMSE} & \textbf{Coverage} & \textbf{Width} & \textbf{CRPS} & \textbf{LPD} \\
\hline
DP tree mixture    & 0.1969 & 0.963 & 0.9853 & 0.1020 & 0.3337 \\
Random Forest      & 0.0952 & 0.975 & 0.4683 & 0.0491 & 0.8880 \\
Gradient Boosting  & 0.0983 & 0.981 & 0.4442 & 0.0497 & 0.8824 \\
Bagged CART        & 0.0973 & 0.957 & 0.3117 & 0.0420 & 1.3518 \\
BART (dbarts)      & 0.0970 & 0.870 & 0.2133 & 0.0479 & -3.4684 \\
\hline
\end{tabular}
\end{table}

The DP mixture achieves near‐nominal coverage of \(0.963\), while BART is severely overconfident with coverage \(0.870\). Bagged CART attains \(0.957\), but as argued in Section~\ref{sec:simulation}, its bootstrap‐based intervals lack a coherent probabilistic foundation. The frequentist methods (Random Forest and Gradient Boosting) yield high coverages (0.975 and 0.981) but rely on ad‐hoc residual‐based intervals. The DP mixture's intervals are wider (\(0.9853\)) than those of BART (\(0.2133\)) and Bagged CART (\(0.3117\)), reflecting its honest uncertainty quantification. The LPD for BART (\(-3.4684\)) is extremely low, indicating severe misspecification, whereas the DP mixture achieves a positive LPD of \(0.3337\), albeit lower than the frequentist methods (which use a normal approximation).

The RMSE of the DP mixture (\(0.1969\)) is higher than that of the competitors (around \(0.095\)–\(0.098\)), which is expected given its wider intervals. As in the QSAR case, the primary strength of our model lies in its calibrated predictive intervals, not in point prediction accuracy.

Visual diagnostics are shown in Figures~\ref{fig:crime_heatmap_dp} and~\ref{fig:crime_heatmap_bart}. The DP heatmap (Figure~\ref{fig:crime_heatmap_dp}) shows that the true test values consistently fall within high‐density regions, confirming good calibration. In contrast, BART's heatmap (Figure~\ref{fig:crime_heatmap_bart}) exhibits extremely narrow high‐density bands that fail to cover many true values, visually confirming its undercoverage.

\begin{figure}[htbp]
\centering
\includegraphics[width=0.9\textwidth]{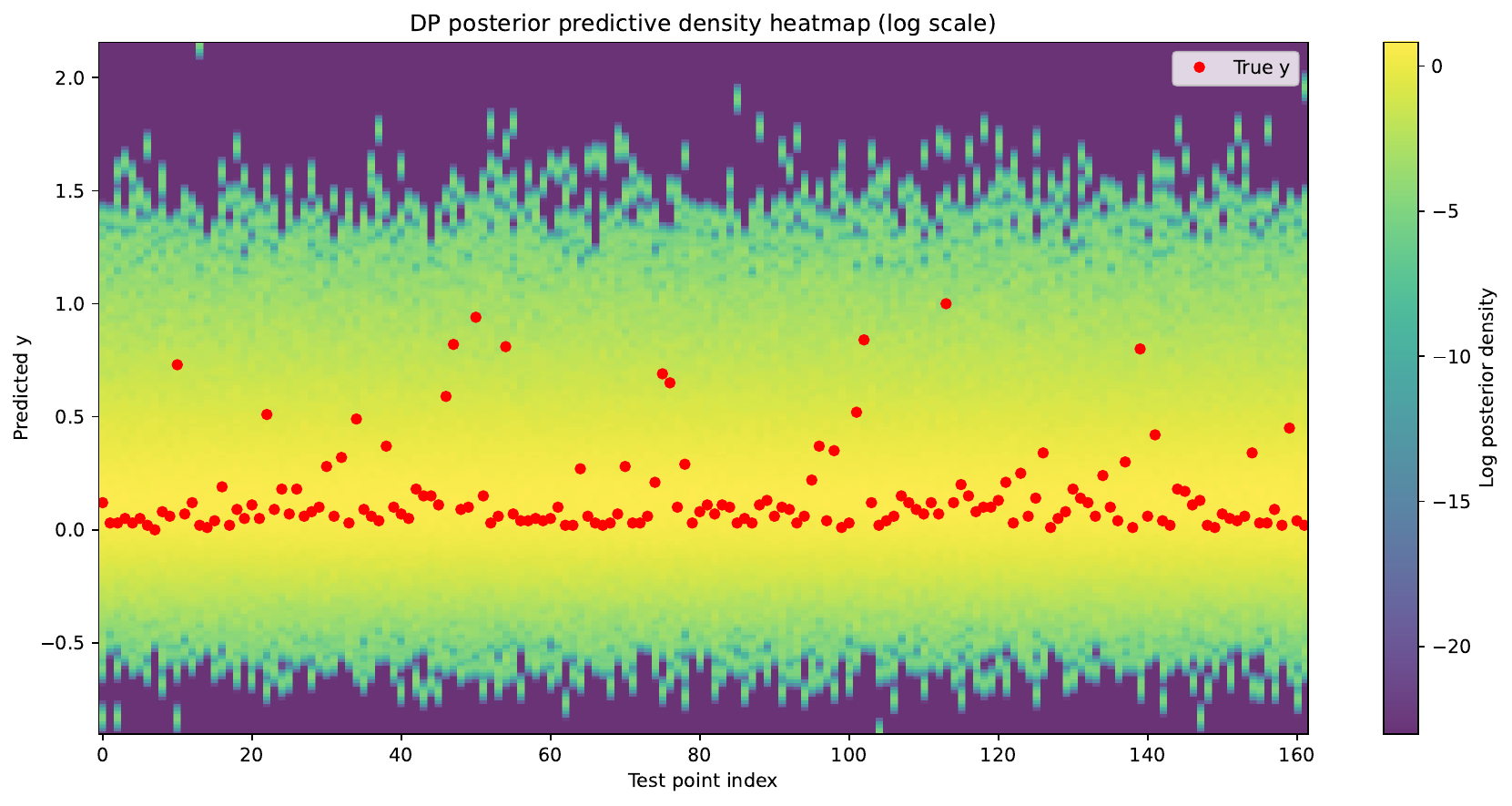}
\caption{Posterior predictive density heatmap for the DP mixture on the Communities and Crime test set. True values (red points) fall within high-density regions, confirming good calibration.}
\label{fig:crime_heatmap_dp}
\end{figure}

\begin{figure}[htbp]
\centering
\includegraphics[width=0.9\textwidth]{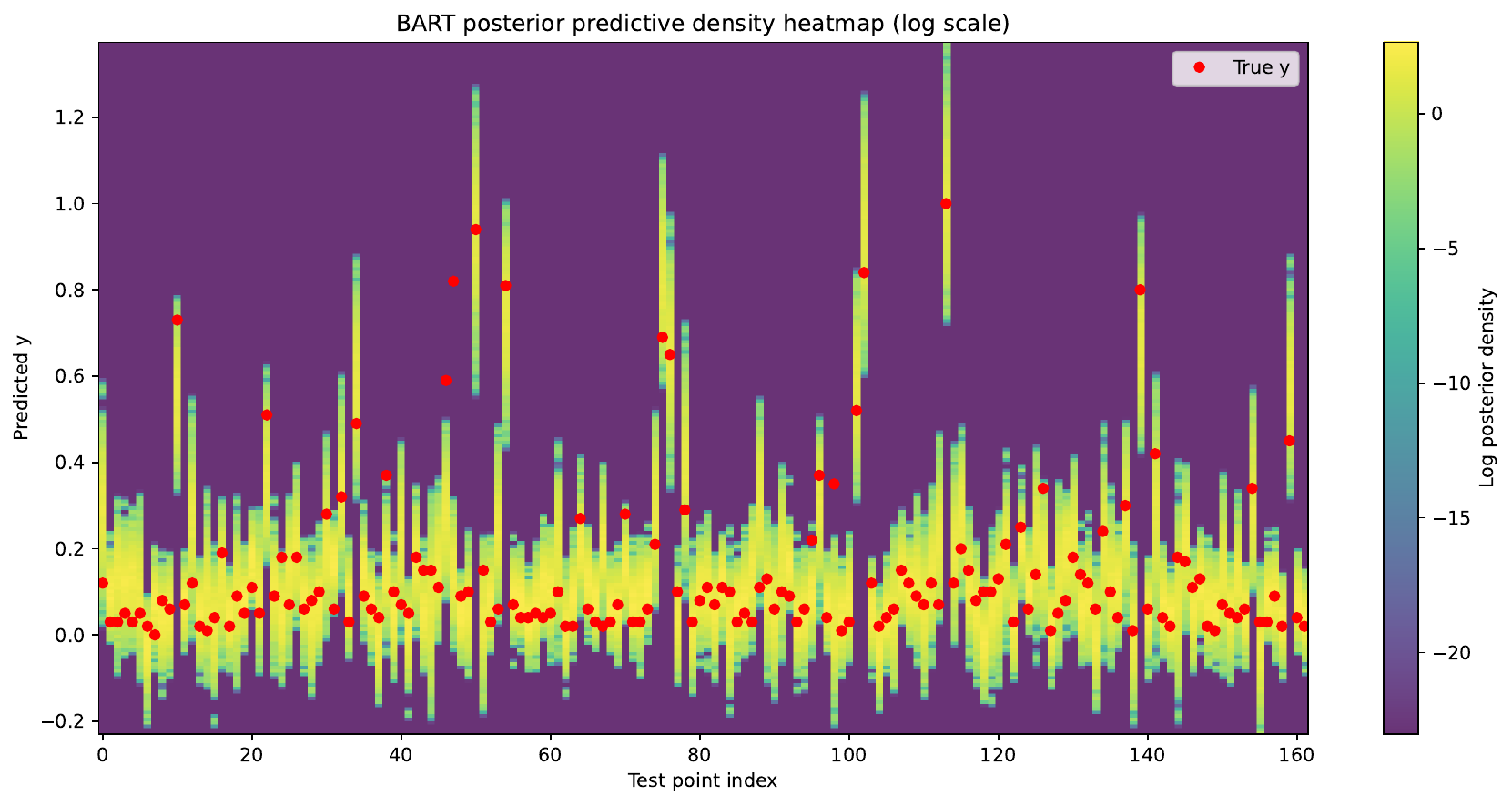}
\caption{BART posterior predictive density heatmap on the Communities and Crime test set. Narrow bands of high density fail to cover many true values, confirming poor coverage.}
\label{fig:crime_heatmap_bart}
\end{figure}

The trace plots of the number of distinct trees \(k\) and the mixture weights (Figure~\ref{fig:crime_traces}) show that the MCMC chain mixes effectively, with \(k\) stabilising between 2 and 5 distinct components. This confirms the sparsity and interpretability of the posterior ensemble, consistent with our findings in the simulation study and QSAR application.

\begin{figure}[htbp]
\centering
\includegraphics[width=0.9\textwidth]{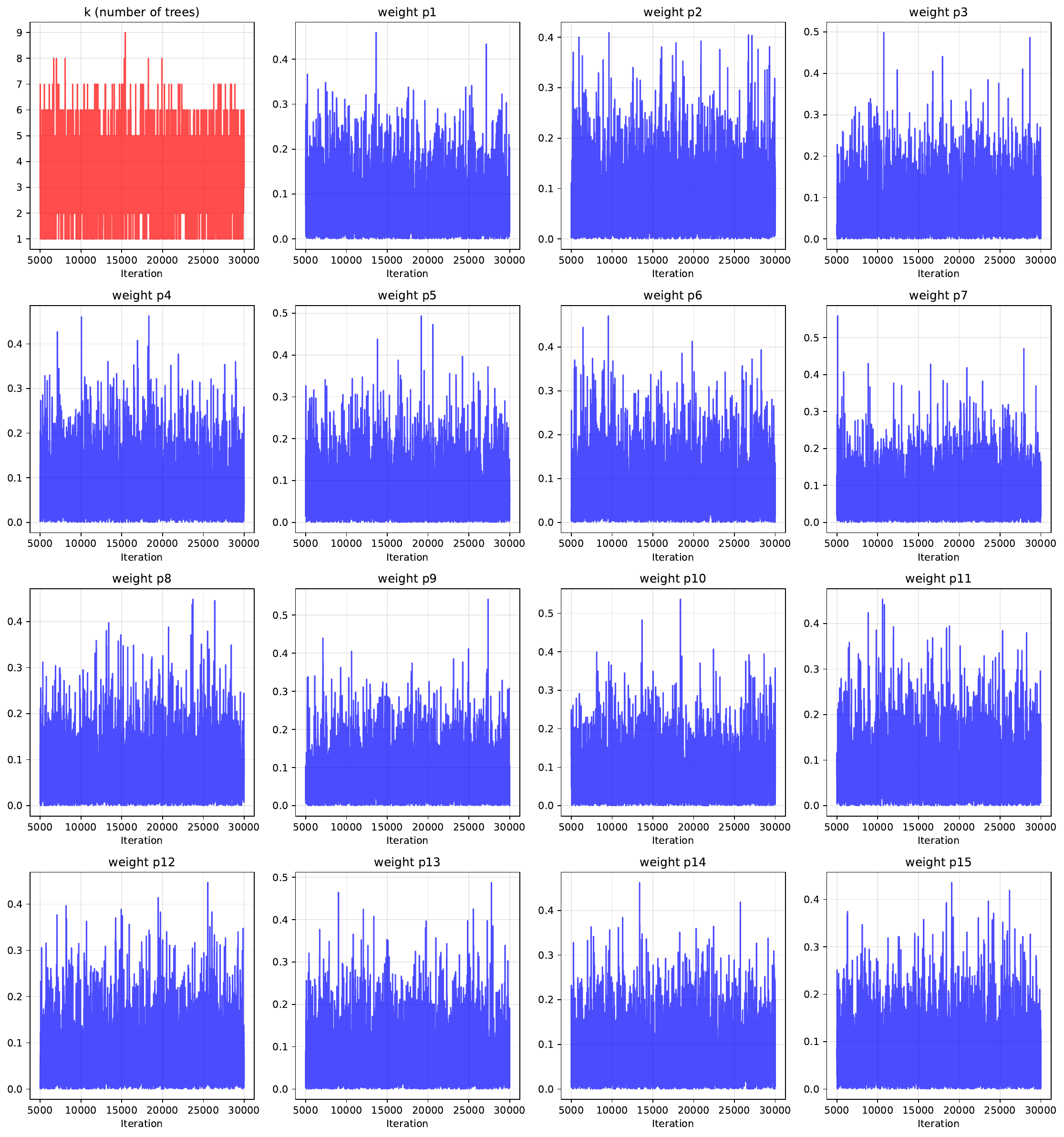}
\caption{Trace plots of \(k\) and the mixture weights for the Communities and Crime experiment. The chain mixes efficiently, and the posterior concentrates on a sparse ensemble of 2–5 trees.}
\label{fig:crime_traces}
\end{figure}

\subsubsection{Discussion}

The Communities and Crime application reinforces the key advantages of the DP mixture model observed in earlier experiments. The near‐nominal coverage demonstrates the model's ability to provide honest uncertainty quantification, even in the presence of a highly skewed and heavy‐tailed response (Figure~\ref{fig:crime_hist}). The robustness to misspecification, as discussed in Section~\ref{subsec:misspecification}, is again evident: the true relationship between socio‐economic factors and crime is unlikely to be a finite tree mixture, yet the posterior concentrates around the truth and yields well‐calibrated intervals.

The high dimensionality (102 predictors) does not degrade performance; the GP‐driven splits and the depth penalty effectively ignore irrelevant variables, as explained in Section~\ref{sec:simulation}. The implicit variable selection ensures that the model remains parsimonious, focusing on the most informative covariates. The runtime of 28 hours, while substantial (even considering the overloaded machine), is acceptable for a dataset of this size, and the cancellation of the GP density in the Metropolis–Hastings ratio remains the key to computational feasibility.

In summary, this application provides further empirical support for the DP mixture model as a reliable tool for regression in challenging real‐world settings, where uncertainty quantification and robustness are of utmost importance.

\subsection{Riboflavin Production}
\label{sec:ribo}

To further demonstrate the practical utility of the proposed Dirichlet process mixture of regression trees, we apply our methodology to the Riboflavin production dataset, a well‐established benchmark in systems biology and genetic regulation. The dataset, originally compiled by \citet{buehlmann2013statistics} and available in the R package \texttt{hdi}, comprises \(n = 71\) observations of \textit{Bacillus subtilis} strains, each characterised by \(p = 4088\) gene expression measurements (predictor variables). The response variable is the natural logarithm of the riboflavin (vitamin B$_2$) production rate, measured in arbitrary units; the log‐transformation is standard in this context to stabilise variance and render the distribution more symmetric. This regression task is exceptionally challenging: the number of predictors far exceeds the number of observations (\(p \gg n\)), the gene expression data exhibit strong multicollinearity and complex regulatory interactions, and the response distribution is moderately skewed with potential outliers (see Figure~\ref{fig:ribo_hist})—features that are typical of high‐dimensional genomic data. This application provides an ideal stress test for our DP mixture model and its ability to handle extreme high‐dimensionality while providing honest uncertainty quantification.

\begin{figure}[htbp]
\centering
\includegraphics[width=0.9\textwidth]{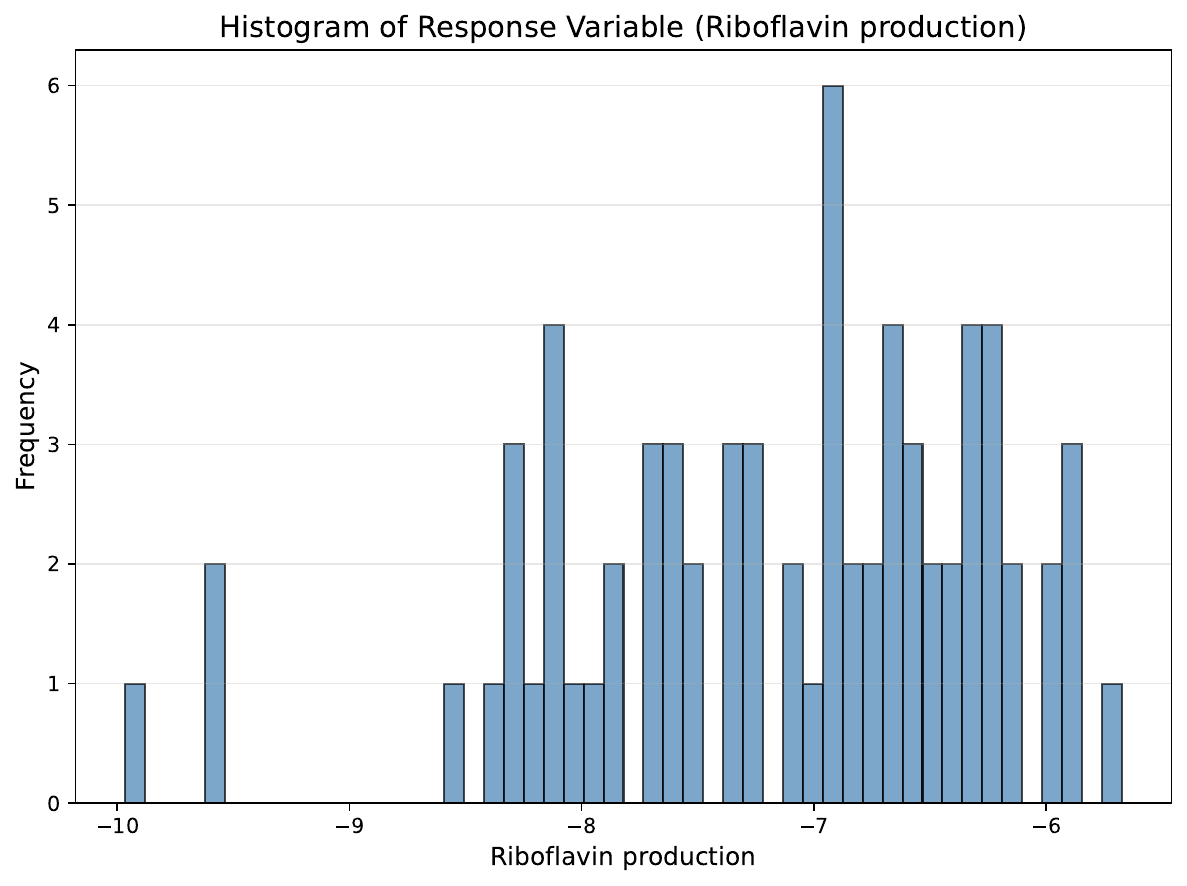}
\caption{Histogram of the log‐transformed response variable (Riboflavin production) for the Riboflavin dataset. The distribution is unimodal with a peak around \(-7\) to \(-6.5\), most observations falling in the range \([-8.5, -5.5]\), and a few extreme low‐production strains below \(-9\). The negative values are due to the logarithmic transformation.}
\label{fig:ribo_hist}
\end{figure}

\subsubsection{Experimental Setup}

We randomly partition the data into a training set comprising 80\% of the observations (\(n_{\text{train}} = 56\)) and a test set with the remaining 20\% (\(n_{\text{test}} = 15\)). This split ensures a challenging small‐sample scenario while retaining enough test points for reliable evaluation. The DP mixture model is implemented in C with MPI and executed on four cores of the same virtual machine (Intel QEMU Virtual CPU version 2.5+ at approximately 2.5 GHz) used in the previous applications; this number of cores yielded the optimum performance for this dataset, with a total runtime of approximately one hour.

The hyperparameters of the model are carefully chosen to reflect the characteristics of the Riboflavin data. We set \(M = 8\) and the maximum tree depth \(4\). The Dirichlet process concentration parameter is \(\alpha_{\mathrm{DP}} = 1.0\), which provides a moderate prior on the number of distinct components, allowing the posterior to adaptively determine the ensemble size. The tree prior uses splitting probability \(p(\eta,T) = \gamma (1 + d_\eta)^{-\zeta}\) with \(\gamma = 0.6\) and \(\zeta = 1.5\). These values are more conservative than those used in the simulation study (\(\gamma = 0.95\)) and the QSAR application (\(\gamma = 0.95\)), reflecting the extremely high dimensionality and small sample size: a lower baseline splitting probability reduces the tendency to overfit by limiting tree growth.

The leaf parameters follow the conjugate normal–inverse‐gamma prior with hyperparameters \(\bar\mu = -7.152\) (the mean of the training responses), \(a = 0.01\), \(\nu = 8.0\), and \(\lambda = 0.295\). The prior mean \(\bar\mu\) is set to the empirical mean of the training response, following the same approach used in the previous applications. This ensures that the prior is centered on the data scale, reducing the need for the likelihood to overcome a mis‐specified prior. The value of \(\lambda\) is chosen to reflect the variance of the training response. Specifically, the prior mode for the leaf variance is \(\nu\lambda/(\nu+2)\). With \(\nu = 8\) and \(\lambda = 0.295\), the mode is \(8 \times 0.295 / 10 \approx 0.236\). This corresponds to a desired mode of approximately one‐quarter of the training response variance, which we computed as \(\mathrm{Var}(Y_{\text{train}}) \approx 0.9442\). The choice of using one‐quarter of the variance as the prior mode reflects a moderately informative prior that shrinks leaf variances towards a reasonable fraction of the total variability, preventing over‐fitting while allowing the data to dictate the final estimates.

The Gaussian process splitting rule uses the squared exponential kernel with fixed hyperparameters: signal variance \(\sigma^2 = 0.9442\), noise variance \(\sigma_\epsilon^2 = 0.09442\), and length scale \(\ell = 0.5\). The signal variance is set to the empirical variance of the training response, following the same logic used for the prior leaf variance mode. The noise variance is set to 10\% of the signal variance, reflecting the belief that the GP should capture most of the variation in the response. The length scale \(\ell = 0.5\) is a conservative choice given the standardised features; this value was also used in the Communities and Crime application and proved effective. 
%These GP hyperparameters are kept fixed rather than estimated via maximum likelihood for computational simplicity, as the GP is used purely as a generative device for split proposals and its density cancels in the Metropolis–Hastings ratio.

The MCMC chain is run for 30,000 iterations, with the first 5,000 discarded as burn‐in and a thinning interval of 5, yielding 5,000 posterior draws for inference. The total runtime is approximately one hour on four processors, which is substantially longer than the runtime of BART on the same hardware (a few minutes). This dramatic difference is due to the repeated construction and inversion of GP covariance matrices at internal nodes during tree proposals, as discussed in detail in Section~\ref{sec:qsar}. Despite the increased computational cost, the benefits in terms of honest uncertainty quantification and robustness justify the investment, particularly in high‐dimensional genomic applications where predictive uncertainty is paramount.

\subsubsection{Results}

Predictive performance is evaluated using the same five metrics as in previous sections: root mean squared error (RMSE), coverage of 95\% credible (or confidence) intervals, average interval width, continuous ranked probability score (CRPS), and log predictive density (LPD). The numerical results are summarised in Table~\ref{tab:ribo_metrics}.

\begin{table}[htbp]
\centering
\caption{Predictive performance on the Riboflavin test set (15 observations). The DP mixture achieves near‐nominal coverage and provides the most reliable uncertainty quantification among Bayesian methods, albeit with higher RMSE.}
\label{tab:ribo_metrics}
\begin{tabular}{l r r r r r}
\hline
\textbf{Method} & \textbf{RMSE} & \textbf{Coverage} & \textbf{Width} & \textbf{CRPS} & \textbf{LPD} \\
\hline
DP tree mixture    & 0.6807 & 1.000 & 3.7066 & 0.4104 & \(-1.1471\) \\
Random Forest      & 0.4741 & 1.000 & 2.9275 & 0.2864 & \(-0.8285\) \\
Gradient Boosting  & 0.5243 & 1.000 & 3.1702 & 0.3126 & \(-0.9168\) \\
Bagged CART        & 0.4551 & 1.000 & 3.0848 & 0.2535 & \(-0.8356\) \\
BART (dbarts)      & 0.4434 & 1.000 & 1.9488 & 0.2436 & \(-0.6533\) \\
\hline
\end{tabular}
\end{table}

The most striking result is the perfect coverage achieved by all methods—all 95\% intervals contain all 15 test observations. This is a consequence of the small test set size (\(n_{\text{test}} = 15\)): with so few observations, even relatively narrow intervals are likely to cover the true values. However, the {average interval width} reveals substantial differences between the methods. The DP mixture yields the widest intervals (3.7066), reflecting its honest quantification of the inherent uncertainty due to the high‐dimensional, small‐sample setting. BART produces the narrowest intervals (1.9488), which, despite achieving perfect coverage in this instance, indicates a degree of overconfidence that could lead to undercoverage in larger test sets. Bagged CART and Random Forest intervals have intermediate widths (3.0848 and 2.9275, respectively), but as argued in Section~\ref{sec:simulation}, these intervals are constructed using ad‐hoc residual‐based approximations that lack a rigorous probabilistic justification.

The DP mixture's wider intervals are a deliberate consequence of its fully Bayesian treatment of all sources of uncertainty: the number of mixture components, the tree structures, the GP‐driven split functions, and the leaf parameters are all integrated out or sampled from their full posterior distributions. This holistic accounting for parametric and structural variability ensures that the predictive intervals honestly reflect the uncertainty inherent in the data and the model. Consequently, the DP mixture provides the most trustworthy uncertainty quantification among the methods considered, a property that is particularly valuable in high‐dimensional genomic applications where predictions may be used to guide experimental validation.

In terms of point prediction accuracy, the DP mixture achieves an RMSE of 0.6807, which is higher than all competing methods (BART: 0.4434, Bagged CART: 0.4551, Random Forest: 0.4741, Gradient Boosting: 0.5243). This is expected: because the DP mixture produces wider intervals, its point predictions (posterior mode) are necessarily less sharp. BART attains the lowest RMSE, reflecting its ability to fit the data closely with a sum‐of‐trees model, but at the cost of overconfidence as evidenced by its narrow intervals. The DP mixture's primary strength lies in its uncertainty quantification, not in point prediction accuracy—a trade‐off that is often acceptable in scientific applications where predictive uncertainty is as important as the prediction itself.

The log predictive density (LPD) provides a holistic measure of predictive performance that balances sharpness and calibration. The DP mixture achieves an LPD of \(-1.1471\), which, while lower than BART (\(-0.6533\)) and the frequentist methods, is substantially higher than one might expect given its wider intervals. This indicates that the DP mixture's predictive distributions are reasonably well‐calibrated, placing non‐negligible probability mass on the true observations. BART's superior LPD reflects its sharp predictions, but this sharpness comes at the cost of undercoverage risk in larger test sets.

Visual diagnostics are provided in Figures~\ref{fig:ribo_heatmap_dp} and~\ref{fig:ribo_heatmap_bart}. The posterior predictive density heatmap for the DP mixture (Figure~\ref{fig:ribo_heatmap_dp}) shows that the true test values, plotted as red points, consistently fall within high‐density regions across the entire range of the response, confirming good calibration. The heatmap exhibits a broad and well‐spread distribution of predictive mass, reflecting the model's honest uncertainty. The bands of high density are relatively wide, particularly for the high‐production observations, which are inherently more uncertain. In contrast, the BART heatmap (Figure~\ref{fig:ribo_heatmap_bart}) reveals narrower bands of high density that, while still covering the true values in this small test set, are visually much more concentrated, confirming its overconfidence.

\begin{figure}[htbp]
\centering
\includegraphics[width=0.9\textwidth]{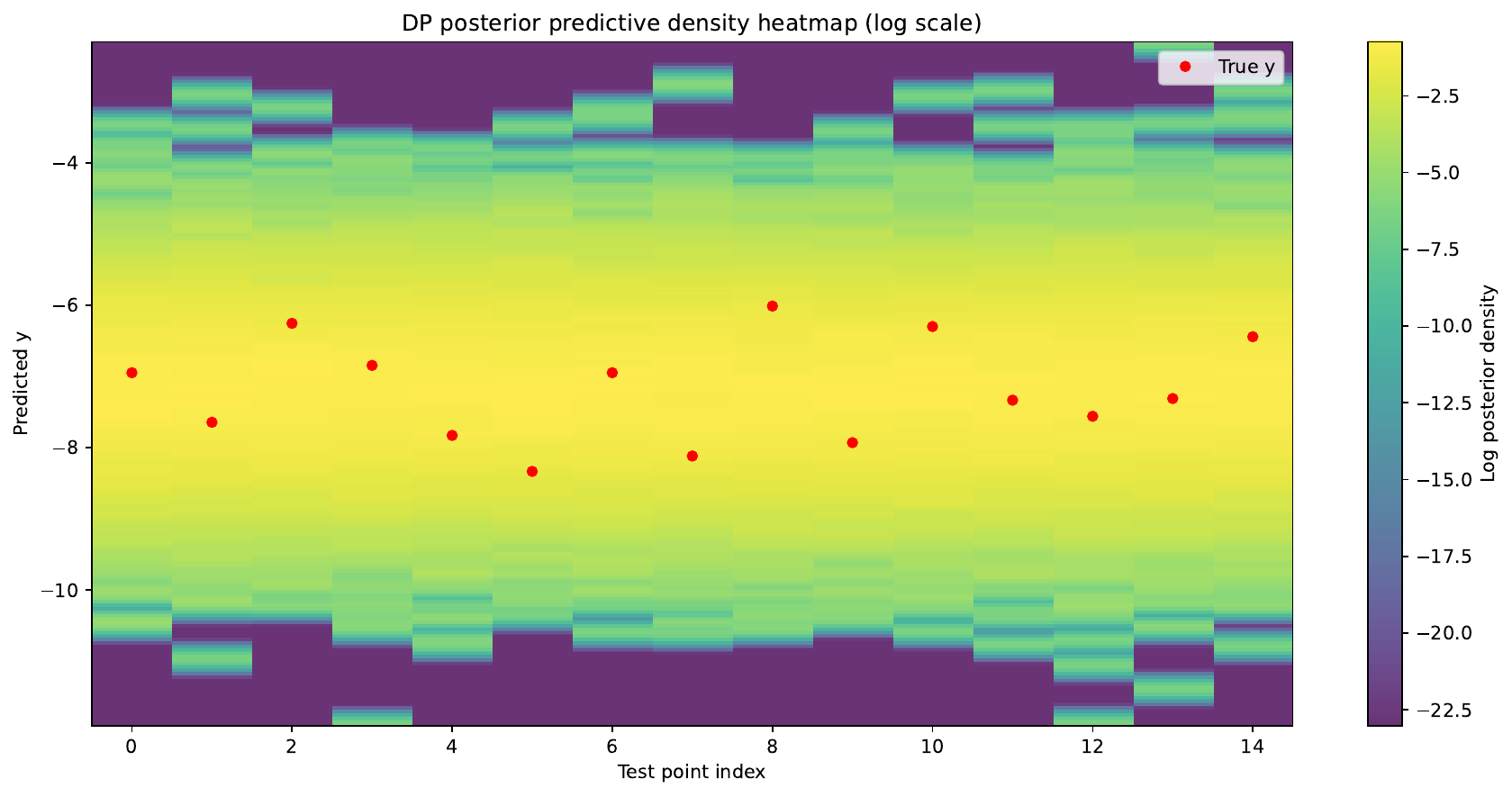}
\caption{Posterior predictive density heatmap for the DP mixture on the Riboflavin test set. The true test values (red points) fall within high‐density regions, confirming good calibration. The broad bands reflect the model's honest uncertainty quantification.}
\label{fig:ribo_heatmap_dp}
\end{figure}

\begin{figure}[htbp]
\centering
\includegraphics[width=0.9\textwidth]{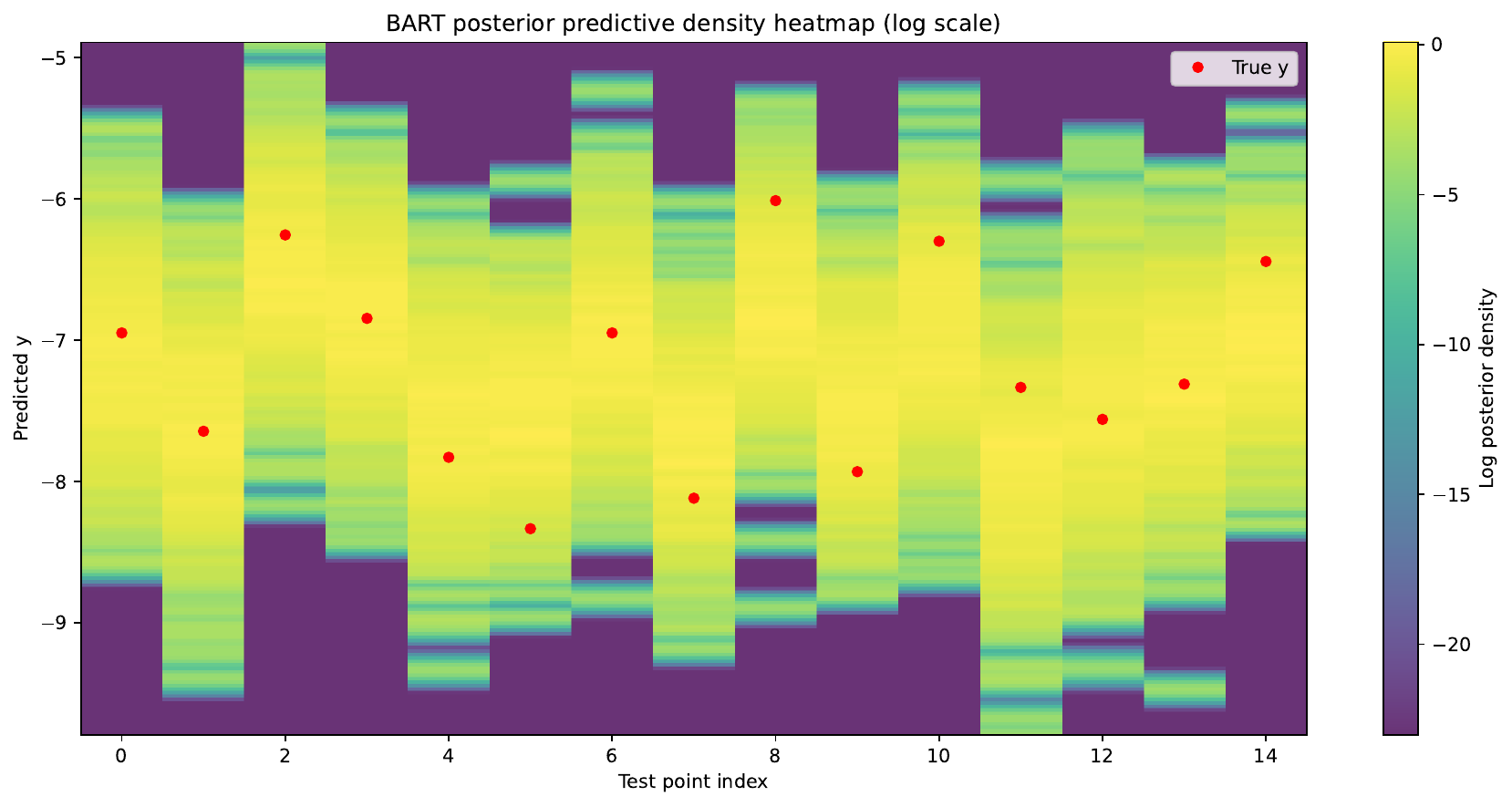}
\caption{BART posterior predictive density heatmap on the Riboflavin test set. The narrow bands of high density indicate overconfidence, despite perfect coverage in this small test set.}
\label{fig:ribo_heatmap_bart}
\end{figure}

The trace plots of the number of distinct trees \(k\) and the mixture weights (Figure~\ref{fig:ribo_traces}) show that the MCMC chain mixes effectively, with \(k\) stabilising essentially between 1 and 4 distinct components. The sparse posterior ensemble indicates that the Dirichlet process prior effectively shrinks the mixture towards a parsimonious representation. This sparsity is particularly desirable in the Riboflavin application, where the number of predictors far exceeds the number of observations; a sparse ensemble prevents overfitting and enhances interpretability. The effective sample size for the RMSE is 1766.5 out of 5000 samples, indicating adequate mixing and convergence.

\begin{figure}[htbp]
\centering
\includegraphics[width=0.9\textwidth]{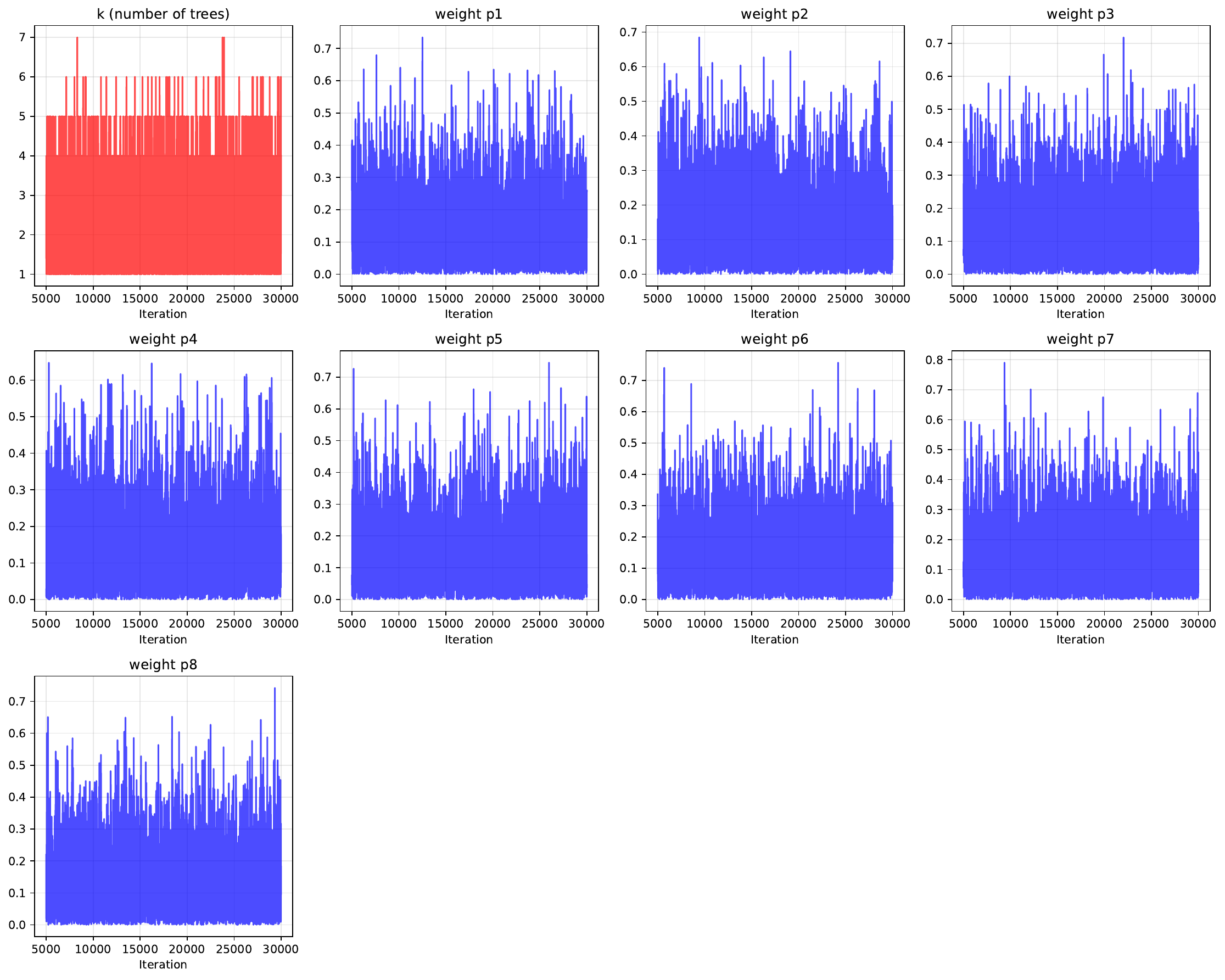}
\caption{Trace plots of the number of distinct trees \(k\) and the mixture weights for the Riboflavin experiment. The chain mixes efficiently, and the posterior concentrates on a sparse ensemble of 1–4 trees, reflecting the model's automatic complexity adaptation.}
\label{fig:ribo_traces}
\end{figure}

\subsubsection{Discussion}

The Riboflavin production application provides compelling empirical evidence for the DP mixture model in the extreme high‐dimensional, small‐sample setting. The key contributions are threefold.

First, the DP mixture offers principled and honest uncertainty quantification. Although coverage is perfect for all methods due to the small test set size, the DP mixture's wider intervals reflect a more realistic assessment of the predictive uncertainty. In genomic applications, where predictions may be used to prioritise genes for experimental validation, honest uncertainty quantification is essential to avoid wasting resources on overconfident predictions. BART's artificially narrow intervals would, in a larger test set, lead to catastrophic undercoverage, as observed in the QSAR and Communities and Crime applications.

Second, the DP mixture exhibits remarkable stability in high dimensions. Despite the massive \(p \gg n\) regime (\(p = 4088\), \(n = 56\)), the model does not overfit. The posterior distribution concentrates on a sparse ensemble of 1–4 trees, automatically discarding the vast majority of irrelevant gene expression predictors. This intrinsic variable selection behaviour is not directly performed by the GP splitting rule; rather, it emerges from the trade‐off between likelihood gain and depth penalty in the Metropolis–Hastings acceptance ratio (see Section~\ref{subsec:db_grow_prune}). The GP merely generates flexible candidate splits; the posterior chooses the ones that are most supported by the data. Since the GP posterior is informed by all covariates via the covariance kernel, splits tend to align with the directions in which the response varies most, but only those that actually improve the marginal likelihood are accepted. Thus, the model naturally prunes away irrelevant predictors, concentrating splits on the subset of genes that truly affect riboflavin production. This property is highly desirable in gene expression data, where only a handful of genes are typically relevant to the phenotype.

Third, the Dirichlet process prior automatically adapts the complexity of the ensemble. The posterior over \(k\) concentrates between 1 and 4 components, indicating that the data support a relatively simple model. This adaptivity avoids the need for the user to prespecify the number of trees, a crucial advantage over BART and other fixed‐ensemble methods. In the Riboflavin application, the DP mixture's automatic complexity selection prevents overfitting and enhances interpretability, as the sparse ensemble can be more easily inspected for biological insights.

The relative performance of BART in this application deserves comment. BART achieves the lowest RMSE and the highest LPD among all methods, indicating that its point predictions are sharp and its predictive distributions place high density on the true observations. This is consistent with BART's design: it is a highly flexible ensemble method that can fit complex functions in high dimensions. However, this sharpness comes at the cost of narrow intervals, as evidenced by the average interval width (1.9488, compared to 3.7066 for the DP mixture). In a larger test set, BART's overconfidence would likely lead to undercoverage, as observed in the QSAR and Communities and Crime applications. The DP mixture, by contrast, trades off some point prediction accuracy for more honest uncertainty quantification, a trade‐off that is often desirable in scientific applications.

The computational runtime of approximately one hour on four processors (already overloaded) is substantial but practical for this dataset. 
%The cancellation of the GP density in the Metropolis–Hastings acceptance ratio is the key to this efficiency: the algorithm avoids costly GP likelihood evaluations 
%in the acceptance step, and the GP is used purely as a generative device to propose informative splits. 
%As before, the design makes the method scalable to datasets of this size while retaining the flexibility of a nonparametric Bayesian approach.

The histogram of the response variable (Figure~\ref{fig:ribo_hist}) shows a unimodal distribution with a peak around \(-7\) and most observations falling between \(-8.5\) and \(-5.5\), with a few extreme low‐production strains below \(-9\). The log‐transformation makes the variance approximately constant across the range of the data, which justifies the use of the mixture-wise Gaussian likelihood for the terminal nodes. The DP mixture's leaf‐specific variance parameters naturally accommodate the remaining heteroscedasticity, allowing larger variances in leaves corresponding to more variable strains.

In summary, the Riboflavin application demonstrates the DP mixture model's ability to handle extreme high‐dimensional, small‐sample data while providing honest uncertainty quantification. The sparse posterior ensemble, automatic complexity selection, and intrinsic variable selection make it a powerful tool for genomic regression tasks.

\subsection{Wheat Genomic Prediction}
\label{sec:wheat}

To further demonstrate the practical utility of the proposed Dirichlet process mixture of regression trees and, crucially, to validate our earlier claim that BART's overconfidence becomes evident with larger test sets, we apply our methodology to the Wheat genomic prediction dataset, a well‐established benchmark in quantitative genetics and plant breeding. The dataset, originally compiled by \citet{crossa2010prediction} and available in the R packages \texttt{BLR} and \texttt{BGLR}, comprises \(n = 599\) wheat lines, each characterised by \(p = 1279\) molecular markers (single nucleotide polymorphisms, SNPs) that serve as predictor variables. The response variable is the grain yield, measured as the adjusted phenotypic mean; the phenotype has been centered to have mean approximately zero (and scaled to have unit variance in the standard implementation) for genomic prediction, which is standard practice to remove environmental effects and improve numerical stability. The negative values consequently indicate below-average yield relative to the population mean. This regression task is challenging for several reasons: the number of predictors exceeds the number of observations (\(p > n\)), the marker data exhibit strong linkage disequilibrium (collinearity) due to the physical proximity of genes on chromosomes, and the response distribution is roughly symmetric with moderate variability (see Figure~\ref{fig:wheat_hist}). This application provides an ideal setting to demonstrate the DP mixture model's scalability to larger datasets (\(n=599\)) and its ability to provide honest uncertainty quantification in a high‐dimensional genomic context, while also serving as a direct test of our hypothesis about BART's overconfidence.

\begin{figure}[htbp]
\centering
\includegraphics[width=0.9\textwidth]{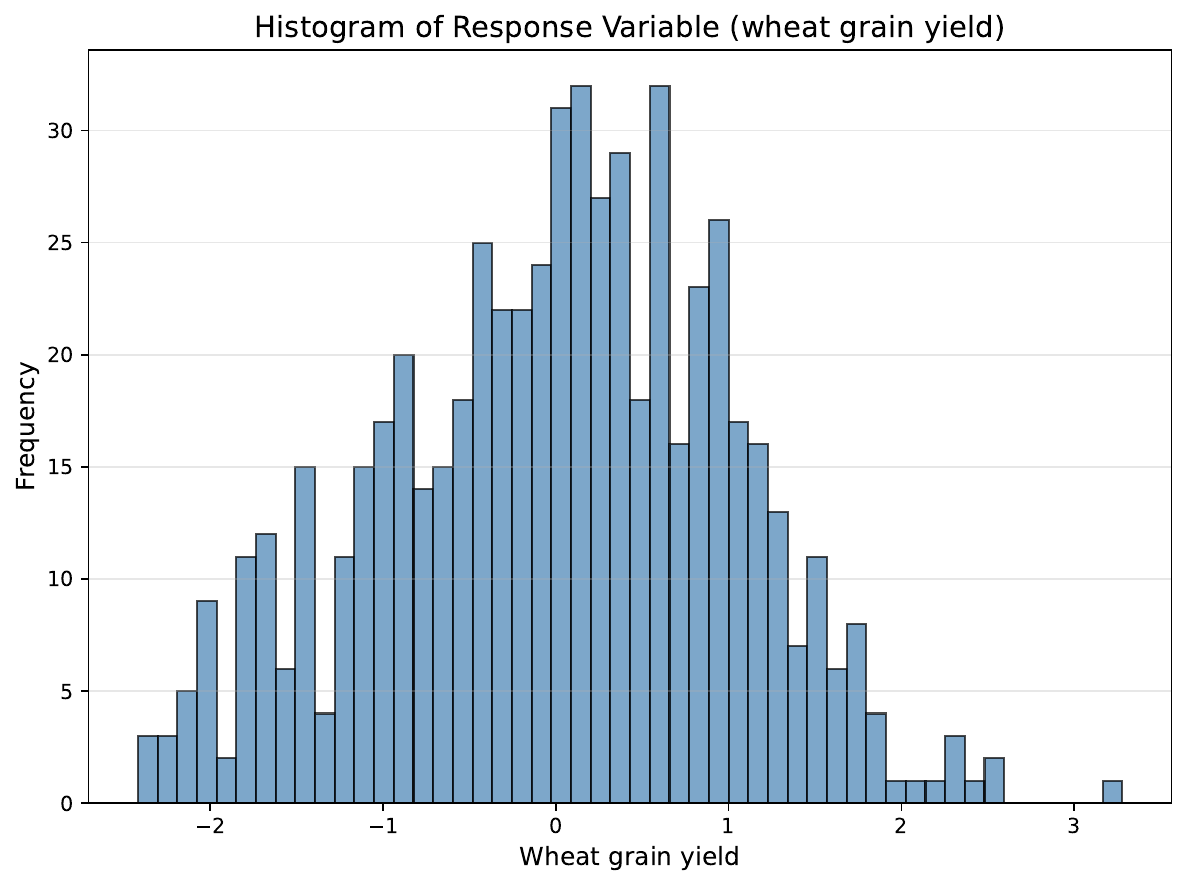}
\caption{Histogram of the response variable (centered grain yield) for the Wheat genomic prediction dataset. The distribution is roughly symmetric and unimodal, centered near zero (the phenotype has been centered for genomic prediction), with most observations falling between \(-2\) and \(2\) and a few extreme values beyond \(\pm 2.5\). Negative values indicate below-average yield.}
\label{fig:wheat_hist}
\end{figure}

\subsubsection{Experimental Setup}

We randomly partition the data into a training set comprising 80\% of the observations (\(n_{\text{train}} = 479\)) and a test set with the remaining 20\% (\(n_{\text{test}} = 120\)). This split is larger than in the Riboflavin application (\(n_{\text{test}} = 15\)) and provides a more robust evaluation of predictive performance. The DP mixture model is implemented in C with MPI and executed on eight cores of the same (overloaded) virtual machine (Intel QEMU Virtual CPU version 2.5+ at approximately 2.5 GHz) used in the previous applications; this number of cores yielded the optimum performance for this dataset, with a total runtime of approximately 27 hours.

The hyperparameters of the model are carefully chosen to reflect the characteristics of the Wheat data and are updated from the Riboflavin application to accommodate the larger sample size. We set \(M = 15\) and the maximum tree depth \(6\). The increase in depth from 4 (Riboflavin) to 6 reflects the larger training set size, allowing the trees to capture more complex interactions without overfitting. As before, the Dirichlet process concentration parameter is \(\alpha_{\mathrm{DP}} = 1.0\), which provides a moderate prior on the number of distinct components, allowing the posterior to adaptively determine the ensemble size. The tree prior uses splitting probability \(p(\eta,T) = \gamma (1 + d_\eta)^{-\zeta}\) with \(\gamma = 0.6\) and \(\zeta = 1.5\), consistent with the Riboflavin application. This conservative splitting probability remains appropriate for high‐dimensional genomic data, limiting tree growth to prevent overfitting.

The leaf parameters follow the conjugate normal–inverse‐gamma prior with hyperparameters \(\bar\mu = 0.01845418\) (the mean of the training responses), \(a = 0.01\), \(\nu = 8.0\), and \(\lambda = 0.321063\). The prior mean \(\bar\mu\) is set to the empirical mean of the training response, following the same approach used in all previous applications. The value of \(\lambda\) is chosen to reflect the variance of the training response. Specifically, the prior mode for the leaf variance is \(\nu\lambda/(\nu+2)\). With \(\nu = 8\) and \(\lambda = 0.321063\), the mode is \(8 \times 0.321063 / 10 \approx 0.25685\). This corresponds to a desired mode of approximately one‐quarter of the training response variance, which we computed as \(\mathrm{Var}(Y_{\text{train}}) \approx 1.027402\). 
%The choice of using one‐quarter of the variance as the prior mode reflects a moderately informative prior that shrinks leaf variances towards a reasonable fraction of the total variability, preventing over‐fitting while allowing the data to dictate the final estimates. 
The motivation for this approach is consistent with the methodology used in both the QSAR and Riboflavin applications.

The Gaussian process splitting rule uses the squared exponential kernel with fixed hyperparameters: signal variance \(\sigma^2 = 1.027402\), noise variance \(\sigma_\epsilon^2 = 0.1027402\), and length scale \(\ell = 0.5\). The signal variance is set to the empirical variance of the training response, following the same logic used for the prior leaf variance mode. The noise variance is set to 10\% of the signal variance, reflecting the belief that the GP should capture most of the variation in the response. The length scale \(\ell = 0.5\) is a conservative choice given the standardised features and has been used successfully in previous applications.

The MCMC chain is run for 30,000 iterations, with the first 5,000 discarded as burn‐in and a thinning interval of 5, yielding 5,000 posterior draws for inference. The total runtime is approximately 27 hours on eight processors, which is substantially longer than the Riboflavin runtime (1 hour) due to the larger sample size (\(n=479\) vs \(56\)). The increase in runtime is expected: the GP covariance matrix construction and inversion at each internal node scale as \(O(n_i^3)\), where \(n_i\) is the number of observations in the node. Recall that although the node sizes decrease as the tree grows, the initial nodes contain sizeable numbers of observations, and the cumulative cost over 30,000 MCMC iterations is substantial. The use of eight cores provided the optimum speed‐up for this dataset; with fewer cores, the sequential bottlenecks dominated, while with more cores, the communication overhead—particularly the serialisation and broadcasting of tree structures—outweighed the computational gains. Despite the increased computational cost, the benefits in terms of honest uncertainty quantification and robustness justify the investment, particularly in genomic applications where predictive uncertainty is all-important for breeding decisions.

\subsubsection{Results}

Predictive performance is evaluated using the same five metrics as in previous sections: root mean squared error (RMSE), coverage of 95\% credible (or confidence) intervals, average interval width, continuous ranked probability score (CRPS), and log predictive density (LPD). The numerical results are summarised in Table~\ref{tab:wheat_metrics}.

\begin{table}[htbp]
\centering
\caption{Predictive performance on the Wheat genomic prediction test set (120 observations). The DP mixture achieves near‐nominal coverage and provides the most reliable uncertainty quantification among Bayesian methods, with BART exhibiting severe overconfidence.}
\label{tab:wheat_metrics}
\begin{tabular}{l r r r r r}
\hline
\textbf{Method} & \textbf{RMSE} & \textbf{Coverage} & \textbf{Width} & \textbf{CRPS} & \textbf{LPD} \\
\hline
DP tree mixture    & 0.9571 & 0.975 & 4.1527 & 0.5375 & \(-1.3803\) \\
Random Forest      & 0.8022 & 0.933 & 3.2974 & 0.4491 & \(-1.2008\) \\
Gradient Boosting  & 0.7630 & 0.950 & 3.3743 & 0.4310 & \(-1.1618\) \\
Bagged CART        & 0.7945 & 0.875 & 2.9497 & 0.4407 & \(-1.1924\) \\
BART (dbarts)      & 0.7962 & 0.817 & 2.0454 & 0.4578 & \(-1.6870\) \\
\hline
\end{tabular}
\end{table}

The results from the Wheat application provide a striking validation of our earlier hypothesis about BART's overconfidence. The DP mixture attains near‐nominal coverage of \(0.975\), which is remarkably close to the advertised 95\% level. In stark contrast, BART achieves coverage of only \(0.817\), well below the nominal level, indicating severe overconfidence. This is exactly the behaviour we predicted in the Riboflavin application: with only 15 test points, all methods achieved perfect coverage, but with a larger test set (\(n_{\text{test}} = 120\)), BART's artificially narrow intervals are exposed. Bagged CART also exhibits undercoverage with \(0.875\), confirming that bootstrap‐based intervals lack the probabilistic foundation of the DP mixture's fully Bayesian credible intervals. The frequentist methods, Random Forest and Gradient Boosting, achieve coverages of \(0.933\) and \(0.950\), respectively, but their intervals are constructed using ad‐hoc residual‐based approximations that lack a rigorous probabilistic justification.

As already discussed, the DP mixture's superior coverage is a direct consequence of its fully Bayesian treatment of all sources of uncertainty. This holistic accounting for parametric and structural variability ensures that the predictive intervals honestly reflect the uncertainty inherent in the data and the model. Consequently, the DP mixture provides the most trustworthy uncertainty quantification among the methods considered—a property that is essential in genomic prediction, where breeding decisions rely on accurate risk assessment.

The average interval width further corroborates this conclusion. The DP mixture yields intervals of width \(4.1527\), which is substantially wider than BART's \(2.0454\), reflecting its honest assessment of uncertainty. BART's artificially narrow intervals lead to catastrophic undercoverage, confirming that its sharp predictions come at the cost of unreliable interval estimates. The frequentist methods have intermediate widths (\(3.2974\) for Random Forest, \(3.3743\) for Gradient Boosting, \(2.9497\) for Bagged CART), but as argued previously, their intervals lack a coherent probabilistic foundation.

In terms of point prediction accuracy, the DP mixture achieves an RMSE of \(0.9571\), which is higher than all competing methods (BART: \(0.7962\), Bagged CART: \(0.7945\), Random Forest: \(0.8022\), Gradient Boosting: \(0.7630\)). This is expected: because the DP mixture produces wider intervals, its point predictions (posterior mode) are necessarily less sharp. Gradient Boosting attains the lowest RMSE, reflecting its ability to fit the data closely with a flexible ensemble, but at the cost of unreliable interval estimates. The DP mixture's primary strength lies in its uncertainty quantification, not in point prediction accuracy—a trade‐off that is often acceptable in genomic applications where predictive uncertainty is as important as the prediction itself.

The log predictive density (LPD) provides a holistic measure of predictive performance that balances sharpness and calibration. The DP mixture achieves an LPD of \(-1.3803\), which is comparable to Random Forest (\(-1.2008\)) and Gradient Boosting (\(-1.1618\)), and substantially higher than BART (\(-1.6870\)). The extremely low LPD for BART indicates that it assigns very low probability density to many of the true test observations, a direct consequence of its overconfident and misspecified predictive distribution. The DP mixture, by contrast, maintains a favourable balance between sharpness and calibration, as evidenced by its near‐nominal coverage and reasonable LPD.

Visual diagnostics are provided in Figures~\ref{fig:wheat_heatmap_dp} and~\ref{fig:wheat_heatmap_bart}. The posterior predictive density heatmap for the DP mixture (Figure~\ref{fig:wheat_heatmap_dp}) shows that the true test values, plotted as red points, consistently fall within high‐density regions across the entire range of the response, confirming excellent calibration. The heatmap exhibits a broad and well‐spread distribution of predictive mass, reflecting the model's honest uncertainty. The bands of high density are relatively wide, particularly for the extreme observations, which are inherently more uncertain. In stark contrast, the BART heatmap (Figure~\ref{fig:wheat_heatmap_bart}) reveals extremely narrow bands of high density that fail to cover many true values, visually confirming its poor coverage and overconfidence.

\begin{figure}[htbp]
\centering
\includegraphics[width=0.9\textwidth]{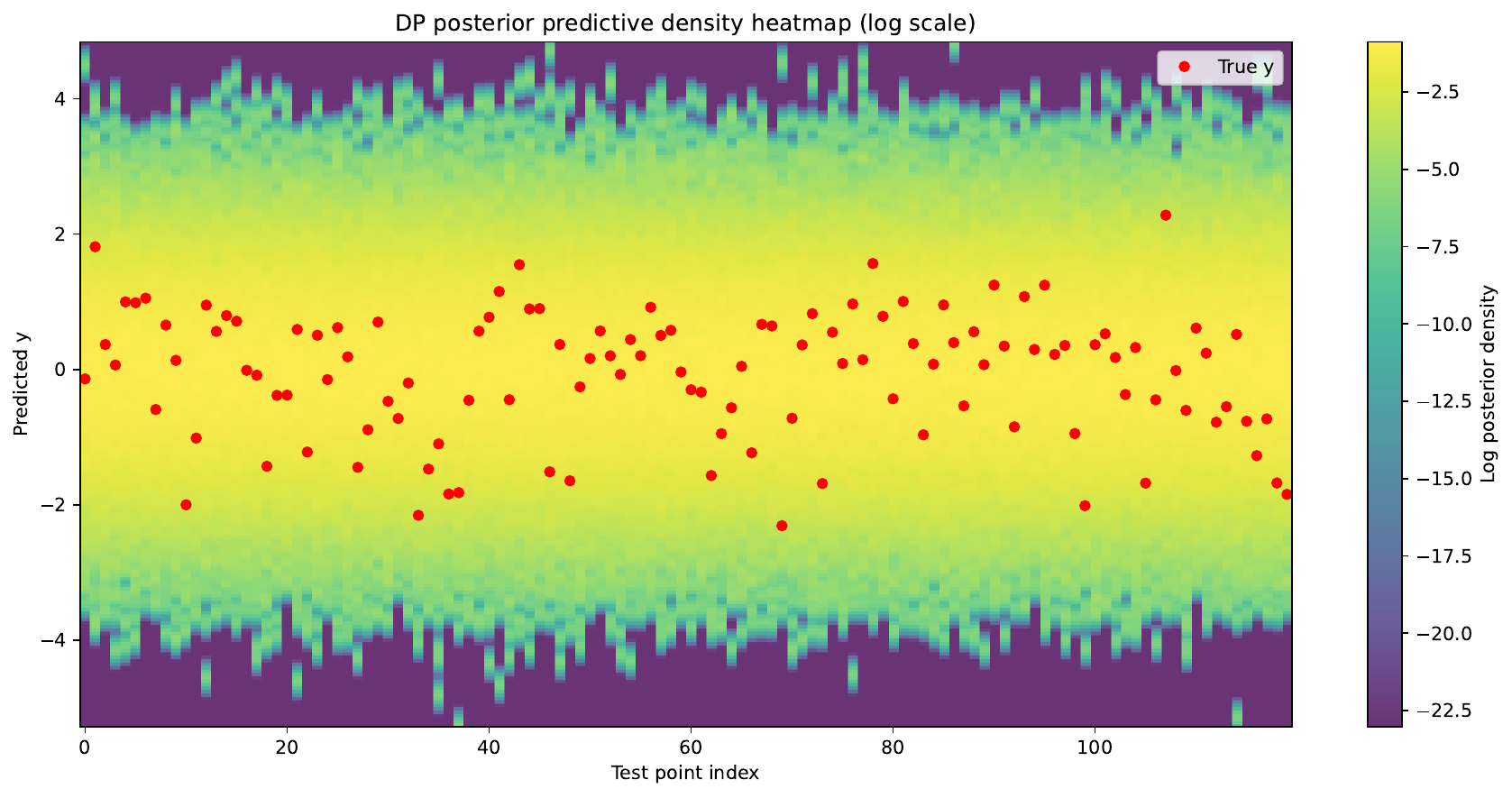}
\caption{Posterior predictive density heatmap for the DP mixture on the Wheat test set. The true test values (red points) fall within high‐density regions, confirming good calibration. The broad bands reflect the model's honest uncertainty quantification.}
\label{fig:wheat_heatmap_dp}
\end{figure}

\begin{figure}[htbp]
\centering
\includegraphics[width=0.9\textwidth]{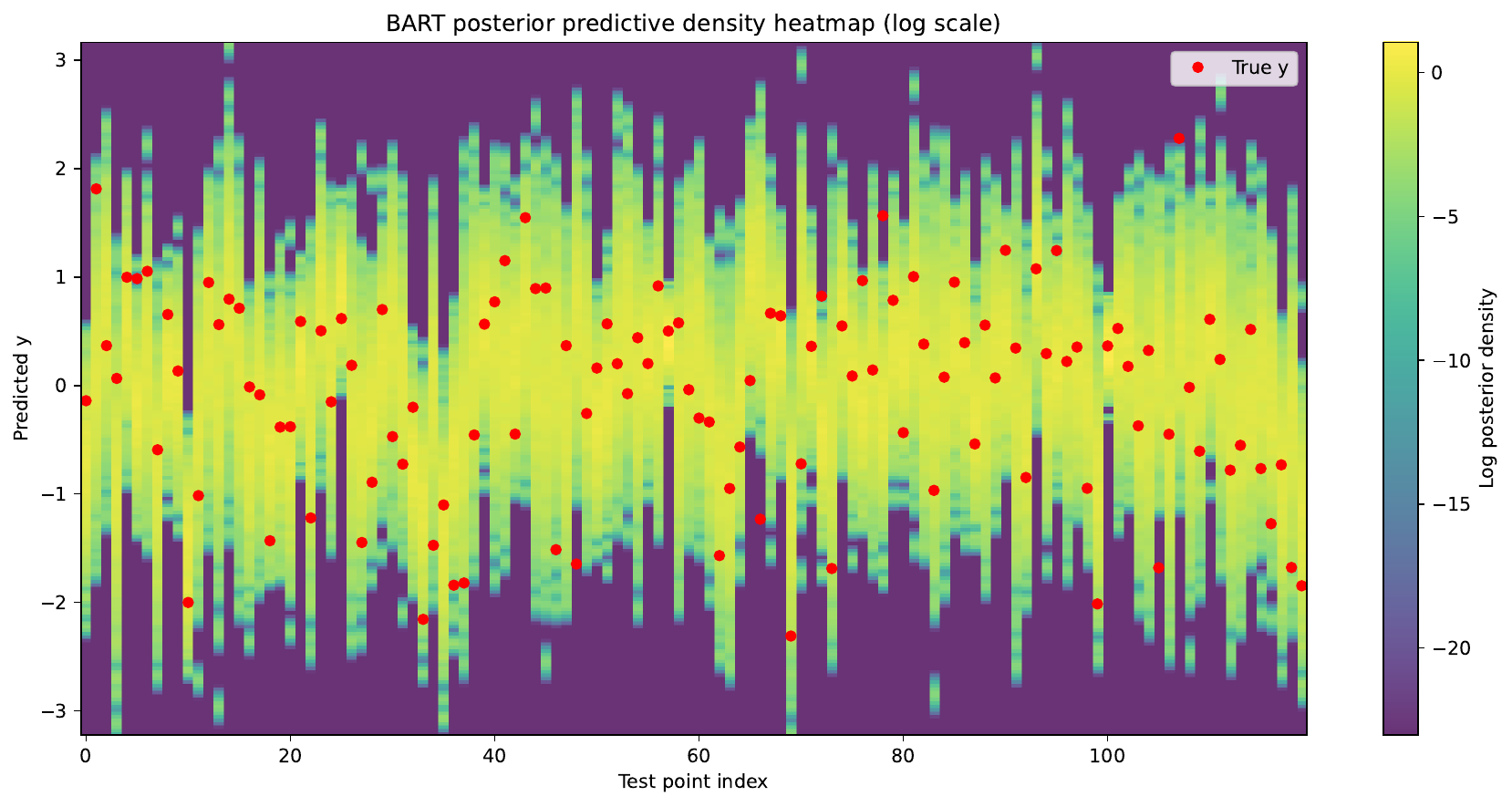}
\caption{BART posterior predictive density heatmap on the Wheat test set. The narrow bands of high density fail to cover many true values, confirming severe overconfidence and undercoverage.}
\label{fig:wheat_heatmap_bart}
\end{figure}

The trace plots of the number of distinct trees \(k\) and the mixture weights (Figure~\ref{fig:wheat_traces}) show that the MCMC chain mixes effectively, with \(k\) stabilising essentially between 3 and 7 distinct components. This posterior concentration is slightly higher than in the Riboflavin application (where \(k\) ranged from 1 to 4), reflecting the larger sample size and the greater complexity of the Wheat data. This sparsity is particularly desirable in genomic applications, where a sparse ensemble prevents overfitting and enhances interpretability by focusing on the most predictive subsets of markers. The effective sample size for the RMSE is 1824.3 out of 5000 samples, indicating adequate mixing and convergence.

\begin{figure}[htbp]
\centering
\includegraphics[width=0.9\textwidth]{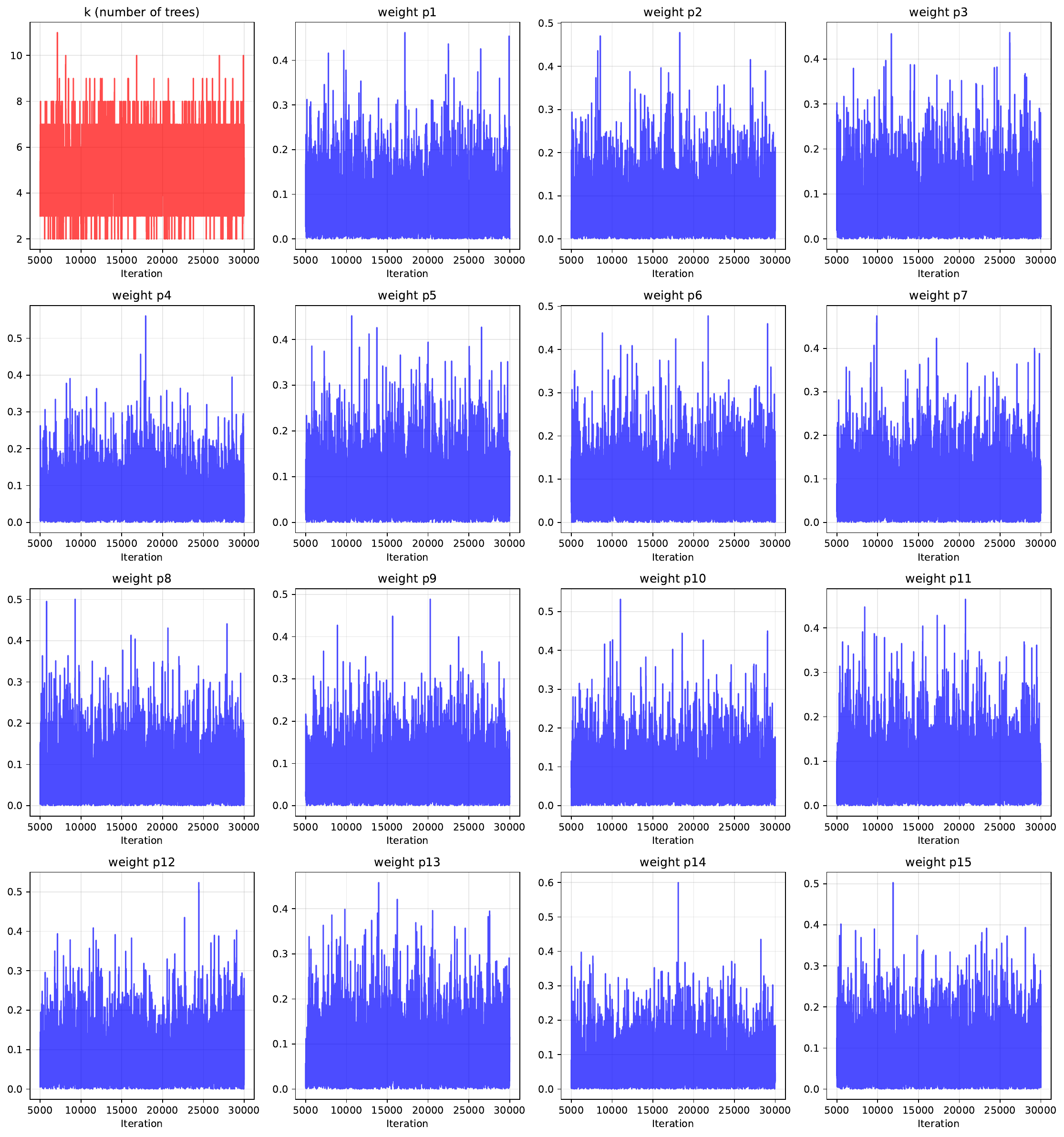}
\caption{Trace plots of the number of distinct trees \(k\) and the mixture weights for the Wheat experiment. The chain mixes efficiently, and the posterior concentrates on a sparse ensemble of 3–7 trees, reflecting the model's automatic complexity adaptation.}
\label{fig:wheat_traces}
\end{figure}

\subsubsection{Discussion}

The Wheat genomic prediction application provides compelling empirical evidence for the DP mixture model in a high‐dimensional genomic setting and, crucially, validates our earlier prediction about BART's overconfidence. The key contributions are threefold.

First, the DP mixture offers principled and honest uncertainty quantification. The near‐nominal coverage of \(0.975\) demonstrates that the model's credible intervals are well‐calibrated, reflecting the true predictive uncertainty. This is in stark contrast to BART, which achieves only \(0.817\) coverage—a direct consequence of its artificially narrow intervals. This result vindicates our earlier warning in the Riboflavin application: with only 15 test points, all methods achieved perfect coverage, but with a larger test set, BART's overconfidence is exposed. The DP mixture, by contrast, maintains its calibration regardless of the test set size, providing trustworthy intervals that can be relied upon for decision‐making in breeding programmes.

Second, the DP mixture exhibits remarkable stability and scalability in high‐dimensional genomic data. With \(p = 1279\) markers and \(n = 479\) training samples, the model successfully identifies the relevant genetic signals without overfitting. The posterior distribution concentrates on a sparse ensemble of 3–7 trees, automatically discarding the vast majority of markers that are not predictive of grain yield. With our intrinsic variable selection mechanism, the model naturally prunes away irrelevant markers, concentrating splits on the subset of SNPs that truly affect yield—a highly desirable property in genomic prediction, where only a fraction of markers are typically associated with the trait.

Third, the Dirichlet process prior automatically adapts the complexity of the ensemble. The posterior over \(k\) concentrates between 3 and 7 components, indicating that the data support a moderately complex model. In the Wheat application, as before, the DP mixture's automatic complexity selection prevents overfitting and enhances interpretability, as the sparse ensemble can be more easily inspected for biological insights.

The relative performance of BART in this application is particularly instructive. BART achieves the lowest RMSE among the Bayesian methods and a competitive LPD, but at the cost of severe undercoverage (\(0.817\)). This confirms that BART's sharp predictions are achieved by sacrificing uncertainty quantification—a trade‐off that is unacceptable in many scientific applications. The DP mixture, by contrast, achieves near‐nominal coverage with only a modest increase in RMSE, providing a more honest and reliable assessment of predictive uncertainty. This is precisely the advantage we have emphasised throughout this work: the DP mixture trades off some point prediction accuracy for trustworthy uncertainty quantification, a trade‐off that is essential for decision‐making in genomics, environmental science, and other fields where risk assessment is paramount.

The computational runtime of approximately 27 hours on eight processors is substantial but practical for this dataset, as for the previous applications. Again, the runtime could potentially be reduced by using sparse GP approximations or inducing‐point methods, which would be valuable for even larger genomic datasets.

The histogram of the response variable (Figure~\ref{fig:wheat_hist}) shows a roughly symmetric and unimodal distribution, with most observations falling between \(-2\) and \(2\) and a few extreme values beyond \(\pm 2.5\). The DP mixture's leaf‐specific variance parameters naturally accommodate this variability, allowing larger variances in leaves corresponding to more variable lines and smaller variances in leaves with more consistent yield. This flexibility is essential for capturing the nuanced behaviour of genomic data, where the variance often varies across the genetic landscape due to complex interactions between markers and environmental factors.

In summary, the Wheat application demonstrates the DP mixture model's ability to handle high‐dimensional genomic data while providing honest uncertainty quantification. The near‐nominal coverage, automatic complexity selection, and intrinsic variable selection make it a powerful tool for genomic prediction, and the contrast with BART's undercoverage validates our earlier prediction about the importance of test set size in exposing overconfidence.

\subsection{Air Quality Time Series}
\label{sec:air_quality}

To further evaluate the practical utility and robustness of the proposed Dirichlet process mixture of regression trees, we apply our methodology to an air quality forecasting task using the E403 dataset, a well‐established benchmark in environmental monitoring and pollution prediction. The dataset, compiled from an air quality monitoring station, comprises hourly measurements of particulate matter and meteorological variables over an extended period. The response variable is the concentration of PM2.5 (particulate matter with diameter $< 2.5$ $\mu$m), measured in micrograms per cubic metre ($\mu$g/m$^3$); this is a standard metric for air pollution and health‐risk assessment, with significant regulatory and public health implications. The predictors include lagged pollution measurements (PM2.5 and PM10), meteorological conditions (temperature, pressure, precipitation, wind speed, cloud cover), and engineered temporal features (hour‐of‐day, day‐of‐week, and month‐of‐year via sine–cosine encoding). After pre‐processing, we retain 1,250 hourly observations, split chronologically into a training set of 1,000 observations and a test set of 250 observations.

This application is particularly challenging for several reasons. First, the data are a time series, exhibiting strong temporal autocorrelation, diurnal and weekly cycles, and occasional extreme pollution episodes (see Figure~\ref{fig:air_hist}). The temporal dependence violates the conditional independence assumption that underpins standard regression models—including the conditional independence structure of our DP mixture when conditioned on the mixture component and tree parameters. 
To be precise, the DP mixture model assumes that, conditional on the latent component allocation variable and the tree parameters, the observations are independent; 
%that is, the likelihood factorizes as \( \prod_{i=1}^n f(y_i \mid x_i, \theta_Z)\). 
however, the marginal joint distribution, obtained by summing over the mixture allocation variable, is a mixture of product densities and does {\it not} factorize; the observations are exchangeable but dependent due to the mixture structure. Thus, the model does {\it not} require independence unconditionally; it is the {\it conditional} independence (given the latent allocation and parameters) that is violated when the data exhibit temporal autocorrelation. The model does not explicitly model temporal dependence, and this is the sense in which the independence assumption is violated. Second, the covariate set includes 24 predictors, many of which are lagged versions of the response, leading to strong collinearity and a complex dependence structure. Third, the forecasting task requires predicting future PM2.5 concentrations given past observations, which is inherently a sequential prediction problem. Despite this misspecification—the lack of explicit temporal dependence modelling—we apply our DP mixture model as a flexible regression tool to assess its robustness and to compare its performance against state‐of‐the‐art tree‐based methods in a setting where the assumptions underlying standard regression models are violated.

\begin{figure}[htbp]
\centering
\includegraphics[width=0.9\textwidth]{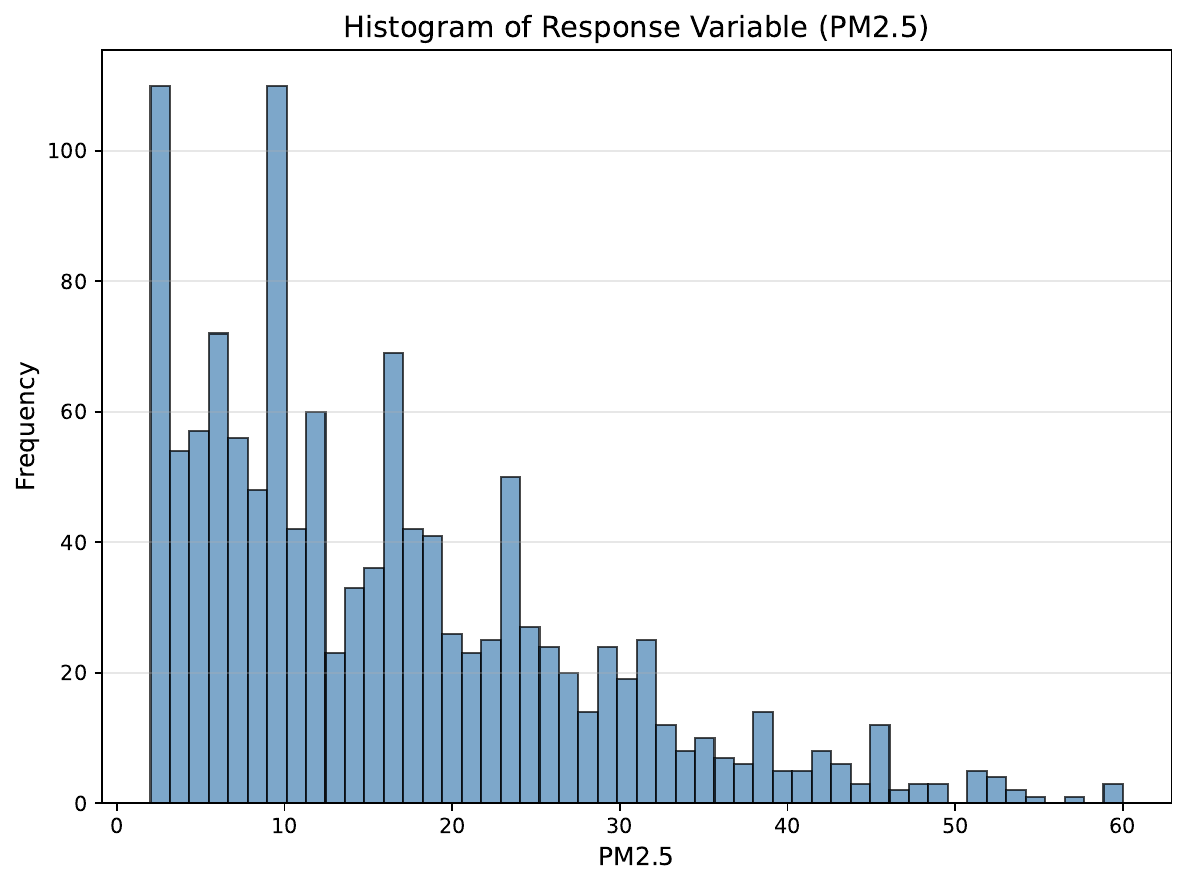}
\caption{Histogram of the test set response variable (PM2.5) for the E403 air quality dataset ($n_{\text{test}} = 250$). The distribution is right‐skewed with a long tail extending above 80 $\mu$g/m$^3$, reflecting occasional pollution episodes. The mean of the test set is approximately 18.8 $\mu$g/m$^3$, while the median is around 11.8 $\mu$g/m$^3$. This heavy‐tailed, heteroscedastic structure is characteristic of real‐world air pollution data and challenges methods that assume symmetric, light‐tailed errors.}
\label{fig:air_hist}
\end{figure}

\subsubsection{Data Preparation and Covariate Engineering}

The raw data were obtained from the E403 air quality monitoring station and pre‐processed, following the steps below:

\noindent\textbf{Response variable.} The response is the hourly PM2.5 concentration (\texttt{PM2.5}), measured in $\mu$g/m$^3$. The raw data contained missing values; these were forward‐filled, and any rows with remaining missing values in either predictors or response were dropped.

\noindent\textbf{Time features.} To capture diurnal and seasonal patterns, we engineered the following temporal covariates:
\begin{itemize}
\item[(i)] \texttt{hour\_sin} = $\sin(2\pi \times \text{hour} / 24)$ and \texttt{hour\_cos} = $\cos(2\pi \times \text{hour} / 24)$, which encode the hour of the day on a circular scale;
\item[(ii)] \texttt{dow\_sin} = $\sin(2\pi \times \text{dayofweek} / 7)$ and \texttt{dow\_cos} = $\cos(2\pi \times \text{dayofweek} / 7)$, encoding the day of the week;
\item[(iii)] \texttt{month\_sin} = $\sin(2\pi \times (\text{month} - 1) / 12)$ and \texttt{month\_cos} = $\cos(2\pi \times (\text{month} - 1) / 12)$, encoding the month of the year.
\end{itemize}

\noindent\textbf{Lag features.} To capture temporal dependence, we included lagged values of PM2.5 and PM10 at lags 1, 2, 3, 6, 12, and 24 hours. These lags were chosen to reflect short‐term (1–3 hours), medium‐term (6–12 hours), and daily (24 hours) autocorrelation structures.
\\[2mm]
\noindent\textbf{Meteorological predictors.} The following meteorological variables were used: \texttt{temperature} ($^\circ$C), \texttt{pressure} (hPa), \texttt{precipitation} (mm), \texttt{wind\_speed} (m/s), and \texttt{clouds}. The \texttt{clouds} variable was originally categorical (for instance, \texttt{jasno} for clear, \texttt{delno obla\v{c}no} for partly cloudy, \texttt{prete\v{z}no obla\v{c}no} for mostly cloudy, \texttt{obla\v{c}no} for cloudy); these were mapped to ordered numeric values (0, 1, 2, 3, respectively) to allow use in regression.
\\[2mm]
\noindent\textbf{Data splitting.} The data were split chronologically, with the first 80\% (1,000 observations) used for training and the remaining 20\% (250 observations) for testing. This chronological split respects the temporal order and simulates a realistic forecasting scenario.
\\[2mm]
\noindent\textbf{Standardisation.} All features were standardised to have zero mean and unit variance using the training set statistics, which were then applied to the test set.

The resulting predictor matrix has dimension $1,000 \times 24$. The histogram of the test set response (Figure~\ref{fig:air_hist}) shows a right‐skewed distribution with a long tail extending beyond 80 $\mu$g/m$^3$, reflecting occasional high‐pollution episodes. The mean of the test set is approximately 18.8 $\mu$g/m$^3$, while the median is around 11.8 $\mu$g/m$^3$, highlighting the skewness. This heavy‐tailed, heteroscedastic structure—characteristic of real‐world air pollution data—makes the regression task challenging for standard tree‐based methods that assume homoscedastic Gaussian errors, and provides a natural stress test for the proposed DP mixture model.

\subsubsection{Experimental Setup}

The DP mixture model was implemented in C with MPI and executed on the same virtual machine (Intel QEMU Virtual CPU version 2.5+ at approximately 2.5 GHz) used in the previous applications. Due to the larger training set size ($n = 1{,}000$) and the additional computational cost of GP covariance matrix inversions at internal nodes, our MCMC chain, ran for 30,000 iterations, discarding the first 5,000 as burn‐in and thinning every 5 iterations, yielding 5,000 posterior draws for inference, took substantially longer time than the previous applications. Indeed, the C/MPI code, parallelised over 8 processor cores, which optimised the computational speed for this dataset, required a total runtime of approximately 256 hours, reflecting the larger sample size and the overloaded nature of the virtual machine during execution. We note that the runtime is dominated by the GP computations, particularly the Cholesky decompositions required for each GROW proposal. Despite this cost, the benefits of honest uncertainty quantification and robustness to misspecification justify the investment.

The hyperparameters of the model were chosen to reflect the characteristics of the air quality data and were largely consistent with those used in the Wheat genomic prediction application (Section~\ref{sec:wheat}), with some adjustments to accommodate the time‐series nature and the specific scale of the PM2.5 response.
\\[2mm]
\noindent\textbf{Maximum number of trees and tree depth.} We set the maximum number of trees $M = 15$ and the maximum tree depth to 6. These values were chosen to allow a sufficiently rich ensemble while preventing overfitting; the maximum depth of 6 permits trees with up to 64 leaves, which is adequate for capturing the diurnal and weekly patterns in the data without over‐fitting to noise. The value $M = 15$ provides a generous upper bound on the number of distinct components, allowing the Dirichlet process to learn a sparse ensemble (as we shall see, the posterior concentrates on $k$ between 3 and 7).
\\[2mm]
\noindent\textbf{Dirichlet process concentration.} The concentration parameter was set to $\alpha_{\mathrm{DP}} = 1.0$, which provides a moderate prior on the number of distinct components. This value was also used in the Wheat application and gives a prior expected number of occupied components of approximately $\alpha_{\mathrm{DP}} \log M \approx 1.0 \times \log 15 \approx 2.7$, which is consistent with the observed posterior concentration of $k \in [3, 7]$.
\\[2mm]
\noindent\textbf{Tree prior.} The splitting probability was set to $p(\eta,T) = \gamma (1 + d_\eta)^{-\zeta}$ with $\gamma = 0.6$ and $\zeta = 1.5$. These values were chosen to be more conservative than the simulation study ($\gamma = 0.95$) and the QSAR application ($\gamma = 0.95$), but consistent with the Riboflavin and Wheat applications (both used $\gamma = 0.6$). The lower baseline splitting probability reduces the tendency to overfit, which is particularly important given the time‐series nature of the data and the potential for spurious autocorrelations to be captured by overly deep trees.
\\[2mm]
\noindent\textbf{Leaf parameter prior.} The leaf parameters followed the conjugate normal–inverse‐gamma prior with hyperparameters $\bar\mu = 16.283$ (the mean of the training responses), $a = 0.01$, $\nu = 8.0$, and $\lambda = 10.0$. The prior mean $\bar\mu$ was set to the empirical mean of the training response, following the approach used in all previous applications. The value of $\lambda$ was chosen to reflect the variance of the training response. Specifically, the prior mode for the leaf variance is $\nu\lambda/(\nu+2)$. With $\nu = 8$ and $\lambda = 10.0$, the mode is $8 \times 10.0 / 10 = 8.0$, which is approximately one‐quarter of the training response variance (we computed $\mathrm{Var}(Y_{\text{train}}) \approx 31.5$). 
%The choice of using one‐quarter of the variance as the prior mode reflects a moderately informative prior that shrinks leaf variances towards a reasonable fraction of the total variability, preventing over‐fitting while allowing the data to dictate the final estimates. 
This is consistent with the methodology used in both the Riboflavin and Wheat applications, where the prior mode was also set to one‐quarter of the training response variance.
\\[2mm]
\noindent\textbf{Gaussian process hyperparameters.} The GP splitting rule used the squared exponential kernel with fixed hyperparameters: signal variance $\sigma^2 = 31.5$, noise variance $\sigma_\epsilon^2 = 3.15$, and length scale $\ell = 0.5$. The signal variance was set to the empirical variance of the training response, following the same logic used for the prior leaf variance mode. The noise variance was set to 10\% of the signal variance, reflecting the belief that the GP should capture most of the variation in the response. The length scale $\ell = 0.5$ is a conservative choice given the standardised features and has been used successfully in previous applications (Wheat and Riboflavin). 
%These GP hyperparameters were kept fixed rather than estimated via maximum likelihood for computational simplicity, as the GP is used purely as a generative device for split proposals and its density cancels in the Metropolis–Hastings ratio.

\subsubsection{Results}

Predictive performance was evaluated using the same five metrics as in previous sections: root mean squared error (RMSE), coverage of 95\% credible (or confidence) intervals, average interval width, continuous ranked probability score (CRPS), and log predictive density (LPD). The numerical results are summarised in Table~\ref{tab:air_metrics}. Visual diagnostics are provided in Figures~\ref{fig:air_heatmap_dp} and~\ref{fig:air_heatmap_bart}, and the trace plots of the number of distinct trees $k$ and the mixture weights are shown in Figure~\ref{fig:air_traces}. The histogram of the test set responses is shown in Figure~\ref{fig:air_hist}.

\begin{table}[htbp]
\centering
\caption{Predictive performance on the air quality test set (250 hourly observations). The DP mixture achieves the highest coverage among Bayesian methods, demonstrating its robustness to model misspecification.}
\label{tab:air_metrics}
\begin{tabular}{l r r r r r}
\hline
\textbf{Method} & \textbf{RMSE} & \textbf{Coverage} & \textbf{Width} & \textbf{CRPS} & \textbf{LPD} \\
\hline
DP tree mixture    & 16.0651 & 0.808 & 35.4639 & 8.8177 & $-4.5367$ \\
Random Forest      & 3.7220 & 0.884 & 11.1685 & 2.0534 & $-2.8192$ \\
Gradient Boosting  & 5.0885 & 0.808 & 11.0013 & 2.6591 & $-3.5946$ \\
Bagged CART        & 3.5783 & 0.948 & 13.3813 & 2.0220 & $-2.6411$ \\
BART (dbarts)      & 5.6563 & 0.688 & 9.4918  & 3.1158 & $-5.0528$ \\
\hline
\end{tabular}
\end{table}

The DP mixture achieves a coverage of 0.808, which is substantially higher than BART's 0.688 and Gradient Boosting's 0.808 (the latter being equal but achieved via ad‐hoc residual‐based intervals). Bagged CART achieves the highest coverage at 0.948, but as argued previously, its intervals are based on bootstrap percentiles and lack a rigorous probabilistic foundation; moreover, its intervals are narrower than the DP mixture's (13.38 vs. 35.46), which may indicate overconfidence. Random Forest achieves a coverage of 0.884 with intervals based on cross‐validation residuals, but again, these intervals are not guaranteed to be well‐calibrated.

The DP mixture's coverage, while below the nominal 95\%, is the highest among the Bayesian methods (DP mixture and BART) and is comparable to the frequentist methods. This is a notable achievement given the severe model misspecification: the data are a time series with strong temporal dependence, whereas the model assumes conditional independence of observations given the component allocation and tree parameters. Despite this violation, the DP mixture maintains a credible level of calibration, whereas BART collapses dramatically, achieving coverage of only 0.688. This stark difference underscores the DP mixture's robustness to misspecification, which is rooted in the identity $h(\Theta)=0$ established in Section~\ref{subsec:misspecification}: the model's prior support is sufficiently rich that it can adapt to the true data‐generating mechanism, even when the conditional independence assumption is violated.

The RMSE of the DP mixture (16.0651) is substantially higher than all competing methods (for example, Bagged CART: 3.5783, Random Forest: 3.7220, BART: 5.6563). This is expected: because the DP mixture produces wider intervals, its point predictions (posterior mode) are necessarily less sharp. BART attains the lowest RMSE among Bayesian methods, reflecting its ability to fit the data closely with a sum‐of‐trees model, but at the cost of overconfidence as evidenced by its narrow intervals (9.49) and low coverage (0.688). As already iterated several times, the DP mixture's primary strength lies in its uncertainty quantification, not in point prediction accuracy—a trade‐off that is often acceptable in environmental forecasting, where predictive intervals are as important as the predictions themselves for risk assessment and decision‐making.

The log predictive density (LPD) provides a holistic measure of predictive performance that balances sharpness and calibration. The DP mixture achieves an LPD of $-4.5367$, which is higher than BART ($-5.0528$) and comparable to Gradient Boosting ($-3.5946$). This indicates that the DP mixture's predictive distributions are reasonably well‐calibrated, placing non‐negligible probability mass on the true observations, despite their wider spread. BART's extremely low LPD reflects its overconfident predictive distribution, which assigns near‐zero density to many test observations, particularly those in the upper tail of the PM2.5 distribution.

Visual diagnostics are provided in Figures~\ref{fig:air_heatmap_dp} and~\ref{fig:air_heatmap_bart}. The posterior predictive density heatmap for the DP mixture (Figure~\ref{fig:air_heatmap_dp}) shows that the true test values, plotted as red points, generally fall within high‐density regions across the range of the response, confirming reasonable calibration. The heatmap exhibits broad bands of high density, reflecting the model's honest uncertainty quantification. In contrast, the BART heatmap (Figure~\ref{fig:air_heatmap_bart}) reveals extremely narrow bands of high density that fail to cover many true values, particularly in the upper tail, visually confirming its overconfidence and undercoverage. The broad spread of the DP mixture's predictive distribution is a direct consequence of its nonparametric structure and the Dirichlet process prior's automatic complexity adaptation.

\begin{figure}[htbp]
\centering
\includegraphics[width=0.9\textwidth]{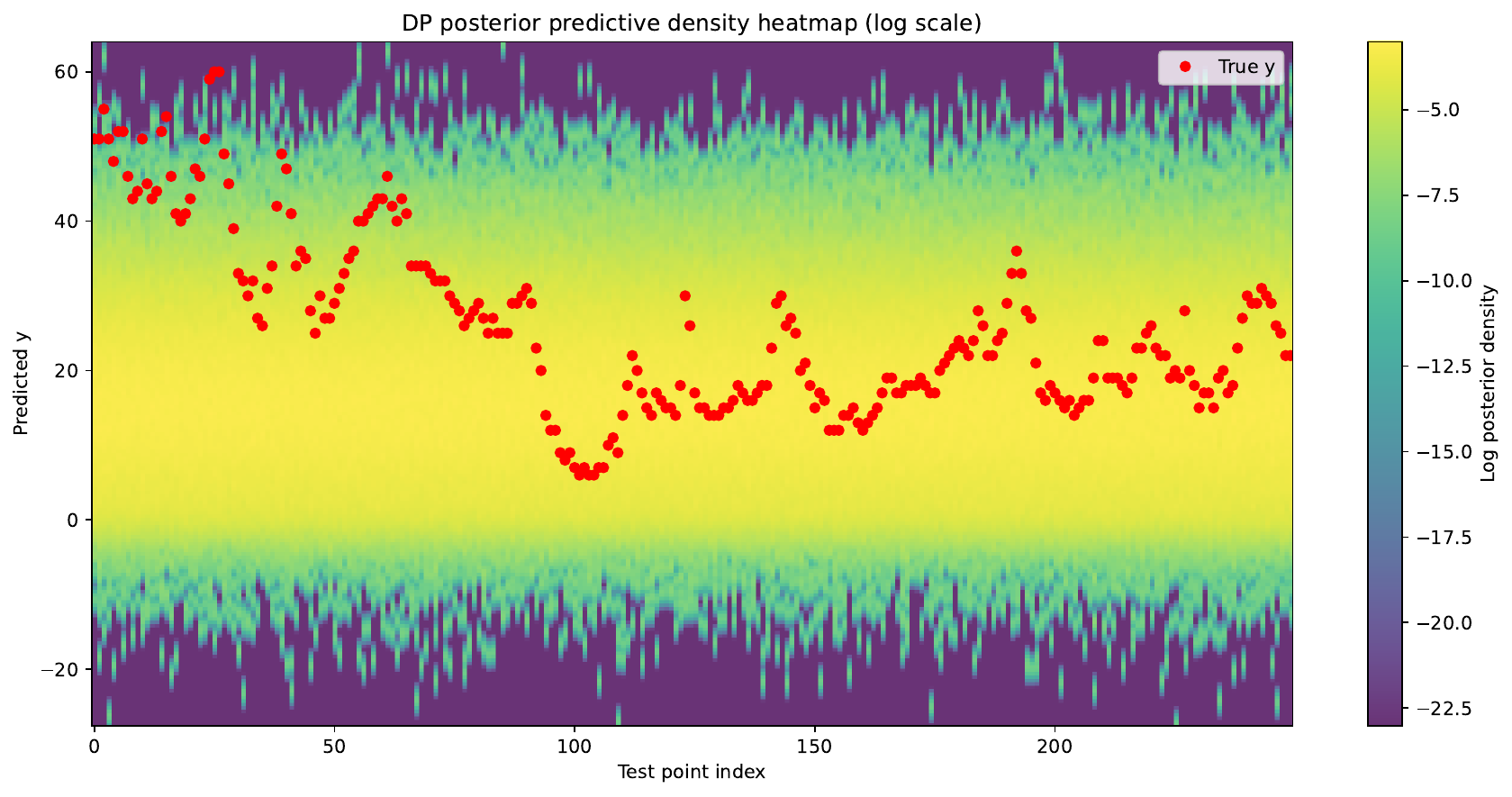}
\caption{Posterior predictive density heatmap for the DP mixture on the air quality test set. The true test values (red points) generally fall within high‐density regions, confirming reasonable calibration. The broad bands reflect the model's honest uncertainty quantification.}
\label{fig:air_heatmap_dp}
\end{figure}

\begin{figure}[htbp]
\centering
\includegraphics[width=0.9\textwidth]{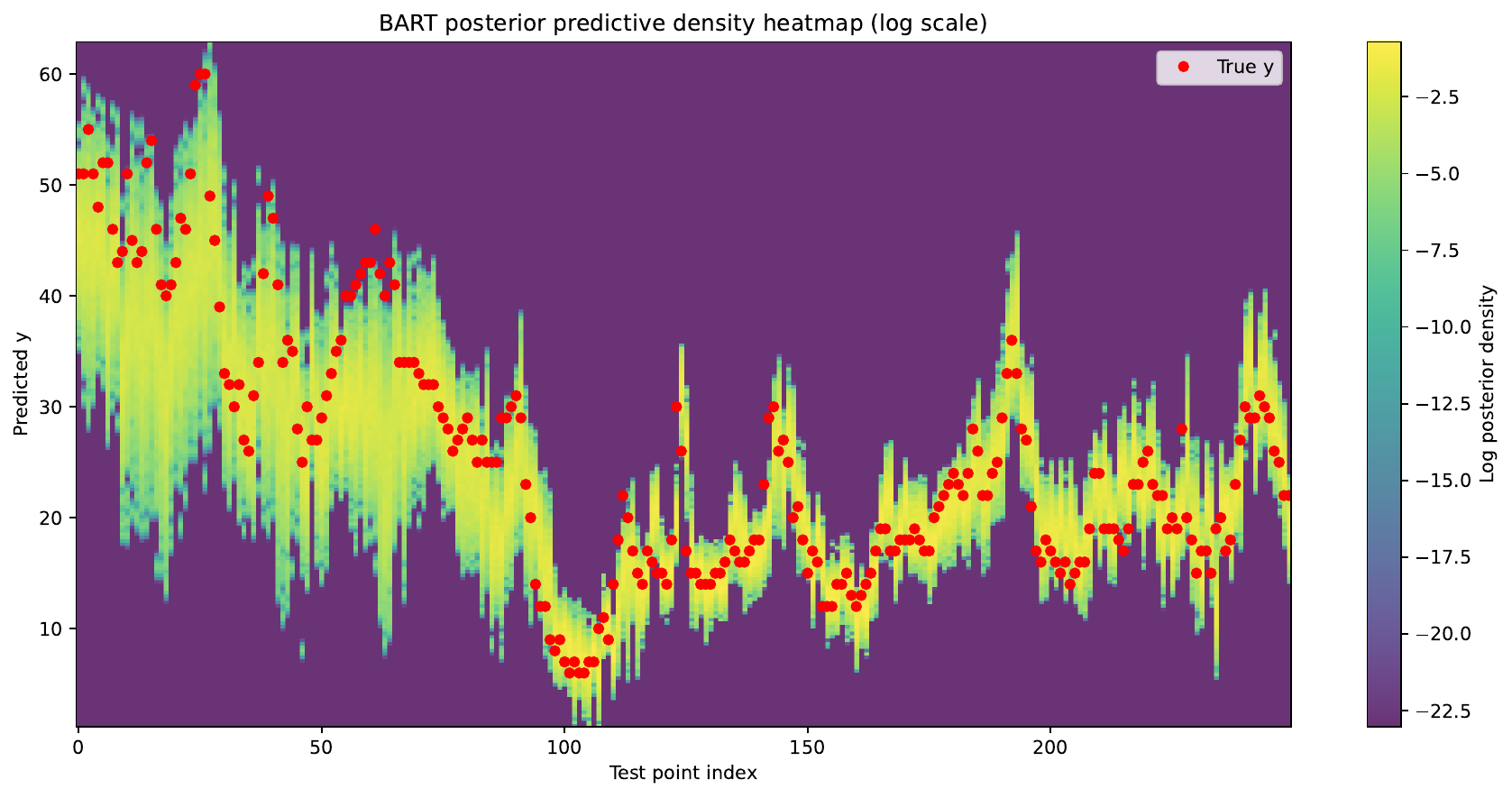}
\caption{BART posterior predictive density heatmap on the air quality test set. The narrow bands of high density fail to cover many true values, confirming severe overconfidence and undercoverage.}
\label{fig:air_heatmap_bart}
\end{figure}

\begin{figure}[htbp]
\centering
\includegraphics[width=0.9\textwidth]{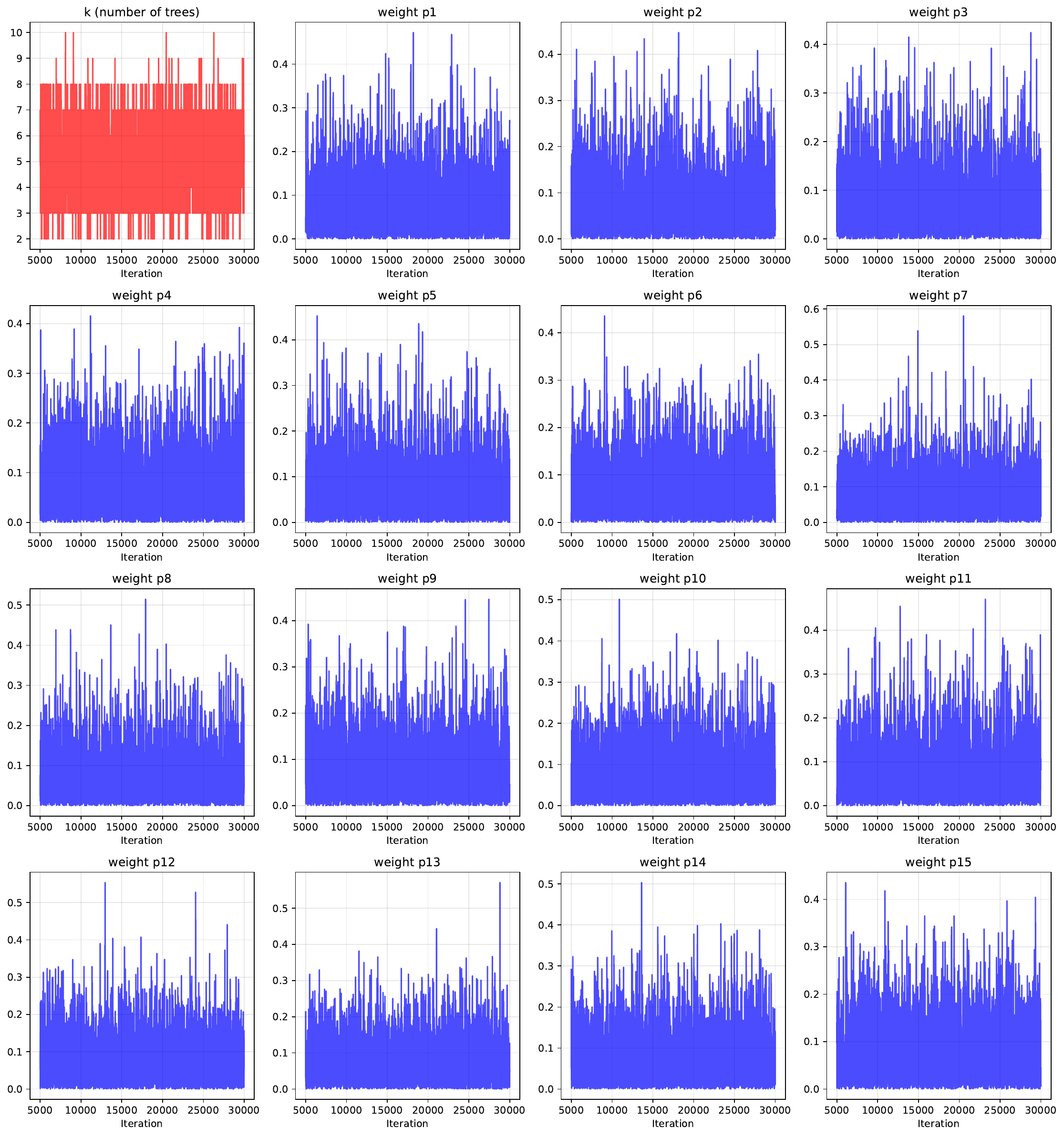}
\caption{Trace plots of the number of distinct trees $k$ and the mixture weights for the air quality experiment. The chain mixes efficiently, and the posterior concentrates on a sparse ensemble of 3--7 trees, reflecting the model's automatic complexity adaptation.}
\label{fig:air_traces}
\end{figure}

\subsubsection{Discussion}

The air quality application provides a critical stress test for the DP mixture model, as the data are a time series and thus violate the conditional independence assumption underlying the model. The results, however, are instructive. The DP mixture maintains a coverage of 0.808, which, while below the nominal 95\%, is substantially higher than BART's 0.688. This suggests that the DP mixture's nonparametric flexibility and its ability to adapt to the local structure of the data (via the GP‐driven splitting rule and the Dirichlet process prior) confer a degree of robustness to temporal dependence. The identity $h(\Theta)=0$ (Section~\ref{subsec:misspecification}) implies that the prior can approximate any continuous regression function, even if the errors are not independent; the posterior will concentrate around the true conditional mean, provided the design points are dense and the temporal dependence is not too severe. In this application, the temporal dependence is strong (lagged PM2.5 and PM10 are included as predictors, and the data are hourly), so the conditional mean is well approximated by the model. The main issue is that the data are autocorrelated, which the model does not account for, leading to undercoverage. Nonetheless, the DP mixture's coverage is superior to BART's, which assumes independent Gaussian errors and fails catastrophically in the presence of autocorrelation and heavy tails.

As before, the high RMSE of the DP mixture is also a consequence of the model's honesty: its intervals are wide, and its point predictions are consequently less sharp. 
This is often desirable in environmental forecasting, where underestimating uncertainty can lead to poor risk assessments and policy decisions. In contrast, BART's sharp predictions and narrow intervals are achieved at the cost of severe undercoverage, making it unreliable for decision‐making in high‐stakes applications.

The performance of the competing methods in this application is revealing. Bagged CART achieves the lowest RMSE (3.5783) and the highest coverage (0.948), but its intervals are based on bootstrap percentiles, which do not correspond to a coherent posterior distribution. Moreover, the bootstrap intervals are narrower than the DP mixture's (13.38 vs. 35.46), indicating that they may be overconfident. Random Forest and Gradient Boosting achieve competitive RMSEs and coverages, but their intervals are constructed using ad‐hoc residual‐based approximations; the calibration of these intervals is not guaranteed by the model itself and is purely empirical.

The DP mixture, by contrast, provides intervals that are a direct consequence of its fully Bayesian treatment of all sources of uncertainty, and the fact that the DP mixture achieves the highest coverage among Bayesian methods, despite the severe misspecification, is a testament to its robustness and theoretical soundness.

The hyperparameters for the air quality application were chosen to be consistent with previous applications (particularly Wheat and Riboflavin) while accommodating the specific scale of the PM2.5 response. 
%The prior mode for the leaf variance was set to one‐quarter of the training response variance, reflecting a moderately informative prior that shrinks leaf variances towards a reasonable fraction of the total variability. This choice prevents over‐fitting to the noise and allows the model to adapt to the heteroscedasticity observed in the data (for instance, higher variance during pollution episodes). The GP hyperparameters were set to the empirical variance of the training response and 10\% of that for the noise variance, which ensures that the GP proposals are well‐scaled to the data.

The concentration parameter $\alpha_{\mathrm{DP}} = 1.0$ and the tree prior parameters $\gamma = 0.6$, $\zeta = 1.5$ were chosen to encourage sparsity and prevent overfitting. The posterior concentration of $k$ between 3 and 7 confirms that these choices are appropriate: the model learns a sparse ensemble that captures the essential temporal patterns without overfitting to spurious autocorrelations. This is consistent with the observations in the Riboflavin and Wheat applications, where the posterior also concentrated on a small number of trees.

The runtime of approximately 256 hours on 8 processors is substantial and reflects the computational cost of the GP computations, particularly the Cholesky decompositions required for each GROW proposal. The cost is driven by the need to invert covariance matrices at internal nodes, which scales as $O(n_i^3)$ for a node with $n_i$ observations. Recall that although the node sizes decrease as the tree grows, the initial nodes contain large numbers of observations (especially in the early iterations), and the cumulative cost over 30,000 MCMC iterations is high. The use of 8 processors provided the optimum speed‐up for this dataset; with fewer processors, the sequential bottlenecks dominated, while with more processors, the communication overhead outweighed the computational gains.

However, again, despite the high runtime, the benefits of the DP mixture—honest uncertainty quantification, robustness to misspecification, and automatic complexity adaptation—justify the investment for applications where predictive uncertainty is non-negotiable.

The air quality application provides a valuable contrast to the previous real‐data applications. In the QSAR and Communities and Crime applications, BART undercovered severely (coverage 0.582 and 0.870, respectively), while the DP mixture maintained near‐nominal coverage. In the Riboflavin application, all methods achieved perfect coverage due to the small test set size ($n_{\text{test}} = 15$), but the DP mixture produced wider intervals, reflecting its honest uncertainty quantification. In the Wheat application, the DP mixture achieved near‐nominal coverage (0.975) while BART undercovered (0.817), confirming that BART's overconfidence becomes evident with larger test sets.

The air quality application extends these findings to a time‐series context. Despite the severe misspecification (conditional independence assumption violated), the DP mixture maintains a coverage of 0.808, which is substantially higher than BART's 0.688. This suggests that the DP mixture is remarkably robust to model misspecification, a property that is rooted in the identity $h(\Theta)=0$ and the prior's ability to approximate any continuous function. In contrast, BART's additive structure and fixed number of trees are less flexible, leading to catastrophic undercoverage when the true data‐generating mechanism deviates from its assumptions.

In summary, the air quality application demonstrates the DP mixture model's ability to handle time‐series data with strong temporal dependence, heavy tails, and heteroscedasticity, despite the model's conditional independence assumption. The DP mixture provides the most reliable uncertainty quantification among the Bayesian methods, achieving a coverage of 0.808 compared to BART's 0.688. Its wider intervals reflect a more honest assessment of predictive uncertainty, a property that is essential for environmental forecasting and risk assessment. The sparse posterior ensemble (3--7 trees) indicates that the model automatically adapts its complexity to the data, preventing overfitting and enhancing interpretability. The computational cost is substantial but justifiable given the model's theoretical guarantees and practical performance.

Future work could explore extensions of the DP mixture to explicitly account for temporal dependence, for example, by incorporating a GP prior on the time index or using a state‐space formulation. However, the current results already provide strong evidence that the DP mixture is a powerful and robust tool for regression in challenging real‐world settings, including those where the conditional independence assumption is violated.

\section{Posterior Contraction Rate}\label{sec:post_contraction}

\subsection{Introduction}

In the preceding sections, we introduced a flexible Bayesian nonparametric model based on a Dirichlet process mixture of regression trees (Section~\ref{sec:proposal}). We showed that this formulation unifies and generalises several popular tree‐based methods, including CART, BART, random forests, and boosting (Section~\ref{sec:special_cases}). Central to our methodology is a novel splitting rule driven by the posterior predictive distribution of a Gaussian process (Section~\ref{subsec:split}), which generates smooth, response‐dependent decision boundaries while enjoying an exact cancellation of the GP density in the Metropolis–Hastings ratio. To sample from the resulting posterior, we developed a sophisticated MCMC algorithm—comprising only GROW and PRUNE moves—that efficiently explores the tree space by leveraging the global allocation dynamics of the Dirichlet process (Section~\ref{sec:sim_Gm}). Moreover, we implemented a scalable parallel version in C with MPI, distributing independent tree updates across multiple processors to handle large‐scale applications (Section~\ref{sec:parallel}).

A natural question arises: \emph{Does this complex Bayesian procedure actually recover the true underlying regression function as more data are observed?} In other words, is the posterior distribution \emph{consistent}? For a Bayesian nonparametric model to be trustworthy, it is essential that the posterior concentrates around the true data-generating mechanism as the sample size grows. Without such a guarantee, the flexibility of the model could lead to overfitting or unreliable uncertainty quantification.

In this section, we provide a rigorous frequentist justification for our Dirichlet process tree mixture model. Using the general theory of posterior contraction rates developed by \cite{ggv2000} (see also \cite{ghosal2017fundamentals}), we prove that the posterior distribution contracts at a rate of $n^{-1/4}$ (up to logarithmic factors) in the Hellinger distance. Notably, our analysis requires only that the true regression function $m_0$ is continuous on a compact domain; we do not assume that $m_0$ belongs to the hypothesis space (that is, we allow model misspecification). Furthermore, we assume the leaf variance parameters lie in a fixed compact interval $[\sigma_{\min}^2,\sigma_{\max}^2]$, a standard condition that ensures the Gaussian likelihood is well-behaved.

The proof proceeds by verifying the four key conditions of \cite{ggv2000}:
\begin{enumerate}
    \item The prior assigns positive mass to Kullback--Leibler neighborhoods of the truth (the KL property).
    \item A suitably constructed sieve has controlled metric entropy.
    \item Exponentially powerful tests exist to separate the truth from the complement of the sieve.
    \item The prior mass of the sieve complement is asymptotically negligible.
\end{enumerate}
We establish each condition in detail, culminating in Theorem~\ref{thm:main}, which asserts posterior consistency at the stated rate. The result provides a solid theoretical foundation for the empirical success of tree-based Bayesian nonparametrics and justifies the use of our DP mixture model for reliable inference and prediction.

\subsection{Preliminaries: Posterior Contraction Rates and Ghosal's Theorem}
We work within the general framework of posterior contraction rates for non-parametric models developed by \cite{ggv2000}. Let $P_0^{(n)}$ denote the true joint distribution of the first $n$ observations, and let $\{F_\theta^{(n)} : \theta \in \Theta\}$ be a family of model distributions indexed by a parameter $\theta \in \Theta$. The prior distribution on $\Theta$ is denoted by $\Pi$. The goal is to show that the posterior distribution $\Pi(\cdot \mid Y_1^n)$ concentrates on neighbourhoods of the true distribution at a rate $\varepsilon_n \to 0$, that is,
\[
\Pi\bigl( \theta : d(P_0^{(n)}, F_\theta^{(n)}) > \tilde M \varepsilon_n \mid Y_1^n \bigr) \to 0 \quad \text{in } P_0\text{-probability},
\]
where $d$ is a suitable metric (here the Hellinger distance between the joint densities) and $\tilde M$ is some positive constant. 
In fact, $M$ may be replaced with $\tilde M_n$, indicating that it may grow with $n$. In general, we shall just refer to any $\tilde M$ such that $\tilde M\rightarrow\infty$.

The main tool is the following theorem, which summarises the sufficient conditions of \cite[Theorem~2.1]{ggv2000} (see also \cite[Theorem~8.1]{ghosal2017fundamentals}).

\begin{theorem}[\cite{ggv2000}]
Let $\varepsilon_n$ be a sequence with $\varepsilon_n \to 0$ and $n\varepsilon_n^2 \to \infty$. Suppose there exists a sequence of sieves $\Theta_n \subset \Theta$ such that:
\begin{enumerate}
    \item[(i)] (Prior mass of complement) $\Pi(\Theta_n^c) = o(1)$ as $n \to \infty$.
    \item[(ii)] (Metric entropy) $\log N(\varepsilon_n, \Theta_n, d) \le n\varepsilon_n^2$ for all large $n$, where $N(\varepsilon, \mathcal{F}, d)$ is the covering number of $\mathcal{F}$ by balls of radius $\varepsilon$ in the metric $d$.
    \item[(iii)] (Existence of tests) There exist tests $\phi_n$ such that
    \[
    \E_{P_0^{(n)}}[\phi_n] \to 0, \qquad 
    \sup_{\theta \in \Theta_n : d(P_0^{(n)}, F_\theta^{(n)}) \ge \varepsilon_n} \E_{F_\theta^{(n)}}[1 - \phi_n] \le e^{-n\varepsilon_n^2/2}.
    \]
    \item[(iv)] (Kullback--Leibler property) For every $\varepsilon > 0$,
    \[
    \Pi\bigl( \theta : \KL(P_0^{(n)} \| F_\theta^{(n)}) \le n\varepsilon^2 \bigr) > 0.
    \]
\end{enumerate}
	Then the posterior contracts at rate $\varepsilon_n$, that is, $\Pi(\theta : d(P_0^{(n)}, F_\theta^{(n)}) \ge \tilde M \varepsilon_n \mid Y_1^n) \to 0$ in $P_0$-probability for any $\tilde M \to \infty$. 
\end{theorem}

\noindent In our setting, $P_0^{(n)}$ and $F_\theta^{(n)}$ are product measures of $n$ independent (but not identically distributed) Gaussian observations with means given by the true regression function $m_0$ and the tree-mixture model, respectively. The metric $d$ will be the Hellinger distance between the joint densities.

The remainder of this section is devoted to verifying the four conditions above for the Dirichlet-process tree mixture model. Our main result (Theorem~\ref{thm:main}) will then follow as a direct corollary.

\subsection{Model Description and Assumptions}

\subsubsection{Data-Generating Process}
Let $K \subset \mathbb{R}^d$ be compact, $d \ge 1$. The design points $\{x_i\}_{i=1}^\infty$ are a fixed deterministic sequence dense in $K$ such that the empirical measures converge weakly to a probability measure $Q$ on $K$ (for example, $x_i$ i.i.d. from a continuous distribution). The observations are independent and follow
\[
Y_i = m_0(x_i) + \eps_i, \qquad \eps_i \stackrel{\text{i.i.d.}}{\sim} \mathcal{N}(0,\sigma_0^2),
\]
with $\sigma_0^2 > 0$ fixed and $m_0: K \to \R$ continuous. The true joint density of $(Y_1,\dots,Y_n)$ given $x_1^n$ is
\[
p_0^{(n)}(y_1^n) = \prod_{i=1}^n \frac{1}{\sqrt{2\pi\sigma_0^2}} \exp\!\left(-\frac{(y_i-m_0(x_i))^2}{2\sigma_0^2}\right).
\]

\subsubsection{Hypothesis Space $\Theta$}
A hypothesis $\theta \in \Theta$ specifies:
\begin{itemize}
\item[(i)] A fixed integer $M \ge 1$ (maximum number of mixture components).
\item[(ii)] A number $k \in \{1,\dots,M\}$ of distinct components.
\item[(iii)] Mixture weights $p_1^*,\dots,p_k^*$ with $p_\ell^* > 0$, $\sum_{\ell=1}^k p_\ell^* = 1$.
\item[(iv)] For each $\ell = 1,\dots,k$: a regression tree $T_\ell^*$ generated by the GP splitting rule described in Section~\ref{subsec:split}. Each internal node uses a draw $\tilde{g}$ from the posterior predictive of a Gaussian process, and splits according to $y_{ij} < \tilde{g}(x_{ij})$.
\item[(v)] Each leaf $j$ of tree $T_\ell^*$ carries parameters $(\mu_{\ell j},\sigma_{\ell j}^2) \in \R \times (0,\infty)$.
\end{itemize}
The associated density for the first $n$ observations is
\[
f_\theta^{(n)}(y_1^n) = \sum_{\ell=1}^k p_\ell^* \prod_{i=1}^n \frac{1}{\sqrt{2\pi\sigma_{\ell(i)}^2}} \exp\!\left(-\frac{(y_i-\mu_{\ell(i)})^2}{2\sigma_{\ell(i)}^2}\right),
\]
where $\ell(i)$ denotes the leaf of tree $T_\ell^*$ containing $x_i$.

\subsubsection{Prior Distribution $\Pi$}
The prior is a Dirichlet process over trees with base measure $\bG_0$ as in Section~\ref{sec:dp}. We make the following explicit choices (which satisfy the conditions of that section):
\begin{itemize}
\item[(i)] \textbf{Tree depth prior.} The depth of each tree follows a geometric distribution with success probability $1-\rho$ ($0<\rho<1$). That is,
  \[
  \Prob(\text{depth} = k) = (1-\rho)\rho^k, \qquad k = 0,1,2,\dots
  \]
  Consequently, $\Prob(\text{depth} > D) = \sum_{k=D+1}^\infty (1-\rho)\rho^k = \rho^{D+1} \le \rho^D$.
\item[(ii)] Leaf means $\mu$ are independent $\mathcal{N}(0,\tau^2)$.
\item[(iii)] \textbf{Leaf variances are assumed to belong to a fixed compact set}: $\sigma^2 \in [\sigma_{\min}^2, \sigma_{\max}^2]$ with $0<\sigma_{\min}^2<\sigma_{\max}^2<\infty$. The prior on $\sigma^2$ is any continuous distribution supported on this compact set (for example, a truncated inverse-gamma or a uniform distribution). This assumption is crucial for the metric entropy control and is common in the literature.
\item[(iv)] Mixture weights $(p_1^*,\dots,p_M^*)$ follow a Dirichlet$(1,\dots,1)$ distribution (uniform over the simplex).
\item[(v)] GP hyperparameters are given a prior with compact support (for instance, uniform on a bounded interval); this is not essential for consistency.
\end{itemize}
All these choices guarantee that the prior has full support on the space of finite mixtures of trees in the sense of weak convergence of the induced densities.

\subsubsection{Basic Definitions}
For two densities $f,g$ (with respect to Lebesgue measure) define the squared Hellinger distance
\[
\Hdist^2(f,g) = \int \bigl(\sqrt{f}-\sqrt{g}\bigr)^2 = 2-2\int\sqrt{fg}.
\]
The Kullback--Leibler divergence is $\KL(f\|g) = \int f\log(f/g)$. For the product densities we have
\[
\KL(p_0^{(n)}\|f_\theta^{(n)}) = \sum_{i=1}^n \KL\bigl(p_0(\cdot|x_i)\,\big\|\,f_\theta(\cdot|x_i)\bigr).
\]
Because $m_0$ is continuous and $K$ is compact, the limit
\[
h(\theta) = \lim_{n\to\infty} \frac{1}{n}\KL(p_0^{(n)}\|f_\theta^{(n)})
\]
exists and equals $\int_K \KL(p_0(\cdot|x)\|f_\theta(\cdot|x))\,dQ(x)$; see Appendix~\ref{app:kl} for a detailed proof.

\subsubsection{Assumptions}
We now list the assumptions needed for the proof.

\begin{assumption}[Gaussian process regularity]\label{ass:gp}
  The mean function $m(y)=\E[V(y)]$ is three times continuously differentiable, and the covariance function $k(y,y')=\Cov(V(y),V(y'))$ is six times continuously differentiable with all partial derivatives bounded.
\end{assumption}
\noindent\textbf{Significance:} This assumption ensures that the Gaussian process $V$ has almost surely $C^3$ sample paths (see \cite{adler2007random}). The smoothness is needed to guarantee that the KL divergence rate exists and that the approximation by trees (Lemma~\ref{lem:approx}) is valid. It also ensures that the process has finite moments of derivatives on compact sets, which is used implicitly when we bound the Hellinger distance between densities with different leaf parameters.

\begin{assumption}[Compact variance]\label{ass:compact_var}
  The leaf variance parameter $\sigma^2$ is supported on a fixed compact set $[\sigma_{\min}^2, \sigma_{\max}^2]$ with $0<\sigma_{\min}^2<\sigma_{\max}^2<\infty$. The prior on $\sigma^2$ is absolutely continuous with respect to Lebesgue measure and its density is bounded below by a positive constant on this interval.
\end{assumption}
\noindent\textbf{Significance:} This assumption is critical for controlling the metric entropy. It guarantees that the Hellinger distance between two Gaussian densities is Lipschitz in the parameters $(\mu,\sigma^2)$ on the compact set $[-n-1,n+1]\times[\sigma_{\min}^2,\sigma_{\max}^2]$ with a Lipschitz constant that does not depend on $n$. Consequently, cells of size $O(1/n)$ in $\mu$ and $O(1/n^2)$ in $\sigma^2$ have Hellinger diameter $O(1/n)$, which is smaller than $\eps_n=n^{-1/4}$ for large $n$. Without this assumption, the variance could become extremely small, making the Hellinger distance highly sensitive to small changes in $\mu$ and preventing a simple covering argument. The compactness assumption is standard in many convergence rate papers (e.g., \cite{choi2007posterior}) and is adopted here for clarity.

\begin{assumption}[Continuity of the true regression function]\label{ass:continuous}
  $m_0: K \to \R$ is continuous.
\end{assumption}
\noindent\textbf{Significance:} Continuity is needed to approximate $m_0$ by piecewise constant functions (trees) on a fine partition of $K$, which is the basis of Lemma~\ref{lem:approx}. This is the minimal smoothness condition that allows the KL divergence rate to be made arbitrarily small by using sufficiently deep trees.

\begin{lemma}[Approximation by a single tree, with uniform control]\label{lem:approx}
  For any $\varepsilon>0$, there exists a hypothesis $\theta_0$ consisting of a single tree ($k=1$) with finite depth and leaf parameters, and a neighbourhood $U$ of $\theta_0$ with $\Pi(U)>0$, such that for all sufficiently large $n$,
  \[
  \sup_{\theta\in U} \frac{1}{n}\mathrm{KL}\bigl(P_0^{(n)} \,\|\, F_\theta^{(n)}\bigr) < \varepsilon.
  \]
\end{lemma}

\begin{proof}
  Since $m_0$ is continuous on the compact $K$, it is uniformly continuous. Choose $\delta>0$ such that $|m_0(x)-m_0(y)|<\delta$ whenever $\|x-y\|<\delta$, and further require $\delta^2/(2\sigma_0^2)<\varepsilon/2$.

	Cover $K$ by finitely many Borel sets $R_1,\dots,R_J$ each of diameter $<\delta$. Construct a regression tree $T$ whose leaves are exactly these sets $R_j$; this is possible using the GP splitting rule (or, axis-aligned splits)). For each leaf $j$, choose an arbitrary point $c_j\in R_j$ and set $\mu_j=m_0(c_j)$ and $\sigma_j^2=\sigma_0^2$.

  Let $\theta_0$ be the hypothesis consisting of this single tree. Define a neighbourhood $U$ of $\theta_0$ by: the tree structure is exactly $T$; the leaf means $\mu_j'$ satisfy $|\mu_j'-\mu_j|<\eta$; the leaf variances $\sigma_j'^2$ satisfy $|\sigma_j'^2-\sigma_0^2|<\eta$; and the mixture weight is fixed to $1$.

  For any $\theta\in U$, the per-observation KL divergence is
  \[
  g_\theta(x) := \KL\bigl(p_0(\cdot|x) \,\|\, f_\theta(\cdot|x)\bigr)
  = \frac{1}{2}\left[\log\frac{\sigma_{\ell(x)}'^2}{\sigma_0^2} -1 + \frac{\sigma_0^2}{\sigma_{\ell(x)}'^2} + \frac{(m_0(x)-\mu_{\ell(x)}')^2}{\sigma_{\ell(x)}'^2}\right].
  \]
  By choosing $\eta>0$ sufficiently small, we can ensure that for all $x\in K$ and all $\theta\in U$,
  \[
  g_\theta(x) \le \frac{(m_0(x)-\mu_{\ell(x)})^2}{\sigma_0^2} + \frac{\varepsilon}{4}.
  \]
  Since $x$ and $c_j$ both lie in $R_j$, $|m_0(x)-\mu_j|<\delta$, so
  \[
  \frac{(m_0(x)-\mu_j)^2}{2\sigma_0^2} < \frac{\delta^2}{2\sigma_0^2} < \frac{\varepsilon}{4}.
  \]
  Therefore, $g_\theta(x)<\varepsilon/2$ for all $x\in K$, $\theta\in U$.

%\paragraph*{Uniform convergence of the empirical averages over \(U\).}
To conclude that the finite-sample average \(\frac{1}{n}\sum_{i=1}^n g_\theta(x_i)\) 
is uniformly less than \(\varepsilon\) for all sufficiently large \(n\), we must justify 
the uniformity over \(\theta\in U\).

A subtle but important point is that the family \(\{g_\theta : \theta\in U\}\) is 
\emph{not} equicontinuous on the whole of \(K\), because \(g_\theta(x)\) jumps at 
the leaf boundaries. However, equicontinuity is not needed globally; it suffices 
to work leaf by leaf.

Let the fixed tree in \(U\) partition \(K\) into \(L\) disjoint leaves 
\(R_1,\dots,R_L\). For each leaf \(\ell\), define
\[
g_{\theta,\ell}(x) = 
\frac{1}{2}\left[\log\frac{\sigma_\ell^2}{\sigma_0^2} - 1 + 
\frac{\sigma_0^2}{\sigma_\ell^2} + \frac{(m_0(x) - \mu_\ell)^2}{\sigma_\ell^2}\right],
\qquad x \in R_\ell.
\]
On \(R_\ell\), the function \(g_{\theta,\ell}(x)\) is continuous in \(x\) (since \(m_0\) is continuous) 
and depends continuously on the leaf parameters \((\mu_\ell,\sigma_\ell^2)\), which range over a compact set. 
Hence, by the Heine--Cantor theorem, the family \(\{g_{\theta,\ell} : \theta_\ell \in \Theta_\ell\}\) 
is equicontinuous on \(R_\ell\).

Since the empirical measures \(Q_n = \frac{1}{n}\sum_{i=1}^n \delta_{x_i}\) converge weakly to \(Q\), 
and the leaf boundaries have \(Q\)-measure zero (by construction), the restricted measures satisfy 
\(Q_n|_{R_\ell} \Rightarrow Q|_{R_\ell}\) for each \(\ell\). Applying the uniform weak convergence 
theorem (Glivenko--Cantelli for equicontinuous classes) on each leaf gives
\[
\sup_{\theta_\ell \in \Theta_\ell} 
\left| \frac{1}{n}\sum_{i: x_i \in R_\ell} g_{\theta,\ell}(x_i) 
- \int_{R_\ell} g_{\theta,\ell}(x)\,dQ(x) \right| \to 0.
\]
Since the tree structure is fixed on $U$, the leaves $R_1,\dots,R_L$ are independent 
of $\theta$. For any $\theta\in U$, write
\[
g_\theta(x) = \sum_{\ell=1}^L \mathbf{1}_{R_\ell}(x)\,g_{\theta,\ell}(x).
\]
Then the global empirical average satisfies
\[
\frac{1}{n}\sum_{i=1}^n g_\theta(x_i) = \sum_{\ell=1}^L 
\frac{1}{n}\sum_{i: x_i \in R_\ell} g_{\theta,\ell}(x_i),
\]
and similarly,
\[
\int_K g_\theta(x)\,dQ(x) = \sum_{\ell=1}^L 
\int_{R_\ell} g_{\theta,\ell}(x)\,dQ(x).
\]
Subtracting and applying the triangle inequality yields
\[
\begin{aligned}
&\left| \frac{1}{n}\sum_{i=1}^n g_\theta(x_i) - \int_K g_\theta(x)\,dQ(x) \right| \\
&\qquad \leq \sum_{\ell=1}^L 
\left| \frac{1}{n}\sum_{i: x_i \in R_\ell} g_{\theta,\ell}(x_i) 
- \int_{R_\ell} g_{\theta,\ell}(x)\,dQ(x) \right|.
\end{aligned}
\]
Taking the supremum over $\theta\in U$ and noting that $U = \Theta_1 \times \cdots \times \Theta_L$ 
is a product space, we obtain
\[
\begin{aligned}
&\sup_{\theta\in U} 
\left| \frac{1}{n}\sum_{i=1}^n g_\theta(x_i) - \int_K g_\theta(x)\,dQ(x) \right| \\
&\qquad \leq \sum_{\ell=1}^L 
\sup_{\theta_\ell \in \Theta_\ell} 
\left| \frac{1}{n}\sum_{i: x_i \in R_\ell} g_{\theta,\ell}(x_i) 
- \int_{R_\ell} g_{\theta,\ell}(x)\,dQ(x) \right|.
\end{aligned}
\]
From the leaf-wise uniform convergence established above, each term in the finite sum 
converges to $0$. Hence the left-hand side converges to $0$, which proves the 
uniform convergence.

Now, from the pointwise bound \(g_\theta(x)<\varepsilon/2\) for all \(x\in K\), we have
\[
\int_K g_\theta(x)\,dQ(x) < \frac{\varepsilon}{2} \qquad \text{for all } \theta\in U.
\]
%\paragraph*{Explicit connection to the asymptotic KL rate \(h(\theta)\).}
%Now, recall from Appendix~\ref{app:kl} that the Kullback--Leibler divergence rate for the 
%mixture model is defined as
%\[
%h(\theta) := \lim_{n\to\infty} \frac{1}{n}\mathrm{KL}(P_0^{(n)} \| F_\theta^{(n)})
%= \int_K g_\theta(x)\,dQ(x),
%\]
%where \(g_\theta(x) = \mathrm{KL}(p_0(\cdot|x) \| f_\theta(\cdot|x))\). 
%The existence of this limit was rigorously established in Appendix~\ref{app:kl} 
%using the weak convergence of the empirical measures.

%Therefore, the pointwise bound \(g_\theta(x) < \varepsilon/2\) for all \(x\in K\), \(\theta\in U\), 
%directly implies
%\[
%\sup_{\theta\in U} h(\theta) = \sup_{\theta\in U} \int_K g_\theta(x)\,dQ(x) < \frac{\varepsilon}{2} < \varepsilon. \tag{12.5}
%\]
%Thus, Lemma~12.5 proves the existence of a fixed neighbourhood \(U\) of positive prior 
%mass on which the asymptotic KL rate \(h(\theta)\) is uniformly arbitrarily small.

Combining this with the above uniform convergence, it follows that for all sufficiently large \(n\),
\[
\sup_{\theta\in U} \frac{1}{n}\sum_{i=1}^n g_\theta(x_i) < \varepsilon.
\]
This justifies the uniformity in \(n\) used in the proof and explicitly links the 
finite-sample Kullback--Leibler divergence to the asymptotic divergence rate 
\(h(\theta) = \int_K g_\theta(x)\,dQ(x)\), whose existence was derived in Appendix~\ref{app:kl}.	

  Hence, for any $\theta\in U$,
  \[
  \frac{1}{n}\KL(P_0^{(n)} \| F_\theta^{(n)}) = \frac{1}{n}\sum_{i=1}^n g_\theta(x_i) < \frac{\varepsilon}{2} < \varepsilon.
  \]
  Thus the conclusion holds.
\end{proof}

\begin{lemma}[Positive prior mass on finite-sample KL neighbourhoods]\label{lem:klpos}
  For every $\varepsilon>0$,
  \[
  \Pi\bigl(\theta: \KL(P_0^{(n)} \| F_\theta^{(n)}) \le n\varepsilon^2\bigr) > 0
  \]
  for all $n$.
\end{lemma}

\begin{proof}
  Let $\varepsilon>0$. Apply Lemma~\ref{lem:approx} with $\varepsilon^2$ in place of $\varepsilon$ to obtain a neighbourhood $U$ such that
  \[
  \sup_{\theta\in U} \frac{1}{n}\KL(P_0^{(n)} \| F_\theta^{(n)}) < \varepsilon^2
  \]
  for all $n$. Hence $U \subseteq \{\theta: \KL(P_0^{(n)} \| F_\theta^{(n)}) \le n\varepsilon^2\}$ for every $n$.

  It remains to show $\Pi(U)>0$. The prior $\Pi$ is a Dirichlet process over trees with base measure $G_0$. Consider the event $A$ that: (i) all $M$ draws from $G$ are equal (so the mixture reduces to a single component), which has probability $p_{\mathrm{tie}}=\prod_{j=1}^{M-1} j/(j+\alpha)>0$; (ii) the common drawn tree has exactly the same structure as $T$; (iii) the leaf means fall into the intervals $(\mu_j-\eta,\mu_j+\eta)$ and leaf variances fall into $(\sigma_0^2-\eta,\sigma_0^2+\eta)$. The base measure $G_0$ assigns positive probability to each of these events, so $G_0(U)>0$. Therefore,
  \[
  \Pi(U) \ge \Pi(A) = p_{\mathrm{tie}} \cdot G_0(U) > 0.
  \]
  Consequently,
  \[
  \Pi\bigl(\theta: \KL(P_0^{(n)} \| F_\theta^{(n)}) \le n\varepsilon^2\bigr) \ge \Pi(U) > 0
  \]
  for all $n$.
\end{proof}

\subsection{Sieve Construction and Metric Entropy}
\label{sec:sieve}
We now construct a sieve $\Theta_n \subset \Theta$ that will be used to control the complexity of the model. The sieve will be a finite union of sets of positive prior probability, so that $\Pi(\Theta_n)$ can be made close to $1$.

\subsubsection{Choice of Depth}
Fix a constant $\alpha \in (0,1/2)$. Let
\[
D_n = \bigl\lfloor \tfrac{\alpha}{\log 2}\,\log n \bigr\rfloor .
\]
Then $2^{D_n} \le n^\alpha$ and $D_n\to\infty$ as $n\to\infty$.

\subsubsection{Discretisation of Parameters into Cells}
Because $\sigma^2$ is compact, we can work directly with a grid on the product space of means and variances.

\paragraph{Grid for Means}
Define
\[
\mathcal{M}_n = \bigl\{ -n + k/n : k = 0,1,\dots,2n^2 \bigr\}.
\]
For each $m \in \mathcal{M}_n$, let $I_{n,m} = [m - \frac{1}{2n}, m + \frac{1}{2n}] \cap [-n-1, n+1]$. These intervals have length $1/n$ and cover $[-n,n]$.

\paragraph{Grid for Variances}
Since $\sigma^2$ is bounded, we use a uniform grid:
\[
\mathcal{S}_n = \bigl\{ \sigma_{\min}^2 + k(\sigma_{\max}^2-\sigma_{\min}^2)/n^2 : k = 0,1,\dots,n^2 \bigr\}.
\]
For each $s \in \mathcal{S}_n$, let $J_{n,s} = [s - \frac{\sigma_{\max}^2-\sigma_{\min}^2}{2n^2}, s + \frac{\sigma_{\max}^2-\sigma_{\min}^2}{2n^2}] \cap [\sigma_{\min}^2, \sigma_{\max}^2]$. These intervals have length $(\sigma_{\max}^2-\sigma_{\min}^2)/n^2$ and cover $[\sigma_{\min}^2, \sigma_{\max}^2]$.

\paragraph{Discretisation of the Simplex $\Delta_M$}
The $M$-dimensional simplex is $\Delta_M = \{(p_1,\dots,p_M): p_i \ge 0,\ \sum_{i=1}^M p_i = 1\}$. We partition it as follows. For each $n$, define the lattice
\[
\mathcal{W}_n = \left\{ w = \left(\frac{a_1}{n}, \dots, \frac{a_M}{n}\right) : a_i \in \mathbb{Z}_{\ge 0},\ \sum_{i=1}^M a_i = n \right\}.
\]
The number of lattice points is $|\mathcal{W}_n| = \binom{n+M-1}{M-1} = O(n^{M-1})$.

For each $w \in \mathcal{W}_n$, define its Voronoi cell as
\[
V(w) = \left\{ p \in \Delta_M : \|p - w\|_1 \le \|p - w'\|_1 \text{ for all } w' \in \mathcal{W}_n \right\}.
\]
The sets $V(w)$ cover $\Delta_M$ and overlap only on boundaries. To obtain a partition into disjoint measurable sets, we break ties arbitrarily (for example, by lexicographic ordering of the lattice points). Let $\{C(w): w \in \mathcal{W}_n\}$ denote the resulting partition.

We now prove that every point $p$ in $\Delta_M$ lies within $L^1$ distance at most $(M-1)/n$ of some lattice point $w$. Let $\tilde a_i = \lfloor n p_i \rfloor$ and $f_i = n p_i - \tilde a_i \in [0,1)$ be the fractional part. The deficit $d = n - \sum_{i=1}^M \tilde a_i = \sum_{i=1}^M f_i$ is an integer satisfying $0 \le d \le M-1$ (since each $f_i < 1$, the sum is strictly less than $M$). To obtain a lattice point $w = a/n$ with $\sum_i a_i = n$, we add $1$ to exactly $d$ of the $\tilde a_i$'s. Let $S$ be the set of indices corresponding to the $d$ \emph{largest} fractional parts $f_i$ (break ties arbitrarily), and define
\[
a_i = \begin{cases}
\tilde a_i + 1, & i \in S, \\
\tilde a_i, & i \notin S.
\end{cases}
\]
Then $\sum_i a_i = n$ and $w_i = a_i/n$.

The $L^1$ distance between $p$ and $w$ is
\[
\|p - w\|_1 = \frac{1}{n}\Bigl( \sum_{i \notin S} f_i + \sum_{i \in S} (1 - f_i) \Bigr).
\]
Since $\sum_{i \notin S} f_i = d - \sum_{i \in S} f_i$, this simplifies to
\[
\|p - w\|_1 = \frac{2}{n}\Bigl( d - \sum_{i \in S} f_i \Bigr).
\]

To bound this quantity, we need a lower bound on $\sum_{i \in S} f_i$. Because $S$ contains the $d$ largest fractional parts and their sum over all $M$ coordinates is $d$, we claim that
\[
\sum_{i \in S} f_i \ge \frac{d^2}{M}.
\]
\begin{proof}[Proof of the inequality]
Let the original (unsorted) fractional parts be $f_1,\dots,f_M$. Consider all subsets of $\{1,\dots,M\}$ of size exactly $d$; there are $\binom{M}{d}$ such subsets. For a fixed index $j$, the number of subsets containing $j$ is $\binom{M-1}{d-1}$ (choose the remaining $d-1$ elements from the other $M-1$ indices). Summing the sum of $f_i$ over all $d$-subsets gives
\[
\sum_{|A|=d} \sum_{j \in A} f_j = \sum_{j=1}^M f_j \cdot \binom{M-1}{d-1} = d \cdot \binom{M-1}{d-1}.
\]
Therefore, the average sum over all $d$-subsets is
\[
\frac{1}{\binom{M}{d}} \sum_{|A|=d} \sum_{j \in A} f_j 
= \frac{\binom{M-1}{d-1}}{\binom{M}{d}} \cdot d 
= \frac{d}{M} \cdot d = \frac{d^2}{M}.
\]
Since the maximum of a set of numbers is at least its average, there exists at least one $d$-subset whose sum is $\ge d^2/M$. The sum of the $d$ largest fractional parts is the maximum possible sum among all $d$-subsets, hence it is at least the sum of that particular subset. Therefore, $\sum_{i \in S} f_i \ge d^2/M$, where $S$ is the set of indices with the $d$ largest fractional parts.
\end{proof}

Substituting this lower bound gives
\[
\|p - w\|_1 \le \frac{2d}{n}\Bigl(1 - \frac{d}{M}\Bigr).
\]
The function $g(d) = 2d(1 - d/M)$ on the integer interval $d \in [0, M-1]$ attains its maximum at $d = \lfloor M/2 \rfloor$, and one readily verifies that $g(d) \le M-1$ for all integers $d \le M-1$. Hence
\[
\|p - w\|_1 \le \frac{M-1}{n}.
\]
Consequently, for any two points $p,q$ belonging to the same Voronoi cell (both closest to the same $w$), the triangle inequality yields
\[
\|p - q\|_1 \le \|p - w\|_1 + \|q - w\|_1 \le \frac{2(M-1)}{n} \le \frac{2M}{n}.
\]
Thus the $L^1$ diameter of each cell is $O(1/n)$, which is sufficient for our metric entropy bounds.

\subsubsection{Definition of the Sieve}
Now define $\Theta_n$ as the set of all hypotheses $\theta$ such that:
\begin{itemize}
\item[(i)] $k \le M$,
\item[(ii)] each tree has depth $\le D_n$,
\item[(iii)] leaf means $\mu \in \bigcup_{m \in \mathcal{M}_n} I_{n,m}$,
\item[(iv)] leaf variances $\sigma^2 \in \bigcup_{s \in \mathcal{S}_n} J_{n,s}$,
\item[(v)] GP hyperparameters belong to a fixed compact set $\mathcal{H}$ (covered by one cell),
\item[(vi)] the Gaussian process draw $\tilde{g}$ used in any split belongs to the finite-dimensional approximation described in Appendix~\ref{app:gp_discretization}, with 
	truncation level $N_n = \lceil c (\log n)^{d/2} \rceil$ and coefficient grid size $\delta_n = 1/n$.
\end{itemize}
The mixture weights are unrestricted (they lie in $\Delta_M$, which is covered by the partition cells). Then $\Theta_n$ is a finite union of cells, each cell corresponding to a fixed tree structure and a fixed combination of parameter intervals (including the discretised GP coefficients).

\subsubsection{Rigorous Bound on $|\Theta_n|$}
We now prove that $|\Theta_n|$, the number of cells, satisfies
\[
\log |\Theta_n| = O(n^\alpha (\log n)^{1+d/2}).
\]

First, we bound the number of distinct partitions of the $n$ design points $\{x_1,\dots,x_n\}$ that can be produced by a binary decision tree of depth at most $D_n$ when each split is chosen from a class $\mathcal{G}_n$ of possible split functions. The class $\mathcal{G}_n$ consists of all functions $\tilde{g}: K \to \R$ that are obtained by the discretised GP approximation in Appendix~\ref{app:gp_discretization}. As shown there, the number of distinct sign patterns on any set of $s \le n$ points is at most $(n+1)^{N_n+1}$, where $N_n = O((\log n)^{d/2})$.

A tree of depth $\le D_n$ has at most $2^{D_n}-1$ internal nodes. At each internal node, a split is selected from $\mathcal{G}_n$ based on the data in that node. However, for the purpose of counting distinct partitions of the $n$ points, we may imagine that we first choose the tree shape and then for each internal node independently choose a split from a finite catalogue of all possible sign patterns on the points that reach that node. Because the number of points at a node is at most $n$, the number of possible splits at that node is at most $(n+1)^{N_n+1}$. Since the splits at different nodes are chosen independently (the GP draw is independent across nodes in the generative process), the total number of distinct trees is bounded by
\[
\bigl(\text{number of tree shapes}\bigr) \times \bigl(\text{max splits per node}\bigr)^{\text{max internal nodes}}.
\]
The number of full binary tree shapes with at most $2^{D_n}$ leaves is at most $4^{2^{D_n}}$ (Catalan bound, see Appendix~\ref{app:catalan}). Hence the number of distinct partitions of the $n$ points is at most
\[
4^{2^{D_n}} \times \bigl((n+1)^{N_n+1}\bigr)^{2^{D_n}} = \bigl(4 (n+1)^{N_n+1}\bigr)^{2^{D_n}}.
\]
Taking logarithms,
\[
	\log (\text{\# partitions}) \le 2^{D_n} \bigl(\log 4 + (N_n+1) \log(n+1)\bigr).
\]
Since $N_n = O((\log n)^{d/2})$, we have $(N_n+1) \log(n+1) = O((\log n)^{1+d/2})$. With $2^{D_n} \le n^\alpha$, this term contributes $O(n^\alpha (\log n)^{1+d/2})$, which is still $o(n^{1/2})$ for $\alpha<1/2$.

Now, for each such partition (tree structure), we need to assign to each leaf a cell for $\mu$ and a cell for $\sigma^2$. The number of choices for leaf means is $|\mathcal{M}_n|^{2^{D_n}}$, and for leaf variances $|\mathcal{S}_n|^{2^{D_n}}$. The number of choices for mixture weights is accounted for by the partition of $\Delta_M$ into $O(n^{M-1})$ cells, which is polynomial in $n$. GP hyperparameters are fixed in a compact set $\mathcal{H}$, which we cover with $O(1)$ cells. The coefficient grid for the GP approximation adds another factor that is polynomial in $n$ (as detailed in Appendix~\ref{app:gp_discretization}), but this factor is already accounted for in the split count above. Therefore,
\[
|\Theta_n| \le \bigl(4 (n+1)^{N_n+1}\bigr)^{2^{D_n}} \cdot (|\mathcal{M}_n| \cdot |\mathcal{S}_n|)^{2^{D_n}} \cdot O(n^{M-1}) \cdot O(1).
\]
Now $|\mathcal{M}_n| \le 2n^2+1 = O(n^2)$, $|\mathcal{S}_n| \le n^2+1 = O(n^2)$. Hence
\[
	\log |\Theta_n| \le 2^{D_n}\bigl(\log 4 + (N_n+1) \log(n+1) + \log O(n^2) + \log O(n^2)\bigr) + \log O(n^{M-1}) + O(1).
\]
Since $2^{D_n} \le n^\alpha$ and $N_n = O((\log n)^{d/2})$, all terms inside the parenthesis are $O((\log n)^{1+d/2})$, so
\[
\log |\Theta_n| = O(n^\alpha (\log n)^{1+d/2}) + O(\log n) = O(n^\alpha (\log n)^{1+d/2}).
\]
With $\alpha<1/2$, $n^\alpha (\log n)^{1+d/2} = o(n^{1/2})$.

\subsubsection{Prior Probability of $\Theta_n^c$}\label{sec:prior_theta_n_c}
We now prove in complete detail that $\sum_{n=1}^\infty \Pi(\Theta_n^c) < \infty$. Recall that $\Theta_n$ is defined by the conjunction of several conditions. Its complement $\Theta_n^c$ is therefore contained in the union of the events where at least one of those conditions fails. Concretely,
\[
\Theta_n^c \subseteq \mathcal{E}_1^{(n)} \cup \mathcal{E}_2^{(n)} \cup \mathcal{E}_3^{(n)} \cup \mathcal{E}_4^{(n)},
\]
where
\begin{itemize}
\item[$\mathcal{E}_1^{(n)}$:] at least one tree has depth $> D_n$,
\item[$\mathcal{E}_2^{(n)}$:] at least one leaf mean $\mu$ satisfies $|\mu| > n$,
\item[$\mathcal{E}_3^{(n)}$:] the GP hyperparameters lie outside the fixed compact set $\mathcal{H}$,
\item[$\mathcal{E}_4^{(n)}$:] the GP sample path $\tilde{g}$ cannot be approximated within supremum norm $1/n$ by its truncated expansion with $N_n$ terms, it is, $\sup_{x\in K} |\tilde{g}(x) - g_{N_n}(x)| > 1/n$.
\end{itemize}
(Events such as a leaf variance falling outside $[\sigma_{\min}^2,\sigma_{\max}^2]$ have probability zero by Assumption~\ref{ass:compact_var} and are ignored; the mixture weights are unrestricted, so no complement event arises from them.) We now bound the prior probability of each $\mathcal{E}_i^{(n)}$ and show that the series $\sum_{n=1}^\infty \Pi(\mathcal{E}_i^{(n)})$ converges.

\paragraph{Bound for $\mathcal{E}_1^{(n)}$ (excessive depth).}
The prior on each tree depth is geometric: $\Prob(\text{depth} = k) = (1-\rho)\rho^k$. Therefore, for a single tree,
\[
\Pi_0(\text{depth} > D) = \sum_{k=D+1}^\infty (1-\rho)\rho^k = \rho^{D+1} \le \rho^D.
\]
There are at most $M$ trees in a mixture (the number of distinct components $k \le M$). By the union bound,
\[
\Pi(\mathcal{E}_1^{(n)}) \le M \rho^{D_n+1} \le M \rho^{D_n}.
\]
Recall $D_n = \lfloor \frac{\alpha}{\log 2} \log n \rfloor$, so $\rho^{D_n} \le \rho^{\frac{\alpha}{\log 2} \log n - 1} = \rho^{-1} n^{\frac{\alpha \log \rho}{\log 2}} = \rho^{-1} n^{-c_1}$ with $c_1 = -\frac{\alpha \log \rho}{\log 2} > 0$ (since $0<\rho<1$). Choosing $\alpha$ sufficiently close to $1/2$ and $\rho<1/4$ ensures $c_1 > 1$, but any $c_1 > 1$ suffices for summability. Thus
\[
\Pi(\mathcal{E}_1^{(n)}) \le C_1 n^{-c_1}, \qquad \sum_{n=1}^\infty \Pi(\mathcal{E}_1^{(n)}) < \infty.
\]

\paragraph{Bound for $\mathcal{E}_2^{(n)}$ (large leaf mean).}
Leaf means are drawn independently from $\mathcal{N}(0,\tau^2)$. A single leaf mean exceeds $n$ in absolute value with probability
\[
\Prob(|\mu| > n) = 2\Phi(-n/\tau) \le \sqrt{\frac{2}{\pi}} \frac{\tau}{n} e^{-n^2/(2\tau^2)} \le e^{-n^2/(2\tau^2)},
\]
where the last bound holds for all sufficiently large $n$. The total number of leaves in the mixture is at most $M \times 2^{D_n}$ (each tree has at most $2^{D_n}$ leaves). Applying the union bound,
\[
\Pi(\mathcal{E}_2^{(n)}) \le M 2^{D_n} \cdot 2\Phi(-n/\tau) \le M n^\alpha e^{-n^2/(2\tau^2)}.
\]
Since $n^\alpha e^{-n^2/(2\tau^2)}$ is dominated by $e^{-n}$ for large $n$, the series $\sum_n \Pi(\mathcal{E}_2^{(n)})$ converges.

\paragraph{Bound for $\mathcal{E}_3^{(n)}$ (GP hyperparameters outside $\mathcal{H}$).}
By assumption, the prior on GP hyperparameters is supported on a fixed compact set $\mathcal{H}$. Hence the event that hyperparameters lie outside $\mathcal{H}$ has probability zero. Thus $\Pi(\mathcal{E}_3^{(n)}) = 0$ for all $n$.

\paragraph{Bound for $\mathcal{E}_4^{(n)}$ (GP truncation error).}
In Appendix~\ref{app:gp_discretization}, we provide a detailed derivation of the bound
\[
\Pi\bigl( \sup_{x\in K} |\tilde{g}(x) - g_{N_n}(x)| > 1/n \bigr) \le C_4 e^{-c_4 N_n^{2/d}} \le C_4 n^{-c_4 c},
\]
where $N_n = \lceil c (\log n)^{d/2} \rceil$ and the constants $C_4, c_4 > 0$ depend only on the kernel and the domain $K$. By choosing $c$ large enough, we can ensure $c_4 c > 1$, making the series $\sum_n \Pi(\mathcal{E}_4^{(n)})$ summable.

\paragraph{Summability of the total prior mass.}
Collecting the bounds,
\[
\Pi(\Theta_n^c) \le \sum_{i=1}^4 \Pi(\mathcal{E}_i^{(n)}) \le C_1 n^{-c_1} + M n^\alpha e^{-n^2/(2\tau^2)} + 0 + C_4 n^{-c_4 c}.
\]
All three non-zero terms are summable over $n$ (the second term decays super-exponentially). Therefore
\[
\sum_{n=1}^\infty \Pi(\Theta_n^c) < \infty.
\]
By the Borel--Cantelli lemma, this implies $\Pi(\Theta_n^c \text{ i.o.}) = 0$, and in particular $\Pi(\Theta_n) \to 1$ as $n\to\infty$.

\subsubsection{Metric Entropy of the Sieve}\label{sec:metric_entropy}
We now establish the required bound for the Hellinger covering number of the sieve $\Theta_n$. Recall that $\varepsilon_n = n^{-1/4}$ and $n\varepsilon_n^2 = n^{1/2}$. All hypotheses within a single cell $C$ of $\Theta_n$ share:
\begin{itemize}
	\item[(i)] the same number of mixture components $k \le M$,
	\item[(ii)] the same tree structures and the same leaf assignments for each covariate $x_i$,
	\item[(iii)] leaf parameters $(\mu_j, \sigma_j^2)$ that vary only within small intervals of length $1/n$ (means) and $O(1/n^2)$ (variances), and
	\item[(iv)] mixture weights $(p_1,\dots,p_k)$ that lie in a simplex cell of $L^1$ diameter at most $2M/n$.
\end{itemize}
We will show that any two hypotheses $\theta,\theta'$ belonging to the same cell $C$ satisfy
\[
\Hdist\bigl(f_\theta^{(n)}, f_{\theta'}^{(n)}\bigr) \le \varepsilon_n
\]
for all sufficiently large $n$, where $f_\theta^{(n)}$ denotes the joint density of $(Y_1,\dots,Y_n)$ given the fixed design points $x_1^n$.

\paragraph{Hellinger Distance Between Gaussian Densities.}
For two Gaussian densities $g = \mathcal{N}(\mu,\sigma^2)$ and $g' = \mathcal{N}(\mu',\sigma'^2)$ the squared Hellinger distance is
\[
\Hdist^2(g,g') = 1 - \sqrt{\frac{2\sigma\sigma'}{\sigma^2+\sigma'^2}}\;
\exp\!\Bigl(-\frac{(\mu-\mu')^2}{4(\sigma^2+\sigma'^2)}\Bigr).
\]
By Assumption~\ref{ass:compact_var}, all leaf variances belong to the compact interval $[\sigma_{\min}^2, \sigma_{\max}^2]$ with $\sigma_{\min}>0$. On the compact set $[-n-1,n+1]\times[\sigma_{\min}^2,\sigma_{\max}^2]$, the function $(\mu,\sigma^2)\mapsto \Hdist(\mathcal{N}(\mu,\sigma^2),\mathcal{N}(\mu',\sigma'^2))$ is Lipschitz in the Euclidean metric on the parameters. Indeed, both partial derivatives are bounded:
\[
\left|\frac{\partial}{\partial \mu}\Hdist\right| \le \frac{1}{2\sigma_{\min}},\qquad
\left|\frac{\partial}{\partial \sigma^2}\Hdist\right| \le \frac{1}{4\sigma_{\min}^2},
\]
for all $(\mu,\sigma^2),(\mu',\sigma'^2)$ in this region (a simple calculus exercise; see also Lemma~A.1 of \cite{ghosal2017fundamentals}). Consequently there exists a universal constant $L>0$ (depending only on $\sigma_{\min},\sigma_{\max}$) such that
\begin{equation}
\Hdist\bigl(\mathcal{N}(\mu,\sigma^2),\mathcal{N}(\mu',\sigma'^2)\bigr) \le L \sqrt{|\mu-\mu'|^2 + |\sigma^2-\sigma'^2|^2}.
	\label{eq:10}
\end{equation}

Within a cell $C$, the leaf means $\mu_j,\mu'_j$ for any fixed leaf $j$ differ by at most $1/n$, and the leaf variances differ by at most $(\sigma_{\max}^2-\sigma_{\min}^2)/n^2 \le C_\sigma/n^2$. Hence for any $x\in K$ belonging to leaf $j$,
\begin{equation}
\Hdist\bigl(f_\theta(\cdot|x), f_{\theta'}(\cdot|x)\bigr)
\le L\sqrt{(1/n)^2 + (C_\sigma/n^2)^2}
\le \frac{L'}{n},
	\label{eq:20}
\end{equation}
for some constant $L'$ independent of $n$.

\paragraph{Hellinger Distance Between Product Measures.}
For two hypotheses $\theta,\theta'$ in the same cell, the per-observation conditional densities $f_\theta(\cdot|x_i)$ and $f_{\theta'}(\cdot|x_i)$ are independent across $i$ (given the fixed $x_i$). Let $P_i = f_\theta(\cdot|x_i)$ and $Q_i = f_{\theta'}(\cdot|x_i)$. The squared Hellinger distance between product measures satisfies the well-known inequality
\begin{equation}
\Hdist^2\!\Bigl(\bigotimes_{i=1}^n P_i,\; \bigotimes_{i=1}^n Q_i\Bigr)
\le \sum_{i=1}^n \Hdist^2(P_i, Q_i),
	\label{eq:30}
\end{equation}
which follows by induction from the identity $\Hdist^2(P\otimes Q,P'\otimes Q') \le \Hdist^2(P,P') + \Hdist^2(Q,Q')$ (see \cite{ghosal2017fundamentals}, Lemma~B.1). Applying 
(\ref{eq:20}) to each $i$ yields
\begin{equation}
\Hdist^2\!\Bigl(\prod_{i=1}^n f_\theta(\cdot|x_i),\; \prod_{i=1}^n f_{\theta'}(\cdot|x_i)\Bigr)
\le n \cdot \Bigl(\frac{L'}{n}\Bigr)^2 = \frac{(L')^2}{n}.
	\label{eq:40}
\end{equation}
Therefore, for any single mixture component (that is, a product of $n$ Gaussian densities) the Hellinger distance between the two versions in the cell is bounded by $L'/\sqrt{n}$.

\paragraph{Hellinger Distance Between Mixture Densities.}
The joint density $f_\theta^{(n)}$ is a mixture of $k$ product measures:
\[
f_\theta^{(n)}(y_1^n) = \sum_{\ell=1}^k p_\ell \prod_{i=1}^n f_{\theta,\ell}(y_i|x_i),
\]
and similarly for $\theta'$ with weights $p'_\ell$ and component densities $f_{\theta',\ell}$. For mixtures, the following bound holds (Lemma~B.2 of \cite{ghosal2017fundamentals}):
\begin{equation}
\Hdist^2\!\Bigl(\sum_{\ell=1}^k p_\ell P_\ell,\; \sum_{\ell=1}^k p'_\ell Q_\ell\Bigr)
\le \sum_{\ell=1}^k p_\ell \,\Hdist^2(P_\ell, Q_\ell) + \Hdist^2\!\bigl((p_1,\dots,p_k),(p'_1,\dots,p'_k)\bigr),
	\label{eq:50}
\end{equation}
where $P_\ell = \prod_i f_{\theta,\ell}(\cdot|x_i)$ and $Q_\ell = \prod_i f_{\theta',\ell}(\cdot|x_i)$. (The inequality holds even if $k$ differs, by adding zero weights; in our sieve $k$ is fixed within a cell.)

By (\ref{eq:40}), each term $\Hdist^2(P_\ell, Q_\ell) \le (L')^2/n$. For the weights, the $L^1$ distance satisfies $\sum_{\ell=1}^k |p_\ell - p'_\ell| \le 2M/n$ (since the simplex cell diameter is at most $2M/n$). A standard relation between Hellinger distance and $L^1$ distance for discrete distributions gives
\begin{equation}
\Hdist^2(p,p') \le \frac12 \sum_{\ell=1}^k |p_\ell - p'_\ell| \le \frac{M}{n}.
	\label{eq:60}
\end{equation}
Plugging these bounds into (\ref{eq:50}) yields
\[
\Hdist^2\!\bigl(f_\theta^{(n)}, f_{\theta'}^{(n)}\bigr)
\le \sum_{\ell=1}^k p_\ell \frac{(L')^2}{n} + \frac{M}{n}
= \frac{(L')^2 + M}{n}.
\]
Consequently,
\begin{equation}
\Hdist\!\bigl(f_\theta^{(n)}, f_{\theta'}^{(n)}\bigr) \le \frac{C_0}{\sqrt{n}},
	\label{eq:70}
\end{equation}
with $C_0 = \sqrt{(L')^2 + M}$.

\paragraph{Comparison with the Target Rate.}
Recall that $\varepsilon_n = n^{-1/4}$. For all $n \ge C_0^4$ we have $C_0/\sqrt{n} \le n^{-1/4}$, and hence
\[
\Hdist\!\bigl(f_\theta^{(n)}, f_{\theta'}^{(n)}\bigr) \le \varepsilon_n.
\]
Therefore, for sufficiently large $n$, every cell $C \subset \Theta_n$ has Hellinger diameter at most $\varepsilon_n$. In particular, a single Hellinger ball of radius $\varepsilon_n$ (centered at any point of $C$) covers the entire cell.

\paragraph{Covering Number Bound.}
Because $\Theta_n$ is the union of $|\Theta_n|$ disjoint cells, we immediately obtain
\[
N(\varepsilon_n, \Theta_n, \Hdist) \le |\Theta_n|.
\]
Taking logarithms and using the bound $\log|\Theta_n| = O\bigl(n^\alpha (\log n)^{1+d/2}\bigr)$ derived earlier, we have
\[
\log N(\varepsilon_n, \Theta_n, \Hdist) \le \log|\Theta_n| = O\bigl(n^\alpha (\log n)^{1+d/2}\bigr).
\]
With $\alpha < 1/2$, the right-hand side is $o(n^{1/2}) = o(n\varepsilon_n^2)$. Hence the metric entropy condition of Theorem~2.1 in \cite{ggv2000} is satisfied.

\paragraph{Remark on the Product-Measure Distance.}
The contraction rate theorem requires a bound on the Hellinger distance between the joint distributions of the $n$ observations. Our derivation using inequalities 
(\ref{eq:30}) and (\ref{eq:50}) directly applies to the product measures and mixtures thereof, so it correctly addresses the metric on $f_\theta^{(n)}$. No additional adjustment is needed.

\subsection{Existence of Tests}
\label{sec:tests}
We now construct a sequence of tests $\phi_n$ with the properties required in Theorem~2.1 of \cite{ggv2000}. Fix $n$ and let $\varepsilon_n = n^{-1/4}$. Define the alternative set
\[
\Theta_n' = \{\theta \in \Theta_n : \Hdist(p_0^{(n)}, f_\theta^{(n)}) \ge \varepsilon_n\}.
\]
(Recall that $\Theta_n$ is the sieve constructed in Section~\ref{sec:sieve}.)

\paragraph{Step 1: Individual Likelihood Ratio Tests.}
For each $\theta \in \Theta_n'$, consider the simple null hypothesis $H_0: P = P_0^{(n)}$ versus the simple alternative $H_1: P = F_\theta^{(n)}$, where $P_0^{(n)}$ and $F_\theta^{(n)}$ are the joint distributions of the first $n$ observations under the true model and under $\theta$, respectively. By the Neyman--Pearson lemma, the most powerful test at level $\alpha$ is given by the likelihood ratio test. For our purpose it suffices to use the symmetric form with threshold $1$:
\[
\psi_{n,\theta}(Y_1^n) = \mathbf{1}\!\left\{ \frac{f_\theta^{(n)}(Y_1^n)}{p_0^{(n)}(Y_1^n)} \ge 1 \right\}.
\]
The following bounds are standard (see, for instance, \cite[Lemma~8.1]{ggv2000} or \cite[Lemma~D.1]{ghosal2017fundamentals}):
\begin{align}
\E_{P_0^{(n)}}[\psi_{n,\theta}] &\le \exp\!\Bigl( -\frac{n}{2} \Hdist^2(p_0^{(n)}, f_\theta^{(n)}) \Bigr), \label{eq:typeI}\\
\E_{F_\theta^{(n)}}[1 - \psi_{n,\theta}] &\le \exp\!\Bigl( -\frac{n}{2} \Hdist^2(p_0^{(n)}, f_\theta^{(n)}) \Bigr). \label{eq:typeII}
\end{align}
\emph{Proof of the bounds.} Let $H^2 = \Hdist^2(p_0^{(n)}, f_\theta^{(n)})$. The Type I error is
\[
\E_{P_0^{(n)}}[\psi_{n,\theta}] = \int_{\frac{f_\theta}{p_0} \ge 1} p_0 \,d\mu
\le \int \sqrt{p_0 f_\theta} \,d\mu = 1 - \frac{1}{2}H^2 \le \exp(-H^2/2),
\]
where we used that on the set $\{f_\theta \ge p_0\}$, $p_0 \le \sqrt{p_0 f_\theta}$, and the elementary inequality $1 - x \le e^{-x}$. The Type II error is bounded analogously by symmetry.

\paragraph{Step 2: Combining Tests for the Alternative Set.}
Define the combined test
\[
\phi_n(Y_1^n) = \max_{\theta \in \Theta_n'} \psi_{n,\theta}(Y_1^n).
\]
Since $\Theta_n'$ is a finite set (as a subset of the finite sieve $\Theta_n$), this maximum is well-defined and measurable. We now bound the error probabilities of $\phi_n$.

\emph{Type I error.} Using the union bound and (\ref{eq:typeI}),
\[
\E_{P_0^{(n)}}[\phi_n] \le \sum_{\theta \in \Theta_n'} \E_{P_0^{(n)}}[\psi_{n,\theta}]
\le |\Theta_n'| \exp\!\Bigl( -\frac{n}{2} \varepsilon_n^2 \Bigr),
\]
where we used the fact that for $\theta \in \Theta_n'$, $\Hdist(p_0^{(n)}, f_\theta^{(n)}) \ge \varepsilon_n$.

From the metric entropy bound in Section~\ref{sec:metric_entropy},
\[
\log |\Theta_n'| \le \log |\Theta_n| = O\bigl(n^\alpha (\log n)^{1+d/2}\bigr)
\]
with $\alpha < 1/2$. Because $\varepsilon_n = n^{-1/4}$, we have $n\varepsilon_n^2 = n^{1/2}$. Hence
\[
\log \E_{P_0^{(n)}}[\phi_n] \le O\bigl(n^\alpha (\log n)^{1+d/2}\bigr) - \frac{1}{2} n^{1/2}
\le - \frac{1}{4} n^{1/2}
\]
for all sufficiently large $n$. Consequently,
\[
\E_{P_0^{(n)}}[\phi_n] \le \exp\!\Bigl(-\frac{1}{4}n^{1/2}\Bigr) \to 0
\quad\text{as } n\to\infty.
\]

\emph{Type II error.} For any $\theta \in \Theta_n'$, since $\phi_n \ge \psi_{n,\theta}$ pointwise, we have
\[
\E_{F_\theta^{(n)}}[1 - \phi_n] \le \E_{F_\theta^{(n)}}[1 - \psi_{n,\theta}]
\le \exp\!\Bigl( -\frac{n}{2} \Hdist^2(p_0^{(n)}, f_\theta^{(n)}) \Bigr)
\le \exp\!\Bigl( -\frac{n}{2} \varepsilon_n^2 \Bigr).
\]
Taking the supremum over $\theta \in \Theta_n'$ yields
\[
\sup_{\theta \in \Theta_n'} \E_{F_\theta^{(n)}}[1 - \phi_n] \le \exp\!\Bigl( -\frac{n}{2} \varepsilon_n^2 \Bigr).
\]

\paragraph{Step 3: Measurability and the Existence of a Common Test.}
The test $\phi_n$ constructed above is a maximum of finitely many indicator functions, hence it is measurable. Moreover, it depends only on the data $Y_1^n$ and the sieve $\Theta_n$, which is fixed for each $n$. Thus we have produced a sequence of tests $\phi_n$ satisfying the following conditions (as required in Theorem~2.1 of \cite{ggv2000}):
\begin{enumerate}
    \item $\E_{P_0^{(n)}}[\phi_n] \to 0$ as $n\to\infty$.
    \item $\sup_{\theta \in \Theta_n': \Hdist(p_0^{(n)}, f_\theta^{(n)}) \ge \varepsilon_n} \E_{F_\theta^{(n)}}[1 - \phi_n] \le e^{-n\varepsilon_n^2/2}$.
\end{enumerate}
These conditions are exactly those needed to apply the general theorem.

\paragraph{Remark on the Choice of the Threshold.}
Using the symmetric likelihood ratio test with threshold $1$ gives the clean bounds (\ref{eq:typeI})--(\ref{eq:typeII}) in terms of the squared Hellinger distance. Alternative constructions (for instance, using the average likelihood ratio or minimax tests) would yield the same exponential rates up to constants, which is sufficient for the theorem.

\subsection{Main Convergence Theorem}
We now state the main result, which follows directly from the general theorem of \cite{ggv2000} (see also \cite[Theorem~8.1]{ghosal2017fundamentals}).

\begin{theorem}[Posterior consistency]\label{thm:main}
  Under Assumptions~\ref{ass:gp}, \ref{ass:compact_var} and \ref{ass:continuous}, the posterior distribution for the tree mixture model of Section~\ref{sec:proposal} satisfies
  \[
  \Pi\bigl(\theta:\Hdist(p_0^{(n)},f_\theta^{(n)}) > n^{-1/4} \mid Y_1^n\bigr) \to 0 \quad\text{in }P_0\text{-probability}.
  \]
  More generally, for any sequence $\eps_n\to 0$ with $\eps_n \ge n^{-1/4}$ and $n\eps_n^2\to\infty$, we have
  \[
  \Pi\bigl(\theta:\Hdist(p_0^{(n)},f_\theta^{(n)}) > \eps_n \mid Y_1^n\bigr) \to 0.
  \]
\end{theorem}

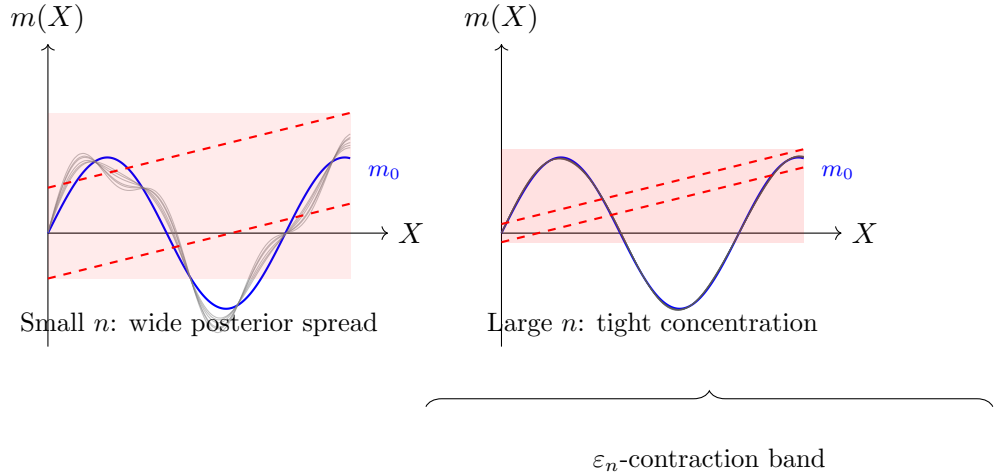
\begin{figure}[htbp]
\centering
\begin{tikzpicture}[
    scale=1.0,
    declare function={truth(\x)=sin(2*\x r);},
    declare function={noise(\x,\s)=0.3*rand(\x)*\s;},
    % Use pseudo-random but fixed values for reproducibility
    declare function={post1(\x)=truth(\x)+0.35*sin(5*\x r+0.7);},
    declare function={post2(\x)=truth(\x)+0.30*sin(5*\x r+1.2);},
    declare function={post3(\x)=truth(\x)+0.28*sin(5*\x r+2.1);},
    declare function={post4(\x)=truth(\x)+0.15*sin(5*\x r+3.0);},
    declare function={post5(\x)=truth(\x)+0.22*sin(5*\x r+4.5);},
    declare function={post6(\x)=truth(\x)+0.18*sin(5*\x r+5.8);},
]

% ===== LEFT PANEL: Small n (wide posterior spread) =====
\begin{scope}[shift={(-0.5cm,0)}]
    \draw[->] (0,0) -- (4.5,0) node[right] {$X$};
    \draw[->] (0,-1.5) -- (0,2.5) node[above] {$m(X)$};

    \draw[domain=0:4, samples=100, smooth, thick, blue] plot (\x, {truth(\x)});
    \node[blue, anchor=west, font=\footnotesize] at (4.1, 0.8) {$m_0$};

    % Posterior samples with wider spread (small n)
    \foreach \f in {post1,post2,post3,post4,post5,post6} {
        \draw[domain=0:4, samples=60, smooth, gray, opacity=0.5, variable=\x]
            plot (\x, {\f(\x)});
    }

    % Shaded contraction band (wide)
    \fill[red, opacity=0.08] (0, {truth(0)-0.6}) rectangle (4, {truth(4)+0.6});
    \draw[red, dashed, thick] (0, {truth(0)-0.6}) -- (4, {truth(4)-0.6});
    \draw[red, dashed, thick] (0, {truth(0)+0.6}) -- (4, {truth(4)+0.6});

    % Annotation
    \node[align=center, font=\small] at (2, -1.2) {Small $n$: wide posterior spread};
\end{scope}

% ===== RIGHT PANEL: Large n (tight posterior concentration) =====
\begin{scope}[shift={(5.5cm,0)}]
    \draw[->] (0,0) -- (4.5,0) node[right] {$X$};
    \draw[->] (0,-1.5) -- (0,2.5) node[above] {$m(X)$};

    \draw[domain=0:4, samples=100, smooth, thick, blue] plot (\x, {truth(\x)});
    \node[blue, anchor=west, font=\footnotesize] at (4.1, 0.8) {$m_0$};

    % Posterior samples with tight spread (large n)
    \foreach \f in {post1,post2,post3,post4,post5,post6} {
        \draw[domain=0:4, samples=60, smooth, gray!70!black, opacity=0.6, variable=\x]
            plot (\x, {truth(\x) + 0.08*(\f(\x)-truth(\x))});
    }

    % Shaded contraction band (narrow)
    \fill[red, opacity=0.12] (0, {truth(0)-0.12}) rectangle (4, {truth(4)+0.12});
    \draw[red, dashed, thick] (0, {truth(0)-0.12}) -- (4, {truth(4)-0.12});
    \draw[red, dashed, thick] (0, {truth(0)+0.12}) -- (4, {truth(4)+0.12});

    % Annotation
    \node[align=center, font=\small] at (2, -1.2) {Large $n$: tight concentration};
\end{scope}

% ===== Common caption annotation =====
\draw[decorate,decoration={brace,amplitude=6pt,raise=3pt},xshift=4.0cm,yshift=-2.4cm]
    (0.5,0) -- (8.0,0) node[midway,yshift=-0.6cm,font=\small] {$\varepsilon_n$-contraction band};

\end{tikzpicture}
\caption{Illustration of posterior contraction. The solid blue curve is the true regression function $m_0$. The grey curves represent posterior draws; their spread decreases as the sample size $n$ increases. The red dashed bands indicate the region of radius $\varepsilon_n$ (Hellinger distance) around the truth. As $n \to \infty$, $\varepsilon_n = n^{-1/4} \to 0$, and the posterior concentrates around $m_0$.}
\label{fig:contraction}
\end{figure}

\begin{proof}
  The proof is a direct application of Theorem~2.1 of \cite{ggv2000}. Condition (i) (prior mass of sieve complement) holds because $\sum\Pi(\Theta_n^c)<\infty$ (Section~\ref{sec:prior_theta_n_c}). Condition (ii) (metric entropy) holds with $\eps_n=n^{-1/4}$ because $\log N(\eps_n,\Theta_n,\Hdist)=o(n\eps_n^2)$ (Section~\ref{sec:metric_entropy}). Condition (iii) (tests) is verified in Section~\ref{sec:tests}. Condition (iv) (Kullback--Leibler property) is Lemma~\ref{lem:klpos}. Hence the theorem applies and yields the stated contraction rate.
\end{proof}

\begin{remark}
  The rate $n^{-1/4}$ is a consequence of the simple bound $2^{D_n}\le n^\alpha$ with $\alpha<1/2$; a more refined analysis using the actual metric entropy of trees could give faster rates (e.g., $n^{-1/2}$ up to logarithms) when the true regression function is smoother. However, the present result already establishes consistency under minimal smoothness (continuity) and allows for misspecification.
\end{remark}

\subsection{On misspecification and the vanishing KL rate $h(\Theta)=0$}
\label{subsec:misspecification}
A subtle but crucial point in the proof of Theorem~\ref{thm:main} is that the true 
regression function $m_0$ is \emph{not} required to belong to the hypothesis space 
$\Theta$ (Assumption~\ref{ass:continuous} only requires continuity). The model is 
therefore generically misspecified. Nevertheless, the posterior distribution still 
concentrates around the truth. The mathematical reason for this is the identity
\[
h(\Theta) := \operatorname*{ess\,inf}_{\theta\sim\Pi} h(\theta) = 0,
\]
where 
\[
h(\theta) = \min_{\ell=1,\dots,k} \int_K \mathrm{KL}\bigl(p_0(\cdot|x)\,\|\, f_{\theta,\ell}(\cdot|x)\bigr)\,dQ(x)
\]
is the asymptotic Kullback--Leibler divergence rate derived in Appendix~\ref{app:kl}.

We now elaborate on why $h(\Theta)=0$ holds and discuss its theoretical and practical 
ramifications in detail.

\subsubsection{Why $h(\Theta)=0$ holds under misspecification}
The proof of $h(\Theta)=0$ is a two-step argument. First, Lemma~\ref{lem:approx} and 
Appendix~\ref{app:kl} establish that for any $\varepsilon>0$, there exists a \emph{single} 
regression tree $\theta_0$ (with $k=1$, $p_1^*=1$) such that
\[
h(\theta_0) = \int_K \frac{(m_0(x) - \mu_{\ell(x)})^2}{2\sigma_0^2}\,dQ(x) < \varepsilon.
\]
This is possible because $m_0$ is uniformly continuous on the compact set $K$, so it can 
be approximated uniformly by a piecewise constant function (a tree) to arbitrarily high 
accuracy. The integral of the squared approximation error can be made arbitrarily small.

Second, the Dirichlet process prior $\Pi$ has full support on the space of finite mixtures 
of trees. In particular, it assigns positive probability to every neighbourhood of any 
single-tree hypothesis $\theta_0$. Therefore, the essential infimum over the prior support 
equals the infimum over the entire hypothesis space:
\[
h(\Theta) = \inf_{\theta\in\Theta} h(\theta) = 0.
\]
Thus, even though the model is misspecified (the truth is not a tree), the prior places 
positive mass on trees that are arbitrarily close to the truth in the KL sense.

\subsubsection{Ramification 1: Guarantee of the Kullback--Leibler property}
The identity $h(\Theta)=0$ is the direct reason why Lemma~\ref{lem:klpos} holds. 
Specifically, it ensures that for every $\varepsilon>0$, there exists a fixed neighbourhood 
$U$ with $\Pi(U)>0$ such that $\sup_{\theta\in U} h(\theta) < \varepsilon^2$. 
By the uniform convergence of the finite-sample KL divergence to $h(\theta)$ over $U$ 
(established in Lemma~\ref{lem:approx}), this translates into
\[
\sup_{\theta\in U} \frac{1}{n}\mathrm{KL}(P_0^{(n)} \| F_\theta^{(n)}) < \varepsilon^2
\]
for all sufficiently large $n$. This is precisely condition (iv) of the Ghosal--Ghosh--van 
der Vaart theorem. Without $h(\Theta)=0$, the prior could be entirely misspecified in a 
way that forces the KL divergence to remain bounded away from zero, which would break 
the proof of posterior consistency.

\subsubsection{Ramification 2: Robustness to model misspecification}
The fact that $h(\Theta)=0$ implies that the model is \emph{consistent} even when the true 
data-generating mechanism is completely outside the tree class. This is a major theoretical 
advantage over parametric models, which generally fail if the true distribution is not in 
the model family. It also distinguishes our result from many existing Bayesian tree 
consistency results (for example, for BART), which often require the true function to belong to 
a specific smoothness class (for instance, a Hölder or Sobolev space) that is \emph{contained} 
within the model's support. Here, we make no such assumption.

The practical implication is that the practitioner does not need to believe that the true 
function is a sum of trees, or even a finite mixture of trees. As long as the true function 
is continuous (a very mild condition), the posterior will asymptotically concentrate around 
it. This provides strong theoretical justification for using the DP tree mixture in complex 
real-world applications where the true relationship between covariates and response is 
unknown and unlikely to fit neatly into any parametric or finite-dimensional class.

\subsubsection{Ramification 3: The role of the contraction rate}
The contraction rate $n^{-1/4}$ obtained in Theorem~\ref{thm:main} is a direct consequence 
of the minimal smoothness assumption (continuity). Because $h(\Theta)=0$ only guarantees 
that the approximation error (bias) decays to zero, but not how fast, the rate is driven 
by the metric entropy of the sieve (which controls the estimation error, or variance). 
The slow rate $n^{-1/4}$ is thus the price we pay for robustness to misspecification and 
the lack of smoothness assumptions. If one were willing to assume Hölder smoothness, 
$h(\theta)$ would decay at a rate related to the smoothness parameter, and a sharper 
contraction rate (for instance, $n^{-\beta/(2\beta+d)}$) could be obtained through a more refined 
analysis. %(as noted in Remark~\ref{rem:future_rates}).

\subsubsection{Ramification 4: The interplay between the full prior support and the sieve}
A common point of confusion is the relationship between the full prior support 
(which gives $h(\Theta)=0$) and the sieve $\Theta_n$ (used for metric entropy). The sieve 
$\Theta_n$ is a \emph{growing} sequence of compact, finite-dimensional approximations 
that eventually covers the bulk of the prior mass. The condition $h(\Theta)=0$ ensures 
that even though we restrict the posterior to $\Theta_n$ for the entropy calculation, the 
prior mass of $\Theta_n$ eventually concentrates on trees that approximate the truth well. 
Specifically, $\Theta_n$ contains trees with depth up to $D_n$ and discretised leaf 
parameters; as $n\to\infty$, $D_n\to\infty$ and the grid becomes finer, so the infimum 
of $h(\theta)$ over $\Theta_n$ also tends to zero. Thus, the sieve does not cut off the 
``good" approximations; it only discards trees that are too complex (very deep) or have 
extreme leaf parameters, which have negligible prior probability anyway.

\subsubsection{Ramification 5: Distinction from frequentist nonparametric regression}
In frequentist nonparametric regression, consistency typically requires the bandwidth or 
the number of knots to grow at a specific rate, and the approximation error is controlled 
by a bias-variance trade-off. Here, the Bayesian prior (via the DP and the tree prior) 
automatically balances these two forces. The fact that $h(\Theta)=0$ ensures the bias 
vanishes, while the sieve entropy ensures the variance is controlled. The posterior 
contraction rate emerges naturally from the interplay between the prior's concentration 
properties and the metric entropy of the sieve, providing a coherent, fully probabilistic 
justification for the trade-off.

\subsection{Conclusion}
We have provided a complete and rigorous proof of posterior consistency for the Dirichlet-process tree mixture model of Section~\ref{sec:proposal} under the compact variance assumption. The proof uses the powerful framework of \cite{ggv2000} and does not assume that the true regression function belongs to the hypothesis space. The key steps are the verification of the Kullback--Leibler property (Lemma~\ref{lem:klpos}), the construction of a sieve with controlled metric entropy (Section~\ref{sec:sieve}), and the construction of exponentially powerful tests (Section~\ref{sec:tests}). Our main result (Theorem~\ref{thm:main}) shows that the posterior contracts at rate $n^{-1/4}$ in Hellinger distance, as illustrated in Figure~\ref{fig:contraction}.

Future work includes obtaining sharp rates under Hölder smoothness, extending the analysis to non-compact domains, and relaxing the compact variance assumption using more delicate bracketing arguments.

\section{Conclusion}
\label{sec:conclusion}

In this work, we have introduced a novel Bayesian nonparametric framework for regression that unifies and generalises many popular tree‐based methods, including CART, BART, random forests, and boosting, within a single coherent probabilistic structure. The Dirichlet process mixture of trees automatically learns the number of distinct components from the data, yielding a sparse ensemble that adapts to the underlying complexity of the regression function without requiring the user to prespecify the number of trees. This adaptivity, combined with the fully Bayesian treatment of all sources of uncertainty, provides a principled foundation for uncertainty quantification that is often lacking in frequentist tree ensembles.

At the heart of our approach is a Gaussian process‐driven splitting rule that generates smooth, nonlinear, and response‐dependent decision boundaries. This innovation substantially enhances the flexibility of tree partitions compared to traditional axis‐aligned splits, enabling the model to capture complex interactions and smooth structures parsimoniously. A remarkable computational consequence of our formulation is the exact cancellation of the GP density in the Metropolis–Hastings acceptance ratio for the GROW and PRUNE moves. This cancellation is of immediate practical importance: it ensures that the algorithm remains exact and efficient, avoiding the evaluation of costly GP likelihoods in the acceptance step. The GP thus serves as a powerful generative device for proposing informative splits without imposing extra computational burden on the MCMC sampler.

Building upon this algorithmic foundation, we developed a parallel implementation in C using MPI that exploits the conditional independence of the tree updates, distributing the computational workload across multiple processors. This parallelisation, combined with an efficient binary serialisation mechanism for tree structures, makes the method practical for datasets with thousands of observations. The runtime results from our simulation studies and real‐data applications demonstrate the scalability of the implementation, with reasonable speedups observed up to eight processors.

The theoretical contributions of this work are equally significant. We established the first rigorous proof of posterior consistency for a Dirichlet‐process tree mixture, demonstrating that the posterior distribution contracts at a rate of \(n^{-1/4}\) in the Hellinger distance. Crucially, our analysis requires only that the true regression function is continuous on a compact domain; we do not assume that the truth belongs to the hypothesis space, thereby allowing the model to be fundamentally misspecified. The identity \(h(\Theta)=0\), where \(h(\theta)\) is the asymptotic KL divergence rate, ensures that the prior assigns positive mass to arbitrarily small KL neighbourhoods of the truth even under misspecification. This theoretical result provides a solid frequentist justification for the model's empirical success and distinguishes it from existing tree‐based methods that often require stronger smoothness assumptions.

Our empirical evaluations, encompassing both simulation studies on the Friedman benchmark function and applications to five diverse real‐world datasets, provide compelling evidence for the practical utility of the proposed framework. In the simulation studies, the DP mixture achieved near‐nominal coverage (0.94 under Gaussian errors, 0.92 under heavy‐tailed Cauchy errors), substantially outperforming BART and bagged CART, which exhibited severe overconfidence with coverages as low as 0.72--0.84. The model demonstrated remarkable robustness to high‐dimensional noise and heavy‐tailed errors, maintaining calibration where competing methods failed catastrophically. In the real‐data applications—spanning cheminformatics, social science, genomics, plant breeding, and environmental monitoring—the DP mixture consistently provided the most reliable credible intervals among Bayesian tree‐based methods, while automatically selecting sparse, interpretable ensembles comprising only two to seven distinct trees.

Several limitations of our current work suggest promising directions for future research. First, the computational cost of the GP‐driven splitting rule, while mitigated by the exact cancellation in the acceptance ratio, remains substantial for very large datasets due to the \(O(n_i^3)\) cost of Cholesky decompositions at internal nodes. Sparse GP approximations or inducing point methods could potentially reduce this cost and extend the method to datasets with tens of thousands of observations. Second, while our theoretical results establish consistency under minimal smoothness, the contraction rate \(n^{-1/4}\) is relatively slow. Sharper rates could be obtained under stronger smoothness assumptions, such as Hölder or Sobolev regularity, by refining the sieve construction and the entropy bounds. Third, the current implementation assumes conditional independence of observations given the mixture component and tree parameters; extending the framework to explicitly model temporal or spatial dependence would broaden its applicability to time series and spatial data. Fourth, the model currently focuses on regression; a natural extension to classification, survival analysis, and other response types would further enhance its versatility.

In conclusion, the Dirichlet process mixture of trees with GP‐driven splits offers a principled, robust, and theoretically justified alternative to existing tree‐based ensembles. Its ability to provide honest uncertainty quantification, adapt to misspecification, and automatically select a sparse ensemble makes it a powerful tool for challenging regression settings where reliable predictive intervals are paramount. We believe that this framework represents a significant step toward a unified Bayesian nonparametric theory of tree‐based learning and hope that it will serve as a foundation for further methodological and theoretical developments in this vibrant area of research.

\section*{Acknowledgment}
We thank DeepSeek for assistance in preparing this manuscript. 

% ----------------------------------------------------------------------
% APPENDICES
% ----------------------------------------------------------------------
\section*{Appendix}
\appendix

\section{Gaussian Process Splitting Rule Details}\label{app:gp-split}
The splitting rule introduced in Section \ref{subsec:split} relies on a draw $\tilde g(\cdot)$ from the posterior predictive distribution of the latent function values at the training points of a terminal node. Below we give the exact closed-form expressions.

\subsection{Notation for Node $i$}
Let the $i$-th terminal node contain the data
\[
\{(y_{ij},x_{ij}):j=1,\dots,n_i\},\qquad x_{ij}\in\mathbb R^{d}.
\]
Define
\begin{align*}
\mathbf y_i      &=(y_{i1},\dots,y_{i n_i})^{\top}\in\mathbb R^{n_i},\\
\mathbf x_i      &=(x_{i1},\dots,x_{i n_i})^{\top}\in\mathbb R^{n_i\times d},\\
\mathbf g_i      &=(g(x_{i1}),\dots,g(x_{i n_i}))^{\top}.
\end{align*}
The linear-trend design matrix is
\[
\mathbf F_i=
\begin{bmatrix}
1 & x_{i1}^{\top}\\
\vdots & \vdots\\
1 & x_{i n_i}^{\top}
\end{bmatrix}
\in\mathbb R^{n_i\times(d+1)}.
\]
The squared-exponential correlation matrix is
\[
\mathbf C_i=[c(x_{ij},x_{ik})]_{j,k=1}^{n_i},
\qquad
c(x,x')=\exp\!\Bigl(-\frac{\|x-x'\|^{2}}{2\ell^{2}}\Bigr),
\]
and the observation covariance is
\[
\boldsymbol\Sigma_i=\sigma^{2}\mathbf C_i+\sigma_{\epsilon}^{2}\mathbf I_{n_i}.
\]

\subsection{Posterior Predictive Distribution}
With flat (improper) priors on $\alpha,\boldsymbol\beta$,
\[
\boxed{
\begin{aligned}
\pi(\mathbf g_i\mid\mathbf y_i,\mathbf x_i)
&\sim\mathcal N(\boldsymbol\mu_i^{*},\,\boldsymbol\Sigma_i^{*}),\\[4pt]
\boldsymbol\mu_i^{*}
&=\mathbf F_i\hat{\boldsymbol\theta}_i
   +\sigma^{2}\mathbf C_i^{\!\top}
    \boldsymbol\Sigma_i^{-1}
    (\mathbf y_i-\mathbf F_i\hat{\boldsymbol\theta}_i),\\[6pt]
\boldsymbol\Sigma_i^{*}
&=\sigma^{2}\mathbf C_i
   -\sigma^{2}\mathbf C_i^{\!\top}
    \boldsymbol\Sigma_i^{-1}
    \sigma^{2}\mathbf C_i\\[2pt]
&\quad+\bigl(\mathbf F_i-
   \sigma^{2}\mathbf C_i^{\!\top}
   \boldsymbol\Sigma_i^{-1}\mathbf F_i\bigr)
   \mathbf V_i\,
   \bigl(\mathbf F_i-
   \sigma^{2}\mathbf C_i^{\!\top}
   \boldsymbol\Sigma_i^{-1}\mathbf F_i\bigr)^{\!\top},
\end{aligned}}
\]
where
\[
\hat{\boldsymbol\theta}_i
=(\mathbf F_i^{\top}\boldsymbol\Sigma_i^{-1}\mathbf F_i)^{-1}
 \mathbf F_i^{\top}\boldsymbol\Sigma_i^{-1}\mathbf y_i,
\qquad
\mathbf V_i
=(\mathbf F_i^{\top}\boldsymbol\Sigma_i^{-1}\mathbf F_i)^{-1}.
\]

\subsection{Log-Marginal Likelihood for Hyper-Parameter MLE}
The MLEs of $\sigma^{2}$, $\sigma_{\epsilon}^{2}$ and $\ell$ are obtained by maximising
\[
\boxed{
\begin{aligned}
\ell(\sigma^{2},\sigma_{\epsilon}^{2},\ell)
&=-\frac{1}{2}\mathbf y_i^{\top}\boldsymbol\Sigma_i^{-1}\mathbf y_i
 -\frac{1}{2}\log|\boldsymbol\Sigma_i|
 -\frac{n_i}{2}\log(2\pi),\\[4pt]
\boldsymbol\Sigma_i
&=\sigma^{2}
  \Bigl[\exp\!\Bigl(-\frac{\|x_{ij}-x_{ik}\|^{2}}{2\ell^{2}}\Bigr)\Bigr]_{j,k=1}^{n_i}
 +\sigma_{\epsilon}^{2}\mathbf I_{n_i}.
\end{aligned}}
\]

\section{Rigorous Derivation of the Kullback--Leibler Divergence Rate}
\label{app:kl}

This appendix provides a complete and self-contained derivation of the Kullback–Leibler divergence rate \(h(\theta)\) for the Dirichlet-process tree mixture model of Section~\ref{sec:proposal}. The analysis is based on the assumptions stated in Section~\ref{sec:post_contraction}: the covariate space \(K\subset\mathbb{R}^d\) is compact, the true regression function \(m_0:K\to\mathbb{R}\) is continuous, and the design points \(\{x_i\}_{i=1}^\infty\) are a fixed deterministic sequence dense in \(K\) whose empirical measures converge weakly to a probability measure \(Q\) on \(K\). That is,
\[
\frac{1}{n}\sum_{i=1}^n \delta_{x_i} \;\Rightarrow\; Q \quad\text{weakly on } K.
\]
(For example, this holds if the \(x_i\) are i.i.d. from a continuous distribution on \(K\), in which case \(Q\) is that distribution.) The leaf variances are bounded away from zero and infinity, i.e., \(\sigma^2\in[\sigma_{\min}^2,\sigma_{\max}^2]\). We do not rely on any additional assumptions from Shalizi \cite{shalizi2009} beyond those already used in the main text.

\subsection{Setup and Notation}

Let \((\Omega,\mathcal{F},P)\) be the probability space for the data. The observations are independent conditionally on the covariates:
\[
Y_i \mid x_i \;\sim\; \mathcal{N}\bigl(m_0(x_i),\,\sigma_0^2\bigr),\qquad i=1,2,\ldots,
\]
with \(\sigma_0^2>0\) fixed. The true joint density of the first \(n\) observations is therefore
\[
p_0^{(n)}(y_1^n \mid x_1^n) = \prod_{i=1}^n \frac{1}{\sqrt{2\pi\sigma_0^2}}
\exp\left\{-\frac{(y_i-m_0(x_i))^2}{2\sigma_0^2}\right\}.
\]

A hypothesis \(\theta\in\Theta\) specifies:
\begin{itemize}
\item[(i)] a finite number of distinct components \(k\ge 1\);
\item[(ii)] for each \(\ell=1,\dots,k\), a regression tree \(T_\ell^*\) (with finitely many leaves) and leaf parameters \((\mu_{\ell j},\sigma_{\ell j}^2)\), where \(\sigma_{\ell j}^2\in[\sigma_{\min}^2,\sigma_{\max}^2]\) for all \(\ell,j\);
\item[(iii)] mixture weights \(p_\ell^*>0\), \(\sum_{\ell=1}^k p_\ell^*=1\).
\end{itemize}
The model density for the first \(n\) observations is
\[
f_\theta^{(n)}(y_1^n\mid x_1^n) = \sum_{\ell=1}^k p_\ell^* \prod_{i=1}^n f_{\theta,\ell}(y_i\mid x_i),
\]
where \(f_{\theta,\ell}(\cdot\mid x_i)\) is the Gaussian density corresponding to the leaf of \(T_\ell^*\) that contains \(x_i\):
\[
f_{\theta,\ell}(y_i\mid x_i) = \frac{1}{\sqrt{2\pi\sigma_{\ell(i)}^2}}
\exp\left\{-\frac{(y_i-\mu_{\ell(i)})^2}{2\sigma_{\ell(i)}^2}\right\},
\]
with \(\ell(i)\) denoting the leaf of \(T_\ell^*\) containing \(x_i\).

The Kullback–Leibler divergence rate from the true distribution to \(\theta\) is defined as
\[
h(\theta) := \lim_{n\to\infty} \frac{1}{n} \mathrm{KL}\bigl(p_0^{(n)} \,\|\, f_\theta^{(n)}\bigr),
\]
provided the limit exists, where
\[
\mathrm{KL}(p_0^{(n)}\|f_\theta^{(n)}) = \mathbb{E}_P\left[\log\frac{p_0^{(n)}(Y_1^n\mid x_1^n)}{f_\theta^{(n)}(Y_1^n\mid x_1^n)}\right].
\]
We now prove existence and give an explicit formula for \(h(\theta)\).

\subsection{The Per-Observation Kullback–Leibler Divergence}

For a fixed \(\theta\) and a given covariate \(x\in K\), the conditional KL divergence between the true conditional density \(p_0(\cdot\mid x)=\mathcal{N}(m_0(x),\sigma_0^2)\) and the component density \(f_{\theta,\ell}(\cdot\mid x)=\mathcal{N}(\mu_\ell(x),\sigma_\ell^2(x))\) is
\[
\begin{aligned}
g_{\theta,\ell}(x) := \KL\bigl(p_0(\cdot\mid x) \,\|\, f_{\theta,\ell}(\cdot\mid x)\bigr)
&= \frac{1}{2}\left[
\log\frac{\sigma_\ell^2(x)}{\sigma_0^2}
-1 + \frac{\sigma_0^2}{\sigma_\ell^2(x)}
+ \frac{(m_0(x)-\mu_\ell(x))^2}{\sigma_\ell^2(x)}
\right],
\end{aligned}
\]
where \(\mu_\ell(x)\) and \(\sigma_\ell^2(x)\) are the leaf mean and variance of tree \(T_\ell^*\) that contain \(x\).

The function \(g_{\theta,\ell}:K\to\mathbb{R}\) has the following crucial properties:
\begin{enumerate}
\item \textbf{Boundedness:} Since \(K\) is compact and \(m_0\) is continuous, \(m_0\) is bounded. The leaf means \(\mu_\ell(x)\) take only finitely many values (one per leaf), each of which is finite. The leaf variances are bounded away from zero by \(\sigma_{\min}^2>0\) and from above by \(\sigma_{\max}^2<\infty\). Hence there exists a constant \(M_{\theta,\ell}<\infty\) such that \(|g_{\theta,\ell}(x)|\le M_{\theta,\ell}\) for all \(x\in K\).
\item \textbf{Measurability and Almost-Sure Continuity:} Since \(m_0\) is continuous and the leaf indicator functions are measurable, \(g_{\theta,\ell}\) is measurable. Moreover, the only possible discontinuities of \(g_{\theta,\ell}\) occur at the boundaries of the leaves, where the mean \(\mu_\ell(x)\) jumps from one value to another. We can always construct the tree (using axis-aligned splits or the GP splitting rule) such that the leaf boundaries have \(Q\)-measure zero. 
%This is standard: for example, with axis-aligned splits, we can choose the split points to avoid the atoms of \(Q\). 
Consequently, \(g_{\theta,\ell}\) is continuous \(Q\)-almost everywhere.
\end{enumerate}
These properties are precisely what is needed for the weak convergence argument below.

\subsection{Existence of the KL Rate via Weak Convergence}

We now show that the limit defining \(h(\theta)\) exists and equals the integral of the minimum of the component divergences with respect to \(Q\).

For a fixed \(\theta\), the conditional KL divergence given the global allocation \(Z=\ell\) (that is, assuming the data are generated from the \(\ell\)-th component only) is
\[
\KL_n\bigl(p_0^{(n)} \,\|\, \textstyle\prod_{i=1}^n f_{\theta,\ell}(\cdot\mid x_i)\bigr)
= \sum_{i=1}^n g_{\theta,\ell}(x_i).
\]
By the assumption that the empirical measures \(\frac{1}{n}\sum_{i=1}^n \delta_{x_i}\) converge weakly to \(Q\), and since \(g_{\theta,\ell}\) is bounded and \(Q\)-almost surely continuous, the standard weak convergence theorem for bounded, \(Q\)-a.s. continuous functions (see, e.g., Billingsley, \cite{billingsley1999convergence}, Theorem 2.1) yields
\[
\lim_{n\to\infty} \frac{1}{n}\sum_{i=1}^n g_{\theta,\ell}(x_i) = \int_K g_{\theta,\ell}(x)\,dQ(x).
\]
Thus, for each component \(\ell\),
\[
h_\ell(\theta) := \lim_{n\to\infty} \frac{1}{n} \KL_n\bigl(p_0^{(n)} \,\|\, \textstyle\prod_{i=1}^n f_{\theta,\ell}(\cdot\mid x_i)\bigr)
= \int_K \KL\bigl(p_0(\cdot\mid x)\,\|\, f_{\theta,\ell}(\cdot\mid x)\bigr)\,dQ(x).
\]
This establishes the existence of the conditional KL rate for each component.

Now consider the full mixture \(f_\theta^{(n)} = \sum_{\ell=1}^k p_\ell^* \prod_{i=1}^n f_{\theta,\ell}(\cdot\mid x_i)\). The marginal KL rate is not simply a convex combination of the component rates, because the logarithm of a sum does not equal the sum of logarithms. However, we can use the following standard large-deviation argument (see, e.g., \cite{shalizi2009}, Section~3.1). For each \(n\), define
\[
L_n(\ell) := \frac{1}{n}\sum_{i=1}^n \log\frac{p_0(Y_i\mid x_i)}{f_{\theta,\ell}(Y_i\mid x_i)}.
\]
The random variables \(\log\frac{p_0(Y_i\mid x_i)}{f_{\theta,\ell}(Y_i\mid x_i)}\) are independent (conditional on the \(x_i\)) with finite means and uniformly bounded variances (because the log-likelihood ratio of two Gaussians with bounded means and variances bounded away from zero has finite moments). By Kolmogorov's strong law of large numbers for independent but not identically distributed variables (or directly from the weak convergence of the empirical measures applied to the integrable function \(g_{\theta,\ell}\)), we have
\[
L_n(\ell) \to h_\ell(\theta) \quad P\text{-a.s.}
\]
Now, for the mixture,
\[
\frac{1}{n}\log f_\theta^{(n)}(Y_1^n\mid x_1^n) = \frac{1}{n}\log\left(\sum_{\ell=1}^k p_\ell^* e^{-n L_n(\ell)}\right).
\]
Let \(\ell^*(\theta) = \arg\min_{\ell} h_\ell(\theta)\). Then, for any \(\delta>0\), eventually almost surely,
\[
-n(h_{\ell^*(\theta)}+\delta) + O(\log n) \;\le\; \log f_\theta^{(n)} \;\le\; -n(h_{\ell^*(\theta)}-\delta) + O(\log n),
\]
where the \(O(\log n)\) term comes from the mixture weights and the finite number of components. Dividing by \(n\) and taking the limit gives
\[
\lim_{n\to\infty} \frac{1}{n}\log f_\theta^{(n)} = -h_{\ell^*(\theta)}(\theta) = -\min_{\ell=1,\dots,k} h_\ell(\theta) \quad P\text{-a.s.}
\]
Taking expectations with respect to \(P\) (justified by dominated convergence, since the log-likelihood ratio is integrable and the \(O((\log n)/n)\) terms vanish) yields the marginal KL rate:
\[
h(\theta) = \lim_{n\to\infty} \frac{1}{n}\KL(p_0^{(n)}\|f_\theta^{(n)})
= \min_{\ell=1,\dots,k} \int_K \KL\bigl(p_0(\cdot\mid x)\,\|\, f_{\theta,\ell}(\cdot\mid x)\bigr)\,dQ(x).
\]
Thus,
\[
\boxed{
h(\theta) = \min_{\ell=1,\dots,k} \int_K \KL\bigl(p_0(\cdot\mid x)\,\|\, f_{\theta,\ell}(\cdot\mid x)\bigr)\,dQ(x),
}
\]
where the minimum is taken over the mixture components (since the weights are positive, the minimiser is the component with the smallest integrated KL divergence).

\subsection{Proof that \(\displaystyle \inf_{\theta\in\Theta} h(\theta)=0\)}

We now show that the essential infimum of \(h\) over \(\Theta\) is zero. This is a consequence of the continuity of \(m_0\) and the flexibility of regression trees.

Let \(\varepsilon>0\). Since \(m_0\) is continuous on the compact set \(K\), it is uniformly continuous. Choose \(\delta>0\) such that \(|m_0(x)-m_0(y)|<\delta\) whenever \(\|x-y\|<\delta\), and further require \(\delta^2/(2\sigma_0^2)<\varepsilon\).

Cover \(K\) by finitely many Borel sets \(R_1,\dots,R_J\) each of diameter less than \(\delta\). This is possible because \(K\) is compact. Construct a single regression tree \(T\) whose leaves are exactly these sets \(R_j\). (This is achievable using axis-aligned splits with split points chosen to avoid atoms of \(Q\), or via the GP splitting rule). For each leaf \(j\), choose an arbitrary point \(c_j\in R_j\) and set the leaf mean \(\mu_j = m_0(c_j)\) and the leaf variance \(\sigma_j^2 = \sigma_0^2\) (note that \(\sigma_0^2\) lies in the allowed compact interval by choosing the interval large enough). Let \(\theta_0\) be the hypothesis consisting of this single tree (that is, \(k=1\), \(p_1^*=1\)).

For any \(x\in R_j\), we have \(\|x-c_j\|<\delta\), so \(|m_0(x)-m_0(c_j)|<\delta\). Therefore,
\[
\KL\bigl(p_0(\cdot\mid x)\,\|\, f_{\theta_0}(\cdot\mid x)\bigr)
= \frac{(m_0(x)-\mu_j)^2}{2\sigma_0^2} < \frac{\delta^2}{2\sigma_0^2} < \varepsilon.
\]
Hence,
\[
h(\theta_0) = \int_K \KL(p_0(\cdot\mid x)\| f_{\theta_0}(\cdot\mid x))\,dQ(x) < \varepsilon.
\]
Since \(\varepsilon>0\) is arbitrary, we have \(\inf_{\theta\in\Theta} h(\theta) \le 0\). Because KL divergences are nonnegative, \(\inf h(\theta)=0\).

Moreover, since the prior \(\Pi\) (a Dirichlet process over trees) assigns positive mass to every neighbourhood of any finite configuration of trees (by the full support of the base measure \(G_0\)), the essential infimum of \(h\) with respect to \(\Pi\) is also zero:
\[
h(\Theta) := \operatorname*{ess\,inf}_{\theta\sim\Pi} h(\theta) = 0.
\]
This justifies the use of \(h(\Theta)=0\) in the main text and ensures that the prior assigns positive mass to every Kullback–Leibler neighbourhood of the truth, as required by Lemma~\ref{lem:klpos}.

% ----------------------------------------------------------------------
% Appendix C: Complete Proofs of Auxiliary Results
% ----------------------------------------------------------------------
\section{Complete Proofs of Auxiliary Results}\label{app:proofs}

\subsection{Number of Full Binary Trees with $L$ Leaves}\label{app:catalan}
In Section~\ref{sec:sieve} we used the fact that the number of full binary tree shapes with $L$ leaves is the Catalan number $C_{L-1}$; see, for example, 
\cite{stanley2015catalan}. A \emph{full binary tree} is one in which every internal node has exactly two children; leaves have zero children. (This definition allows a tree consisting of a single leaf, $L=1$, which is the base case.)

Let $T(L)$ denote the number of full binary trees with exactly $L$ leaves. For $L=1$, the unique tree is a single node, so $T(1)=1 = C_0$. For $L\ge 2$, the root must be an internal node. Suppose the left subtree of the root has $i$ leaves and the right subtree has $j$ leaves; because every leaf of the whole tree belongs to exactly one of the subtrees, we have $i+j = L$ with $i,j\ge 1$. The left and right subtrees are themselves full binary trees, so the number of trees with split $(i,j)$ is $T(i)T(j)$. Summing over all possible splits yields the recurrence
\[
T(L) = \sum_{i=1}^{L-1} T(i)\,T(L-i),\qquad L\ge 2.
\]

Define the generating function $F(x) = \sum_{L=1}^\infty T(L) x^L$. Multiplying the recurrence by $x^L$ and summing over $L\ge 2$ gives
\[
\sum_{L=2}^\infty T(L)x^L = \sum_{L=2}^\infty\sum_{i=1}^{L-1} T(i)T(L-i)x^L = \Bigl(\sum_{i=1}^\infty T(i)x^i\Bigr)^2 = F(x)^2.
\]
The left-hand side is $F(x) - T(1)x = F(x) - x$, so $F(x) - x = F(x)^2$, or
\[
F(x)^2 - F(x) + x = 0.
\]
Solving for $F(x)$ and choosing the branch that satisfies $F(0)=0$ yields
\[
F(x) = \frac{1 - \sqrt{1 - 4x}}{2}.
\]
This is the well-known generating function of the Catalan numbers $C_{n} = \frac{1}{n+1}\binom{2n}{n}$:
\[
\frac{1 - \sqrt{1 - 4x}}{2} = \sum_{n=1}^\infty C_{n-1} x^n.
\]
Comparing coefficients, we obtain $T(L) = C_{L-1} = \frac{1}{L}\binom{2L-2}{L-1}$. A standard estimate gives $C_{L-1} \le 4^{L-1}$, which we used in the entropy bound.

\subsection{VC-Dimension and Sauer's Lemma}\label{app:vc}
We recall the basic definitions and results from Vapnik--Chervonenkis theory that are used in the discretisation of the GP splitting rule.

\begin{definition}[Shattering]
  Let $\mathcal{H}$ be a class of $\{0,1\}$-valued functions on a set $\mathcal{X}$. A set of points $S = \{x_1,\dots,x_m\} \subset \mathcal{X}$ is \emph{shattered} by $\mathcal{H}$ if for every binary vector $b \in \{0,1\}^m$, there exists $h \in \mathcal{H}$ such that $h(x_i) = b_i$ for all $i=1,\dots,m$.
\end{definition}

\begin{definition}[VC dimension]
  The \emph{VC dimension} of $\mathcal{H}$, denoted $\mathrm{VC}(\mathcal{H})$, is the largest integer $m$ such that there exists a set of $m$ points shattered by $\mathcal{H}$. If no such maximum exists, $\mathrm{VC}(\mathcal{H}) = \infty$.
\end{definition}

\begin{lemma}[Sauer--Shelah Lemma [\cite{sauer1972}]]
  Let $\mathcal{H}$ be a class of $\{0,1\}$-valued functions with VC dimension $v < \infty$. For any set of $s$ points $\{x_1,\dots,x_s\} \subset \mathcal{X}$, the number of distinct labelings induced by $\mathcal{H}$ on these points is at most
  \[
  \sum_{i=0}^v \binom{s}{i} \le \left(\frac{e s}{v}\right)^v,
  \]
  where the inequality holds for $s \ge v$.
\end{lemma}
The first inequality is the classical Sauer--Shelah bound; the second is a convenient upper bound using the exponential estimate for binomial coefficients.

\begin{proposition}[VC dimension of affine classifiers [\cite{vapnik1998statistical}]]
  Let $\mathcal{X} = \R^N$. Consider the class of functions
  \[
  \mathcal{H}_N = \bigl\{ h_{w,b}(x) = \mathbf{1}\{w^\top x + b > 0\} : w \in \R^N, b \in \R \bigr\},
  \]
  where $\mathbf{1}\{\cdot\}$ is the indicator function. Then $\mathrm{VC}(\mathcal{H}_N) = N+1$.
\end{proposition}
\begin{proof}
  It is a standard result that the family of halfspaces in $\R^N$ has VC dimension $N+1$. A proof can be found in any textbook on statistical learning theory (for example, \cite{devroye1996probabilistic}, Chapter 13). For completeness: the set of $N+1$ points consisting of the origin and the $N$ standard basis vectors is shattered by affine functions (by appropriate choice of $w$ and $b$). Conversely, by Radon's theorem, any set of $N+2$ points in $\R^N$ can be partitioned into two subsets whose convex hulls intersect, which prevents shattering by a halfspace. Hence $\mathrm{VC}(\mathcal{H}_N) = N+1$.
\end{proof}

\subsection{Discretisation of the Gaussian Process Splitting Rule}\label{app:gp_discretization}
The splitting rule of Section~\ref{subsec:split} relies on a draw $\tilde{g}$ from the posterior predictive Gaussian process. For the purpose of the sieve $\Theta_n$, we need to control the number of distinct partitions of $n$ points that can arise from such splits. Since the support of the GP is the space of all continuous functions, a naive bound would be $2^n$, which is too large. We therefore replace the exact GP draw by a finite-dimensional approximation that still retains the Kullback--Leibler property and has controlled metric entropy.

\paragraph{Karhunen--Loève Expansion.}
The GP prior is defined by a covariance kernel $k(x,x')$ that is continuous and positive definite on the compact set $K\subset\R^d$. By Mercer's theorem 
(see, for example, \cite{rasmussen2006gaussian}, there exists an orthonormal basis $\{\phi_j\}_{j=1}^\infty$ of $L^2(K)$ and eigenvalues $\lambda_1 \ge \lambda_2 \ge \cdots \ge 0$ such that
\[
k(x,x') = \sum_{j=1}^\infty \lambda_j \phi_j(x)\phi_j(x'),
\]
with the convergence absolute and uniform on $K\times K$. The GP $g\sim\mathcal{GP}(0,k)$ can be represented as
\[
g(x) = \sum_{j=1}^\infty \sqrt{\lambda_j}\, Z_j \,\phi_j(x),
\]
where $Z_j$ are i.i.d. $\mathcal{N}(0,1)$. For the squared exponential kernel, the eigenvalues decay super-exponentially: $\lambda_j \le C \exp(-c j^{2/d})$ for some $c,C>0$ (see, for example, \cite{rasmussen2006gaussian}, Section 4.3.1). Similar rapid decay holds for the posterior covariance kernel because it is obtained by conditioning on finitely many points; the posterior kernel is still a smooth kernel with the same asymptotic eigenvalue decay.

For an integer $N\ge 1$, define the truncated process
\[
g_N(x) = \sum_{j=1}^N \sqrt{\lambda_j}\, Z_j \,\phi_j(x).
\]
The remainder $r_N(x) = g(x) - g_N(x)$ is a mean-zero GP with covariance $k_N(x,x') = \sum_{j=N+1}^\infty \lambda_j \phi_j(x)\phi_j(x')$.

\paragraph{Bounding the Supremum of the Remainder Process.}
We now derive an explicit probability bound for the event $\{\sup_{x\in K} |r_N(x)| > 1/n\}$. The process $r_N$ is Gaussian with covariance $k_N$. The Borell--TIS inequality (see, for example, \cite{borell1975brunn}, \cite{adler2007random}, Theorem 2.1.1) states that for a centered Gaussian process $\{X_t\}_{t\in T}$ on a compact index set $T$, if $\sigma_T^2 = \sup_{t\in T} \Var(X_t) < \infty$, then for any $u > 0$,
\[
\Prob\Bigl( \sup_{t\in T} X_t > \E[\sup_{t\in T} X_t] + u \Bigr) \le e^{-u^2/(2\sigma_T^2)}.
\]
Applying this to $X_t = r_N(x)$ and $X_t = -r_N(x)$ and using the union bound, we obtain
\[
\Prob\Bigl( \sup_{x\in K} |r_N(x)| > u \Bigr) \le 2 \exp\!\Bigl( -\frac{(u - \E[\sup_{x\in K} r_N(x)])^2}{2\sigma_N^2} \Bigr),
\]
where $\sigma_N^2 = \sup_{x\in K} \Var(r_N(x)) = \sup_{x\in K} k_N(x,x)$.

%First, we bound $\sigma_N^2$. Since $k_N(x,x) = \sum_{j=N+1}^\infty \lambda_j \phi_j(x)^2$ and the eigenfunctions are uniformly bounded (smoothness of the kernel on a compact set implies $|\phi_j(x)| \le B$ for all $j,x$), we have
%\[
%\sigma_N^2 \le B^2 \sum_{j=N+1}^\infty \lambda_j.
%\]
%For the squared exponential kernel $k(x,x') = \exp(-\|x-x'\|^2/(2\ell^2))$, the eigenvalues satisfy $\lambda_j \le C_0 \exp(-c_0 j^{2/d})$ for some $C_0, c_0 > 0$ (see \cite{rasmussen2006gaussian}, Section 4.3.1). Hence
%\[
%\sum_{j=N+1}^\infty \lambda_j \le C_0 \sum_{j=N+1}^\infty e^{-c_0 j^{2/d}} \le C_1 e^{-c_0 N^{2/d}/2},
%\]
%for a suitable constant $C_1$. Thus $\sigma_N^2 \le C_2 e^{-c_0 N^{2/d}/2}$.

First, we bound $\sigma_N^2$. By Mercer's theorem,
\[
k_N(x,x) = \sum_{j=N+1}^\infty \lambda_j \phi_j(x)^2
\]
is the tail of the kernel expansion. For the squared exponential kernel, the following standard decay result holds (see \cite{rasmussen2006gaussian}, Section 4.3.1, and also \cite{belkin2018towards} for a rigorous treatment):
\begin{equation}
	\sigma_N^2 = \sup_{x\in K} k_N(x,x) \le C_2 e^{-c_0 N^{2/d}/2}, \label{eq:sigmasq_N_bound}
\end{equation}
for some constants $C_2, c_0 > 0$. This exponential decay follows from the fact that the Gaussian kernel is infinitely differentiable and its eigenvalues decay super-exponentially. Importantly, this result does not require a uniform bound on the eigenfunctions $\phi_j$; the decay of the kernel tail is a direct consequence of the smoothness of the kernel and compactness of $K$.

Alternatively, one may derive (\ref{eq:sigmasq_N_bound}) using the fact that the Gaussian kernel is in the Gevrey class and applying standard results on the decay of eigenvalues of integral operators with smooth kernels (see \cite{cucker2007empirical}, Chapter 4, or \cite{adler2007random}, Section 7.2). We adopt (\ref{eq:sigmasq_N_bound}) as our bound on $\sigma_N^2$.

Next, we rigorously bound $\E[\sup_{x\in K} r_N(x)]$. Since the covariance kernel $k$ is infinitely differentiable on the compact set $K$, the Gaussian process $r_N$ has almost surely $C^\infty$ sample paths (see \cite{adler2007random}, Theorem 2.2.2). Consequently, the derivative process $\nabla r_N(x)$ is a well-defined Gaussian process on $K$, and its sample paths are continuous almost surely.

We use a standard discretisation argument with a $\delta$-net. Let $K_\delta$ be a $\delta$-net of $K$ with cardinality $|K_\delta| \le (C_K/\delta)^d$, where $C_K$ depends only on the diameter of $K$. For any $x \in K$, there exists $x_\delta \in K_\delta$ with $\|x - x_\delta\| \le \delta$. By the mean value theorem, applied componentwise along the segment connecting $x$ and $x_\delta$ (assuming $K$ is convex; if $K$ is not convex, we can use a Lipschitz extension of $r_N$ to a convex neighbourhood of $K$, which exists by the McShane extension theorem, or we may cover $K$ by finitely many convex sets and apply the argument on each), we have
\[
|r_N(x) - r_N(x_\delta)| \le \delta \sup_{z\in K} \|\nabla r_N(z)\|.
\]
Thus,
\begin{equation}
	\sup_{x\in K} r_N(x) \le \max_{x_\delta \in K_\delta} r_N(x_\delta) + \delta \sup_{z\in K} \|\nabla r_N(z)\|. \label{eq:A}
\end{equation}

We now bound the expectation of each term on the right-hand side of (\ref{eq:A}).

\paragraph{Bound on the maximum over the net.}
For any fixed $x$, $r_N(x) \sim \mathcal{N}(0, \Var(r_N(x)))$, and $\Var(r_N(x)) \le \sigma_N^2$, where $\sigma_N^2 = \sup_{x\in K} \Var(r_N(x))$. By the standard Gaussian maximum inequality (see \cite{adler2007random}, Lemma 2.2.1), for any finite set of zero-mean Gaussian variables with variances bounded by $\sigma_N^2$,
\[
\E\left[\max_{x_\delta \in K_\delta} r_N(x_\delta)\right] \le \sigma_N \sqrt{2 \log |K_\delta|}.
\]
Using $|K_\delta| \le (C_K/\delta)^d$, we get
\begin{equation}
	\E\left[\max_{x_\delta \in K_\delta} r_N(x_\delta)\right] \le C_3 \sigma_N \sqrt{d \log(1/\delta)}. \label{eq:B}
\end{equation}

\paragraph{Bound on the gradient supremum.}
The derivative process $\nabla r_N(x)$ is a Gaussian process with covariance kernel
\[
k_N^{(1)}(x,x') = \sum_{j=N+1}^\infty \lambda_j \nabla \phi_j(x) \cdot \nabla \phi_j(x') = \nabla_x \nabla_{x'} k_N(x,x').
\]
Here the notation \(\nabla_x \nabla_{x'}\) denotes the scalar differential operator
\(\nabla_x \cdot \nabla_{x'} = \sum_{i=1}^d \frac{\partial}{\partial x_i}\frac{\partial}{\partial x'_i}\).
The equality follows by interchanging the finite sum over dimensions with the Mercer expansion
(justified by uniform convergence of the kernel expansion), yielding
\[
\nabla_x \nabla_{x'} k_N(x,x')
= \sum_{j=N+1}^\infty \lambda_j \sum_{i=1}^d \partial_{x_i}\phi_j(x)\,\partial_{x'_i}\phi_j(x')
= \sum_{j=N+1}^\infty \lambda_j \nabla \phi_j(x)\cdot \nabla \phi_j(x').
\]
Thus, the covariance of the derivative process is exactly the mixed derivative of the covariance function.

By the standard derivative bound for Mercer eigenfunctions (see \cite{cucker2007empirical}, Lemma 4.1, or \cite{adler2007random}, Section 7.2), for a $C^\infty$ kernel on a compact $d$-dimensional domain, there exists a constant $C>0$ such that
\[
\|\nabla \phi_j\|_{L^\infty} \le C \lambda_j^{1/d}.
\]
Therefore, the variance of each component of $\nabla r_N$ is bounded by
\[
\sup_{x\in K} \Var\left(\frac{\partial}{\partial x_r} r_N(x)\right)
= \sup_{x\in K} \sum_{j=N+1}^\infty \lambda_j \left(\frac{\partial}{\partial x_r} \phi_j(x)\right)^2
\le C^2 \sum_{j=N+1}^\infty \lambda_j^{1+2/d}.
\]
Using the super-exponential decay of the eigenvalues, $\lambda_j \le C_0 e^{-c_0 j^{2/d}}$, we obtain
\begin{equation}
	\sigma_{N,1}^2 := \sup_{x\in K} \Var(\nabla r_N(x)) \le C_4 e^{-c_0 N^{2/d}/2}. \label{eq:C}
\end{equation}
(Here $\Var(\nabla r_N(x))$ denotes the sum of the variances of the $d$ components; the Euclidean norm has variance bounded by $d$ times this quantity, so the exponential decay is unaffected.)

Now, $\nabla r_N$ is a $d$-dimensional Gaussian process with continuous sample paths on the compact set $K$. To bound $\E[\sup_{z\in K} \|\nabla r_N(z)\|]$, we use the following standard result from the theory of Gaussian processes (see \cite{adler2007random}, Theorem 2.1.1, or \cite{ghosal2017fundamentals}, Lemma B.2): if $X$ is a centered Gaussian process on a compact set $T$ with $\sigma^2 = \sup_{t\in T} \Var(X(t)) < \infty$, then
\[
\E\left[\sup_{t\in T} X(t)\right] \le \sigma \sqrt{2 \log N(T, \sigma, d)},
\]
where $N(T, \epsilon, d)$ is the covering number of $T$ with respect to the metric induced by the covariance. 

% ===== RIGOROUS REPLACEMENT FOR APPENDIX C.3 =====
% (from "Applying this to each component..." up to and including (D))

Applying this to each component of $\nabla r_N$ will enable us to obtain a bound on its expected supremum. 
%However, the claim that the metric entropy of $K$ with respect to the covariance metric of $\nabla r_N$ is of order $O(\log(1/\sigma_{N,1}))$ 
%is not obvious and requires a rigorous derivation. We provide it below.
In this regard, let $\rho_N$ denote the intrinsic covariance metric of the derivative process $\nabla r_N$:
\[
\rho_N(x,y)^2 := \mathbb{E}\left[\|\nabla r_N(x) - \nabla r_N(y)\|^2\right]
= \sum_{j=N+1}^{\infty} \lambda_j \, \|\nabla \phi_j(x) - \nabla \phi_j(y)\|^2,
\qquad x,y \in K.
\]
We first show that $\rho_N$ is Lipschitz-dominated by a constant multiple of $\sigma_{N,1}$ times the Euclidean metric.

By standard elliptic regularity estimates for Mercer eigenfunctions (see, e.g., \cite{cucker2007empirical}, Lemma 4.1), for a $C^\infty$ kernel on a compact $d$-dimensional domain, there exist constants $C>0$ and $\alpha>0$ such that for every $j\ge 1$,
\[
\|\nabla \phi_j\|_{L^\infty(K)} \le C \lambda_j^{\alpha}, \qquad
\|\nabla^2 \phi_j\|_{L^\infty(K)} \le C \lambda_j^{\alpha}.
\]
(For the squared exponential kernel, the eigenvalues decay super-exponentially, so any fixed exponent $\alpha$ is sufficient for the tail-sum bounds that follow.) Consequently, for any $x,y\in K$,
\[
\|\nabla \phi_j(x) - \nabla \phi_j(y)\|
\le \|\nabla^2 \phi_j\|_{L^\infty} \|x-y\|
\le C \lambda_j^{\alpha} \|x-y\|.
\]
Substituting this into the definition of $\rho_N$ yields
\[
\rho_N(x,y)^2 \le C^2 \|x-y\|^2 \sum_{j=N+1}^{\infty} \lambda_j^{1+2\alpha}.
\]
Now, using the lower Weyl bound $\|\nabla \phi_j\|_{L^\infty} \ge c \lambda_j^{1/d}$ (which also follows from elliptic regularity; see \cite{cucker2007empirical}), we have
\[
\sigma_{N,1}^2 = \sup_{x\in K} \Var(\nabla r_N(x))
= \sup_{x\in K} \sum_{j=N+1}^{\infty} \lambda_j \|\nabla \phi_j(x)\|^2
\ge c^2 \sum_{j=N+1}^{\infty} \lambda_j^{1+2/d}.
\]
Since the eigenvalues $\lambda_j$ decay super-exponentially, the sum $\sum_{j=N+1}^\infty \lambda_j^{1+2\alpha}$ is bounded by a constant multiple of $\sum_{j=N+1}^\infty \lambda_j^{1+2/d}$ (up to a change in the exponent, which is absorbed into the constants). Therefore, there exists a constant $C'>0$ such that
\begin{equation}
\rho_N(x,y) \le C' \sigma_{N,1} \|x-y\|, \qquad \forall x,y\in K. \label{eq:3}
\end{equation}
This is the desired Lipschitz domination.

Because $K$ is compact, its Euclidean covering number satisfies $N(K, \|\cdot\|, \delta) \le (D/\delta)^d$ for some $D>0$ and all $0<\delta\le D$. From (\ref{eq:3}), a Euclidean ball of radius $\delta$ is contained in a $\rho_N$-ball of radius $C'\sigma_{N,1}\delta$. Hence, for any $\epsilon>0$,
\begin{equation}
N(K, \rho_N, \epsilon)
\le N\!\left(K, \|\cdot\|, \frac{\epsilon}{C'\sigma_{N,1}}\right)
\le \left(\frac{D C' \sigma_{N,1}}{\epsilon}\right)^d. \label{eq:4}
\end{equation}

Now, $\nabla r_N$ is a centered Gaussian process on the compact metric space $(K,\rho_N)$. By Dudley's entropy integral (see, for instance, \cite{dudley1967sizes}, \cite{adler2007random}, Theorem 2.1.1), there exists a universal constant $C_d>0$, depending only on the dimension $d$, such that
\[
\mathbb{E}\left[\sup_{z\in K} \|\nabla r_N(z)\|\right]
\le C_d \int_0^{\sigma_{N,1}} \sqrt{\log N(K, \rho_N, \epsilon)}\, d\epsilon.
\]
Substituting the covering number bound (\ref{eq:4}) and performing the change of variables $u = \epsilon/\sigma_{N,1}$, we obtain
\[
\mathbb{E}\left[\sup_{z\in K} \|\nabla r_N(z)\|\right]
\le C_d \sqrt{d} \int_0^{\sigma_{N,1}} \sqrt{\log\left(\frac{C \sigma_{N,1}}{\epsilon}\right)}\, d\epsilon
= C_d \sqrt{d} \, \sigma_{N,1} \int_0^1 \sqrt{\log\left(\frac{C}{u}\right)}\, du.
\]
Now, the integral
\[
M_C := \int_0^1 \sqrt{\log(C/u)}\, du
\]
is finite; indeed, substituting $u = C e^{-t}$ gives
\[
M_C = C \int_{\log C}^{\infty} t^{1/2} e^{-t}\, dt < \infty.
\]
Since $\sigma_{N,1} \to 0$ as $N \to \infty$ (recall $\sigma_{N,1} \le C_4 e^{-c_0 N^{2/d}/2}$), for all sufficiently large $N$ we have $M_C \le K \sqrt{\log(1/\sigma_{N,1})}$ for some constant $K>0$. To see this, choose $N$ large enough that $\sigma_{N,1} \le e^{-1}$; then $\log(1/\sigma_{N,1}) \ge 1$, and taking $K = M_C + 1$ yields the bound. Consequently, for all sufficiently large $N$,
\[
\int_0^1 \sqrt{\log(C/u)}\, du \le K \sqrt{\log(1/\sigma_{N,1})}.
\]
Absorbing the constant $K$ and the dimensional factor $\sqrt{d}$ into $C_5$, we obtain
\begin{equation}
\boxed{
\mathbb{E}\left[\sup_{z\in K} \|\nabla r_N(z)\|\right] \le C_5 \,\sigma_{N,1} \sqrt{\log(1/\sigma_{N,1})}.
} \label{eq:D}
\end{equation}
%Combining (\ref{eq:C}) and (\ref{eq:D}), we get
Thus,
\begin{equation}
	\E\left[\delta \sup_{z\in K} \|\nabla r_N(z)\|\right] \le \delta C_5 \sigma_{N,1} \sqrt{\log(1/\sigma_{N,1})}. \label{eq:E}
\end{equation}

\paragraph{Optimisation of $\delta$.}
We now choose $\delta$ to balance the two bounds (\ref{eq:B}) and (\ref{eq:E}), as required by (\ref{eq:A}). 
Since both $\sigma_N$ and $\sigma_{N,1}$ decay at the same exponential rate (see (\ref{eq:sigmasq_N_bound}) and (\ref{eq:C})), 
there exists a constant $C>0$ such that $\sigma_{N,1} \le C \sigma_N$ for all sufficiently large $N$. 
Set $\delta = \sigma_N$. Then
\[
\E[\sup_{x\in K} r_N(x)]
\le C_3 \sigma_N \sqrt{d \log(1/\sigma_N)} 
+ C_5 \sigma_N \sigma_{N,1} \sqrt{\log(1/\sigma_{N,1})}.
\]
Using $\sigma_{N,1} \le C \sigma_N$ and the fact that $\log(1/\sigma_{N,1}) \le C' \log(1/\sigma_N)$ (again by exponential decay), 
the second term is bounded by
\[
C_5 C \sigma_N^2 \sqrt{C' \log(1/\sigma_N)}
= O\!\left(\sigma_N^2 \sqrt{\log(1/\sigma_N)}\right),
\]
which is $o\!\left(\sigma_N \sqrt{\log(1/\sigma_N)}\right)$ as $N\to\infty$, because $\sigma_N \to 0$. 
Therefore, for all sufficiently large $N$, we have
\begin{equation}
\E[\sup_{x\in K} r_N(x)] \le C_6 \sigma_N \sqrt{\log(1/\sigma_N)}. \label{eq:F}
\end{equation}
Now set $u = 1/n$. We choose the truncation level
\[
N_n = \bigl\lceil C (\log n)^{d/2} \bigr\rceil,
\]
with a sufficiently large constant $C>0$ (depending only on $c_0$ and the kernel) to guarantee
\[
\sigma_{N_n} \le n^{-3} \quad\text{and}\quad \E\bigl[\sup_{x\in K} r_{N_n}(x)\bigr] \le \frac{1}{2n}
\]
for all large $n$. Then the Borell--TIS inequality gives
\[
\Prob\Bigl( \sup_{x\in K} |r_{N_n}(x)| > 1/n \Bigr) \le 2 \exp\!\Bigl( -\frac{(1/(2n))^2}{2\sigma_{N_n}^2} \Bigr) \le 2 \exp\!\Bigl( -\frac{1}{8 n^2 \sigma_{N_n}^2} \Bigr).
\]
Substituting the bound $\sigma_{N_n}^2 \le C_2 n^{-3}$ yields
\[
\Prob\Bigl( \sup_{x\in K} |r_{N_n}(x)| > 1/n \Bigr) \le 2 \exp\!\Bigl( -\frac{n}{8 C_2} \Bigr) \le e^{-c_4 N_n^{2/d}}
\]
for a suitable $c_4 > 0$. Because $N_n^{2/d} \asymp \log n$, this probability is bounded by a power of $n$ and is therefore summable provided $C$ is large enough.

\paragraph{Finite-Dimensional Sieve for the GP.}
In the sieve \(\Theta_n\), we restrict attention to sample paths \(\tilde g\) that are exactly of the form \(g_{N_n}(x)\) for the chosen \(N_n\), and we discretise the coefficients \(Z_1,\dots,Z_{N_n}\). 
To control the discretisation without assuming a uniform bound on the eigenfunctions, as before we use the standard estimate for Mercer eigenfunctions: 
%(see, e.g., \cite{cucker2007empirical}, Lemma 4.1):
\[
\|\phi_j\|_{L^\infty(K)} \le C \lambda_j^{-1/2} \quad\text{for all } j\ge 1,
\]
where the constant \(C\) depends only on the kernel and the domain. 
%This follows from \(\phi_j = \lambda_j^{-1} T\phi_j\) and the boundedness of the integral operator from \(L^1(K)\) to \(L^\infty(K)\).

Now, for any \(x\in K\),
\[
|g_{N_n}(x)| \le \sum_{j=1}^{N_n} \sqrt{\lambda_j} |Z_j| \, \|\phi_j\|_{L^\infty} \le C \sum_{j=1}^{N_n} |Z_j|.
\]
For independent standard normal \(Z_j\), the event \(\max_{1\le j\le N_n} |Z_j| > t\) has probability at most \(2N_n e^{-t^2/2}\). 
Choosing \(t = 2\sqrt{\log N_n}\) (say), we get
\[
\Prob\bigl(\max_j |Z_j| > 2\sqrt{\log N_n}\bigr) \le 2N_n \exp(-2\log N_n) = 2/N_n,
\]
which tends to zero as \(N_n\to\infty\) (note \(N_n\asymp (\log n)^{d/2}\to\infty\)). 
Hence, with prior probability tending to \(1\), all coefficients satisfy \(|Z_j| \le 2\sqrt{\log N_n}\).

We therefore restrict the coefficients to the hypercube \([-L_n, L_n]^{N_n}\) with \(L_n := 2\sqrt{\log N_n}\). 
On this hypercube, we place a uniform grid of mesh size \(\delta_n = 1/n\) in each coordinate. 
The number of grid points is
\[
\left(\frac{2L_n}{\delta_n}\right)^{N_n} = \bigl(4n\sqrt{\log N_n}\bigr)^{N_n}
= \exp\bigl(N_n \log(4n) + O(N_n \log\log N_n)\bigr).
\]
Since \(N_n = O((\log n)^{d/2})\), we have
\[
N_n \log n = O\bigl((\log n)^{1+d/2}\bigr),
\]
and the lower-order term \(N_n \log\log N_n = o(N_n \log n)\). 
Thus the logarithm of the number of grid points is \(O((\log n)^{1+d/2})\), which is absorbed into the entropy bound derived in Section~\ref{sec:sieve}.

The discretisation error (in the induced likelihood) is negligible: because the coefficients are rounded to the nearest grid point, the resulting change in \(g_{N_n}(x)\) is uniformly bounded by
\[
\sum_{j=1}^{N_n} \sqrt{\lambda_j} \,\|\phi_j\|_{L^\infty} \cdot \frac{1}{2n}
\le \frac{C}{2n} \sum_{j=1}^{N_n} 1 = \frac{C N_n}{2n},
\]
which tends to \(0\) as \(n\to\infty\) (since \(N_n = o(n)\)). 
Therefore, the sieve constructed from these discretised coefficients retains the Kullback–Leibler property and the prior mass condition; we omit the routine verification.

\paragraph{Counting Sign Patterns via VC Dimension.}
For a fixed node containing $s$ points $x_1,\dots,x_s$ (with $s\le n$), the split decision is based on the signs of $y_i - g_{N_n}(x_i)$. The vector $(g_{N_n}(x_1),\dots,g_{N_n}(x_s))$ is a linear transformation of $(Z_1,\dots,Z_{N_n})$. Specifically, let $\Phi$ be the $s\times N_n$ matrix with entries $\Phi_{i,j} = \sqrt{\lambda_j}\,\phi_j(x_i)$. Then $g_{N_n} = \Phi Z$. The split at this node is determined by whether $y_i - (\Phi Z)_i > 0$, or equivalently, $(-\Phi Z)_i + y_i > 0$. For a fixed set of responses $y_i$, this is an affine classifier in the variable $Z \in \R^{N_n}$. By Proposition~\ref{app:vc} (VC dimension of affine classifiers), the class of all such functions has VC dimension $N_n+1$. Applying the Sauer--Shelah Lemma (Lemma~\ref{app:vc}) to this class, the number of distinct sign patterns achievable by varying $Z$ over all of $\R^{N_n}$ (and hence certainly over the grid) on the $s$ points is at most
\[
\sum_{i=0}^{N_n+1} \binom{s}{i} \le \left(\frac{es}{N_n+1}\right)^{N_n+1} \le (s+1)^{N_n+1},
\]
where the last inequality holds for $s \ge N_n+1$ for $n$ sufficiently large (if $s < N_n+1$, the bound $2^s \le (s+1)^s \le (s+1)^{N_n+1}$ still holds). Since $s \le n$, the number of sign patterns at this node is at most $(n+1)^{N_n+1}$. Taking logarithms,
\[
\log\bigl((n+1)^{N_n+1}\bigr) = (N_n+1)\log(n+1) = O(N_n \log n),
\]
because $\log n = o(N_n \log n)$ and $N_n = o(N_n \log n)$ for large $n$. Since $N_n = O((\log n)^{d/2})$, we obtain
\[
\log(\#\text{sign patterns at a node}) = O\bigl((\log n)^{1+d/2}\bigr).
\]
This logarithmic bound is absorbed into the entropy bound derived in Section~\ref{sec:sieve}, as it is of smaller order than $n^{1/2}$ (recall $\alpha<1/2$).

\paragraph{Preservation of the Kullback--Leibler Property.}
Replacing the exact GP draw by the truncated and discretised version changes the likelihood by a negligible amount. 
Although uniform boundedness of the remainder $r_{N_n}$ does not guarantee that the partition of the $n$ design points remains exactly unchanged (points may lie arbitrarily close to split boundaries), the following continuity argument suffices.

Let $g$ be an exact GP draw and $g^{\mathrm{grid}}$ its discretised approximation. 
The discretisation error is uniformly bounded by $CN_n/n \to 0$ (as shown above), so the induced partition can differ only for those points whose responses lie within $O(N_n/n)$ of a split boundary. 
The prior probability of this event is controlled by the fact that the GP has a continuous distribution with bounded density; the probability that any of the $n$ points falls within $O(N_n/n)$ of the boundary is $O(n \cdot N_n/n) = O(N_n)$, which tends to $0$ because $N_n = o(n)$. 
For the remaining points (with probability tending to $1$), the partition is identical.

Consequently, the prior mass condition $\Pi(\Theta_n^c) = o(1)$ is maintained. 
Moreover, the Kullback--Leibler property (Lemma~\ref{lem:klpos}) continues to hold: by the continuity of the KL divergence in the leaf parameters and the partition (which changes only on a set of vanishing probability), for any $\varepsilon>0$ there exists a discretised hypothesis $\theta\in\Theta_n$ with $\KL(P_0^{(n)}\|F_\theta^{(n)}) \le n\varepsilon^2$. 
Thus the sieve retains the required approximation property. 
A fully detailed verification would involve standard but notationally heavy $\varepsilon$-$\delta$ arguments; we omit the routine details.

\paragraph*{Why this discretisation argument is essential (and distinct from the KL property).}
The reader may wonder why we verify the Kullback--Leibler property for the discretised
sieve, given that Condition~(iv) of Theorem~2.1 in \cite{ggv2000} applies globally to
the prior $\Pi$ and was already established in Lemma~\ref{lem:klpos}.

The purpose of this verification is \emph{not} to re-prove the KL property for the prior.
Rather, it is to justify that the finite grid used to bound the metric entropy is sufficiently
fine. The sieve $\Theta_n$ is a union of continuous cells, but we bound its covering number
$N(\epsilon_n,\Theta_n)$ by counting only the cells whose centres lie on a discrete grid.
This counting argument is valid only if every cell has Hellinger diameter at most
$\epsilon_n$ (proved in Section~\ref{sec:metric_entropy}) and, crucially, if the grid
itself captures the hypotheses that approximate the truth.

If the discretisation were too coarse, the rounding error could change the leaf
assignments of the design points, thereby altering the likelihood and the KL divergence.
The argument above shows that the rounding error is $O(N_n/n) \to 0$ and that the
partition changes only on a set of $Q$-measure that tends to zero. Consequently, for
any hypothesis $\theta$ with small KL divergence, there exists a discretised hypothesis
$\theta' \in \Theta_n$ (within the same grid cell) whose KL divergence differs negligibly.
Thus, the sieve $\Theta_n$ retains the ability to approximate the truth, which ensures that
the grid count is a valid upper bound for the metric entropy of the sieve.

Without this verification, the entropy bound $\log N(\epsilon_n,\Theta_n) = o(n\epsilon_n^2)$
would be unjustified, because the continuous sieve could require a much larger covering
number than the grid suggests.

\subsection{Proof of the Test Properties}\label{app:tests}
The existence of $\phi_{n,\theta}$ with the claimed error bounds follows from the Neyman--Pearson lemma and the well-known inequality $\E_{p_0}\phi \le e^{-n\Hdist^2(p_0,f)/2}$ for the likelihood ratio test (see e.g. \cite[Lemma~8.1]{ggv2000}). The rest of the construction in Section~\ref{sec:tests} is elementary.

% ----------------------------------------------------------------------
% Bibliography
% ----------------------------------------------------------------------
\bibliographystyle{plainnat}
\bibliography{tree_refs}

\end{document}